%% file: main.tex
\documentclass[12pt, english]{article}
\pdfoutput=1

\usepackage{preamble}

\usepackage{lastpage}
\DeclareUnicodeCharacter{1D464}{$w$}
\usepackage[section]{placeins}
\usepackage{multirow}

\newcommand{\ignore}[1]{}
\newcommand{\eps}{\mathrm{eps}}

\makeatother

\begin{document}
    
    \title{Wasserstein F-tests for Fr\'echet regression on Bures-Wasserstein manifolds}
    \author{Haoshu Xu\thanks{Graduate Group in Applied Mathematics and Computational Science, University of Pennsylvania, Philadelphia, PA 19104, USA; email: \texttt{haoshuxu@sas.upenn.edu}}
    \and 
    Hongzhe Li\thanks{Department of Biostatistics, Epidemiology and Informatics, University of Pennsylvania, Philadelphia, PA 19104, USA; email: \texttt{hongzhe@pennmedicine.upenn.edu}.}
    }

    \maketitle
    
    \abstract{  
        This paper addresses the problem of regression analysis where the outcome is a random covariance matrix and the covariates are Euclidean. The study is situated within the framework of Fr\'echet regression on the Bures-Wasserstein manifold, which is pertinent to fields like single-cell genomics and neuroscience, where covariance matrices are observed across many samples. 
        Fr\'echet regression on the Bures-Wasserstein manifold is formulated as estimating the conditional Fr\'echet mean given covariates $x$.  A non-asymptotic $\sqrt{n}$-rate of convergence (up to $\log n$ factors) is obtained for our estimator, uniformly for $\norm{x} \lesssim \sqrt{\log n}$, which is crucial for  deriving the asymptotic null distribution and assessing the power of our proposed statistical test for the null hypothesis of no association.   Additionally, a central limit theorem for the point estimate is derived, offering insights into testing covariate effects. The null distribution of the test statistic is shown to converge to a weighted sum of independent chi-square distributions. The test’s power is also demonstrated against a sequence of contiguous alternatives. Simulation results validate the accuracy of the asymptotic distributions. 
        Finally, the proposed methods are applied to a single-cell gene expression dataset, illustrating changes in gene co-expression networks with age.
    }
    \medskip
    
    \noindent\textbf{Keywords:} Fr\'echet regression; Functional calculus;  Hypothesis testing;  Optimal transport;   Wasserstein distance.

    \clearpage
    \tableofcontents
    \medskip

    \section{Introduction}
        \label{sec: intro}
        \input{intro.tex}

    \section{Preliminaries}
        \label{sec: prelim}
        \input{prelim.tex}

    \section{Problem Formulation}
        \label{sec: formulation}
        \input{formulation.tex}

    \section{Statistical Inference for Fr\'echet Regression on Bures-Wasserstein manifold}
        \label{sec: main}
        \input{main_results.tex}

    \section{Algorithm and Numerical Experiments}
        \label{sec: numerics}
        \input{numeric.tex}

    \section{Application to Single-cell 
    Gene Co-expression Networks}
        \label{sec: single_cell}
        \input{single_cell.tex}

    \section{Discussion}
        \label{sec: discuss}
        \input{discuss}

    \vskip 0.2in
    \bibliographystyle{apalike}
    \bibliography{bib}

\newpage

\appendix

    \section{Background on optimal transport and functional calculus}
        \label{sec: appx_bg}
        \input{appx_bg}

    \section{Technical lemmas}
    In this section, we collect technical lemmas and relevant notations for the proof. First, properties of differentials of $T^S_Q$ are collected in Appendix \ref{subsec: Techlem_dkT}. Then various concentration results are given in Appendix \ref{subsec: Techlem_conc}. Finally, properties of $F(\cdot,\cdot)$ and $Q^*(\cdot)$ are summarized in Appendix \ref{subsec: Techlem_QF}.    

        \subsection{Differentials of optimal transport maps}
        \label{subsec: Techlem_dkT}
        \input{tech_lem_dT.tex}

        \subsection{Concentration inequalities and uniform convergence}
        \label{subsec: Techlem_conc}
        \input{tech_lem_conc.tex}

        \subsection{Properties of $F$ and $Q^*$}
        \label{subsec: Techlem_QF}
        \input{tech_lem_other.tex}

    \section{Proof of Theorem \ref{thm: unfm_conv} and Corollary \ref{cor: W_moment}}
    \label{sec: prf_unfm_conv}
    \input{prf_unfm_conv.tex}

    \section{Proof of Theorem \ref{thm: estimation}}
    \label{sec: proof_est}
    \input{proof_of_est.tex}

    \section{Proof of Corollary \ref{cor: clt_ind}}
    \label{sec: prf_clt_cor}
    \input{clt_cor.tex}

    \section{Proof of Theorem \ref{thm: test_dist}}
    \label{sec: proof_null}
    \input{proof_test_dist.tex}

    \section{Proof of Proposition \ref{prop: size}}
    \label{sec: prf_size}
    \input{prf_size.tex}

    \section{Proof of Theorem \ref{thm: power}}
    \label{sec: proof_power}
    \input{proof_power.tex}

    \section{Additional simulations}
    \label{sec: more_sim}
    \input{more_sim}

\end{document}

%% file: intro.tex
In modern data analysis, positive definite matrices frequently arise across various fields, including medical imaging \citep{dryden09, fillard07}, neuroscience \citep{friston11_brain,kong20,kong21}, signal processing \citep{arnaudon13}, and computer vision \citep{caseiro12}. For example, in large-scale single-cell RNA-seq data, individual-specific covariance matrices can be estimated and interpreted as co-expression networks among a set of genes. In neuroimaging, covariance matrices (or correlation matrices after standardization) of multiple brain regions are used to summarize functional connectivity, providing insights into brain network organization. A central challenge in these applications is performing regression analysis where the covariance matrix serves as the outcome variable in relation to a set of covariates.

Several regression models have been proposed for covariance matrix outcomes. \citet{Chiu1996} developed a method that models the elements of the logarithm of the covariance matrix as a linear function of the covariates, but this approach requires estimating a large number of parameters. \citet{HoffNiu} proposed a regression model where the covariance matrix is expressed as a quadratic function of the explanatory variables. \citet{Zou2017} linked the matrix outcome to a linear combination of similarity matrices derived from the covariates and examined the asymptotic properties of different estimators under this framework. \citet{Zhao_CAP} introduced the Covariate Assisted Principal (CAP) regression model for multiple covariance matrix outcomes, focusing on identifying linear projections of the covariance matrices associated with the covariates. \citet{kong20} and \citet{kong21} developed linear and nonparametric regression methods for high-dimensional, low-rank matrices using nuclear norm regularization. However, these approaches often either impose specific structural assumptions on the covariance matrices or require the estimation of a large number of parameters.

In this paper, we focus on general regression analysis where the responses are symmetric positive-definite (SPD) matrices and the predictors are Euclidean within the framework of the Fr\'echet regression model\citep{muller16}.
This model defines a global regression function that links response data in an arbitrary metric space with Euclidean predictors. A key consideration in this context is the choice of an appropriate metric for SPD matrices. Various metrics are available, including the Log-Euclidean metric \citep{logEuc}, the Cholesky metric \citep{dryden09} and the Log-Cholesky metric \citep{lin19}. Each of these metrics has distinct properties and implications for regression analysis. However, the Bures-Wasserstein metric $W$, introduced by \citet{bures69}, stands out for its desirable mathematical and geometric properties. It is defined for any pair of SPD matrices $A,B \in \cS_d^{++}$ as 
\begin{equation}
    W(A,B) = \sbr{\tr{A} + \tr{B} - 2 \tr \rbr{A^{1/2}BA^{1/2}}^{1/2}}^{1/2}.
    \label{eqn: W_gaussian}
\end{equation}
This distance possesses desirable properties under both scaling and rotation: for positive scalar $c>0$ and orthogonal matrix $O \in \cO_d$, it satisfies
\begin{align*}
    W(cA,cB) &= c^{1/2} W(A,B)\\
    W(OAO^\top, OBO^\top) &= W(A,B)
\end{align*}
Known as the Bures distance in quantum information theory \citep{bures69}, this metric is equivalent to the Wasserstein distance between two centered Gaussian distributions with specified covariance matrices. The Wasserstein distance, a fundamental concept in optimal transport (OT) theory \citep{villani03, villani09}, which blends optimization, analysis, and geometry, has become highly valuable in various statistical and machine learning applications \citep{imbens06, deb23_rank, wgan17, ot17_domain_adap, hallin21}. This equivalence provides the Bures-Wasserstein distance with a clear and meaningful probabilistic interpretation as the minimal "effort" required to transport one Gaussian distribution, defined by its covariance matrix, to another.
In single-cell genomics, most measurement technologies are destructive, meaning that the same cell cannot be observed multiple times or profiled over time. Consequently, measurements at each time point are often modeled as distributions, making optimal transport techniques well-suited for analyzing the associated dynamics \citep{schie19_ot,bunne22_prox,bunne23_aligned_sb,bunne23_sb_gau,bunne23_neural_ot}. Given its connection to optimal transport theory, the Bures-Wasserstein metric is a natural choice for studying gene expression covariance matrices. When equipped with this metric, the space $\cS^{++}_d$ becomes a Riemannian manifold, known as the Bures-Wasserstein manifold \citep{bhatia19_BW}.

In contrast, the Log-Euclidean metric, the Cholesky metric  and the Log-Cholesky metric lack clear practical interpretations and do not naturally align with the structure of distributions or physical models. 
In addition, the Log-Euclidean metric $d_{\mathrm{LE}}$ is insensitive to scaling; specifically, $d_{\mathrm{LE}} (cA,cB) = d_{\mathrm{LE}}(A,B)$ for any $c>0$, making it less suitable for applications where the magnitude of matrices is a critical factor. 
The Log-Cholesky metric $d_{\mathrm{LC}}$ treats the diagonal and off-diagonal components of the Cholesky factor $L_A$ differently, resulting in $d_{\mathrm{LC}}(cA,cB)$ not being proportional to $d_{\mathrm{LC}}(A,B)$ for $c>0$, which further limits its applicability.
Additionally, both the Cholesky metric $d_{\mathrm{C}}$ and the Log-Cholesky metric $d_{\mathrm{LC}}$ lack rotational invariance; specifically, for $O \in \cO_d$,
\begin{align*}
     d_{\mathrm{C}}(OAO^\top, OBO^\top)&\neq d_{\mathrm{C}}(A,B),\\
    d_{\mathrm{LC}}(OAO^\top, OBO^\top) &\neq d_{\mathrm{LC}}(A,B).
\end{align*}
These limitations reduce the practical utility of these metrics for applications involving SPD matrices where scaling and rotation are essential, such as in aligning data from different sources \citep{rong}.

\ignore{
    In this paper, we focus on regression analysis where the responses are valued in $(\cS^{++}_d,W)$ and the predictors are Euclidean. 
    the affine-invariant metric \citep{moak05, Fletcher07},

    A fundamental consideration in the study of regression models for covariance matrices is the choice of metric. 
    Various metrics on the space $\cS^{++}_d$ of $d \times d$ positive definite matrices have been explored, including the log-Euclidean metric [lin19], Cholesky metric [aoas09] and log-Cholesky metric [lin19].

    the trace metric \citep{lang99}, affine-invariant metric \citep{moak05,Fletcher07} and log-Cholesky metric \citep{lin19}. Among these, the Bures-Wasserstein metric $W$, originally introduced by \citet{bures69} , is defined for any pair of matrices $Q,S \in \cS^{++}_d$ as

    This metric has garnered interest in quantum information theory \citep{bures69}, where it is referred to as the Bures distance, and it coincides with the Wasserstein distance between two centered Gaussian distributions with corresponding covariance matrices. The Wasserstein distance, a special case of optimal transport (OT) theory—which lies at the intersection of optimization, analysis, and geometry—measures the minimal cost required to transport mass from one probability distribution to another \citep{villani03,villani09}. This metric has proven valuable in various statistical and machine learning tasks \citep{imbens06,deb23_rank,wgan17,ot17_domain_adap,hallin21}.
}

Having selected the Bures-Wasserstein distance as the metric for SPD matrices in the Fr\'echet regression model, we consider independent and identically distributed (i.i.d.) pairs of predictor and response variables $(X_1,Q_1),\ldots, (X_n,Q_n) \in \RR^p \times \cS^{++}_d$. The objective is to conduct inference, particularly to test the effect of the covariate $X$ on the response variable~$Q$. 
Previous research on the Fr\'echet regression model has primarily focused on consistency in the asymptotic regime \citep{muller16,muller22_unfm}, with inference being considered only in the specific case of one-dimensional (1D) density curves as response variables under the Wasserstein metric \citep{petersen21_WFtest}. Notably, for any pair of 1D distributions $\mu,\nu \in \cP(\RR)$ with distribution functions $F_{\mu}, F_{\nu}$, the 2-Wasserstein distance between them is given by $\sbr{\int_0^1\abs{F^{-1}_{\mu}(t)-F^{-1}_{\nu}(t)}^2 d t}^{1/2}$. Therefore, the 1D Wasserstein space is unique in that it  has zero sectional curvature \citep{ambrosio}, and can be \textit{isometrically} embedded into a Hilbert space. This flat geometry allows for closed-form expressions in 1D Fr\'echet inference and provides well-understood solutions to various problems
\citep{zemel_vic_16,muller23_W_regression,bigot17_pca}. 
However, most metric spaces of interest are nonlinear and exhibit nonzero curvature. For example, the Wasserstein space of distributions in $d$ dimension is positively curved for $d>1$ \citep{ambrosio}, with the Bures-Wasserstein manifold being a special case. Consequently, \textit{isometric} embeddings of such curved metric spaces into a Hilbert space cannot be assumed, making it difficult to derive distributional results for Fr\'echet regression in these spaces. This complexity poses additional challenges for statistical inference, particularly in providing guarantees on significance levels and power.

\subsection{Main contribution}
We focus on Fr\'echet regression on the Bures-Wasserstein manifold, and our main contributions are threefold.

First, we establish a non-asymptotic $\sqrt{n}$-rate of convergence (up to $\log n$ factors) for the regression estimate $\hQ(x)$ uniformly over the region $\norm{x} \lesssim \sqrt{\log n}$. To the best of our knowledge, this is the first non-asymptotic uniform convergence result for Fr\'echet regression over a potentially diverging region. Beyond the standard assumptions of light tails for $X$ and $Q$, we only require well-separation and a local curvature lower bound, which are mild conditions and are verified in a simple case.
These results are crucial for later deriving the asymptotic null distribution and power of our proposed test for the association between covariance matrices and covariates.

Second, we derive a central limit theorem for the point estimate $\hQ(x)$, leading to the construction of a pointwise confidence region. The covariance operator of the limiting Gaussian distribution is shown to have contributions from two sources: the intrinsic variability in $Q$ and the imperfect information regarding $X$.

Third, we carefully construct a test statistic with a tractable asymptotic null distribution, which is represented as a weighted sum of $\chi^2_p$ distributions.
The weights are determined by the covariance of the tangent vector, which can be interpreted as a generalization of classical noise variance. The proposed test is also shown to be powerful against a sequence of contiguous alternatives. To the best of our knowledge, this is the first test developed for Fr\'echet regression on a space with nonzero sectional curvature.

Finally, we validate our theoretical results through numerical simulations and real applications. 

\subsection{Related works}

\paragraph{Statistical OT} In addition to advances in computational optimal transport (OT) \citep{cuturi_13,19_computational_ot,altschuler_near-linear_2017}, there has been a growing interest in the statistical aspects of OT, where the stability of estimated densities \citep{weed19_cost_est}, Wasserstein distances \citep{aop19_op_cst_clt,weed19_entropy_cost,weed23_entropy_cost_clt,altschuler_asymptotics_2022}, transport maps \citep{rig21_map_optimal,weed22_map_ent,weed22_map_plug,weed22_entrp_map_clt,weed23_map_disc_est,weed23_map_clt}, and Fr\'echet means \citep{agueh11_Wbary,pass17_barycenter,legouic17_existence_W_bary, rig22_fast_alex,alt21_averaging} are investigated in the presence of sampling noise.

For the Wasserstein Fr\'echet mean, \citet{rig22_fast_alex} established a parametric rate of convergence for the empirical Fr\'echet mean in the more general Alexandrov spaces, which include the 2-Wasserstein space as a special case, by introducing a bi-extendibility condition that translates into regularity conditions on the Kantorovich potentials. This condition was later relaxed by \citet{che20_gd} when establishing the linear rate of convergence for gradient descent algorithms over the Wasserstein space.

Taking this further, \citet{zemel_vic_16} and \citet{agueh17_clt} established central limit theorems for the empirical Fr\'echet mean of 1D distributions. Subsequently, a central limit theorem for multivariate Gaussians was established by exploiting the first-order differentiability of optimal transport maps \citep{aap21}.

\paragraph{Fr\'echet mean} The Fr\'echet mean is a natural generalization of the concept of an average to abstract metric spaces. For its properties in general curved metric spaces, see \citet{12_Bary_alex,yokota16,rig22_fast_alex} and references therein. The existence and uniqueness of the Fr\'echet mean in the context of Riemannian manifolds and Wasserstein spaces are established in \citet{agueh11_Wbary,pass17_barycenter,legouic17_existence_W_bary}. The asymptotic properties of the empirical Fr\'echet mean on a Riemannian manifold are addressed in \citet{bha03, bha05,rig22_fast_alex}.

\paragraph{Fr\'echet regression} 
The Fr\'echet regression model can be viewed as an extension of the Fr\'echet mean by incorporating a weighted average and was first introduced by \citet{muller16}. \citet{petersen21_WFtest} proposed an F-test specifically for the case of 1D density responses. In their approach, uniform convergence was not necessary for inference because the problem could be embedded into a Hilbert space, allowing for explicit expressions. However, in our case, demonstrating uniform convergence is crucial to ensure that the contributions from the remainder term in the Taylor expansion are negligible.

It is worth noting that while the uniform convergence of Fr\'echet regression is also addressed in \citet{muller16} and \citet{muller22_unfm}, their results are asymptotic and only uniform over a fixed compact set of $x$. In contrast, our result is non-asymptotic, with uniformity achieved within a compact set that expands in diameter to accommodate the potential unboundedness of $\supp X$.

\paragraph{Regression models on manifolds} It is also important to note that the Fr\'echet regression model is defined purely in terms of distance, making it applicable to any abstract metric space. Meanwhile, a separate line of research focuses on regression on manifolds \citep{yuan_local_2012,jrssb17_RSS,lin23,muller23_W_regression}. These regression models are grounded in the concept of tangent spaces from differential geometry. Because tangent spaces are linear, they allow regression on manifolds to essentially reduce to classical linear regression within these tangent spaces.

\subsection{Organization}
The remainder of the paper is organized as follows. In Section \ref{sec: prelim}, we provide the necessary background on optimal transport. Section \ref{subsec: model formulation} formulates the Fr\'echet regression model, followed by the assumptions outlined in Section \ref{subsec: assumption}. The main results are presented in Section \ref{sec: main}, where we demonstrate the uniform convergence of our estimator in Section \ref{subsec: est} and propose our test along with theoretical guarantees in Section \ref{subsec: test}. Finally, Section \ref{sec: numerics} presents a Riemannian gradient descent algorithm and numerical simulations to validate our theory. The proofs of our theorems and technical lemmas are provided in the Appendix. We conclude with a discussion in Section \ref{sec: discuss}.

\subsection{Notation}
We denote by $\ZZ$ and $\RR_+$ the set of integers and the set of non-negative real numbers.
For any $a,b \in  \RR$, we write $a \vee b = \max \rbr{a,b}$ and $a \wedge b = \min \rbr{a,b}$.
For any $x>0$, we write $\log^{+}(x):=\log (x) \vee 1$.
For any integer $K \geq 1$, $[K]=\cbr{1,\ldots,K}$.
Given $z_i \in \RR$ for $i \in [n]$, the set $\cbr{z_1,\ldots,z_n}$ is denoted by $z_1^n$.
The Euclidean norm on $\RR^p$ is denoted $\norm{\cdot}$. For any $x \in \RR^p$ and $L>0$, let  $B_x(L)=B(x,L)=\cbr{y \in \RR^d :\norm{x-y} \leq L}$.
We denote by $\cS_d, \cS_d^+, \cS_d^{++}$ the set of all $d \times d$ symmetric, positive semi-definite and positive definite matrices. For any real $a<b$, we define $\cS_d(a,b):= \cbr{A \in \cS_d: a I_d \prec A \prec b I_d}$. The subscript $d$ is omitted when it's clear from context. 
Given any $A \in \cS_d$, denote the largest and smallest eigenvalue of $A$ by $\eigmax{A}$ and $\eigmin{A}$.
The Frobenius norm and operator norm of a matrix $A$ is denoted by $\Fnorm{A}$ and $\Onorm{A}$.
Given any matrix $A \in \RR^{m,n}$, let $\vecc A \in \RR^{mn}$ denote the vectorized $A $ obtained by stacking columns of $A$.
Given a random variable $X$ and $\alpha>0$, the $\psi_{\alpha}$-"norm" of $X$, denoted by $\Ynorm{\alpha}{X}$, is defined in Appendix \ref{subsec: Techlem_conc}.
The support of a probability distribution is denoted by $\supp(\cdot)$.

Given normed spaces $Y$ and $Z$, let $L(Y;Z)$ denote the space of all bounded linear operator from $Y$ to $Z$. Given a function $\phi:Y \to Z$ and integer $k \geq 0$, the $k$-th differential $d^k\phi$, its operator norm $\norm{d^k\phi}$ and symmetric norm $\opnorm{d^k \phi}$ are defined in Section \ref{sec: prelim}.

Finally, the quantities $C$ and $c$ will refer to constants whose value may change from line to line. 
Given sequences $(a_n)^{\infty}_{n=1}$ and $(b_n)^{\infty}_{n=1}$, we write $a_n \lesssim b_n$ if there exists $C>0$ such that $a_n \leq C b_n$, and we also write $a_n \asymp b_n$ if $b_n \lesssim a_n \lesssim b_n$. 
The constant $C$ is always permitted to depend on $d$ and other problem parameters when they are clear from context.

%% file: prelim.tex

We provide a concise overview of fundamental concepts in optimal transport, along with associated differential properties, specifically focusing on the case of centered Gaussian distributions.

Given a Polish space $(E,d)$, let $\cP_2(E)$ denote the collection of all (Borel) probability measures $\mu$ on $E$ such that $\mathbb{E}_{X \sim \mu}d(X,y)^2<\infty$ for some $y \in E$. One can show that the definition of $\cP_2(E)$ is independent of the choice of $y$. We specialize to the case when $E = \RR^d$ with Euclidean distance. For any pair of measures $\mu, \nu \in \cP_2(\RR^d)$, let $\Pi(\mu, \nu)$ be the set of couplings of between $\mu$ and $\nu$, that is, the collection of probability measures $\pi$ on $\RR^d \times \RR^d$ such that if $(X,Y)\sim \pi$, then $X \sim \mu$ and $Y \sim \nu$. The 2-Wasserstein distance between $\mu$ and $\nu$ is defined as 
\begin{equation}
    W(\mu, \nu):= \sbr{\inf _{\pi \in \Pi(\mu, \nu)} \EE_{(X, Y) \sim \pi} \norm{X-Y}^2}^{1/2}
    \label{def: wstn}
\end{equation}
Here, with a slight abuse of notation, we use $W$ to denote both the Wasserstein distance between two distributions and the Bures-Wasserstein distance between two PSD matrices. This notation is justified by the fact that both distances coincide and have a closed-form expression (\ref{eqn: W_gaussian}) when we identify centered Gaussian distributions with their covariance matrices.

Let $\cP_{2, \mathrm{ac}}(\RR^d)$ denote the subset of measures in $\cP_2(\RR^d)$ that are absolutely continuous with respect to the Lebesgue measure.
Given $\mu_0, \mu_1 \in \cP_{2, \mathrm{ac}}(\RR^d)$, Brenier's theorem guarantees the existence of  a unique optimal coupling $\pi^{\star} \in \Pi(\mu_0, \mu_1)$ that achieves the minimum in (\ref{def: wstn}) and that it is induced by the optimal transport map $T_{\mu_0}^{\mu_1}: \RR^d \to \RR^d$ in the sense that $T_{\mu_0}^{\mu_1}(X) \sim \mu_1$ whenever $X \sim \mu_0$. Specifically, when $\mu_0,\mu_1$ are centered Gaussian distributions with covariance matrices $Q, S$, the optimal transport map is given by the linear map
\begin{equation}
    T_{Q}^{S} = S^{ 1/2}\rbr{S^{ 1/2} Q S^{ 1/2}}^{-1/2} S^{1/2} = Q^{ -1/2}\rbr{Q^{ 1/2} S Q^{ 1/2}}^{1/2} Q^{-1/2}
\end{equation}
For completeness, additional background on the geometry of optimal transport is provided in Appendix \ref{subsec: appx_bg_ot}.

\citet{aap21} showed that for any fixed $S \in \cS^{++}_d$, $T_Q^S$ is (Fr\'echet) differentiable with respect to $Q$, with the differential at $Q$ denoted by $dT^S_Q$. They also showed that for fixed $S \in \cS^{++}_d$, the squared Wasserstein distance $W^2(\cdot,S): \cS^{++}_d \to \RR$ is twice differentiable, with the corresponding 1st and 2nd differential $d W^2(Q,S),d^2 W^2(Q,S)$ satisfying
\begin{equation}
    \begin{aligned}
        &d W^2(Q, S)(X)=\inner{I-T_Q^S}{X} \\
        & d^2 W^2(Q, S)(X, Y)=-\inner{X}{dT_Q^S(Y)}
    \end{aligned}
\end{equation}
In this paper, higher order differentials $d^k T^S_Q$ are essential for the development of the theory. Hence additional background on functional calculus are provided in Appendix \ref{subsec: appx_bg_func} for self-containedness.

%% file: formulation.tex
In this section, we formulate the Fr\'echet regression model in Section \ref{subsec: model formulation} and outline the assumptions in Section \ref{subsec: assumption}.

\subsection{Fr\'echet regression on Bures-Wasserstein manifold}
\label{subsec: model formulation}
Given a metric space $(\cY,d)$, let $\PP$ be a probability distribution over $\RR^p \times \cY$ that generates random pairs $(X,Y)$. The primary challenge in formulating a regression model between $Y$ and $X$ lies in the fact that the metric-space-valued response  $Y$ does not lend itself to linear operations. The Fr\'echet regression model \citep{muller16} seeks to generalize classical linear regression to a general metric space, building upon the concept of the Fr\'echet mean, which we introduce first.

The Fr\'echet mean $\EE_{\text{Fr\'echet}} Y$, generalizes the notion of an average in a general metric space $(\cY,d)$ and is defined as:
\begin{align}
    \EE_{\text{Fr\'echet}} Y:= \argmin_{y \in \cY} \EE_{Y} d^2(y,Y) 
    \label{eqn: frechet_mean}
\end{align}
The above definition is motivated by the fact that when $\cY$ is a Euclidean space,  $\EE_{\text{Fr\'echet}} Y$ coincides with the classical expectation $\EE Y$. Hence, we will omit the subscript and denote the Fr\'echet mean simply as $\EE Y$ hereafter.

With the notion of the Fr\'echet mean established, the Fr\'echet regression model proposed by \citet{muller16} is defined as:
\begin{equation}
    \EE_{\mathrm{\text{Fr\'echet}}} \sbr{Y | X=x} = \argmin_{y \in \cY} \EE_{(X,Y) \sim \PP} \sbr{w(x,X)d^2(y,Y)}
    \label{eqn: frechet_metric}
\end{equation}
where the weight function $w$ is defined as
\begin{equation}
    w(x, X)=1+(x-\mu)^{\top} \Sigma^{-1}(X-\mu), \quad \mu=E(X), \Sigma=\operatorname{Var}(X)
    \label{eqn: frechet_weight}
\end{equation}
In this context, the conditional expectation $\EE_{\mathrm{\text{Fr\'echet}}} \sbr{Y | X=x}$ on the left-hand side of (\ref{eqn: frechet_metric}) is a natural extension of \eqref{eqn: frechet_mean}, defined as the conditional minimizer:
\begin{align*}
    \EE_{\mathrm{\text{Fr\'echet}}} \sbr{Y| X=x} := \argmin_{y \in \cY} \EE_{Y} \sbr{d^2(y,Y) | X=x}
\end{align*}
The objective function $\EE \sbr{w(x,X) d^2(y,Y)}$ on the right-hand side of (\ref{eqn: frechet_metric}) generalizes the Fr\'echet mean (\ref{eqn: frechet_mean}) by introducing a weighted expectation, similar in spirit to kernel estimators in non-parametric statistics \citep{larry}.
The specific choice of the weight $w(x,X)$ in (\ref{eqn: frechet_weight}) is motivated by the requirement that the minimizer corresponds to the desired conditional expectation of $Y$ at $x$ under the classical linear regression model. Specifically, 
when $\cY = \RR$ and $\EE[Y|X=x]= a^\top (x-\mu)+b$, one can verify that
\begin{align}
    a^\top (x-\mu)+b = \argmin_{y \in \RR} \EE_{Y} \sbr{w(x,X)\rbr{y-Y}^2}
    \label{eqn: linear}
\end{align}
For further details on the existence and uniqueness of the various concepts defined above, see \citet{muller16} and the references therein.

When specialized to the Bures-Wasserstein manifold $(\cY,d) = (\cS_d^{++},W)$, we denote $Q^*(x):= \EE_{\mathrm{\text{Fr\'echet}}} \sbr{Q|X=x}$, and our model assumes
\begin{equation}
    Q^*(x) = \argmin_{S \in \cS_d^{++}} F(x,S), \quad \text{where } F(x,S):=\EE \left[  w(x,X) W^2(S,Q) \right]
    \label{eqn: frechet_assumption}
\end{equation}
For detailed discussions on the existence and uniqueness of both the population and empirical Fr\'echet mean on the Bures-Wasserstein manifold, we refer readers to \citet{agueh11_Wbary}, \citet{aap21}, and \citet{invitation}.

In the general case, there is no closed-form expression for $Q^*(x)$, similar to the situation with Bures-Wasserstein barycenters \citep{agueh11_Wbary}.
To gain intuition about the types of functional relationships between $Q$ and $X$ that the Fr\'echet regression model (\ref{eqn: frechet_assumption}) captures, consider a special case where $Q$ is concentrated on matrices that commute with each other. In this setting, for any commuting matrices $S, Q \in \cS_d^{++}$, the Bures-Wasserstein distance~(\ref{eqn: W_gaussian}) between $Q$ and $S$  simplifies to $W(Q, S) = \Fnorm{Q^{1/2} - S^{1/2}}$, which forms an isometric Hilbert embedding of $\supp Q$. Under this simplification, the Fr\'echet regression model (\ref{eqn: frechet_assumption}) essentially reduces to a linear regression model on the square roots of the covariance matrices in a similar flavor as (\ref{eqn: linear}):
\begin{align*}
    Q^*(x) = \argmin_{S \in \cS_d^{++}} \EE \sbr{w(x,X) \Fnorm{S^{1/2} - Q^{1/2}}^2}
\end{align*}
A specific example of this is provided in Example \ref{example_c} in Section \ref{sec: numerics}.

\subsection{Assumptions}
\label{subsec: assumption}
To establish rigorous theoretical guarantees for estimation and hypothesis testing, we impose the following model assumptions, which we believe are both theoretically minimal and sufficiently general.

We start with conditions on the marginal distribution of the covariate $X$ and the conditional distribution of $Q$ given $X$, as outlined in Assumption \ref{assumption: X} and \ref{assumption: bdd_Q}.
\begin{assumption}
    $X$ is sub-Gaussian with $\left\|X \right\|_{\psi_2}\leq C_{\psi_2}$ and $\lambda_{\min}(\Sigma)\geq c_{\Sigma}$ for some constants $C_{\psi_2},c_{\Sigma} >0$.
    \label{assumption: X}
\end{assumption}
 
\begin{assumption}
    Given $X=x \in \supp X$, the eigenvalues of $Q$ are bounded away from $0$ and infinity  in the sense that
    \begin{equation}
        \PP\rbr{Q \in \cS_d \rbr{\gamma_{\Lambda}(\norm{x-\mu})^{-1},\gamma_{\Lambda}(\norm{x-\mu})} | X=x} = 1
        \label{eqn: asmptn_bdd_Q}
    \end{equation}
    where $\gamma_{\Lambda}: \RR^+ \to \RR^+$ is defined by
    \begin{equation}
        \gamma_{\Lambda}(t) := c_\Lambda \rbr{t \vee 1}^{C_\Lambda}
        \label{eqn: asmptn_bdd_Q_gamma1}
    \end{equation}
    for some constants $c_\Lambda \geq 1$ and $C_\Lambda \geq 0$.
    \label{assumption: bdd_Q}
\end{assumption}
Assumption \ref{assumption: bdd_Q} implies upper and lower bounds on both the conditional expectation $Q^*(x)$ (see Lemma \ref{lem: Q*x}) and noise. These bounds become progressively weaker as $x \to \infty$.
In the literature on covariance matrix estimation, an upper bound on the population covariance matrix is commonly assumed \citep{cai_optimal_2010}. The lower bound on $\eigmin{Q^*(x)}$ here is analogous to the uniform upper bound on the conditional densities in Fr\'echet regression for 1D density response curves \citep[Assumption T4]{petersen21_WFtest}, since density is inversely proportional to the standard deviation in a 1D location-scale family. Assumptions on both upper and lower bounds are natural in the context of optimal transport. For example, \citet{rig21_map_optimal}, \citet{weed22_map_plug}, and \citet{weed22_map_ent} assumed smoothness and strong convexity of the Brenier potential for optimal transport map estimation, which translates to upper and lower bounds on the eigenvalues of the covariance matrix in the Gaussian case. Similarly, \citet{alt21_averaging} assume both upper and lower bounds on the eigenvalues to ensure a variance inequality proposed in \citet{che20_gd}, which is critical for proving the $\sqrt{n}$-convergence of the empirical barycenter and the linear convergence of a gradient descent algorithm on the Bures-Wasserstein manifold.
\begin{remark}
    The bounds presented here depend on $\gamma_{\Lambda}(\norm{x-\mu})$, which diverge as $x \to \infty$. This is motivated by the behavior of the conditional mean $\EE (Y|X=x)$ in classical linear regression, where it also diverges as $x \to \infty$. Specifically, when $\mu=0$ and $\EE[Y|X] = a^\top X+b$, it can be shown that $\abs{\EE[Y|X]} \leq \gamma_{\Lambda}(X)$, where $c_{\Lambda} = \norm{a}+\abs{b}$ and $C_{\Lambda}=1$. In our context, we posit that a polynomial growth rate $\gamma_{\Lambda}(t)\lesssim t^{C_{\Lambda}}$, which is permissible under (\ref{eqn: asmptn_bdd_Q_gamma1}), is often met in practical applications.
\end{remark}
\begin{remark}
    The bounded noise assumption can be relaxed by assuming that $Q$ has a light tail conditional on $Q^*(X)$. Specifically, we can assume that
    \begin{align*}
        \PP \cbr{ \frac{\eigmin{Q^*(x)}}{\eigmin{Q}} \vee  \frac{\eigmax{Q}}{\eigmax{Q^*(x)}} > t \Big| X=x} \lesssim  \exp \rbr{-c t^\alpha}, \qquad \forall t>0
    \end{align*}
    for some constant $\alpha>0$. This implies that $\eigmax{Q}$ is still bounded from above and $\eigmin{Q}$ is bounded from below. The proof presented in our paper remains valid by incorporating additional concentration arguments.
\end{remark}

The following assumption pertains to the Fr\'echet regression model. To ensure identifiability, we assume:\begin{assumption}
    For any $x \in \supp X$, $Q^*(x)$ is the unique minimizer of $F(x,\cdot)$.
    \label{assumption: frechet}
\end{assumption}

Then, it is natural to impose further assumptions on the minimizer $Q^*(x)$ of $F(x;\cdot)$. Assumption \ref{assumption: minimizer_global} below concerns the global behavior of $F(x,\cdot)$ outside a local neighborhood around $Q^*(x)$, while Assumption \ref{assumption: minimizer_local} focuses locally on the eigenvalue lower bound of the second differential of $F(x,\cdot)$ at $Q^*(x)$.

\begin{assumption}
    \label{assumption: minimizer_global}
     There exist constants $\alpha_F \geq 1$ and $\delta_F > 0$ such that for any $x \in \supp X$ and any $(\delta,\Delta)$ that satisfies $0 \leq \delta\leq \delta_F \leq \Delta$, the following 
    \begin{equation}
        \inf \left\{  F(x,S) - F(x,Q^*(x)): \delta \leq \Fnorm{S-Q^*(x)} \leq \Delta \right\} \geq \frac{\delta^{\alpha_F}}{\gamma_F(\norm{x-\mu},\Delta)}  
        \label{eqn: assmptn_min_sep}
    \end{equation}
    holds where $\gamma_F: \RR^+ \times \RR^+ \to \RR^+$ is defined by
    \begin{align*}
        \gamma_F(t_1,t_2) = c_F \rbr{t_1 \vee 1}^{C_F} \rbr{t_2 \vee 1}^{C_F}
    \end{align*}
    for constants $c_F \geq 1, C_F \geq 0$. 
\end{assumption}

\begin{remark}
    Assumption \ref{assumption: minimizer_global} is motivated by the well-separated-maximizer assumption in M-Estimation \citep[Lemma 3.2.1]{wellner}, which is also assumed in \citep{muller16}. In the special case when $X$ and $Q$ are independent, Lemma \ref{lem: F_ind} demonstrates that Assumption \ref{assumption: minimizer_global} holds with $\delta_F =1$, $\alpha_F = 2$, $C_F =1$ and some constant $c_F$ large enough. Here we allow for a polynomial dependence on $\norm{x-\mu}$ in the definition of $\gamma_{\lambda}$ so that (\ref{eqn: assmptn_min_eigen}) is still expected to hold in the general case when $X$ and $Q$ are not independent.
\end{remark}

\begin{assumption}
    \label{assumption: minimizer_local}
    For any $x \in \supp X$, consider the symmetric linear operator
    \begin{align*}
        \EE \rbr{-w(x,X)dT_{Q^*(x)}^{Q}}: \cS_d \to \cS_d,
    \end{align*}
    which is the second differential of $F(x,.\cdot)$ as $Q^*(x)$. This operator has a lower bound for its minimum eigenvalue given by:
    \begin{align}
        \lambda_{\min}\rbr{-\EE w(x,X)dT_{Q^*(x)}^{Q}} &\geq \frac{1}{\gamma_{\lambda}(\norm{x-\mu})} \label{eqn: assmptn_min_eigen}
    \end{align}
    where $\gamma_{\lambda}: \RR^+ \times \RR^+ \to \RR^+$ is defined by
    \begin{equation*}
        \gamma_{\lambda}(t) := c_{\lambda} \rbr{t \vee 1}^{C_{\lambda}}
    \end{equation*}
    for some constants $c_{\lambda} \geq 1$ and $C_{\lambda} \geq 0$.
\end{assumption}

\begin{remark}
    We provide some clarification on the term $\EE \rbr{-w(x,X)dT_{Q^*(x)}^{Q}}$ here. 
    It is equal to the second differential of $F(x,\cdot):\cS_d^{++}\to \RR$ at $Q^*(x)$.
    Note that under the Fr\'echet regression model (\ref{eqn: frechet_assumption}), $Q^*(x)$ is assumed to be the minimizer of the objective function~$F(x,\cdot):\cS_d^{++}\to \RR$. As a result, the second differential of $F(x,\cdot)$ at $Q^*(x)$ is positive semi-definite. Similar to the non-singular Fisher information assumption in maximum likelihood estimation \citep{vdv}, we further assume that the second differential $\EE \rbr{-w(x,X) T_{Q^*(x)}^Q}$ has positive eigenvalues that are lower-bounded as in (\ref{eqn: assmptn_min_eigen}). For relevant concepts in functional calculus, see Appendix \ref{subsec: appx_bg_func}, and for the justification that $\EE \rbr{-w(x,X)dT_{Q^*(x)}^{Q}}$ is equal to the second differential, see Appendix \ref{subsec: Techlem_dkT}.
\end{remark}


\begin{remark}
    In the special case when $X$ and $Q$ are independent, which is a consequence of Assumption \ref{assumption: cond_ind} and the null hypothesis of no effect (\ref{eqn: testing_null}) below, one can show that (Lemma \ref{lem: F_ind} in Appendix \ref{subsec: Techlem_QF}) Assumption \ref{assumption: minimizer_local} holds for $C_{\lambda} = 0$ and $c_{\lambda}$ large enough. Again, dependence on $\norm{x-\mu}$ is allowed in order to account for the possible unboundedness of $x$ when $X$ and $Q$ are independent.
\end{remark}

Finally, we assume conditional independence between $X$ and $Q$ given $Q^*(X)$ for hypothesis testing. 
However, this assumption is not necessary for the results on uniform convergence (Theorem \ref{thm: unfm_conv}) or the central limit theorem (Theorem \ref{thm: estimation}).
\begin{assumption}
    $X$ and $Q$ are independent conditional on $Q^{*}(X)$.
    \label{assumption: cond_ind}
\end{assumption}

%% file: main_results.tex
Building on the assumptions outlined in Section \ref{sec: formulation}, we now focus on hypothesis testing within the Fr\'echet regression model. We begin by establishing the uniform convergence of the Fr\'echet regression estimator and proving a central limit theorem in Section \ref{subsec: est}. The uniform convergence is not only theoretically important but also crucial for determining the asymptotic size and power of our proposed test. In Section \ref{subsubsec: test_statistic}, we introduce the test statistic, followed by an analysis of its asymptotic null distribution and asymptotic power in Section \ref{subsubsec: test_theory}.

\subsection{Estimation under the Fr\'echet regression model}
\label{subsec: est}

We define the Fr\'echet regression estimator as follows. For $\rho = n^{-1}$, let the empirical mean and the regularized covariance estimator be given by
\begin{equation*}
    \bar{X}=n^{-1} \sum_{i=1}^n X_i, \quad \hSigma_{\rho}=n^{-1} \sum_{i-1}^n\left(X_i-\bar{X}\right)\left(X_i-\bar{X}\right)^{\top} + \rho I_p.
\end{equation*}
The estimator is then defined as
\begin{equation}
    \hQ_{\rho} (x) := \argmin_{S \in \cS^{++}_d} F_{n, \rho}(x,S), \qquad F_{n, \rho}(x,S):= \frac{1}{n}\sum_{i=1}^n \left[  w_{n,\rho} (x,X_i)W^2(S,Q_i) \right]
    \label{eqn: frechet_estimate}
\end{equation}
where $w_{n,\rho}(x,X_i)= 1 + (x-\bar{X})\hSigma_{\rho}^{-1}(X_i - \bar{X})$. 

The estimator was initially studied by \citet{muller16} with $\rho = 0$, where its asymptotic properties in a general metric space were analyzed under additional assumptions on the covering number, based on the theory of M-estimation. Here, we introduce an additional regularization term, $\rho I_p$, to the empirical covariance matrix. This modification offers several benefits:
First, similar to ridge regression, the regularization term enhances numerical stability due to the involvement of the estimated precision matrix in $w_{n,\rho}$.
Next, as shown in Theorem \ref{thm: unfm_conv}, beyond providing a high-probability bound for the estimation error in (\ref{eqn: thm_unfm_supL}), it enables us to establish a uniform non-asymptotic $\sqrt{n}$-rate of convergence for the expected estimation error in (\ref{eqn: thm_unfm_Fmoment}), defined as:
\begin{equation*}
    \mathbb{E} \sup_{\norm{x - \mu} \lesssim \sqrt{\log n}} \Fnorm{\widehat{Q}_{\rho}(x) - Q^*(x)}. 
\end{equation*}
For $\rho=0$, as studied by \citet{muller16}, it is challenging even to prove the existence of such a moment. However, by setting \(\rho = n^{-1}\), we can show that this moment exists and converges uniformly at a parametric rate. Similar results have been reported in the literature on high-dimensional precision matrix estimation \citep{weidong, cai16}. 
Furthermore, the inclusion of the regularization term $n^{-1}I_p$ does not affect the asymptotic central limit theorem (Theorem \ref{thm: estimation}) or the asymptotic level and power of our test (Theorem \ref{thm: test_dist}, \ref{thm: power}). This is because it only introduces a bias on the order of $n^{-1}$ which is negligible compared to the typical stochastic variation of order $n^{-1/2}$. In fact, $\rho$ can be set to any $n^{-c}$ with a constant $c>1/2$. For simplicity, we fix $\rho = n^{-1}$ unless otherwise specified.


Our results on the non-asymptotic uniform convergence of the estimation error are summarized in Theorem \ref{thm: unfm_conv} below.


\begin{theorem}
    \label{thm: unfm_conv}
    Let $1 \leq L_n \leq C_L \sqrt{\log n}$ for some constant $C_L>0$.
    Suppose Assumption~\ref{assumption: X}-\ref{assumption: minimizer_global} hold. Then
    \begin{enumerate}[label = \arabic*.]
        \item For $\rho \in \cbr{0,n^{-1}}$ and any arbitrarily large fixed constant $\tau > 0$, with probability at least $1-C_3 n^{-\tau}$, we have
        \begin{align}
            \sup_{\norm{x -\mu}\leq L_n} \Fnorm{\hQ_{\rho}(x)-Q^*(x)} \leq  C_1 \frac{ \log^{C_2} {n}}{\sqrt{n}}
            \label{eqn: thm_unfm_supL}
        \end{align}
        for some constants $C_1, C_2, C_3>0$ independent of $n$.
        \item For $\rho = n^{-1}$, we have
        \begin{align}
            \EE \sup_{ \norm{x-\mu} \leq L_n } \Fnorm{\hQ_{\rho}(x)-Q^*(x)} \leq C_{4} \frac{\log^{C_{5}}{n}}{\sqrt{n}}
            \label{eqn: thm_unfm_Fmoment}
        \end{align}
        for some constants $C_{4}, C_{5}>0$ independent of $n$.
    \end{enumerate}
\end{theorem}

\begin{remark}
    The factor $\sqrt{\log n}$ in the upper bound of diameter $L_n$ can be refined to $\log^{\xi_L} n$ for an arbitrarily large fixed constant $\xi_L>0$. One implication of Theorem \ref{thm: unfm_conv} is that for $\rho \in \cbr{0,n^{-1}}$, we have
    \begin{align}
        \max_{i \in [n]} \Fnorm{\hQ_{\rho}(X_i) - Q^*(X_i) } \lesssim \frac{\polylog{n}}{\sqrt{n}}
        \label{eqn: thm6_Qi}
    \end{align}
    with probability at least $1 - O\rbr{n^{-\tau}}$ for any fixed $\tau >0$. This result is crucial for deriving the asymptotic size and power of our proposed test.  \citet{muller16} investigated asymptotic uniform convergence over a fixed compact set. In our context, their results state that for a fixed  $L>0$,
    \begin{align*}
        \sup_{ \norm{x-\mu} \leq L} W \rbr{\hQ_{0}(x), Q^*(x)} = O_p\rbr{n^{-\frac{1}{2(\alpha'-1)}}}
    \end{align*}
    for any $\alpha'>2$. However, their results do not guarantee (\ref{eqn: thm6_Qi}) unless the covariate $X$ is bounded by $L$, and therefore are insufficient for deriving the size and power of our test when $X$ has unbound support. Consequently, we employ non-asymptotic bounds and derive (\ref{eqn: thm6_Qi}) from (\ref{eqn: thm_unfm_supL}).
\end{remark}

\begin{remark}
    We use the Frobenius norm to quantify error here instead of the Wasserstein distance for the following reasons:
    \begin{itemize}
        \item Differential properties of $W^2(\cdot,\cdot)$: The Wasserstein distance is used in the objective function (\ref{eqn: frechet_assumption}) due to its favorable properties under scaling and rotation, as well as its meaningful probabilistic interpretation, as discussed in Section \ref{sec: intro}. After adopting the Bures-Wasserstein metric in our formulation of the Fr\'echet regression model, we need to leverage the differential properties of this metric to establish the statistical properties of the estimator. These differential properties are more naturally expressed using the Frobenius norm. For instance, for any $Q, \tQ \in \cS_d^{++}$ and any $A \in \cS_d$, we have
        \begin{align*}
            d_Q W^2(Q, \tQ)(A)=\inner{I_d-T_Q^{\tQ}}{A}
        \end{align*}
        where the Frobenius inner product is used, and $T_Q^{\tQ}$ is the optimal transport map from the normal distribution $N(0,Q)$ to $N(0,\tQ)$. Similarly, existing Taylor expansions for the Bures-Wasserstein metric are also formulated in terms of the Frobenius norm. Hence, it is more straightforward to express the estimation error using the Frobenius norm. For further details on the differential properties of $W$, see Appendix \ref{subsec: Techlem_dkT}.
        \item Relationship Between Frobenius Norm and Wasserstein Distance: The Frobenius norm and the Wasserstein distance are closely related. Specifically, Lemma \ref{lem: diff} in Appendix \ref{subsec: Techlem_dkT} shows that for any $Q, \tQ \in \cS_d^{++}$, the Wasserstein distance $W(Q, \tQ)$ can be upper-bounded in terms of their Euclidean distance. Consequently, we can derive analogous results for the Wasserstein distance, as in Corollary \ref{cor: W_moment} below. In fact, the two rates are equivalent up to a $\polylog{n}$ factor in our setting.
    \end{itemize}
\end{remark}

\begin{corollary}
    \label{cor: W_moment}
    Let $1 \leq L_n \leq C_L \sqrt{\log n}$ for some constant $C_L>0$.
    Suppose Assumption~\ref{assumption: X}-\ref{assumption: minimizer_global} hold. Then
    \begin{enumerate}[label = \arabic*.]
        \item For $\rho \in \cbr{0,n^{-1}}$ and any arbitrarily large fixed constant $\tau > 0$, with probability at least $1-C_{W3} n^{-\tau}$, we have
        \begin{align*}
            \sup_{\norm{x -\mu}\leq L_n} W\rbr{\hQ_{\rho}(x), Q^*(x)} \leq  C_{W1} \frac{ \log^{C_{W2}} {n}}{\sqrt{n}}
        \end{align*}
        for some constants $C_{W1}, C_{W2}, C_{W3}>0$ independent of $n$.
        \item For $\rho = n^{-1}$, we have
        \begin{align*}
            \EE \sup_{ \norm{x-\mu} \leq L_n } W\rbr{\hQ_{\rho}(x), Q^*(x)} \leq C_{W4} \frac{\log^{C_{W5}}{n}}{\sqrt{n}}
        \end{align*}
        for some constants $C_{W4}, C_{W5}>0$ independent of $n$.
    \end{enumerate}
\end{corollary}

To the best of our knowledge, Theorem \ref{thm: unfm_conv} and Corollary \ref{cor: W_moment} provide the first non-asymptotic uniform convergence results for Fr\'echet regression over a potentially diverging region. The proof of Theorem \ref{thm: unfm_conv} is intricate, and we provide an outline here. First, we use an approach similar to \citet{agueh11_Wbary} to establish that the largest eigenvalues of the estimates $\hQ_{\rho}(x)$ are uniformly bounded from above. Next, we show that $\hQ_{\rho}(x)$ converges uniformly at a slow rate by employing the chaining method. Special attention is required here due to the H\"older continuity of the squared distance $W^2$; see Lemma \ref{lem: prf_conv_lem_slow} for details. Finally, we achieve a uniform fast rate of convergence by solving a quadratic inequality related to the convergence rate. A detailed proof of Theorem \ref{thm: unfm_conv} is provided in Appendix \ref{sec: prf_unfm_conv}.

Theorem \ref{thm: unfm_conv} implies the pointwise consistency of the Fr\'echet regression estimator (\ref{eqn: frechet_estimate}). Moving one step further, we establish a central limit theorem that is elusive in general metric spaces with nonzero curvature \citep{muller16}.
In order to present the theorem,  we pause to introduce several notations. First, for any $x \in \RR^p$,  define $\vec{x}:=\begin{pmatrix}
    1 &  x^\top
\end{pmatrix}^\top$. For any random vector $X \in \RR^p$ with covariance matrix $\Sigma$, denote $\vec{\Sigma}= \EE \vec{X}\vec{X}^\top \in \RR^{(p+1) \times (p+1)}$. Next, let $[ \EE ( -w(x,X)dT_{Q^*(x)}^{Q} ) ]^{-1}$  denote the inverse of $\EE ( -w(x,X)dT_{Q^*(x)}^{Q} )$ which is a linear operator in $L(\cS_d;\cS_d)$ (Appendix \ref{subsec: Techlem_dkT}). Last, given elements $x_1,x_2$ of a Hilbert space $\cH$, the tensor product operator $x_1 \otimes x_2: \cH \to \cH$ is defined by $\rbr{x_1 \otimes x_2} y = \inner{x_1}{y} x_2$ for any $y \in \cH$; see \citet{hsing} for properties of the tensor product operator and its role in the central limit theorem for random elements of a Hilbert space.

With these notations in place, the central limit theorem for the Fr\'echet regression estimator $\hQ_{\rho}(x)$ is stated as Theorem \ref{thm: estimation}.

\begin{theorem}
    Let $\rho \in \cbr{0,n^{-1}}$. Suppose Assumption \ref{assumption: X}-\ref{assumption: minimizer_global} hold. Then for any fixed $x \in \supp X$, the following central limit theorem
    \begin{equation}
        \sqrt{n}\left[ \hQ_{\rho}(x) -Q^*(x) \right]  \xlongrightarrow{w} \sbr{\EE_{(X,Q)} \left( -w(x,X)dT_{Q^*(x)}^{Q} \right)}^{-1} Z_x,
        \label{eqn: thm_est_Q}
    \end{equation}
    holds. Here $Z_x \sim \cN(0, \Xi_x)$ is a Gaussian random element of $\RR^{d \times d}$ with covariance operator $\Xi_x$ equal to $\EE V_x \otimes V_x$ where
    \begin{align*}
        &V_x = V_{x,1} + V_{x,2}\\
        &V_{x,1} = w(x,X) \rbr{ T^{Q}_{Q^*(x)} - I_d }\\
        &V_{x,2} = - \rbr{ \vec{x}^\top \vSigma^{-1} (\vec{X} \vec{X}^\top - \vSigma)\vec{\Sigma}^{-1} \otimes I_d} \cdot \rbr{\EE \vec{X} \otimes (T^{Q}_{Q^*(x)} - I_d)}
    \end{align*}

    \label{thm: estimation}
\end{theorem}

\begin{remark}
    Since a linear transformation of a multivariate normal distribution is still normal,
    Theorem \ref{thm: estimation} implies that asymptotically the entries of $\sqrt{n}( \hQ_{\rho}(x) -Q^*(x) )$  jointly follow a multivariate normal distribution in $\RR^{d^2}$ when vectorized. More specifically, we have
    \begin{align}
        \sqrt{n} \rbr{\vecc \hQ_{\rho}(x) - \vecc Q^*(x) } \xlongrightarrow{w} \cN \rbr{0, \Omega_x}
        \label{eqn: clt_vec}
    \end{align}
    Here $\Omega_x = H_x^{-1} \EE (\vecc V_x)(\vecc V_x)^\top H_x^{-1}$ where $H_x \in \RR^{d^2 \times d^2}$ is the matrix representing the invertible linear operator $\EE ( -w(x,X)dT_{Q^*(x)}^{Q} )$; see Appendix \ref{subsec: Techlem_dkT} for the expression of $H_x$.
\end{remark}

\begin{remark}
    Theorem \ref{thm: estimation} can be utilized to obtain confidence regions for entries of $Q^*(x)$. The asymptotic variance can be estimated using the plug-in method as follows:
    \begin{align*}
        \hat{\Xi}_x = \frac{1}{n} \sum_{i=1}^n V_{x,i} \otimes V_{x,i} 
    \end{align*}
    where
    \begin{align*}
        V_{x,i} &= V_{x,1,i} + V_{x,2,i}\\
        V_{x,1,i} &= w_{n,\rho}(x,X_i) \rbr{ T^{Q_i}_{\hQ_{\rho}(x)} - I_d }\\
        V_{x,2,i} &= - \rbr{ \vec{x}^\top \hat{\vSigma}_{\rho}^{-1} (\vec{X}_i \vec{X}_i^\top - \hat{\vSigma})\hat{\vSigma}_{\rho}^{-1} \otimes I_d} \cdot \rbr{ \frac{1}{n}\sum_{j=1}^n \vec{X}_j \otimes (T^{Q_j}_{\hQ_{\rho}(x)} - I_d)}
    \end{align*}
    Similarly, $\sbr{\EE_{(X,Q)} \left( -w(x,X)dT_{Q^*(x)}^{Q} \right)}^{-1}$ can be estimated by
    \begin{align*}
        \sbr{\frac{1}{n} \sum_{i=1}^n \rbr{-w_{n,\rho}(x,X_i)dT_{\hQ_{\rho}(x)}^{Q_i}}  }^{-1}
    \end{align*}
    In Fig. \ref{fig: QQ_clt} in Section \ref{subsec: simulation}, the asymptotic variance of each entry of $Q^*(x)$ is extracted using~(\ref{eqn: clt_vec}) and appropriate indexing
\end{remark}
Theorem \ref{thm: estimation} shows that the covariance operator $\Xi_x$ has contributions from both $V_{x,1}$ and $V_{x,2}$. If the expectation $\mu$ and the covariance matrix $\Sigma$ of the predictor $X$ are known, and we obtain an estimate $\hQ^*_{0}(x)$ of $Q^*(x)$ by directly optimizing an 'oracle' objective function $F_n^*(x,S)$ defined as $F^*_n(x,S)= \frac{1}{n} \sum_{i=1}^n w(x,X_i) W^2(S,Q_i)$,
then $\hQ^*_{0}(x)$ would follow a central limit theorem with covariance operator exactly equal to $\EE V_{1,x} \otimes V_{1,x}$. When $\mu$ and $\Sigma$ are unknown and empirical estimates $\hmu$ and $\hSigma$ are plugged in as in (\ref{eqn: frechet_estimate}), there is an additional contribution to $\Xi_x$ from $V_{x,2}$. Theorem \ref{thm: estimation} is supported by numerical experiments in Section \ref{sec: numerics}, and 
the proof is given in Appendix \ref{sec: proof_est}.
To the best of our knowledge, the expression for $\Xi_x$ cannot be further simplified in general. However, under the hypothesis that $X$ and $Q$ are independent, which can be a consequence of the null hypothesis in Section \ref{subsec: test} and Assumption~\ref{assumption: cond_ind}, the term $V_{x,2}$ vanishes, leaving $\Xi_x$ with contributions only from $V_{x,1}$. This result is summarized in Corollary \ref{cor: clt_ind} below; see Appendix~\ref{sec: prf_clt_cor} for the proof.


\begin{corollary}
    \label{cor: clt_ind}
    Instate the assumptions in Theorem \ref{thm: estimation}. If $X$ and $Q$ are independent, then 
    \begin{equation*}
        \sqrt{n}\left[ \hQ_{\rho}(x) -Q^*(x) \right] \xlongrightarrow{w} \left( \EE_{(X,Q)} \left( -w(x,X)dT_{Q^*(x)}^{Q} \right) \right)^{-1} Z_x',
    \end{equation*}
    holds where $Z_x' \sim \cN(0, \Xi_x')$ and
    \begin{align*}
        \Xi_x' = \EE \sbr{w(x,X) \rbr{ T^{Q}_{Q^*} - I_d } \otimes \rbr{ T^{Q}_{Q^*} - I_d } w(x,X)}
    \end{align*}
\end{corollary}

\subsection{Hypothesis testing}
\label{subsec: test}
In this section, we focus on testing the global null hypothesis of no effects within the Fr\'echet regression model on the Bures-Wasserstein manifold, defined as follows:
\begin{equation}
    \mathcal{H}_0: Q^*(x) \equiv Q^* \quad \text{for some unknown } Q^*,
    \label{eqn: testing_null}
\end{equation}
where $Q^*(x)$ is constant across all $x$. 

We begin by providing the motivation and the formulation of the test statistic in Section \ref{subsubsec: test_statistic}. Subsequently, we explore its asymptotic null distribution in Section \ref{subsubsec: test_theory}, followed by an analysis of the asymptotic size and power of the proposed test.

\subsubsection{Test statistic}
\label{subsubsec: test_statistic}
Note that a key distinction between the Fr\'echet regression model (\ref{eqn: frechet_assumption}) and classical linear or generalized linear regression models is that the conditional expectation in the Fr\'echet model is determined by an optimization problem, rather than through a link function. Consequently, there is no parameter analogous to the slope \(\beta\) found in linear models. Therefore, the null hypothesis can only be tested by directly aggregating comparisons between the estimated predictions $\hQ_{\rho}(X_i)$ and the Fr\'echet mean $Q^*$.

Given the uniform consistency of the Fr\'echet regression estimator $\hQ_{\rho}(x)$ in Theorem \ref{thm: unfm_conv}, we propose the following test statistic for testing the null hypothesis in (\ref{eqn: testing_null}).
\begin{equation}
    \hat{\cT}_{\rho} = \sum_{i=1}^n \Fnorm{\hat{H}_{\rho}\cdot \left( \hQ_{\rho}(X_i) - \hQ_{\rho}(\bar{X}) \right) }^2, \quad \text{where} \quad \hat{H}_{\rho} =-\frac{1}{n} \sum_{i=1}^n dT_{\hQ_{\rho}(\bar{X})}^{Q_i}
    \label{eqn: test_stat}
\end{equation}

To develop some intuition about $\hat{\cT}_{\rho}$, consider that under the null hypothesis (\ref{eqn: testing_null}), the "difference" between $\hQ_{\rho}(X_i)$ and $Q^*$ should be small. Since $Q^*$ is unknown, we estimate it by $\hQ_{\rho}(\bar{X})$. Note that $\hQ_{\rho}(\bar{X})$ coincides with the Fr\'echet mean of $Q_1,...,Q_n$. Specifically, by definition (\ref{eqn: frechet_estimate}), we have $w_{n,\rho}(\bar{X},X_i)=1$, which implies:
\begin{align*}
    \hQ_{\rho}(\bar{X}) = \argmin_{S \in \cS_d^{++}} \frac{1}{n} \sum_{i=1}^n W^2(S,Q_i)
\end{align*}
which is precisely the Fr\'echet mean of $Q_1,\ldots,Q_n$. Therefore, a sensible test statistic would be of the form
\begin{align*}
    \cT_f =\sum_{i=1}^n f\left( \hQ_{\rho}(X_i), \hQ_{\rho}(\bar{X}) \right)
\end{align*}
where $f: \RR^{d \times d} \times \RR^{d \times d} \to \RR$ is some function that measures the "difference" between $\hQ_{\rho}(X_i)$ and $\hQ_{\rho}(\bar{X})$. The function $f$ is expected to satisfy $f(\cdot,\cdot) \geq 0$, $f(Q,Q) = 0$ and $f(Q,S)=f(S,Q)$. Assuming tightness of the second-order Taylor approximation, the uniform consistency of $\hQ_{\rho}(x)$ (Theorem \ref{thm: unfm_conv}) then implies
\begin{equation}
    \cT_f \approx \sum_{i=1}^{n} \left \langle \frac{1}{2}H_f \left( \hQ_{\rho}(X_i) - \hQ_{\rho}(\bar{X}) \right), \hQ_{\rho}(X_i) - \hQ_{\rho}(\bar{X}) \right\rangle
    \label{eqn: test_stat_intuition_quad}
\end{equation}
where $H_f$ is the Hessian of $f$ at $\hQ_{\rho}(\bar{X})$. This justifies the quadratic form of (\ref{eqn: test_stat}). 

Furthermore, the specification of $H_f$ should depend on the detailed  distribution of $\hQ_{\rho}(X_i) - \hQ_{\rho}(\bar{X})$ such that the quadratic form in (\ref{eqn: test_stat_intuition_quad}) has a well-defined asymptotic distribution. To derive the asymptotic distribution of $\hQ_{\rho}(X_i) - \hQ_{\rho}(\bar{X})$, we rely on the optimality condition $\sum_{j}^{} w_{n,\rho}(x,X_j) (T^{Q_j}_{\hQ_{\rho}(x)} - I_d) = 0$. Under the null hypothesis (\ref{eqn: testing_null}),  and assuming tightness of the first order approximation of the optimality conditions at $X_i$ and $\bar{X}$ we have
\begin{equation}
    \begin{aligned}
        0 &\approx \sum_{j}^{} w_{n,\rho}(X_i,X_j) (T^{Q_j}_{Q^*} - I_d) + \sum_{j}^{} w_{n,\rho}(X_i,X_j) dT^{Q_j}_{Q^*} \left( \hQ_{\rho}(X_i) - Q^* \right)\\
        0 &\approx \sum_{j}^{} w_{n,\rho}(\bar{X},X_j) (T^{Q_j}_{Q^*} - I_d) + \sum_{j}^{} w_{n,\rho}(\bar{X},X_j) dT^{Q_j}_{Q^*} \left( \hQ_{\rho}(\bar{X}) - Q^* \right)
    \end{aligned}
\end{equation}
Take the difference and rearrange, we arrive at
\begin{equation}
    \begin{aligned}
        \left( -\frac{1}{n}\sum_{j}^{} dT_{Q^*}^{Q_j} \right) \cdot \sqrt{n}(\hQ_{\rho}(X_i) - Q^*) &\approx \underbrace{\frac{1}{\sqrt{n}}\sum_{j}^{} (w_{n,\rho}(X_i,X_j) - 1)(T^{Q_j}_{Q^*} - I_d)}_{a_1(X_i)}\\
        &+ \underbrace{\frac{1}{\sqrt{n}}\sum_{j}^{} (w_{n,\rho}(X_i,X_j)-1) dT^{Q_j}_{Q^*} \left( \hQ_{\rho}(X_i) - Q^* \right)}_{a_2(X_i)}
    \end{aligned}
\end{equation}
It can be shown that $a_2(X_i)$ is negligible compared to $a_1(X_i)$, leading to
\begin{equation}
    \left( -\frac{1}{n}\sum_{j}^{} dT_{Q^*}^{Q_j} \right)\cdot \sqrt{n}(\hQ_{\rho}(X_i) - Q^*) \approx a_1(X_i).
    \label{eqn: test_intuition_approx_a1}
\end{equation}
Specifically, the intuition behind (\ref{eqn: test_intuition_approx_a1}) is as follows. The null hypothesis (\ref{eqn: testing_null}) and Assumption \ref{assumption: cond_ind} imply that $X$ and $Q$ are independent. As a result, $\frac{1}{\sqrt{n}}\sum_{j}^{} (w(x,X_j)-1) dT^{Q_j}_{Q^*}$ has zero expectation and is of order $O_p(1)$ by the central limit theorem, suggesting that
\begin{align*}
    \frac{1}{\sqrt{n}}\sum_{j}^{} (w_{n,\rho}(X_i,X_j)-1) dT^{Q_j}_{Q^*} = O_p(1)
\end{align*}
by an approximation argument.
Then the consistency of $\hQ_{\rho}(x)$ implies that $a_2(X_i)$ is of order $o_p(1)$, making it negligible compared to $a_1(X_i)$.

Further calculations confirm a tractable asymptotic distribution for $\sum_{}^{} \left\|a_1(X_i) \right\|^2$, which suggests setting $H_f = (-\frac{1}{n}\sum_{j}^{} dT_{Q^*}^{Q_j}) \otimes (-\frac{1}{n}\sum_{j}^{} dT_{Q^*}^{Q_j})$. 
Since $Q^*$ is unknown, $H_f$ is approximated by $\hat{H}_{\rho}\otimes \hat{H}_{\rho}$, leading directly to our test statistic $\hat{\cT}_{\rho}$ as defined in (\ref{eqn: test_stat}).

\begin{remark}
    \citet{petersen21_WFtest} proposed the following test statistic for the case where
    responses are 1D densities, which has a simpler form compared to ours. For $\rho=0$,
    \begin{align*}
        \hat{\cT}_{\rho,\mathrm{1D}} = \sum_{i=1}^n W^2\left( \hQ_{\rho}(X_i), \hQ_{\rho}(\bar{X}) \right)
    \end{align*} 
    However, their results rely heavily on the isometric Hilbert embedding of the 1D Wasserstein space and do not generalize to higher dimensions. In contrast,  our test statistic (\ref{eqn: test_stat}) is applicable in any dimension and is motivated by the Wald statistic used in generalized linear models. Moreover, one can show that $\hat{\cT}_{\rho}$ and $\hat{\cT}_{n,\mathrm{1D}}$ are equivalent in the special case of 1D Gaussian distributions. Since our test statistic (\ref{eqn: test_stat}) can also be viewed as a generalization of the numerator of the global F-test in multiple linear regression, we refer to $\hat{\cT}_{\rho}$ in (\ref{eqn: test_stat}) as the  Wasserstein F -statistic following \citet{petersen21_WFtest}.
\end{remark}

\subsubsection{Theoretical properties}
\label{subsubsec: test_theory}

We now discuss the theoretical guarantees of our test statistic $\hat{\cT}_{\rho}$ and the corresponding test. To begin, Theorem \ref{thm: test_dist} provides the asymptotic null distribution of $\hat{\cT}_{\rho}$.
\begin{theorem}
    \label{thm: test_dist}
    Let $\rho \in \cbr{0,n^{-1}}$. Suppose Assumption \ref{assumption: X}-\ref{assumption: frechet} and \ref{assumption: cond_ind} hold. Then under the null (\ref{eqn: testing_null}), the test statistic $\hat{\cT}_{\rho}$ satisfies
    \begin{equation}
        \hat{\cT}_{\rho} \xlongrightarrow{w} \sum_{i} \lambda_i w_i
        \label{eqn: asymp_null_dist}
    \end{equation}
    where $w_i$ are i.i.d. $\chi_p^2$ random variables, and $\lambda_i$ are the eigenvalues of the following operator:
    \begin{align*}
        \EE \sbr{\rbr{T_{Q_*(\mu)}^Q -I_d} \otimes \rbr{T_{Q_*(\mu)}^Q -I_d}}
    \end{align*}
\end{theorem}

\begin{remark}
    The proof of Theorem \ref{thm: test_dist} relies on the uniform consistency established in Theorem \ref{thm: unfm_conv}, since $\hat{\cT}_{\rho}$ involves estimated predictions $\hQ_{\rho}(\cdot)$ at random covariates $\{X_i\}_{i \in [n]}$. In contrast, the uniform consistency is not needed in \citet{petersen21_WFtest} because the Wasserstein space $W_2(\RR)$ is essentially flat (has zero sectional curvature), which allows for a closed-form expression for $\hat{\cT}_{n,1D}$. However, the Bures-Wasserstein manifold $(\cS_d^{++},W)$ is positively curved when $d>1$ \citep{ambrosio}, and no closed-form expression is available for $\hat{\cT}_{\rho}$. Thus, we rely on uniform consistency to ensure the tightness of the Taylor approximation. For the proof, 
    see Appendix \ref{sec: proof_null}.
\end{remark}

Theorem \ref{thm: test_dist} asserts that $\hat{\cT}_{\rho}$ converges weakly to a weighted sum of $\chi_p^2$s with weights determined by the eigenvalues of the covariance operator $\EE (T_{Q_*}^Q -I_d) \otimes (T_{Q_*}^Q -I_d)$. To get a corresponding test, note that 
the asymptotic null distribution in Theorem \ref{thm: test_dist} depends on unknown parameters, namely the eigenvalues $\lambda_i$, which must be approximated to formulate a rejection region. A natural approach would be to estimate the eigenvalues $\hlambda_i$ of the sample average $\frac{1}{n}\sum_{i=1}^{n} T^{Q_i}_{\hQ_{\rho}(\bar{X})} \otimes T^{Q_i}_{\hQ_{\rho}(\bar{X})}$ and let $\hq_{1-\alpha}$ be the $1-\alpha$ quantile of $\sum_{i=1}^{} \hlambda_{i} w_i $. Then we define our test $\Phi_{\alpha}$ for any $\alpha \in (0,1)$ by
\begin{equation}
    \Phi_{\rho, \alpha} = \1 \rbr{\hat{\cT}_{\rho} > \hq_{1-\alpha}}
    \label{eqn: test}
\end{equation}

With Theorem \ref{thm: test_dist} established, Proposition \ref{prop: size} below shows that
$\Phi_{\rho,\alpha}$ has an asymptotic size of $\alpha$ under the null hypothesis. See Appendix \ref{sec: prf_size} for the proof.
\begin{prop}
    \label{prop: size}
    Suppose Assumption \ref{assumption: X}-\ref{assumption: frechet} and \ref{assumption: cond_ind} hold. Then under the null (\ref{eqn: testing_null}),
    \begin{equation*}
        \PP \rbr{\hat{\cT}_{\rho} > \hq_{1-\alpha}} \to \alpha
    \end{equation*}
    as $n\to \infty$.
\end{prop}

Finally, we turn to an analysis of the power of the test $\Phi_{\rho,\alpha}$ under a sequence of contiguous alternatives. To this end, we denote by $\mathfrak{P}$ the set of distributions of $(X,Q)$ that satisfy Assumption \ref{assumption: X} - \ref{assumption: cond_ind},
\begin{equation*}
    \mathfrak{P}:=\left\{ \PP \in \cP_2\left( \RR^p \times \cS_d^+\right): \PP \text{ satisfies } \text{Assumption } \ref{assumption: X} - \ref{assumption: cond_ind} \right\}
\end{equation*}
For any $\PP \in \mathfrak{P}$, We measure the deviation of $Q^*(x)$ from being a constant function of $x$ by $\EE \dist^2(Q^*(X),Q^*(\mu))$ where $\dist$ is chosen to be either the Wasserstein distance or the Frobenius distance. We define the corresponding alternatives under each distance as follows:
\begin{equation*}
    \begin{aligned}
        H_{1,n}: \PP \in \mathfrak{P}_{F}(a_n) := \left\{ \widetilde{\PP} \in \mathfrak{P}: \EE_{(X,Q) \sim \widetilde{\PP}} \Fnorm{Q^*(X)-Q^*(\mu)}^2 \geq a_n^2\right\},\\
        \tH_{1,n}: \PP \in \mathfrak{P}_{W}(a_n) := \left\{ \widetilde{\PP} \in \mathfrak{P}: \EE_{(X,Q) \sim \widetilde{\PP}} W^2 \rbr{Q^*(X),Q^*(\mu)}\geq a_n^2\right\}.
    \end{aligned}
\end{equation*}
 Theorem \ref{thm: power} shows that $\Phi_{\rho,\alpha}$ is powerful against both $H_{1,n}$ and $\tH_{1,n}$ whenever $a_n \gtrsim n^{-(1/2 - \alpha_2)}$ for some constant $\alpha_2 >0$. The proof is provided in Appendix \ref{sec: proof_power}.
\begin{theorem}
    \label{thm: power}
    Let $\rho \in \cbr{0,n^{-1}}$.
    Consider a sequence of alternative hypotheses
    $H_{1,n}$
    with $a_n$ being a sequence such that $a_n \gtrsim \frac{1}{n^{1/2 - \alpha_2}}$ for some some constant $\alpha_2 >0$.
    Then the worst case power converges uniformly to $1$, that is 
    \begin{equation*}
        \inf_{\PP \in \mathfrak{P}_{F}(a_n)} \PP\left( \hat{\cT}_{\rho} > \hq_{1-\alpha} \right) \to 1
    \end{equation*}
    as $n \to \infty$. The same result also holds for alternative hypotheses $\tH_{1,n}$, defined by the Wasserstein distance, that is
    \begin{equation*}
        \inf_{\PP \in \mathfrak{P}_{W}(a_n)} \PP\left( \hat{\cT}_{\rho} > \hq_{1-\alpha} \right) \to 1
    \end{equation*}
    as $n \to \infty$
\end{theorem}

%% file: numeric.tex
In this section, we propose a Riemannian gradient descent algorithm for optimizing (\ref{eqn: frechet_estimate}) in Section \ref{subsec: algm}. We then present a series of numerical experiments in Section \ref{subsec: simulation} to validate our theoretical results on the central limit theorem (Theorem \ref{thm: estimation}), asymptotic null distribution (Theorem \ref{thm: test_dist}) and power (Theorem \ref{thm: power}). Additionally, in Section \ref{subsec: partial}, we conduct simulations under a setting where $Q$ is not directly observed but is estimated from data, to assess the deviation from the setting with perfect observations.

\subsection{Riemannian gradient descent algorithm}
\label{subsec: algm}

Motivated by the Bures-Wasserstein gradient descent algorithm \citep{che20_gd,alt21_averaging} for the vanilla Bures-Wasserstein barycenter, we propose a gradient descent algorithm to compute $\hQ_{\rho}(x)$ in (\ref{eqn: frechet_estimate}), which is given as Algorithm \ref{alg: gd}.

\begin{minipage}{0.9\textwidth}
    \centering
    \begin{algorithm}[H] 
        \caption{GD for Fr\'echet regression} \label{alg: gd}
        \begin{algorithmic}[1]
            \State{\textbf{Input:} predictors $\{X_i\}_{i=1}^n$, responses $\{Q_i\}_{i=1}^n$, $\rho \in \cbr{0,n^{-1}}$, predictor $x$, learning rate $\eta$, initialization $S_0$, maximum number of iterations $T$, threshold~$\eps$.
            }
            \State{Initialize $S \gets S_0$.}
            \For{$t=1,...,T$}
                \State{Set 
                    \begin{align}
                        G \gets  I_d + \eta \cdot \frac{1}{n} \sum_{i=1}^n w_{n,\rho}(x,X_i) (T_S^{Q_i} -I_d)
                        \label{eqn: alg_gd1}
                    \end{align}
                }
                \State{Set
                    \begin{align}
                        S \gets G S G
                        \label{eqn: alg_gd2}
                    \end{align}
                }
                \If{$\Fnorm{G} < \eps$ } 
                    \State \textbf{break}
                \EndIf
            \EndFor
            \State{\textbf{Output:} $S$.}
        \end{algorithmic}
    \end{algorithm}
\end{minipage}

\bigskip{}

Algorithm \ref{alg: gd} can be viewed as a Riemannian gradient descent algorithm (see Appendix \ref{sec: appx_bg} and \citet{zemel_vic_16,che20_gd,alt21_averaging}). Intuitively, $-\frac{1}{n} \sum_{i=1}^n w_{n,\rho}(x,X_i) (T_S^{Q_i} -I_d)$ is the derivative of the objective function $F_{n,\rho}(x,S)$ in (\ref{eqn: frechet_estimate}) in the tangent space \citep[Corollary 10.2.7]{ambrosio} and (\ref{eqn: alg_gd1}) corresponds to one gradient step in the tangent space with step size $\eta$. Then (\ref{eqn: alg_gd2}) is mapping the gradient step in the tangent space back to $\cS_d^{++}$ through the exponential map (Appendix \ref{sec: appx_bg})  

The algorithm terminates if the Frobenius norm of the gradient $G$ falls below the threshold $\eps$, indicating that the relative change in the gradient update is less than $\eps$. 
For the initialization $S_0$, optimization over the Euclidean space typically starts near the origin~\citep{chen_gradient_2019,du21}.
However, since the space of symmetric positive definite (SPD) matrices, $\cS_d^{++}$, is nonlinear, the natural counterpart of the origin in this space is the identity matrix $I_d$. Therefore, we initialize at $S_0 = I_d$. 
In Appendix \ref{sec: more_sim}, we use numerical simulations to compare this initialization with random initialization and initialization at the mean, and observe that it consistently performs at least as well as these alternatives. 
For the step size $\eta$, \citet{alt21_averaging} observed through numerical simulations that while the convergence rate of Euclidean gradient descent is highly sensitive to its step size, Riemannian gradient descent requires no tuning and works effectively with $\eta = 1$ when computing the Bures-Wasserstein barycenter. 
In our simulations, we also find that $\eta = 1$ performs at least as well as (and often better than) smaller step sizes. See Appendix \ref{sec: more_sim} for details.



\subsection{Simulation setup and results}
\label{subsec: simulation}
To validate our theory and demonstrate the practical applicability of our inferential procedures, we conduct a series of numerical experiments with $\rho = n^{-1}$. We start with two illustrative examples that follow the Fr\'echet regression model. In Example \ref{example_c}, the covariance matrices share a common eigenspace and commute, which effectively reduces the Fr\'echet regression model to a linear regression model on the square roots. In contrast, Example \ref{example_nc} considers the case where the covariance matrices do not commute.

\begin{example}
    \label{example_c}
    Let $\cbr{X_i = \begin{pmatrix} X_{i1} & \cdots & X_{ip} \end{pmatrix}, i \in [n]}$ be i.i.d. random covariates in $\RR^p$ with $X_i \sim \mathrm{Uniform}[-1,1]^p$. The response matrices $Q_1,\ldots,Q_n \in \RR^{d \times d}$ are generated as:
    \begin{align*}
        Q_i = U V_i f(X_i; \delta)^2 V_i U^\top
    \end{align*}
    where $U$ and $\cbr{X_i, V_i}_{i \in [n]}$ are independent, and:
    \begin{itemize}
        \item $f(\cdot; \delta): \RR^p \to \RR^{d \times d}$ is a mapping from $\RR^p$ to diagonal matrices defined by:
        \begin{align*}
            x = \begin{pmatrix} x_1 & \cdots & x_p \end{pmatrix} \mapsto f(x;\delta) = \left( f(x;\delta)_{kl} \right)_{k,l \in [d]},
        \end{align*}
        where
        \begin{align*}
            f(x;\delta)_{kk} = 1.5 + \frac{k}{2} + \delta \cdot \sum_{j=1}^p x_j,
        \end{align*}
        with $\delta \in (-2p^{-1}, 2p^{-1})$ being a parameter that quantifies the deviation of the model from the null hypothesis (\ref{eqn: testing_null}).

        \item $U \in \cO_d$ is a random orthogonal matrix following the Haar measure.
        \item $V_1,\ldots V_n\in \RR^d$ are random diagonal matrices with i.i.d. diagonal entries $V_{i,kk} \sim \mathrm{Uniform}[-0.1,0.1]$.
    \end{itemize}
    It can be verified that the pair $(X,Q)$ satisfies the Fr\'echet regression model with the conditional expectation $Q^*(\cdot)$ satisfying
    \begin{align*}
        Q^*(X_i) = U f(X_i)U^\top
    \end{align*}
\end{example}

\begin{example}
    \label{example_nc}
    Let $X = \begin{pmatrix} X_1 & \cdots & X_p \end{pmatrix}$ be a random covariate in $\RR^p$ with $X \sim \mathrm{Uniform}[-1,1]^p$. The response matrix $Q \in \mathbb{R}^{d \times d}$, where $d$ is an even number, is generated as:
    \begin{align*}
        Q = UV g(X; \delta)^2 V U^\top,
    \end{align*}
    where $U, V$ and $X$ are independent, and:
    \begin{itemize}
        \item \(g(\cdot;\delta): \mathbb{R}^p \to \mathbb{R}^{d \times d}\) is a mapping from $\RR^p$ to diagonal matrices defined by:
        \begin{align*}
            x = \begin{pmatrix} x_1 & \cdots & x_p \end{pmatrix} \mapsto g(x;\delta) = \left( g(x;\delta)_{kl} \right)_{k,l \in [d]},
        \end{align*}
        where
        \begin{align*}
            g(x;\delta)_{kk} = 1.5 + 0.5 \cdot \lceil k/2 \rceil + \delta \cdot \sum_{j=1}^p x_j,
        \end{align*}
        with $\ceil{\cdot}$ denoting the ceiling function, and $\delta \in (-2p^{-1}, 2p^{-1})$ being a parameter that quantifies the deviation of the model from the null hypothesis (\ref{eqn: testing_null}).
    
        \item $U \in \cO_d$ is a random orthogonal matrix with a block-diagonal structure
        \begin{align*}
            U = \operatorname{diag}(U^{(1)}, \ldots, U^{(\floor{d/2})})
        \end{align*}
        where $U^{(1)}, \ldots, U^{(\floor{d/2})}$ are i.i.d. random $2 \times 2$ orthogonal matrices following the Haar measure.
        
        \item $V \in \RR^d$ is a diagonal matrix with i.i.d. diagonal entries $V_{ii}\sim \mathrm{Uniform}[-0.1,0.1]$.
    \end{itemize}
    It can be verified that the pair $(X, Q)$ satisfies the Fr\'echet regression model with
    \begin{align*}
        Q^*(x) = g(x)^2
    \end{align*}
    
\end{example}

With Example \ref{example_c} and \ref{example_nc} in hand, 
we proceed to check the validity of the central limit theorem for $\hQ_{\rho}(x)$ as stated in Theorem \ref{thm: estimation}. 
To this end, we generate random predictor-response pairs $(X,Q)$ based on Example \ref{example_c} with parameters $d = 5$, $p = 5$ and $\delta = 0$; and Example \ref{example_nc} with parameters $d=6$, $p=5$ and $\delta = 0$.
For each trial, we generate $n=200$ samples of $(X_i,Q_i)$ and obtain the Fr\'echet regression estimate $\hQ_{\rho}(x)$ using Algorithm \ref{alg: gd} with $\rho = n^{-1}$. We then compute normalized error
\begin{align*}
    \sbr{\tQ(x)}_{ij} = \frac{\sqrt{n} \left[ \hQ_{\rho}(x) - Q^*(x) \right]_{ij}}{\sqrt{\hat{v}_{x,ij}}}
\end{align*}
where $\hat{v}_{x,ij}$ denotes the plug-in estimate for the asymptotic variance of the $(i,j)$-entry of $\sqrt{n} (\hQ_{\rho}(x) - Q^*(x))$, following Theorem \ref{thm: estimation} and (\ref{eqn: clt_vec}). Specifically, the covariance operator $\Xi_x$ in Theorem \ref{thm: estimation} is estimated as follows:
\begin{align*}
    \hat{\Xi}_x = \frac{1}{n} \sum_{i=1}^n V_{x,i} \otimes V_{x,i} 
\end{align*}
where
\begin{align*}
    V_{x,i} &= V_{x,1,i} + V_{x,2,i}\\
    V_{x,1,i} &= w_{n,\rho}(x,X_i) \rbr{ T^{Q_i}_{\hQ_{\rho}(x)} - I_d }\\
    V_{x,2,i} &= - \rbr{ \vec{x}^\top \hat{\vSigma}_{\rho}^{-1} (\vec{X}_i \vec{X}_i^\top - \hat{\vSigma})\hat{\vSigma}_{\rho}^{-1} \otimes I_d} \cdot \rbr{ \frac{1}{n}\sum_{j=1}^n \vec{X}_j \otimes (T^{Q_j}_{\hQ_{\rho}(x)} - I_d)}
\end{align*}
The estimated covariance $\hat{v}_{x,ij}$ can be extracted from $\hat{\Xi}_x$ by appropriately indexing its elements.

Figure \ref{fig: QQ_clt} presents the Q-Q (quantile-quantile) plots of $\tQ_{ij}(x)$ against the standard normal distribution, based on 200 Monte Carlo trials at $x=0$. According to Theorem \ref{thm: estimation}, $\tQ_{ij}(x)$ should asymptotically follow a standard normal distribution, which would be indicated by a Q-Q plot that exhibits a linear relationship with a slope of one and an intercept of zero.
As shown in Figure \ref{fig: QQ_clt}, the empirical quantiles of $\tQ_{ij}(0)$ align closely with the theoretical quantiles of $\cN(0,1)$, providing strong empirical support for the validity of Theorem \ref{thm: estimation}.
    
\begin{figure}[!htb]
    \centering
    \begin{subfigure}[t]{0.96\textwidth}
        \centering
        \includegraphics[width=\textwidth]{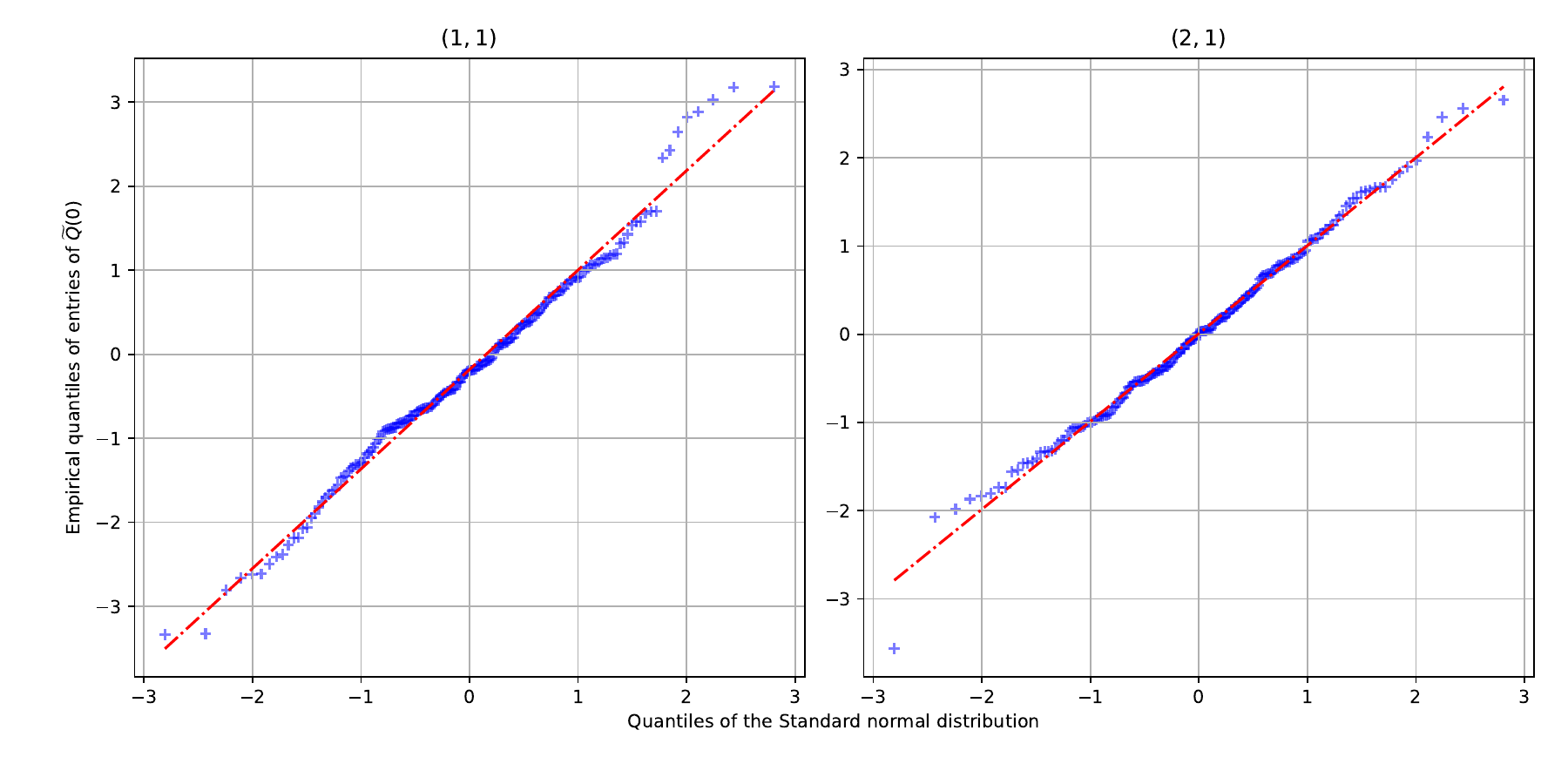}
        \caption{Example \ref{example_c}, $d=5$}
    \end{subfigure}
    \\
    \begin{subfigure}[t]{0.96\textwidth}
        \centering
        \includegraphics[width=\textwidth]{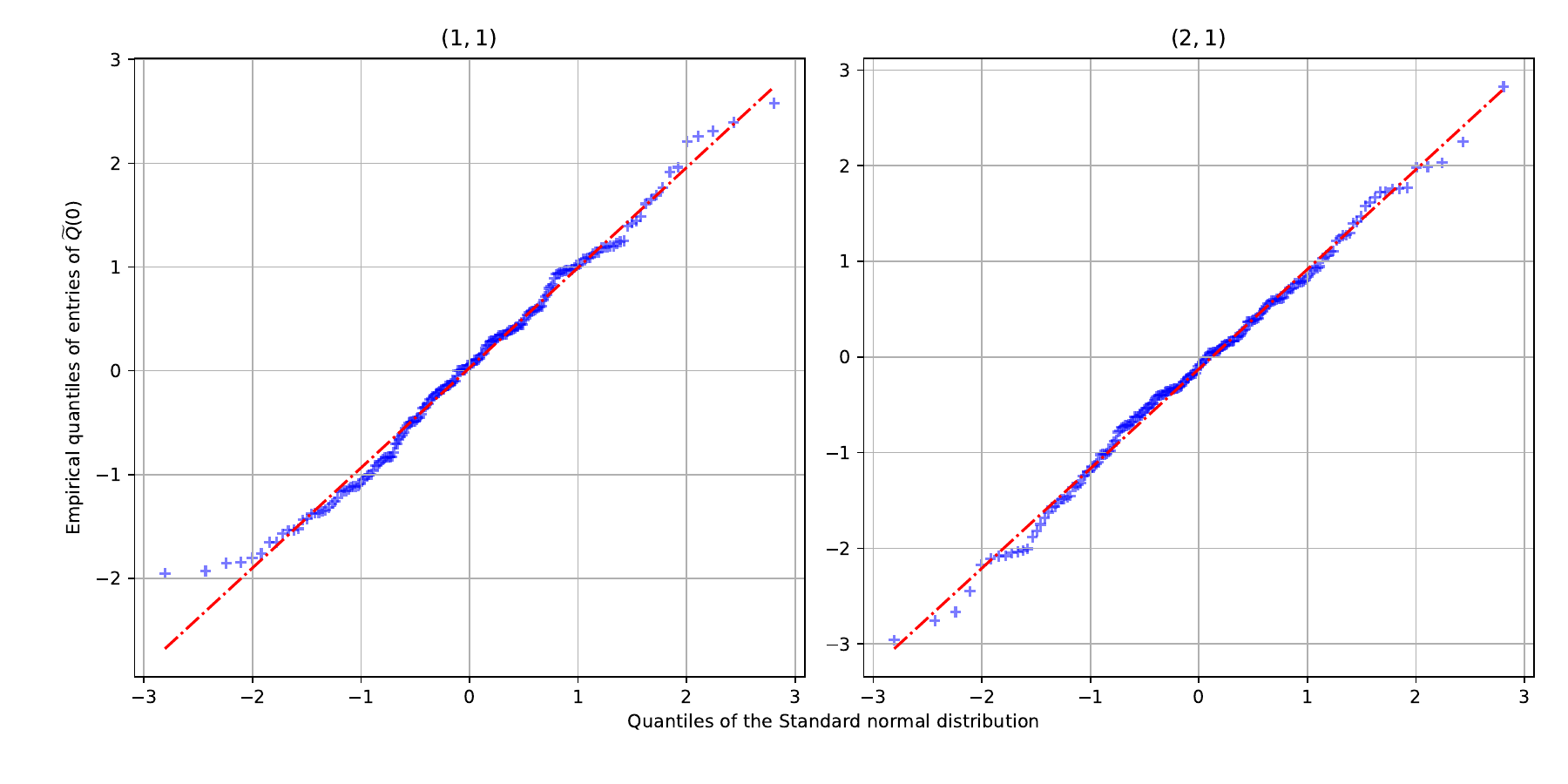}
        \caption{Example \ref{example_nc}, $d=6$}
    \end{subfigure}
    \caption{Q-Q plots of $\tQ_{11}(0)$ and $\tQ_{21}(0)$ with parameters $p=5, \delta = 0, n=200$. (a) Example~\ref{example_c}, $d=5$; (b) Example~\ref{example_nc}, $d=6$.}
    \label{fig: QQ_clt}
\end{figure}

We then turn our attention to Theorem \ref{thm: test_dist}, which characterizes the asymptotic null distribution of the test statistic $\hat{\cT}_{\rho}$, and Theorem \ref{thm: power}, which focuses on the power of the test. Figure \ref{fig: test_c} illustrates the results for Example \ref{example_c} with parameters $n=200$, $p\in \cbr{1,3,5}$, $d \in \cbr{5,10}$; Figure \ref{fig: test_nc} presents results for Example \ref{example_nc} with parameters $n=200$, $p\in \cbr{1,3,5}$, $d \in \cbr{6,10}$. The Q-Q plots demonstrate a linear fit with slope $1$ and intercept $0$.
Additionally, the figures indicate that as $p$ and $d$ increase, the test statistic also increases. The power of the test quickly approaches $1$ as the effect size $\delta$ grows.
We compare the power of our test with that of the distance covariance test ($\mathrm{dcov}$) \citep{dcov07} for testing independence, using the Python $\mathrm{dcor}$ package \citep{dcor_package}. Given that the distance covariance test is nonparametric, our test is expected to outperform $\mathrm{dcov}$ in this context.


\begin{figure}[!htb]
    \centering
    \begin{subfigure}[t]{0.9\textwidth}
        \centering
        \includegraphics[width=\textwidth]{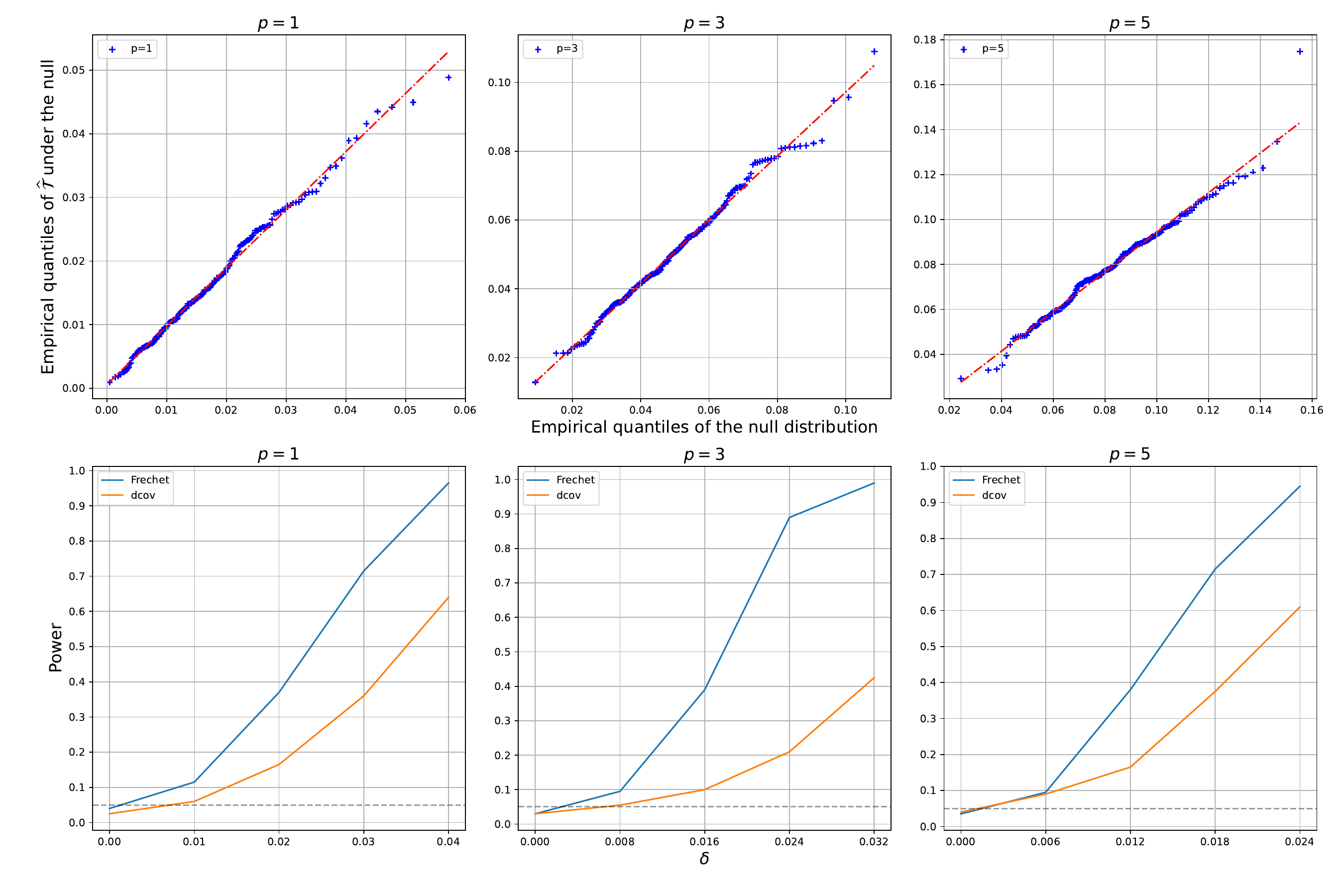}
        \caption{$d=5$}
    \end{subfigure}
    \\
    \begin{subfigure}[t]{0.9\textwidth}
        \centering
        \includegraphics[width=\textwidth]{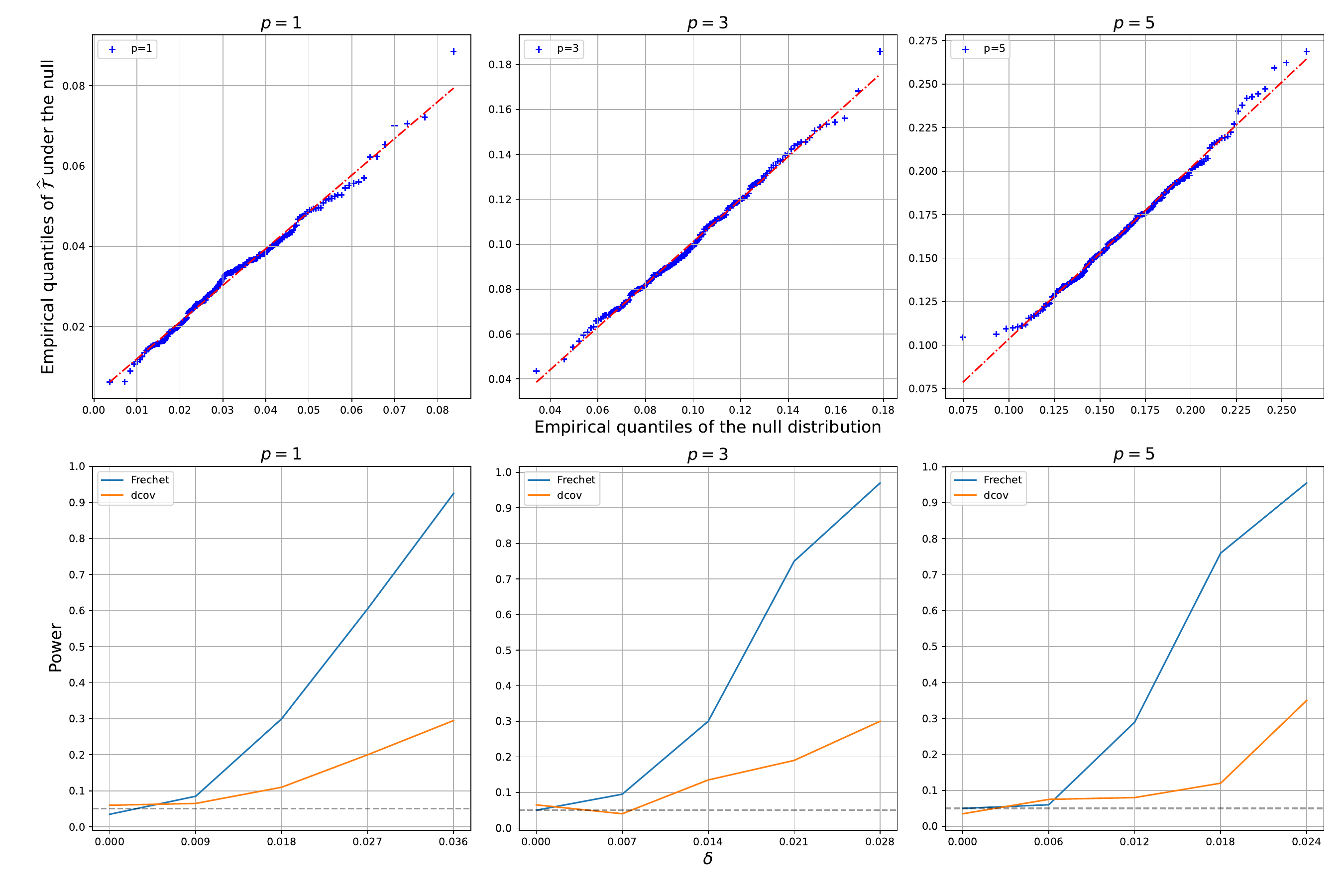}
        \caption{$d=10$}
    \end{subfigure}
    \caption{Q-Q plots of the test statistic $\hat{\cT}_{\rho}$ against its asymptotic null distribution and power curves compared to $\mathrm{dcov}$ as a function of the effect size $\delta$ for Example \ref{example_c}, with parameters $n=200$, $p \in \cbr{1,3,5}$. (a) $d=5$; (b) $d=10$.}
    \label{fig: test_c}
\end{figure}

\begin{figure}[!htb]
    \centering
    \begin{subfigure}[t]{0.9\textwidth}
        \centering
        \includegraphics[width=\textwidth]{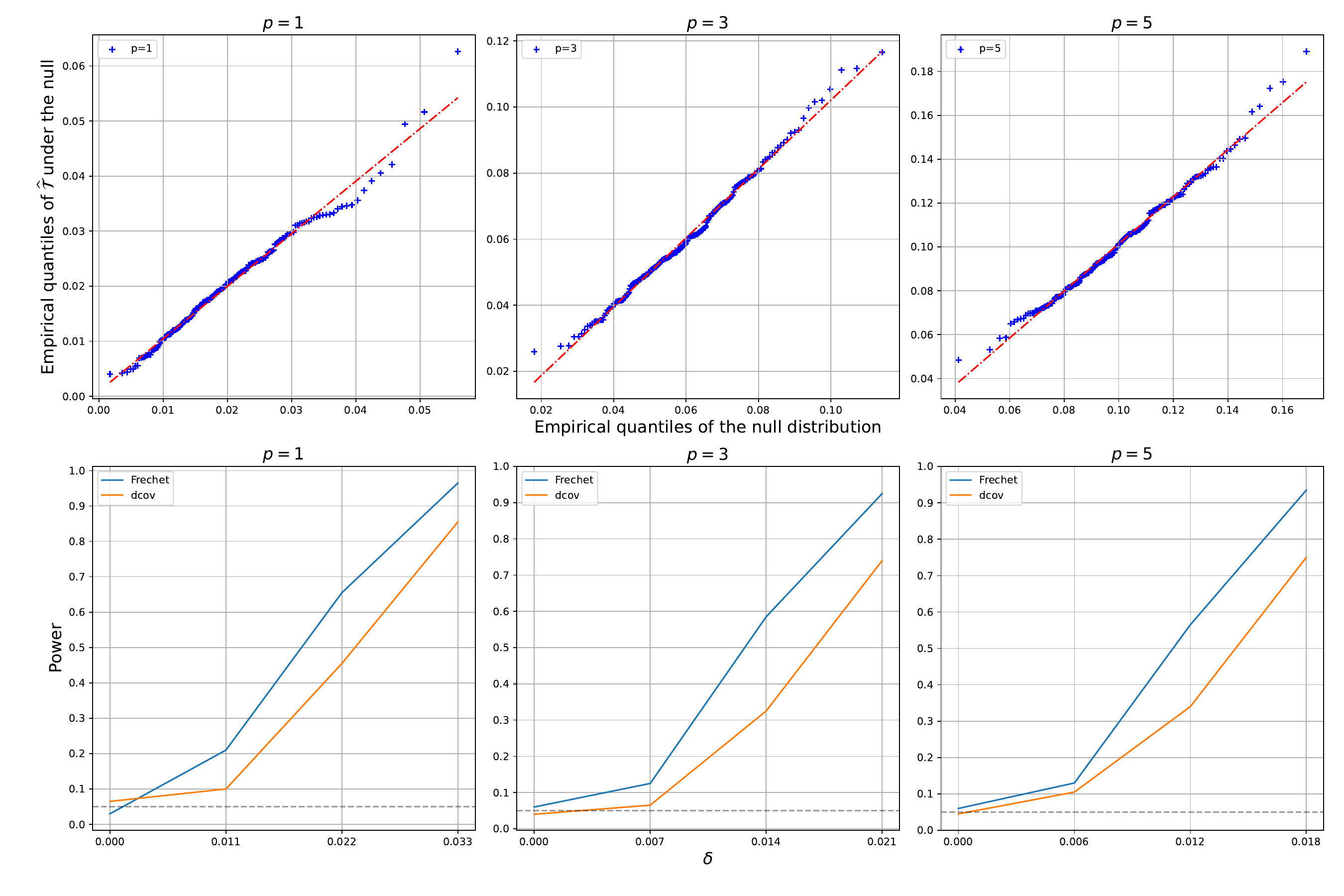}
        \caption{$d=6$}
    \end{subfigure}
    \\
    \begin{subfigure}[t]{0.9\textwidth}
        \centering
        \includegraphics[width=\textwidth]{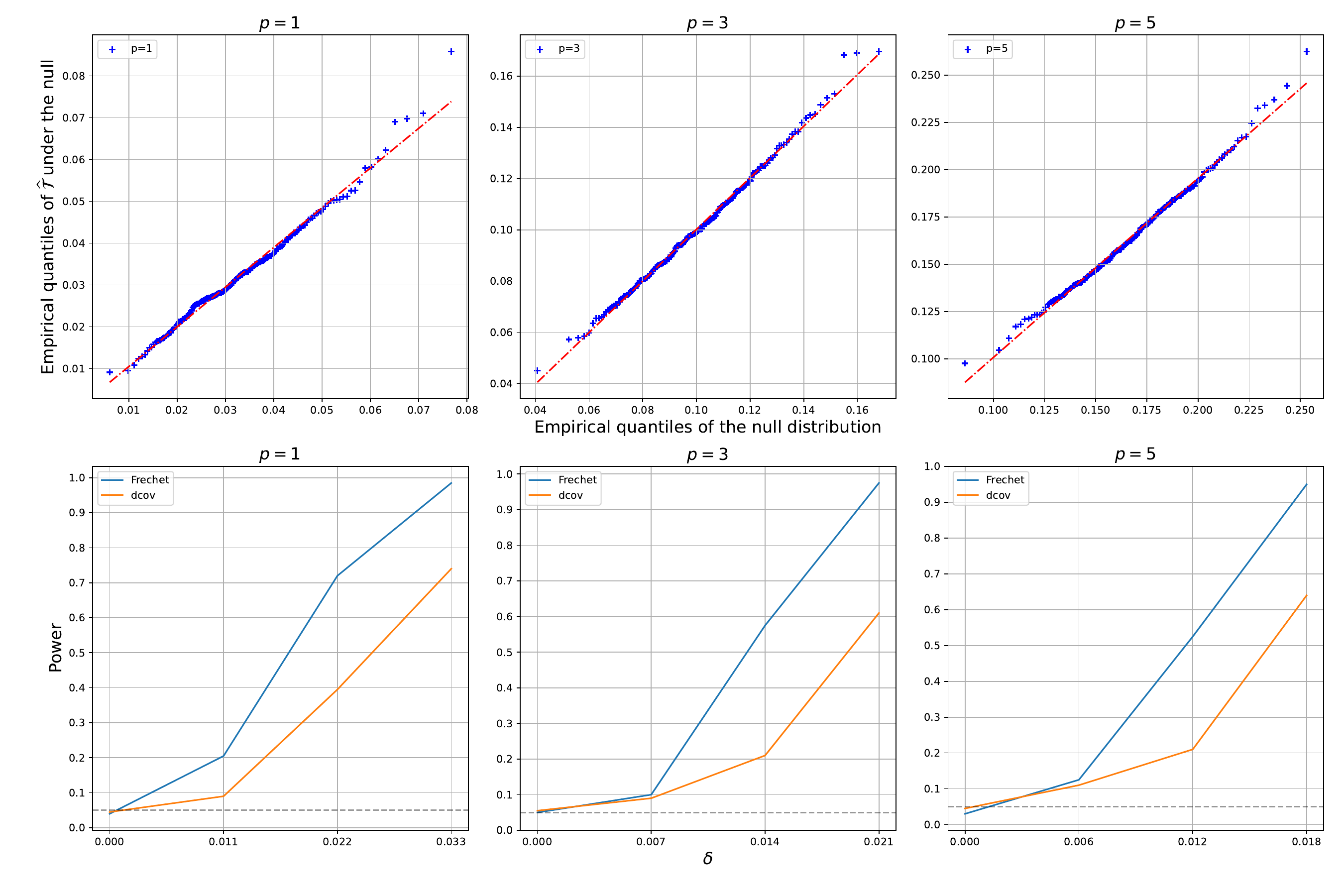}
        \caption{$d=10$}
    \end{subfigure}
    \caption{Q-Q plots of the test statistic $\hat{\cT}_{\rho}$ against its asymptotic null distribution and power curves compared to $\mathrm{dcov}$ as a function of the effect size $\delta$ for Example \ref{example_nc}, with parameters $n=200$, $p \in \cbr{1,3,5}$. (a) $d=6$; (b) $d=10$.}
    \label{fig: test_nc}
\end{figure}

\FloatBarrier

\subsection{Robustness of the results when covariance matrices are estimated}
\label{subsec: partial}
In this section, we examine the numerical performance of our results in a setting where $Q$ is not directly observed but is instead estimated from the data. Specifically, consider the case where we have access only to $\{(X_i;Z_{i1},Z_{i2},...,Z_{i\tn})\}_{i \in [n]}$ with $Z_{i1}, Z_{i2},..., Z_{i \tn} \overset{i.i.d.}{\sim} \cN(0,Q_i)$. In this scenario, a natural plug-in approach for Fr\'echet estimation and testing is to estimate $Q_i$ using the sample covariance, defined as $\bar{Q}_i := \tn^{-1}\sum_{j=1}^{\tn} Z_{i1} Z_{i1}^\top$. This estimate $\bar{Q}_i$ is then substituted for $Q_i$ in downstream estimators.
Specifically, we define the estimator $\hQ_{\rho,\tn}(x)$ and test statistic $\hat{\cT}_{\rho,\tn}$ with $\rho = n^{-1}$ as follows.
\begin{align*}
    \hQ_{\rho,\tn}(x) &= \argmin_{S \in \cS_d^{++}} \frac{1}{n} \sum_{i=1}^n w_{n,\rho}(x,X_i) W^2(S,\bar{Q}_i)\\
    \hat{\cT}_{\rho,\tn} &= \sum_{i=1}^n \Fnorm{\hat{H}_{\rho, \tn}\cdot \left( \hQ_{\rho,\tn}(X_i) - \hQ_{\rho,\tn}(\bar{X}) \right)}^2, \quad \text{where} \quad \hat{H}_{\rho, \tn}=-\frac{1}{n} \sum_{i=1}^n dT_{\hQ_{\rho,\tn}(\bar{X})}^{\bar{Q}_i}
\end{align*}
The asymptotic null distribution of $\hat{\cT}_{\rho,\tn}$ is computed as $\sum_{j} \hlambda_{\tn, j} w_i$ where $w_i$ are i.i.d. $\chi^2_p$ random variables and $\hlambda_{\tn,j}$ are the eigenvalues of the following estimated operator:
\begin{align*}
    \frac{1}{n} \sum_{i=1}^n \rbr{T_{\hQ_{\rho,\tn}(\bar{X})}^{\bar{Q}_i} - I_d} \otimes \rbr{T_{\hQ_{\rho,\tn}(\bar{X})}^{\bar{Q}_i} - I_d}
\end{align*}
We denote by $\hq_{1-\alpha,\tn}$ the $1-\alpha$ quantile of $\sum_{j} \hlambda_{\tn, j} w_i$. This quantile is then used to construct a test:
\begin{align*}
    \Phi_{\rho,\tn,\alpha}= \1 \rbr{\hat{\cT}_{\rho,\tn} > \hq_{1-\alpha,\tn}}
\end{align*}

Similarly, the covariance operator $\hat{\Xi}_{x,\tn}$ is computed as 
\begin{align}
    \hat{\Xi}_{x,\tn} = \frac{1}{n} \sum_{i=1}^n V_{x,\tn,i} \otimes V_{x,\tn,i} 
    \label{eqn: cov_nts}
\end{align}
where
\begin{align*}
    V_{x,i} &= V_{x,\tn, 1,i} + V_{x,\tn, 2,i}\\
    V_{x,\tn ,1,i} &= w_{n,\rho}(x,X_i) \rbr{ T^{\bar{Q}_i}_{\hQ_{\rho,\tn}(x)} - I_d }\\
    V_{x,\tn, 2,i} &= - \rbr{ \vec{x}^\top \hat{\vSigma}_{\rho}^{-1} (\vec{X}_i \vec{X}_i^\top - \hat{\vSigma})\hat{\vSigma}_{\rho}^{-1} \otimes I_d} \cdot \rbr{ \frac{1}{n}\sum_{j=1}^n \vec{X}_j \otimes (T^{\bar{Q}_j}_{\hQ_{\rho,\tn}(x)} - I_d)}
\end{align*}

First, we examine the central limit theorem (Theorem \ref{thm: estimation}) for $(X_i,Q_i)_{i \in [n]}$ satisfying the conditions in Example \ref{example_c} with parameters $\tn \in \cbr{50, 100, +\infty}$, $n=200$, $p=5$, $d=5$ and $ \delta=0$. Here $\tn=+\infty$ indicates $Q_i$ is directly observed.
Using the estimated covariance operator $\hat{\Xi}_{x,\tn}$ defined in (\ref{eqn: cov_nts}), we define the normalized error as:
\begin{align*}
    \sbr{\tQ_{\rho, \tn}(x)}_{i,j} := \frac{\sqrt{n} \left[ \hQ_{\rho, \tn}(x) - Q^*(x) \right]_{ij}}{\sqrt{\hat{v}_{x,\tn;ij}}}
\end{align*}
where $\hat{v}_{x,\tn;ij}$ is extracted from $\hat{\Xi}_{x,\tn}$ with appropriate indexing. Figure \ref{fig: partial_clt} shows the Q-Q plots of $[\tQ_{n,\tn}(x)]_{ij}$ against the standard normal distribution. Generally, as $\tn \to \infty$, the Q-Q plots for finite $\tn$ approach those of $\tn = \infty$, and can still be approximated by a straight line. This observation suggests that the central limit theorem for the entries still holds. Interestingly, the detailed behavior of the bias differs between diagonal and off-diagonal entries. For off-diagonal entries, Figure \ref{fig: partial_clt_offdiag} indicates that $[\hQ_{\rho,\tn}(0)]{ij}$ is unbiased. In contrast, for diagonal entries, Figure \ref{fig: partial_clt_diag} reveals that the Q-Q plots for finite $\tn$ are negatively shifted compared to that of $\tn = \infty$ (blue), suggesting that $[\hQ_{\rho,\tn}(0)]_{ii}$ is negatively biased. In this paper, we focus on the theoretical properties of the Fr\'echet regression estimator (\ref{eqn: frechet_estimate}) in the setting where covariance matrices are directly observed. We leave a detailed theoretical investigation of its refined behavior in the setting where covariance matrices are estimated for future work.

\begin{figure}[!htb]
    \centering
    \begin{subfigure}{\textwidth}
        \includegraphics[width=\textwidth]{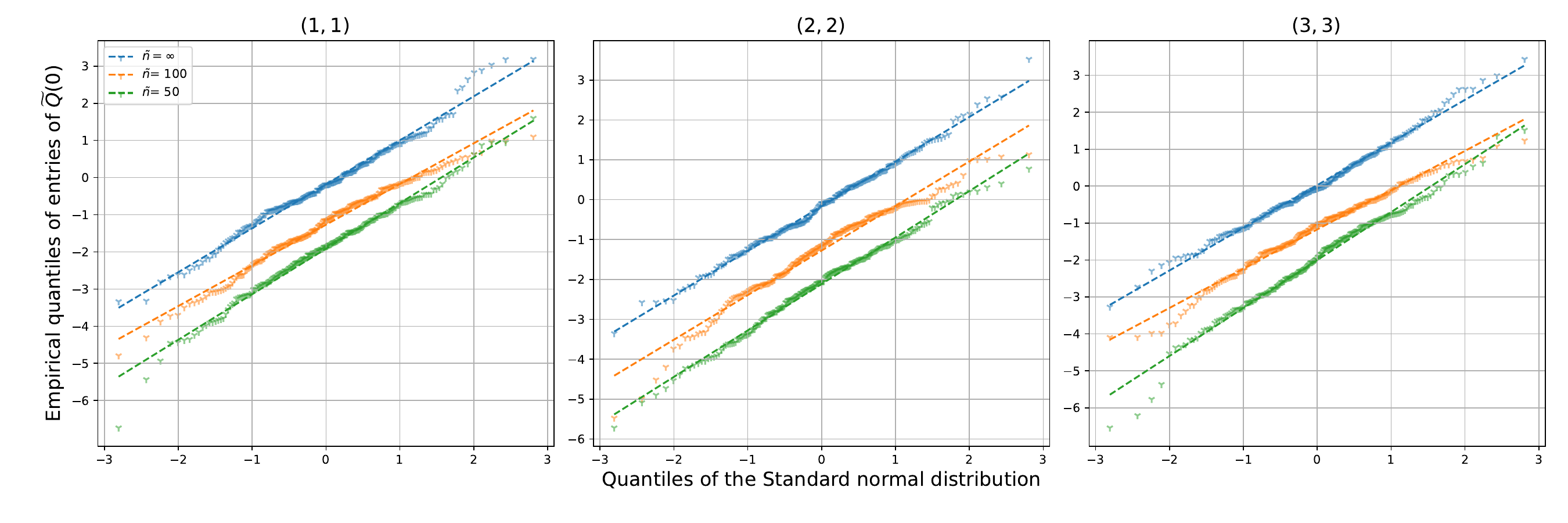}
        \caption{Diagonal entries: $(1,1), (2,2), (3,3)$}
        \label{fig: partial_clt_diag}
    \end{subfigure}
    \\
    \begin{subfigure}{\textwidth}
        \includegraphics[width=\textwidth]{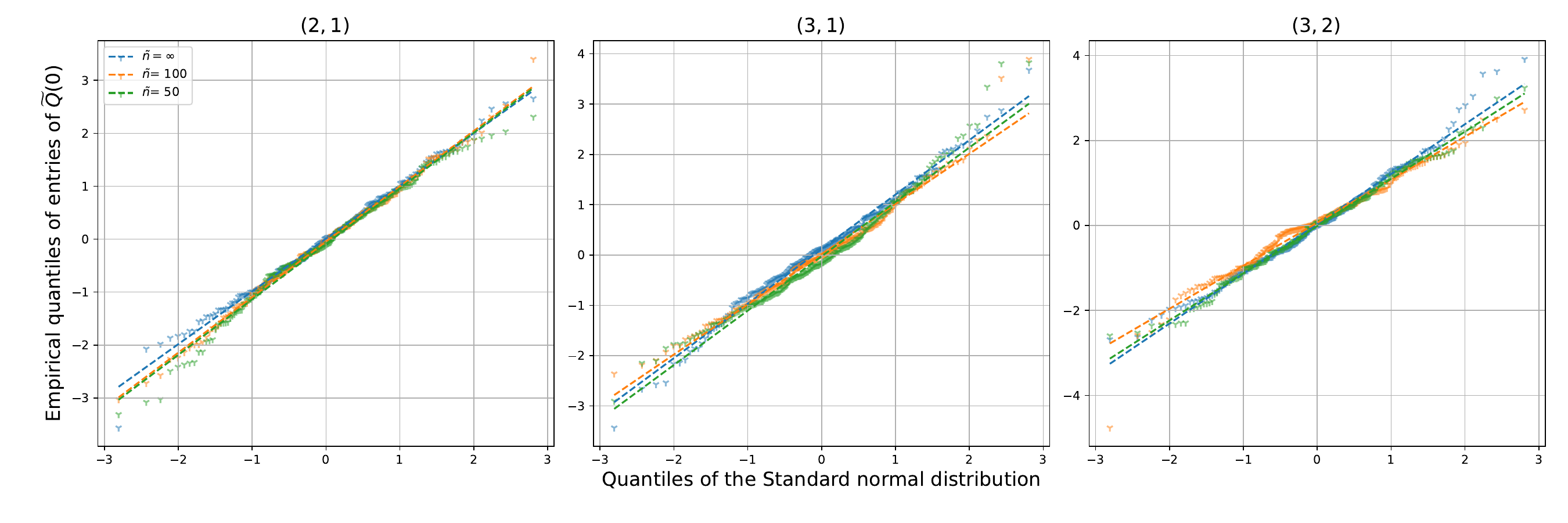}
        \caption{Off-diagonal entries: $(2,1), (3,1), (3,2)$}
        \label{fig: partial_clt_offdiag}
    \end{subfigure}
            
    \caption{Q-Q plots of diagonal and off-diagonal entries of $\tQ_{\rho,\tn}(0)$ with parameters $d=5, n=200$ and varying $\tn \in \{50, 100, \infty\}$.}
    \label{fig: partial_clt}
\end{figure}

Figure \ref{fig: test_nts} shows the Q-Q plot of $\hat{\cT}_{\rho,\tn}$ against $\sum_{j} \hlambda_{\tn, j} w_i$ and the power of the test~$\Phi_{\rho,\tn,\alpha}$ for $(X_i,Q_i)$ satisfying the conditions of Example \ref{example_c} with parameters $\tn \in \cbr{50,100,+\infty}$, $n =200$, $p=5$ and $d=5$. Notably, the Q-Q plots exhibit a slope of $1$ and an intercept of $0$, indicating that $\hat{\cT}_{\rho,\tn}$ follows the asymptotic null distribution $\sum_{j} \lambda_{\tn,j} w_i$, where $w_i$ are i.i.d. $\chi^2_p$ random variables and $\lambda_{\tn,j}$ are the eigenvalues of the following operator:
\begin{align*}
    \EE \sbr{\rbr{T_{Q^*(\mu)}^{\bar{Q}} - I_d} \otimes \rbr{T_{Q^*(\mu)}^{\bar{Q}} - I_d}}
\end{align*}
where $\bar{Q}= \Cov Z$.
This observation suggests that the plug-in method remains valid in this context. The power plot indicates that while the power of the test decreases when $Q_i$ is not directly observed, our test still demonstrates higher power compared to the distance covariance test ($\mathrm{dcov}$) in this setting.

\begin{figure}[!htb]
    \centering
    \includegraphics[width=.9\textwidth]{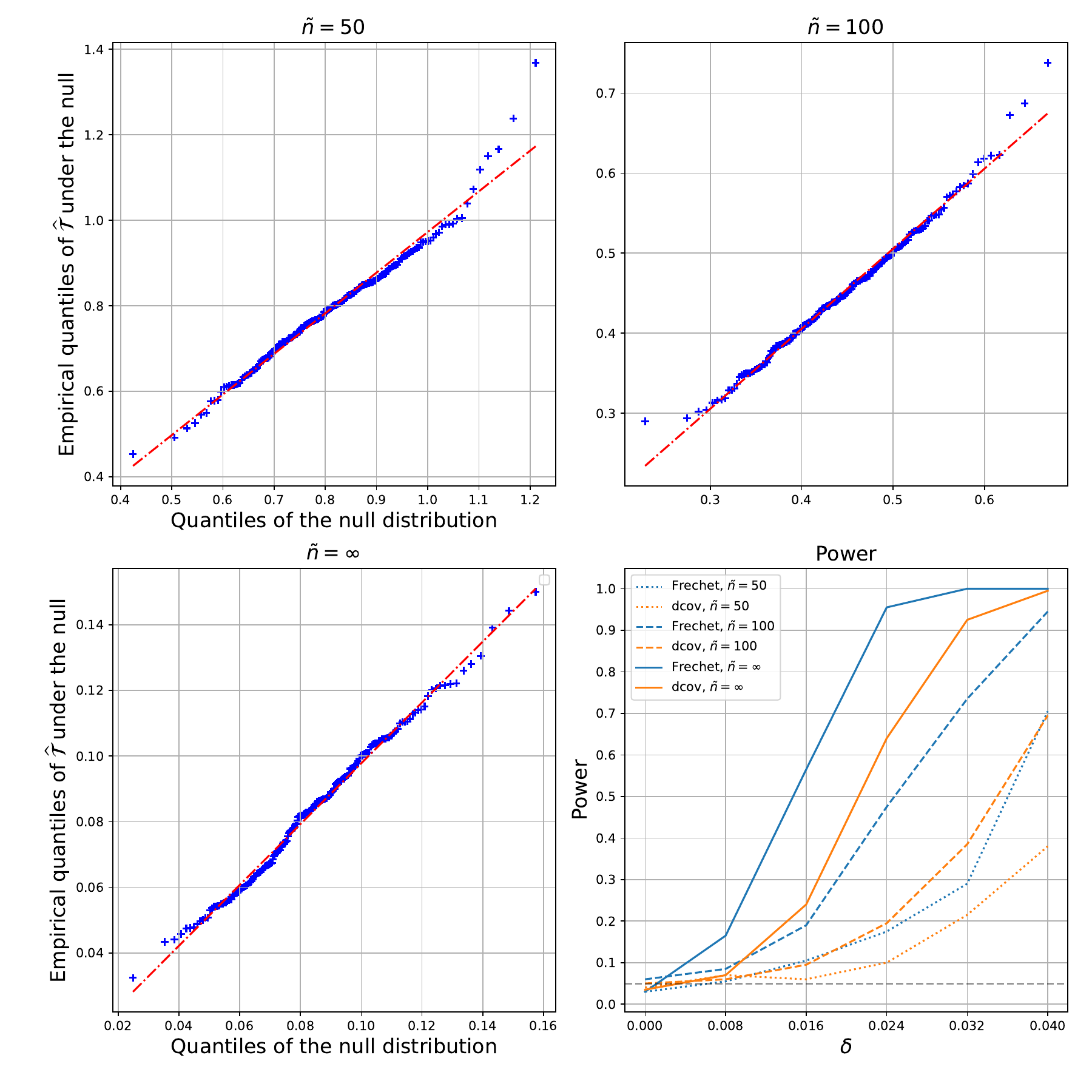}
    \caption{Q-Q plots of the test statistic $\hat{\cT}_{\rho,\tn}$ against the estimated plug-in null distribution $\sum_{j} \hlambda_{\tn, j} w_i$ and power curves of $\Phi_{\rho,\tn,\alpha}$ compared to $\mathrm{dcov}$ as a function of the effect size $\delta$.}

    \label{fig: test_nts}
\end{figure}

            
    

%% file: single_cell.tex
Aging is a complex process of accumulation of molecular, cellular, and organ damage, leading to loss of function and increased vulnerability to disease and death. Nutrient-sensing pathways, namely insulin/insulin-like growth factor signaling and target-of-rapamycin can substantially increase healthy life span of laboratory model organisms \citep{aging, aging1}. These nutrient signaling pathways are conserved in various organisms. We are interested in understanding the co-expression structure of 61 genes in this KEGG nutrient-sensing pathways based on the recently published population scale single cell RNA-seq data of human peripheral blood mononuclear cells (PBMCs) from blood samples of over 982 healthy individuals with ages ranging from 20 to 90 \citep{yazar2022single}.

We focus our analysis on CD4+ naive and central memory T (CD4NC) cells, which is the most common cell type observed in the data. Age-associated changes in CD4 T-cell functionality have been linked to chronic inflammation and decreased immunity \citep{aging2}.  There are a total of 51 genes that are expressed in this cell type.  Even though the Fr\'echet regression still makes sense when the covariance matrix is potentially degenerate (see the remarks after Example \ref{example_nc}), our theory relies on the strict positive definiteness. Hence, we retain only the genes that have nonzero variances at any age, resulting in a total of $37$ genes, see Figure \ref{fig: QQ1-1} for a concise overview of these covariance matrices  for individuals at different ages, showing difference across different ages.  

In genetics research, such covariance matrices represent individual-specific gene co-expression networks. We are interested in testing whether such networks are associated with ages by 
testing whether there is an age effect on the gene expression covariance matrices, i.e. $\cH_0: Q^*(t) \equiv Q^* $ for some $Q^*$. 
The test we propose yields a p-value of 0.00019. When the analysis is performed separately for males and females, we obtain a p-value of 0.00034 for the male group and 0.00012 for the female group, suggesting a strong age effect on the gene expression covariance matrices.

\begin{figure}[!htb]
    \centering
    \includegraphics[width=\textwidth]{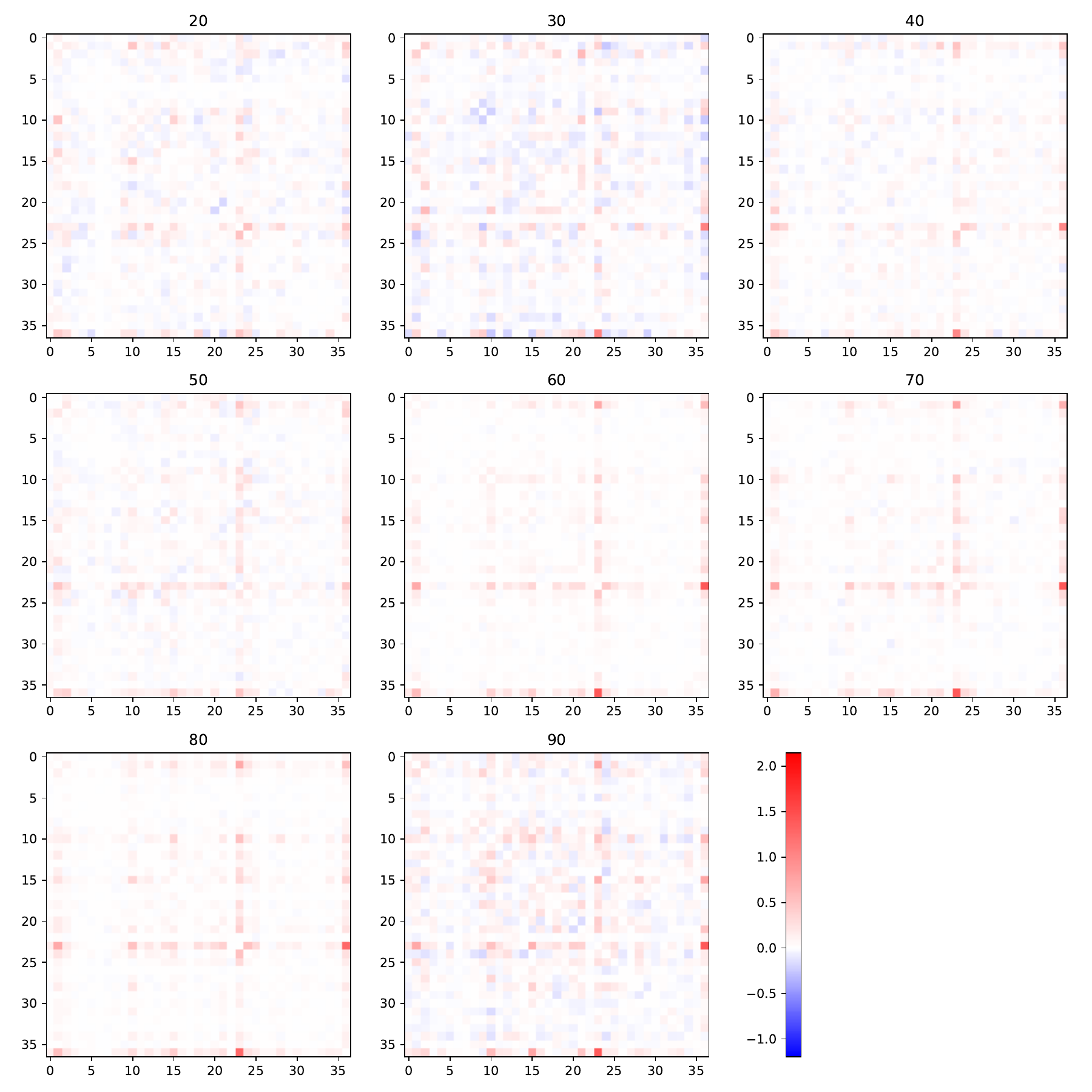}
    \caption{Heatmap of the gene expression covariance matrices with diagonal elements omitted for individual at age $20,30,\ldots,90$, respectively. }
    \label{fig: QQ1-1}
\end{figure}

%% file: discuss.tex
We have develop methods for statistical inference for the Fr\'echet regression on the Bures-Wasserstein manifold, where covariance matrix is treated as the outcome, including the uniform rate of convergence of the conditional Fr\'echet mean and the asymptotic distribution. Based on these results, we have further developed statistical test for testing the association between covariate outcome and Euclidean covariates. These results are further verified using simulations. We have demonstrated the methods by testing the association between gene co-expression and age, indicating the change of co-expressions among a set of genes in nutrient sensing pathway.  The proposed methods have other applications, including in neuroimaging data analysis, where
covariance matrices (or correlation matrices after standardization) of multiple brain regions are used to summarize  as functional connectivity matrices. The proposed methods can be used to identify the factors that are associated with such functional connectivity matrices. 

In this paper,  we assume that the outcome covariance matrices are observed and we only focus on the theoretical properties of the Fr\'echet regression estimator (\ref{eqn: frechet_estimate}) under the complete observation setting.  This is also the setting considered in \cite{muller16} and \cite{petersen21_WFtest}. 
An important direction for future work is to develop  the corresponding theoretical results for scenarios where the covariance matrices must be estimated from the data. Additionally, it would be valuable to relax the assumption on the eigenvalue lower bound and develop methods that can handle singular matrices.


%% file: appx_bg.tex
In this section, we collect relevant background about optimal transport and functional calculus to make the paper more self-contained.

\subsection{Geometry of optimal transport}
\label{subsec: appx_bg_ot}

We begin with the geometry of optimal transport, and then specialize the general concepts to the Bures–Wasserstein manifold.
For introductory expositions of optimal transport, we refer to \citet{villani03,santa15,invitation}. For a more comprehensive treatment, we refer to \citet{ambrosio,villani09}.

Given $\mu_0, \mu_1 \in \cP_{2, \mathrm{ac}}(\RR^d)$, the constant-speed geodesic $\rbr{\mu_{t}}_{t \in [0,1]}$ connecting $\mu_0$ to $\mu_1$ is characterized by
\begin{equation*}
    \mu_t = \sbr{\id + t \rbr{T_{\mu_0}^{\mu_1} - \id}}_{\char"0023} \mu_0, \qquad t \in [0,1]
\end{equation*} 
Here, $\char"0023$ denotes the pushforward operation defined by $T_{\char"0023} \mu (E) = \mu (T^{-1}(E))$ for any Borel set $E \subset \RR^d$.
Then define the tangent vector of the geodesic $(\mu_t)_{t \in [0,1]}$ at $t=0$ to be the mapping $T_{\mu_0}^{\mu_1}- \id$. The tangent space $T_{\mu_0} \cP_{2,\mathrm{ac}}(\RR^d)$ to $\cP_{2,\mathrm{ac}}(\RR^d)$ at $\mu_0$ is defined in \citealp[Thm 8.5.1]{ambrosio} as 
\begin{equation*}
    T_{\mu_0} \cP_{2, \mathrm{ac}}(\RR^d):={\overline{\cbr{\lambda(T_{\mu_0}^{\nu}-\id): \lambda>0, \nu \in \mathcal{P}_{2, \mathrm{ac}}(\RR^d)}}}^{L^2(\mu_0)}
\end{equation*}
Here the overline denotes closure with respect to the $L^2(\mu_0)$ measure.

Given two covariance matrices $Q, S \in \cS^{++}_d$, the constant-speed geodesic connecting the corresponding centered Gaussians is given by
\begin{equation*}
    \rbr{I_d + t \rbr{T_Q^S - I_d}}Q\rbr{I_d + t \rbr{T_Q^S - I_d}}, \quad t \in [0,1]
\end{equation*}
The tangent space $T_{Q}\cS^{++}_d$ can be identified with the space $\cS_d$ of symmetric $d \times d$ matrices. For any $\tS \in T_{Q}\cS^{++}_d$, its norm in the tangent space is given by
\begin{equation*}
    \norm{\tS}_{Q}:=\inner{\tS}{Q \tS}^{1/2}
\end{equation*}

\subsection{Functional calculus}
\label{subsec: appx_bg_func}

To consider higher order differentials of $T^S_Q$ on the Bures-Wasserstein manifold, we give a brief review of some key concepts in functional calculus that are essential for the development, and direct readers to \citet{dudley} for further details.

Let $Y$ and $Z$ be normed spaces with the norm on each denoted by $\norm{\cdot}$, and let $U$ be an open subset of $Y$. Let $L(Y,Z)$ denote the space of all bounded linear operators from $Y$ into $Z$. A function $\phi: U \to Z$ is called (Fr\'echet) differentiable at $u \in U$ if there exists an $L(u) \in L(Y,Z)$ such that for each $y \in Y$ with $u+Y \in Y$,
\begin{equation*}
    \lim_{y \to 0} \frac{\norm{\phi(u+y) - \phi(u) - L(u) y  }}{\norm{y}} = 0
\end{equation*}
The linear operator $L(u)$ is unique, is called the derivative of $\phi$ at $u$, and is denoted by $D\phi(u)$. If the function $\phi$ is differentiable at each $u \in U$, then we say that $\phi$ is differentiable on $U$.

To consider higher order derivatives of $\phi$, let $L^1(Y,Z):= L(Y,Z)$ with the usual operator norm, and let $L^{k+1}(Y,Z):=L(Y,L^k(Y,Z))$ with the operator norm recursively for $k=1,2,\ldots$. For $k\geq 2$, we say that $\phi$ is (Fr\'echet) differentiable of order $k$ at $u$ if $\phi$ has a $(k-1)$st derivative $D^{k-1}\phi(y)$ at each point $y$ of some neighborhood of $u$, and the mapping $D^{k-1}\phi$ is differentiable at $u$. Then $D^k \phi(u)$, the $k$th derivative of $\phi$ at $u$, is defined as the derivative of $D^{k-1}\phi$ at $u$. If $\phi$ is differentiable of order $k$ at $u$ for each $u \in U$ then we say that $\phi$ is differentiable of order $k$ on $U$. A mapping $\phi$ is called a $C^k$ function on $U \subset Y$ if the derivatives of $\phi$ through order $k$ all exist on $U$ and are continuous. A mapping $\phi$ is called a $C^{\infty}$ function on $U \subset X$ if it is a $C^k$ function for each $k$.

It would be convenient to introduce the notion of multilinear mappings in order study the properties of higher order derivatives $D^k \phi$. A function $A: Y^k \to Z$ is called $k$-linear if for each $j\in [k]$, $A(y_1,\ldots,y_k)$ is linear in $y_j$ for any fixed values of $y_i$, $i \neq j$. The function $A$ is called bounded if
\begin{equation}
    \norm{A}:= \sup\left\{ \norm{A(y_1,\ldots,y_k)}: \norm{y_j} \leq 1, j \in [k] \right\}<\infty
    \label{eqn: prelim_mlinear_norm}
\end{equation}
Let $\cM_k(Y,Z)$ be the set of all bounded $k$-linear maps from $Y^k$ to $Z$ with norm define by (\ref{eqn: prelim_mlinear_norm}). 

The space $L^k(Y,Z)$ can be identified with $\cM_k(Y,Z)$ through the natural isomorphism $\Phi_{(k)}: L^k(Y,Z) \to \cM_k(Y,Z)$ defined by 
\begin{equation*}
    \begin{aligned}
        \Phi_{(k)}(A)(y_1, \ldots, y_k) & :=A(y_1)(y_2) \cdots(y_k) \\
        & :=\sbr{\cdots \sbr{A(y_1)}(y_2) \cdots}(y_k)
    \end{aligned}
\end{equation*}
Then the $k$th differential at $u$ is defined by $d^k \phi(u):= \Phi_{(k)}(D^k \phi(u))$.

The $k$th differential $d^k \phi(u)$ is symmetric in the sense that
\begin{equation*}
    d^k\phi(u)(y_{\sigma(1)},\ldots,y_{\sigma(k)}) = d^k \phi(u)(y_1,\ldots,y_k)
\end{equation*}
for any permutation $\sigma$ of $[k]$. For simplicity, we also denote $d^k \phi(u)(y,\ldots,y) $ by $d^k \phi(u) \cdot y^{\otimes k}$.

Let $\cM_{k,s}(Y,Z)$ denote the subspace of all symmetric elements of $\cM_k(Y,Z)$. Then aside from the operator norm inherited from $\cM_{k}(Y,Z)$, one can define another norm $\vertiii{\cdot}$ on $\cM_{k,s}(Y,Z)$ by 
\begin{equation*}
    \vertiii{P} := \sup \left\{ \norm{P(y,y,\ldots,y)}: \norm{y} \leq 1 \right\} 
\end{equation*}

Properties of the high order differentials $d^k T_Q^S$ are investigated in Appendix \ref{subsec: Techlem_dkT}.

%% file: tech_lem_dT.tex
Recall that the optimal transport map between two centered Gaussian distributions $\cN(0,Q)$ and $\cN(0,S)$ has the closed-form expression $T_Q^S = S^{1/2}(S^{1/2}QS^{1/2})^{-1/2} S^{1/2}$. In all relevant analysis, we need to consider differentials of $T_Q^S$ when viewed as a function of $Q$ for fixed $S$, which we denote as $d^kT^S_Q$ for $k \geq 1$.
It is shown in \citet{aap21} that $dT^S_Q$ can be defined as follows.  For any $H \in \cS_d$,
\begin{equation*}
    dT_Q^S(H) :=-S^{1 / 2} U^\top \Lambda^{-1 / 2} \delta \Lambda^{-1 / 2} U S^{1 / 2}
\end{equation*}
where $U^\top \Lambda U$ is an eigenvalue decomposition of $S^{ 1/2} Q S^{ 1/2}$ with $UU^\top = U^\top U = I$ and $\delta=\left(\delta_{i j}\right)_{i, j=1}^d$ with 
\begin{equation*}
    \delta_{i j}= 
    \begin{cases}
        \frac{\Delta_{i j}}{\sqrt{\lambda_i}+\sqrt{\lambda_j}} & 
        i, j \leq \operatorname{rank}(S) \\ 
        0 & \text { otherwise }
    \end{cases}, \qquad
    \Delta=U S^{1 / 2} H S^{1 / 2} U^\top
\end{equation*}
To consider higher order differentials of $T^S_Q$, we start with the map  $Q\mapsto Q^{1/2}$ in Lemma \ref{lem: matrix_sqrt}, next move on to $Q\mapsto Q^{-1/2}$ in Lemma \ref{lem: mtx_inv_sqrt} which then leads to $T_Q^S$ in Lemma \ref{lem: diff}. Connections between $T_Q^S$ and $W^2(Q,S)$ are also investigated in Lemma \ref{lem: diff}.

\begin{lemma}
    The square root functional $\phi: Q \in \cS_d^{++} \to \phi(Q) = Q^{1/2} \in \cS_d^{++}$ is Fr\'echet differentiable at any order on $\cS_d^{++}$. Moreover, for any $Q \in \cS_d^{++}$ and  any $n\geq 0$, we have the estimates
    \begin{equation*}
        \begin{aligned}
            \opnorm{d^{n+1} \phi(Q)}  &\leq C_{d,n} \lambda_{\min}(Q)^{-(n+1/2)}
        \end{aligned}
    \end{equation*}
    Here $C_{d,n}=d^{n/2}\cdot n!\binom{2n}{n} \cdot 2^{-(2n+1)}$.
    \label{lem: matrix_sqrt}
\end{lemma}

\begin{proof}
    See \citet{18_mtrx_sqrt} Theorem 1.1.
    
\end{proof}

\begin{lemma}
    \label{lem: mtx_inv_sqrt}
    The inverse square root functional $\varphi: Q \in \cS_d^{++} \mapsto \varphi(Q) = Q^{-1/2} \in \cS_d^{++}$ is Fr\'echet differentiable at any order on $\cS_d^{++}$. Moreover, for any $A \in \cS_d^{++}, H \in \SS^d$, the following holds.
    \begin{equation*}
        \begin{aligned}
            d \varphi(Q) \cdot H &=- Q^{-1/2} \left( d \phi(Q) \cdot H \right) Q^{-1/2}\\
            d^2 \varphi(Q) \cdot H^{\otimes 2} &=- Q^{-1/2} \left( d^{2} \phi(Q) \cdot H^{\otimes 2} \right) Q^{-1/2} + 2 Q^{-1/2} (d \phi(Q) \cdot H) Q^{-1/2} (d \phi(Q) \cdot H) Q^{-1/2}\\
            d^3 \varphi(Q)\cdot H^{\otimes 3} &= -Q^{-1/2}\left( d^3 \phi(Q) \cdot H^{\otimes 3} \right)Q^{-1/2} + 3 Q^{-1/2}\left( d^2 \phi(Q) \cdot H^{\otimes 2} \right) Q^{-1/2} \left( d \phi(Q) \cdot H\right)Q^{-1/2}\\
            &\quad + 3 Q^{-1/2}\left( d \phi(Q) \cdot H\right) Q^{-1/2}\left( d^2 \phi(Q) \cdot H^{\otimes 2} \right) Q^{-1/2}
            \label{eqn: proof_lem_sqrt_inv}\\
            &\quad -6Q^{-1/2}\left( d \phi(Q) \cdot H\right)Q^{-1/2}\left( d \phi(Q) \cdot H\right)Q^{-1/2}\left( d \phi(Q) \cdot H\right)Q^{-1/2}
        \end{aligned}
    \end{equation*}
    with
    \begin{equation*}
        \begin{aligned}
            \opnorm{d\varphi(Q)} &\leq C_{d,0} \lambda_{\min}(Q)^{-3/2}\\
            \opnorm{d^2 \varphi(Q)} &\leq (C_{d,1} + 2C_{d,0}^2)  \cdot \left(  \lambda_{\min}(Q) \right)^{-5/2}\\
            \opnorm{d^3 \varphi(Q)} &\leq \left( C_{d,2}+ 6 C_{d,1}C_{d,0} + 6 C_{d,0}^3 \right)\lambda_{\min}(Q)^{-7/2}
        \end{aligned}
    \end{equation*}
    where $C_{d,n}$ is defined in Lemma \ref{lem: matrix_sqrt}.

\end{lemma}

\begin{proof}
    By Lemma $\ref{lem: matrix_sqrt}$ we have for infinitesimal $H \in \cS_d$ that
    \begin{equation*}
        \begin{aligned}
            \left( Q+H \right)^{1/2}= Q^{1/2} + \underbrace{d \phi(Q) \cdot H}_{Z_1} + \underbrace{\frac{1}{2} d^2 \phi(Q) \cdot H^{\otimes 2}}_{Z_2} + \underbrace{\frac{1}{6}d^{3} \phi(Q)\cdot H^{\otimes 3}}_{Z_3} + \cdots
        \end{aligned}
    \end{equation*}
    Let $Z=Z_1+Z_2+Z_3+Z_4$, and we obtain for infinitesimal $H \in \cS_d$ that
    \begin{align*}
        &(Q+H)^{-1/2}\\ 
        &= \left( Q^{1/4} (I_d + Q^{-1/4} Z Q^{-1/4}) Q^{1/4}  \right)^{-1}\\
        &= Q^{-1/4} (I_d + Q^{-1/4} Z Q^{-1/4})^{-1} Q^{-1/4}\\
        &\overset{(i)}{=} Q^{-1/4} \left( I_d-  Q^{-1/4} Z Q^{-1/4} + Q^{-1/4} Z Q^{-1/2}ZQ^{-1/4} - Q^{-1/4} Z Q^{-1/2}ZQ^{-1/2}ZQ^{-1/4} + \cdots\right)Q^{-1/4}\\
        &=Q^{-1/2} - Q^{-1/2}ZQ^{-1/2} + Q^{-1/2} Z Q^{-1/2}ZQ^{-1/2} - Q^{-1/2} Z Q^{-1/2}ZQ^{-1/2}ZQ^{-1/2} + \cdots\\
        &\overset{(ii)}{=} Q^{-1/2} - \underbrace{Q^{-1/2}Z_1Q^{-1/2}}_{\text{1st order in } H} - \underbrace{Q^{-1/2}Z_2Q^{-1/2} + Q^{-1/2} Z_1 Q^{-1/2}Z_1 Q^{-1/2}}_{\text{2nd order in } H}\\
        & \underbrace{- Q^{-1/2}Z_3 Q^{-1/2} + Q^{-1/2}Z_2Q^{-1/2} Z_1Q^{-1/2}+Q^{-1/2}Z_1Q^{-1/2} Z_2Q^{-1/2} - Q^{-1/2} Z_1 Q^{-1/2}Z_1Q^{-1/2}Z_1Q^{-1/2}}_{\text{3rd order in }  H} \\
        &+ \text{higher order terms in } H
    \end{align*}
    Here (i) follows from the von Neumann series expansion, and (ii) is obtained by arranging terms according to their orders in $H$. Then (\ref{eqn: proof_lem_sqrt_inv}) follows,
    and we have the following estimate
    \begin{equation*}
        \begin{aligned}
            \Fnorm{d^2 \varphi(Q) \cdot H^{\otimes 2}} &\leq \Fnorm{Q^{-1/2} \left( d^{2} \phi(Q) \cdot H^{\otimes 2} \right) Q^{-1/2}} + 2\Fnorm{Q^{-1/2} (d \phi(Q) \cdot H) Q^{-1/2} (d \phi(Q) \cdot H) Q^{-1/2}}\\
            &\leq \lambda_{\min}(Q)^{-1} \cdot \opnorm{d^2 \phi(Q)}\cdot  \Fnorm{H}^2 + 2 \lambda_{\min}(Q)^{-3/2} \opnorm{d\phi(Q)}^2 \cdot \Fnorm{H}^2\\
            &\leq (C_{d,1} + 2C_{d,0}^2)  \lambda_{\min}(Q)^{-5/2}  \cdot \Fnorm{H}^2
        \end{aligned}
    \end{equation*}
    and similarly
    \begin{align*}
        \left\|d^3 \varphi(Q) \cdot H^{\otimes 3} \right\|_F &\leq \lambda_{\min}(Q)^{-1} \opnorm{d^3 \varphi(Q)} \cdot \Fnorm{H}^3 + 6 \lambda_{\min}(Q)^{-3/2}\opnorm{d^2 \varphi(Q)} \cdot \opnorm{d \varphi(Q)}\cdot  \Fnorm{H}^3\\
        &\quad +6 \lambda_{\min}(Q)^{-2} \opnorm{d\varphi(Q)}^3 \cdot  \Fnorm{H}^3\\
        &\leq \left( C_{d,2}+ 6 C_{d,1}C_{d,0} + 6 C_{d,0}^3 \right)\lambda_{\min}(Q)^{-7/2} \cdot \Fnorm{H}^3
    \end{align*}
    Here the last inequality follows by applying Lemma \ref{lem: matrix_sqrt}.  And we arrive at the desired result.

\end{proof}

\begin{lemma}
    \label{lem: diff}
    The following properties hold for the 2-Wasserstein distance $W(Q,S)$ and the optimal transport map $T_Q^S$. For any $Q,Q_1,Q_2\in \cS_d^{++}$, $S \in \cS_d^+$ and $X,Y \in \cS_d$,
    
    \begin{enumerate}[leftmargin=*,label=\textbf{(\roman*)},ref=\textbf{(\roman*)}]
        \item \label{eqn: diff_W_sup}
        $W^2(Q,S)$ is upper bounded by
        \begin{equation*}
            W^2(Q,S) \leq 2d \left( \lambda_{\max}(Q) + \lambda_{\max}(S) \right)
        \end{equation*}
        
        \item \label{eqn: lem_diff_quad_approx}
        $W^2(Q,S)$ is twice differentiable with 
        \begin{equation*}
            \begin{aligned}
                &d_Q W^2(Q, S)(X)=\inner{I-T_Q^S}{X} \\
                & d_Q^2 W^2(Q, S)(X, Y)=-\inner{X}{dT_Q^S(Y)}
            \end{aligned}
        \end{equation*}
        Moreover, the following quadratic approximation holds:
        \begin{equation*}
            \begin{aligned}
                &\frac{2}{\rbr{1+\eigmaxx{1/2}{Q^{\prime}}}^2} \inner{-dT_{Q_0}^S(Q_1-Q_0)}{Q_1-Q_0} \\
                \leq &W^2(Q_1, S)- W^2(Q_0, S)+ \inner{T_{Q_0}^S-I}{Q_1-Q_0}\\
                \leq &\frac{2}{\rbr{1+\eigminn{1/2}{Q^{\prime}}}^2} \inner{-dT_{Q_0}^S(Q_1-Q_0)}{Q_1-Q_0}
            \end{aligned}
        \end{equation*}
        with $Q':= Q_0^{-1/2}Q_1 Q_0^{-1/2}$.
        
        \item \label{eqn: lem_diff_dT_eigen}
        $dT_Q^S$ is self-adjoint, negative semi-definite and enjoys the following two-sided bound.
        \begin{equation*}
            \frac{\lambda_{\min }^{1 / 2}\left(S^{1 / 2} Q S^{1 / 2}\right)}{2}\left\|Q^{-1 / 2} X Q^{-1 / 2}\right\|_F^2 \leq 
            \left\langle -dT_Q^S(X), X\right\rangle \leq 
            \frac{\lambda_{\max }^{1 / 2}\left(S^{1 / 2} Q S^{1 / 2}\right)}{2}\left\|Q^{-1 / 2} X Q^{-1 / 2}\right\|_F^2
        \end{equation*}
        \item  \label{eqn: lem_diff_W_Frob}
        $W^2(Q_0,Q_1)$ can be upper and lower bounded by the Frobenius norm as follows
        \begin{equation*}
            \frac{1}{2}\frac{\eigmax{Q_0} \eigminn{-2}{Q_0}}{1+\eigminn{-1}{Q_0}\eigmax{Q_1} } \cdot \Fnorm{Q_1-Q_0}^2 \leq W^2(Q_0,Q_1) \leq \frac{\eigmax{Q_0} \eigminn{-2}{Q_0}}{1+\eigmaxx{-1}{Q_0}\eigmin{Q_1}} \cdot \Fnorm{Q_1-Q_0}^2
        \end{equation*}
        \item \label{eqn: lem_diff_dkT_QS}
        $\opnorm{d^k T^S_Q}$can be upper bounded by
        \begin{equation*}
            \begin{aligned}
                \opnorm{d T^S_Q}&\leq \frac{ \eigmin{Q}^{-2} }{2} \cdot \rbr{\eigmax{S^{1 / 2} Q S^{1 / 2}}}^{1/2}\\
                \opnorm{d^2 T^S_Q} &\leq \eigmax{S}^3 (C_{d,1} + 2C_{d,0}^2)  \cdot \rbr{\eigmin{S^{1/2}QS^{1/2}}}^{-5/2}\\
                \opnorm{d^3 T^S_Q} &\leq \eigmax{S}^4 \left( C_{d,2}+ 6 C_{d,1}C_{d,0} + 6 C_{d,0}^3 \right) \rbr{\eigmin{S^{1/2}QS^{1/2}}}^{-7/2}
            \end{aligned}
        \end{equation*}
        Moreover, if $S,Q\in \cS_d(M^{-1},M)$, then we have
        \begin{equation*}
            \begin{aligned}
                \opnorm{dT^S_Q} &\leq \frac{1}{2}M^3\\
                \opnorm{d^2 T^S_Q} &\leq (C_{d,1} + 2C_{d,0}^2) M^{8}\\
                \opnorm{d^3 T^S_Q} &\leq \left( C_{d,2}+ 6 C_{d,1}C_{d,0} + 6 C_{d,0}^3 \right)M^{11}
            \end{aligned}
        \end{equation*}
    \end{enumerate}
    
\end{lemma}

\begin{proof}
    \item
    \paragraph{Proof of \ref{eqn: diff_W_sup}:} By the closed form expression for $W^2(Q,S)$, one has
    \begin{align*}
        W^2(Q,S)&=\tr\left[ Q+S-2(S^{1/2}QS^{1/2})^{1/2} \right]\\
        &\leq \tr\left[ Q+S+2(S^{1/2}QS^{1/2})^{1/2} \right]\\
        &\leq d\left( \lambda_{\max}(Q) + \lambda_{\max}(S) + 2 \lambda_{\max}(Q)^{1/2} \lambda_{\max}(S)^{1/2} \right)\\
        &\leq 2d \left( \lambda_{\max}(Q) + \lambda_{\max}(S) \right)
    \end{align*}

    \paragraph{Proof of \ref{eqn: lem_diff_quad_approx}, \ref{eqn: lem_diff_dT_eigen}:} see \citet[Lemma A.2, A.3, A.4, A.6]{aap21}.

    \paragraph{Proof of \ref{eqn: lem_diff_W_Frob}:}
    Set $S=Q_0$ in (\ref{eqn: lem_diff_quad_approx}), one can obtain
    \begin{equation}
        \begin{aligned}
            \frac{2}{\rbr{1+\eigmaxx{1/2}{Q^{\prime}}}^2} \inner{-dT_{Q_0}^{Q_0}(Q_1-Q_0)}{Q_1-Q_0} &\leq W^2(Q_1,Q_0)\\
            &\leq \frac{2}{\rbr{1+\eigminn{1/2}{Q^{\prime}}}^2} \inner{-dT_{Q_0}^{Q_0}(Q_1-Q_0)}{Q_1-Q_0}
        \end{aligned}
        \label{eqn: diff_prf_SQ0_quad}
    \end{equation}
    where $Q'= Q_0^{-1/2}Q_1 Q_0^{-1/2}$. Next, apply \ref{eqn: lem_diff_dT_eigen} to get that for any $X \in \cS_d$, one has
    \begin{equation}
        \begin{aligned}
            \frac{\eigmin{Q_0}}{2} \eigmaxx{-2}{Q_0} \Fnorm{X}^2 \leq \inner{-dT^{Q_0}_{Q_0}(X)}{X} \leq \frac{\eigmax{Q_0}}{2} \eigminn{-2}{Q_0} \Fnorm{X}^2
        \end{aligned}
        \label{eqn: diff_prf_SQ0_eig}
    \end{equation}
    Combine (\ref{eqn: diff_prf_SQ0_quad}) and (\ref{eqn: diff_prf_SQ0_eig}) to get that
    \begin{align*}
        W^2(Q_1,Q_0) &\leq \frac{\eigmax{Q_0} \eigminn{-2}{Q_0}}{\rbr{1+\eigminn{1/2}{Q^{\prime}}}^2} \cdot \Fnorm{Q_1-Q_0}^2\\
        &\leq \frac{\eigmax{Q_0} \eigminn{-2}{Q_0}}{1+\eigmin{Q^{\prime}}} \cdot \Fnorm{Q_1-Q_0}^2\\
        &\leq \frac{\eigmax{Q_0} \eigminn{-2}{Q_0}}{1+\eigmaxx{-1}{Q_0}\eigmin{Q_1}} \cdot \Fnorm{Q_1-Q_0}^2
    \end{align*}
    and similarly
    \begin{align*}
        W^2(Q_1,Q_0) &\geq \frac{\eigmin{Q_0} \eigmaxx{-2}{Q_0}}{\rbr{1+\eigmaxx{1/2}{Q^{\prime}}}^2} \cdot \Fnorm{Q_1-Q_0}^2\\
        &\geq \frac{1}{2}\frac{\eigmax{Q_0} \eigminn{-2}{Q_0}}{1+\eigmax{Q^{\prime}}} \cdot \Fnorm{Q_1-Q_0}^2\\
        &\geq \frac{1}{2}\frac{\eigmax{Q_0} \eigminn{-2}{Q_0}}{1+\eigminn{-1}{Q_0}\eigmax{Q_1} } \cdot \Fnorm{Q_1-Q_0}^2
    \end{align*}

    \paragraph{Proof of \ref{eqn: lem_diff_dkT_QS}:} First, note that results for $dT^S_Q$ follows directly from \ref{eqn: lem_diff_dT_eigen}. 
    
    Next since $T_Q^S = S^{1/2}(S^{1/2}Q S^{1/2})^{-1/2}S^{1/2}$, applying Lemma \ref{lem: mtx_inv_sqrt} gives that for $H$ small enough, one has
   \begin{align*}
        d^2 T^S_Q \cdot H^{\otimes 2} &= S^{1/2}\left[ d^2 \varphi(S^{1/2}QS^{1/2}) \cdot (S^{1/2}HS^{1/2})^{\otimes 2} \right] S^{1/2}\\
        d^3 T^S_Q \cdot H^{\otimes 3} &= S^{1/2}\left[ d^3 \varphi(S^{1/2}QS^{1/2}) \cdot (S^{1/2}HS^{1/2})^{\otimes 3} \right] S^{1/2}
   \end{align*}
    which implies
    \begin{align*}
        \opnorm{d^2 T_Q^S} &\leq \lambda_{\max}(S)^3 \cdot \opnorm{d^2 \varphi(S^{1/2}QS^{1/2})}\\
        &\leq \lambda_{\max}(S)^3 (C_{d,1} + 2C_{d,0}^2)  \cdot \left(  \lambda_{\min}(S^{1/2}QS^{1/2}) \right)^{-5/2}\\
        \opnorm{d^3 T_Q^S} &\leq \lambda_{\max}(S)^4 \opnorm{d^3 \varphi(S^{1/2}QS^{1/2})}\\
        &\leq \lambda_{\max}(S)^4 \left( C_{d,2}+ 6 C_{d,1}C_{d,0} + 6 C_{d,0}^3 \right)\lambda_{\min}(S^{1/2}QS^{1/2})^{-7/2}
    \end{align*}

\end{proof}

\begin{lemma}
    Let $Y$ and $Z$ be normed spaces with the norm on each denoted by $\norm{\cdot}$. For any $y \in Y$, $P \in \cM_{k,s}(Y,Z)$, $k \geq 2$, let $Py$ denote a $(k-1)$-linear function defined by
    \begin{align*}
        Py(y_1,\ldots,y_{k-1}) = P(y,y_1,\ldots,y_{k-1})
    \end{align*}
    Then $Py \in \cM_{k-1,s}(Y,Z)$ and
    \begin{equation}
        \vertiii{Py} \leq \norm{P} \norm{y} \leq \frac{k^k}{k!}\vertiii{P} \norm{y}
        \label{eqn: lem_fcal_norm}
    \end{equation}
    \label{lem: f_cal}
\end{lemma}

\begin{proof}
    $Py \in \cM_{k-1,s}(Y,Z)$ can be proved by definition. To prove (\ref{eqn: lem_fcal_norm}), note that for any $\norm{\ty} \leq 1$, one has
    \begin{align*}
            \norm{Py(\ty,\ldots,\ty)} &\leq \norm{P} \cdot \norm{y} \cdot \norm{\ty}^{k-1}\\
            &\leq \norm{P} \cdot \norm{y}\\
            &\overset{(i)}{\leq} \frac{k^k}{k!}\vertiii{P} \cdot \norm{y}
    \end{align*}
    Here (i) follows from \citet[Theorem 5.7]{dudley}.
    
\end{proof}

%% file: tech_lem_conc.tex
First, let us introduce some additional notation. Given a random variable $X$ we denote  $\norm{X}_{\psi_\alpha}$ for $\alpha>0$ as follows.
\begin{equation}
    \norm{X}_{\psi_\alpha} := \inf \cbr{\eta>0 : \EE \psi_\alpha \rbr{\abs{X/\eta}} \leq 1}, \quad \text{where } \psi_\alpha(x):= \exp \rbr{x^\alpha}-1 \quad \text{for } x\geq 0
    \label{eqn: sub_alpha_norm}
\end{equation}
Note that by the definition and Markov's inequality, if $\norm{X}_{\psi_\alpha}<\infty$, then
\begin{equation}
    \PP\left( \abs{X} \geq t \right) \leq 2\exp \rbr{- \frac{t^{\alpha}}{\norm{X}_{\psi_\alpha}^{\alpha}}}
    \label{eqn: sub_alpha_P_ineq}
\end{equation}
For $\alpha\geq 1$, $\norm{\cdot}_{\psi_\alpha}$ is a norm \citep{vershynin}, and for $\alpha<1$, it is equivalent to a norm \citep{aop89}.

\begin{lemma}
    \label{lem: norm_equiv}
    Let $X$ be a random variable and $\alpha\geq 1$. Then the following properties are equivalent; the parameters $K_i > 0$ appearing in these properties differ from each other by at most an absolute constant factor depending on $\alpha$.
    \begin{enumerate}
        \item The $\psi_{\alpha}$-norm of $X$ satisfies $\norm{X}_{\psi_\alpha} \leq K_1$.
        \item The tails of $X$ satisfy
        \begin{align*}
            \PP \cbr{\abs{X} \geq t} \leq 2 \exp \rbr{-t^{\alpha}/K_2^{\alpha}} \quad \forall t \geq 0.
        \end{align*}
        \item The moments of $X$ satisfy
        \begin{align*}
            \norm{X}_p \leq K_3 p^{1/\alpha}
        \end{align*}
    \end{enumerate}
\end{lemma}
\begin{proof}
    See \citet[Exercise 2.7.3]{vershynin}.
\end{proof}


For a random vector $X \in \RR^p$, denote 
\begin{equation*}
    \norm{X}_{\psi_\alpha} := \sup_{u \in \RR^p, \norm{u}=1} \norm{u^\top X}_{\psi_\alpha}
\end{equation*}

\begin{lemma}
    \label{lem: gnorm_trick}
    Let $C_1, c_1, \alpha, \beta,B >0$. Then
    \begin{align*}
        \min \cbr{1, B^{\beta} \cdot C_1 \exp \rbr{-c_1 t^{\alpha}}} \leq C_1 \exp \rbr{- c_1 \frac{t^{\alpha}}{1+ \frac{\beta \log \rbr{1 \vee B} }{\log \rbr{1 \vee C_1} } }}
    \end{align*}
\end{lemma}
\begin{proof}
    Without loss of generality, assume $C_1,B \geq 1$. Denote
    \begin{align*}
        g(t)&:= C_1 \exp \rbr{-c_1 t^{\alpha}}\\
        f(t)&:=C_1 \exp \rbr{- c_1 \frac{t^{\alpha}}{ 1+ \frac{\beta \log B}{\log C_1}  } }
    \end{align*}
    $f,g: (0,+\infty) \to \RR$ satisfy the following properties
    \begin{itemize}
        \item $g,f$ are both positive and decreasing for $t \geq 0$.
        \item denote $t_0 = \sbr{\frac{\beta \log B + \log C_1}{c_1}}^{1/\alpha}$, then
        \begin{align*}
            1 \wedge g(t_0) = 1 \leq f(t_0)
        \end{align*}
        \item derivatives satisfy
        \begin{align*}
            \frac{d}{dt} \sbr{\log (g(t))} \leq \frac{d}{dt} \sbr{\log (f(t))}, \quad \forall t \geq t_0
        \end{align*}
    \end{itemize}
    Therefore, 
    \begin{align*}
        1 \wedge g(t) \leq f(t), \quad \forall t >0
    \end{align*}
\end{proof}

\begin{lemma}
    Let $\alpha\geq 1$. Let $X, Q$ be random elements in $\RR^d, \RR^{d \times d}$ and let $\Phi_k $ be a random symmetric operator in $\cL\left( (\RR^{d \times d})^{\times k}; \RR^{d \times d} \right)$ for some $k \in \NN_+$. Then
    \begin{enumerate}[label=\Roman*.]
        \item $\norm{X}_{\psi_\alpha}\leq \norm{\norm{X}_2}_{\psi_\alpha} \leq c_0 (2d+1)^{1/\alpha} \norm{X}_{\psi_\alpha}$ for any $\alpha \geq 1$ where $c_0>0$ is some absolute constant independent of  $d$ and $\alpha$.
        \item If $Q \in \RR^{d \times d}$ is symmetric, then for any $\norm{v}=1$,
        \begin{align*}
            \norm{v^\top Q v}_{\psi_\alpha} \leq \norm{\Onorm{Q}}_{\psi_\alpha}   \leq c_1 \rbr{2d \cdot \frac{\log 3}{\log 2}+1}^{1/\alpha} \sup_{\norm{v}=1} \norm{v^\top Q v}_{\psi_\alpha}
        \end{align*}
        where $c_1>0$ is some absolute constant independent of  $d$ and $\alpha$.
        \item For any $U \in \RR^{d \times d}$ with unit \red{Frobenius} norm
        \begin{equation*}
            \norm{\inner{U}{Q}}_{\psi_\alpha} \leq \norm{\Fnorm{Q}}_{\psi_\alpha} \lesssim \sup_{U: \Fnorm{U}=1} \norm{\inner{U}{Q}}_{\psi_\alpha}
        \end{equation*}
        \item For any $U \in \RR^{d \times d}$ with unit Frobenius norm
        \begin{equation*}
            \norm{ \Fnorm{\Phi_k \cdot U^{\otimes k}} }_{\psi_\alpha} \leq \norm{ \opnorm{\Phi_k} }_{\psi_\alpha}  \lesssim \sup_{U: \Fnorm{U}=1} \norm{ \Fnorm{\Phi_k \cdot U^{\otimes k}} }_{\psi_\alpha} 
        \end{equation*}
    \end{enumerate}
    The constants behind $\lesssim$ only depend on dimension $d$, $k$ and possible on $\alpha$.
    \label{lem: subG_norm}
\end{lemma}

\begin{proof}

    \noindent\textbf{I:} The first inequality $\norm{X}_{\psi_\alpha}\leq \norm{\norm{X}_2}_{\psi_\alpha}$ follows since $\inner{v}{X} \leq \norm{X}$ for any $\norm{v}=1$.

    For the second inequality, by Lemma \ref{lem: norm_equiv}, we can assume without loss of generality that
    $\norm{X}_{\psi_\alpha}$ satisfies $K_2=1$, i.e.
    \begin{align*}
        \sup_{\norm{v}=1} \PP \cbr{\abs{\inner{v}{X}} \geq t}\leq 2 \exp \rbr{-t^{\alpha}}
    \end{align*}
    Following the $1/2$-net argument in Lemma 1 in \citet{jin19}, one can obtain
    \begin{align*}
        \PP \cbr{\norm{X}\geq t} &\leq 4^d \sup_{\norm{v}=1}\PP(\inner{v}{X} \geq t/2)\\
        &\leq \underbrace{2\cdot 4^d \exp \rbr{-t^{\alpha}/2^{\alpha}}}_{=:g(t)}
    \end{align*}
    Lemma \ref{lem: gnorm_trick} implies that for any $t \geq 0$,
    \begin{align*}
        \PP \cbr{\norm{X}\geq t} &\leq 2 \exp \rbr{-\frac{t^{\alpha}}{(2d+1) \cdot 2^{\alpha} } } =:f(t)
    \end{align*}
    Therefore, Lemma \ref{lem: norm_equiv} implies that
    \begin{align*}
        \norm{\norm{X}_2}_{\psi_\alpha} \leq c_0 (2d+1)^{1/\alpha} \norm{X}_{\psi_\alpha}
    \end{align*}
    where $c_0>0$ is some absolute constant independent of  $d$ and $\alpha$.

    \noindent\textbf{II:} We prove the second inequality here. By Lemma \ref{lem: norm_equiv}, we can assume without loss of generality that $\sup_{\norm{v}=1} \norm{v^\top Q v}_{\psi_\alpha}$ satisfies $K_2=1$, i.e.
    \begin{align*}
        \sup_{\norm{v}=1} \PP \cbr{\abs{v^\top Qv} \geq t}\leq 2 \exp \rbr{-t^{\alpha}}
    \end{align*}
    Following the $1/4$-net argument in \citet[Exercise 4.4.3, Theorem 4.4.5.]{vershynin}, one can obtain that for any $t \geq 0$,
    \begin{align*}
        \RR \cbr{\norm{\Onorm{Q}}_{\psi_\alpha} \geq t}  \leq  9^d \cdot 2\exp(-t^{\alpha}/2^{\alpha})
    \end{align*}
    Lemma \ref{lem: gnorm_trick} then implies that
    \begin{align*}
        \RR \cbr{\norm{\Onorm{Q}}_{\psi_\alpha} \geq t}  \leq 2\exp\rbr{-\frac{t^{\alpha}}{(2d \cdot \frac{\log 3}{\log 2}+1) \cdot 2^{\alpha} }}, \quad \forall t \geq 0
    \end{align*}
    Hence  Lemma \ref{lem: norm_equiv} implies that
    \begin{align*}
        \norm{\Onorm{Q}}_{\psi_\alpha} \leq c_1 \rbr{2d \cdot \frac{\log 3}{\log 2}+1}^{1/\alpha} \sup_{\norm{v}=1} \norm{v^\top Q v}_{\psi_\alpha}
    \end{align*}
    where $c_1>0$ is some absolute constant independent of  $d$ and $\alpha$.

    The proofs for statements \textbf{III} and \textbf{IV} follow a similar approach.
\end{proof}

Next, we give some properties of the weights $w(x,X)$ and $w_{n,\rho}(x,X_i)$ that will be crucial to prove uniform concentration as well as a central limit theorem for the estimate $\hQ_{\rho}(x)$. Recall that by definition, $w(x, X)=1+(x-\mu)^{\top} \Sigma^{-1}(X-\mu)$ and $w_{n,\rho}(x)= 1 + (x-\bar{X})\hSigma_{\rho}^{-1}(X_i - \bar{X})$. For any vector $z \in \RR^p$, denote $\vec{z}=(1, z^\top)^\top$ and
\begin{align*}
    \Vec{\Sigma} &:= \EE \Vec{X}\Vec{X}^\top =
    \begin{pmatrix}
        1 & \mu^\top\\
        \mu & \Sigma +\mu \mu^\top
    \end{pmatrix}\\
    \Vec{\Sigma}_{\rho} &:=
    \begin{pmatrix}
        1 & \mu^\top\\
        \mu & \Sigma_{\rho} +\mu \mu^\top
    \end{pmatrix}, \quad \text{where } \Sigma_{\rho} = \Sigma + \rho I_p \\
    \hat{\Vec{\Sigma}}&:=n^{-1} \sum_{i=1}^{n} \vec{X}_i \vec{X}_i^\top = 
    \begin{pmatrix}
        1 & \hmu^\top\\
        \hmu & \hSigma +\hmu \hmu^\top
    \end{pmatrix}\\
    \hat{\vec{\Sigma}}_{\rho} &:=
    \begin{pmatrix}
        1 & \hmu^\top\\
        \hmu & \hSigma_{\rho} +\hmu \hmu^\top
    \end{pmatrix}
\end{align*}
\begin{lemma}
    Suppose $X,X_1,\ldots,X_n \overset{\mathrm{i.i.d.}}{\sim} \PP \in \cP_2(\RR^p)$. Let $\hat{\Vec{\Sigma}}=n^{-1} \sum_{i=1}^{n} \vec{X}_i \vec{X}_i^\top$. Then for any $x \in \RR^p$,
    \begin{align}
        w(x,X)&= \vec{x}^\top \vec{\Sigma}^{-1} \vec{X} \label{eqn: lem_weight_s}\\
        w_{n,\rho}(x,X_i) &= \vec{x}^\top \hat{\Vec{\Sigma}}_{\rho}^{-1} \vec{X} \label{eqn: lem_weight_sin}
    \end{align}
    \label{lem: weights}
\end{lemma}

\begin{proof}
    For (\ref{eqn: lem_weight_s}), by definition one has
    \begin{align*}
        \vec{\Sigma} = 
        \begin{pmatrix}
            1 & \mu^\top\\
            \mu & \Sigma +\mu \mu^\top
        \end{pmatrix}
    \end{align*}
    Computing the inverse of the block matrix $\vec{\Sigma}$ then gives
    \begin{align*}
        \vec{\Sigma}^{-1}=
        \begin{pmatrix}
            1+\mu^\top \Sigma^- \mu & -\mu^\top \Sigma^- \\
            -\Sigma^-\mu & \Sigma^-
        \end{pmatrix}
    \end{align*}
    Finally we arrive at (\ref{eqn: lem_weight_s}) by computing $\vec{x}^\top \Lambda^{-1} \vec{X}$. (\ref{eqn: lem_weight_sin}) follows from similar arguments.    
\end{proof}

Denote $L_{\tau}=C_{\psi_2} \sqrt{(1+\tau) \log n}$.
\begin{lemma}
    \label{lem: s_sin}
    Suppose Assumption \ref{assumption: X}  holds. Set
    \begin{align*}
        W_{0n} &:= W_{0n}(x) = -\bar{X}^\top \hSigma_{\rho}^{-1}(x-\bar{X}) + \mu^\top \hSigma_{\rho}^{-1} (x-\mu)\\
        W_{1n} &:= W_{1n}(x) = \Sigma^{-1}(x-\mu) - \hSigma_{\rho}^{-1} (x-\bar{X})
    \end{align*}
    one has
    \begin{equation}
        w_{n,\rho}(x,X_i)-w(x,X_i) = W_{0n} + W_{1n}^\top X_i 
        \label{eqn: lem_s_sin_1}
    \end{equation}
    Moreover, for any $\rho \in \cbr{0,n^{-1}}$, $L \geq 1$ and $\tau \geq 0$, there exists constant $C>0$ independent of $n$ such that 
    \begin{align}
        &\PP\cbr{\sup_{x\in B_{\mu}(L)}\abs{W_{0n}(x)} > \frac{L}{\sqrt{n}} C(1+\tau)\log n }\lesssim n^{-(1+\tau)}  \label{eqn: lem_s_sin_W0}\\ 
        &\PP\cbr{\sup_{x\in B_{\mu}(L)}\norm{W_{1n}(x)} > \frac{L}{\sqrt{n}} C(1+\tau)\log n}
        \lesssim n^{-(1+\tau)} \label{eqn: lem_s_sin_W1}
    \end{align}
\end{lemma}

\begin{proof}
    Under Assumption \ref{assumption: X}, we can assume without loss of generality that $\mu = 0$ and $\Sigma = I_p$. Then
    \begin{align*}
        w_{n,\rho}(x,X_i) &= 1 + (x-\bar{X})^\top \hSigma_{\rho}^{-1} (X_i - \bar{X})\\
        w(x,X) &= 1+ x^\top X_i
    \end{align*}

    \noindent\textit{Proof of (\ref{eqn: lem_s_sin_1}):} Direct computation gives
    \begin{align*}
        w_{n,\rho}(x,X_i) - w(x,X_i) &= (x-\bar{X})^\top \hSigma_{\rho}^{-1} (X_i - \bar{X}) - x^\top X_i\\
        &= \underbrace{-x^\top \hSigma_{\rho}^{-1}\bar{X} + \bar{X}^\top \hSigma_{\rho}^{-1}\bar{X}   }_{=W_{0n}(x)} + 
        \rbr{\underbrace{(\hSigma_{\rho}^{-1}-I_p)x + \hSigma_{\rho}^{-1} \bar{X} }_{=W_{1n}(x)}}^\top  X_i
    \end{align*}

    \noindent\textit{Proof of (\ref{eqn: lem_s_sin_W0}), (\ref{eqn: lem_s_sin_W1}):} 
        For $W_{0n}(x)$:
        \begin{align*}
            \PP\cbr{\sup_{x\in B_{\mu}(L)}\abs{W_{0n}(x)} > t} &\leq \PP\rbr{ \cbr{\Onorm{\hSigma_{\rho}^{-1}}\leq 2} \cap \cbr{L\norm{\bar{X}}+\norm{\bar{X}}^2 \geq t/2} }\\
            &+\PP\rbr{ \cbr{2<\Onorm{\hSigma_{\rho}^{-1}}\leq n} \cap \cbr{L\norm{\bar{X}}+\norm{\bar{X}}^2 \geq t/n} }\\
            &+ \PP\cbr{ \Onorm{\hSigma_{\rho}^{-1}}> n}\\
            &\leq \underbrace{\PP \rbr{ \cbr{L\norm{\bar{X}}+\norm{\bar{X}}^2 \geq t/2} }}_{\mathbf{(I)}}\\
            &+\underbrace{\PP\rbr{ \cbr{2<\norm{\hSigma_{\rho}^{-1}}\leq n}} \wedge \PP\rbr{  \cbr{L\norm{\bar{X}}+\norm{\bar{X}}^2 \geq t/n} }}_{\mathbf{(II)}}\\
            &+ \underbrace{\PP\cbr{ \norm{\hSigma_{\rho}^{-1}}> n}}_{\mathbf{(III)}}
        \end{align*}
        Note that for any $s \geq 0$, one has 
        \begin{align*}
            \PP \cbr{L\norm{\bar{X}}+\norm{\bar{X}}^2 \geq s} &\leq \PP \cbr{ L\norm{\bar{X}} \geq \frac{\sqrt{n}L}{\sqrt{n}L+1} s} + \PP \cbr{\norm{\bar{X}}^2 \geq \frac{1}{\sqrt{n}L+1} s}\\
            &= \PP\cbr{\norm{\bar{X}} \geq \frac{\sqrt{n}}{\sqrt{n}L+1} s} + \PP \cbr{\norm{\bar{X}}^2 \geq \frac{1}{\sqrt{n}L+1} s}\\
            &\lesssim \exp (-cns^2/L^2) + \exp (-c\sqrt{n}s/L)
        \end{align*}
        \begin{itemize}
            \item If $\rho=0$, then
            \begin{align*}
                \mathbf{(II)} + \mathbf{(III)} &\leq \PP \rbr{\Onorm{\hSigma_{\rho}^{-1}} >2}\\
                &\leq \PP \rbr{\Onorm{\hSigma_{\rho} - I_p} > 1/2}\\
                &\lesssim \exp(-cn)
            \end{align*}
            Hence
            \begin{align*}
                \PP\cbr{\sup_{x\in B_{\mu}(L)}\abs{W_{0n}(x)} > t} &\lesssim  \exp (-cnt^2/L^2) + \exp (-c\sqrt{n}t/L)  + \exp(-cn)\\
                &\lesssim 
                \begin{cases}
                    \exp (-c\sqrt{n}t/L) + \exp(-cn), & t\geq L/\sqrt{n} \\
                    \exp (-cnt^2/L^2) + \exp(-cn), & t < L /\sqrt{n}
                \end{cases}
            \end{align*}
            or equivalently for any $s > 0 $,
            \begin{align*}
                \PP\cbr{\sup_{x\in B_{\mu}(L)}\abs{W_{0n}(x)} > \frac{L}{\sqrt{n}}s } &  \lesssim 
                \begin{cases}
                    \exp(-cs^2) & s<1\\
                    \exp(-cs), &1\leq s\leq n\\
                    \exp(-cn), & s >n
                \end{cases}
            \end{align*}
            \item If $\rho = 1/n$, then $\mathbf{III}=0$. As a result, for any $ s \geq 1$,
            \begin{align*}
                &\quad \ \PP\cbr{\sup_{x\in B_{\mu}(L)}\abs{W_{0n}(x)} > s\frac{L}{\sqrt{n}}}\\
                &\lesssim \exp(-cs) + \exp(-cn) \wedge 
                \begin{cases}
                    \exp(-cs^2/n^2), & 1 \leq s\leq n\\
                    \exp(-cs/n), & s > n
                \end{cases}
                \\
                &= \exp(-cs) + 
                \begin{cases}
                    \exp(-cn), & 1 \leq s \leq n^2\\
                    \exp(-cs/n), & s > n^2
                \end{cases}\\
                &\lesssim
                \begin{cases}
                    \exp(-cs), & 1 \leq s \leq n\\
                    \exp(-cn), & n< s \leq n^2\\
                    \exp(-cs/n), & s >n^2
                \end{cases}
            \end{align*}
        \end{itemize}

        \noindent\textit{For $W_{1n}(x)$:} note that $\hSigma_{\rho}^{-1} -I_p = -\hSigma_{\rho}^{-1} \rbr{\hSigma_{\rho}-I_p}$, then one has
        \begin{align*}
            &\quad \ \PP\cbr{\sup_{x\in B_{\mu}(L)}\norm{W_{1n}(x)} > t}\\
            &\leq \PP \cbr{ \Onorm{\hSigma_{\rho}^{-1}} \Onorm{\hSigma_{\rho}-I_p} L + \Onorm{\hSigma_{\rho}^{-1}} \norm{\bar{X}} \geq t }\\
            &\leq \PP \cbr{ \Onorm{\hSigma_{\rho}^{-1}} \Onorm{\hSigma_{\rho}-I_p} L \geq \frac{L}{L+1}t} + \PP \cbr{\Onorm{\hSigma_{\rho}^{-1}} \norm{\bar{X}} \geq \frac{1}{L+1}t}
        \end{align*}
        Hence for $L \geq 1$, one has
        \begin{align*}
            &\quad \ \PP\cbr{\sup_{x\in B_{\mu}(L)}\norm{W_{1n}(x)} > t}\\
            &\leq \PP \cbr{ \Onorm{\hSigma_{\rho}^{-1}} \Onorm{\hSigma_{\rho}-I_p}  \geq \frac{t}{2L}} + \PP \cbr{\Onorm{\hSigma_{\rho}^{-1}} \norm{\bar{X}} \geq \frac{t}{2L}}\\
            &\leq \PP \rbr{\cbr{\Onorm{\hSigma_{\rho}^{-1}}\leq 2} \cap \cbr{\Onorm{\hSigma_{\rho} - I_p} \geq t/(4L)}}\\
            &+ \PP \rbr{\cbr{2 < \Onorm{\hSigma_{\rho}^{-1}}\leq n} \cap \cbr{\Onorm{\hSigma_{\rho} - I_p} \geq t/(2nL)}}\\
            &+ \PP \rbr{\cbr{\Onorm{\hSigma_{\rho}^{-1}}> n} } \\
            &+ \PP \rbr{\cbr{\Onorm{\hSigma_{\rho}^{-1}}\leq 2} \cap \cbr{\norm{\bar{X}} \geq t/(4L)}}\\
            &+ \PP \rbr{\cbr{2 < \Onorm{\hSigma_{\rho}^{-1}}\leq n} \cap \cbr{\norm{\bar{X}} \geq t/(2nL)}}\\
            &+ \PP \rbr{\cbr{\Onorm{\hSigma_{\rho}^{-1}}> n} } \\
            &\lesssim \underbrace{\PP \cbr{\Onorm{\hSigma_{\rho} - I_p} \geq \frac{t}{4L}} +\PP \cbr{\norm{\bar{X}}\geq \frac{t}{4L}}}_{\mathbf{(i)}}\\
            &+ \underbrace{\PP \cbr{2< \Onorm{\hSigma_{\rho}^{-1}}\leq n} \wedge \sbr{\PP \cbr{\Onorm{\hSigma_{\rho} - I_p} \geq t/(2nL)}+ \PP \cbr{\norm{\bar{X}}\geq t/(2nL)}} }_{\mathbf{(ii)}}\\
            &+ \underbrace{\PP \rbr{\cbr{\Onorm{\hSigma_{\rho}^{-1}}> n} }}_{\mathbf{(iii)}}
        \end{align*}
        \begin{itemize}
            \item If $\rho = 0$, then
            \begin{align*}
                \mathbf{(ii)} + \mathbf{(iii)} \leq \PP \cbr{\Onorm{\hSigma_{\rho}^{-1}}> 2} \lesssim \exp(-cn)
            \end{align*}
            As a result, for any $L \geq 1$,
            \begin{align*}
                \PP\cbr{\sup_{x\in B_{\mu}(L)}\norm{W_{1n}(x)} > t} &\lesssim \exp(-cnt/L) \vee \exp(-cnt^2/L^2) + \exp(-cnt^2/L^2) + \exp(-cn)\\
                &= 
                \begin{cases}
                    \exp(-cnt^2/L^2), & t<L\\
                    \exp(-cn), & t \geq L
                \end{cases}
            \end{align*}
            Therefore, for any $s> 0$, $L \geq 1$,
            \begin{align*}
                \PP\cbr{\sup_{x\in B_{\mu}(L)}\norm{W_{1n}(x)} > \frac{L}{\sqrt{n}}s } \lesssim 
                \begin{cases}
                    \exp(-cs^2), & s<\sqrt{n}\\
                    \exp(-cn), & s>\sqrt{n}
                \end{cases}
            \end{align*}
            \item If $\rho = n^{-1}$, then $\mathbf{(iii)}=0$. As a result, for any $t = sL/\sqrt{n}$ and $L \geq 1$,
            \begin{align*}
                &\quad \ \PP\cbr{\sup_{x\in B_{\mu}(L)}\norm{W_{1n}(x)} > \frac{L}{\sqrt{n}}s }\\
                &\lesssim 
                \begin{cases}
                    \exp(-cs^2), & 0< s < \sqrt{n}\\
                    \exp(-cn), & \sqrt{n} \leq s < n^{3/2}\\
                    \exp(-cs/\sqrt{n}), &   s \geq n^{3/2}
                \end{cases}
            \end{align*}
        \end{itemize}
        Combining results for $W_{0n}(x)$ and $W_{1n}(x)$, one can obtain that for any $\rho \in \cbr{0,n^{-1}}$, $L \geq 1$, $\tau \geq 0$ and $1\leq s \leq \sqrt{n}$,
        \begin{align*}
            \PP\cbr{\sup_{x\in B_{\mu}(L)}\abs{W_{0n}(x)} > \frac{L}{\sqrt{n}}s}  \vee 
            \PP\cbr{\sup_{x\in B_{\mu}(L)}\norm{W_{1n}(x)} > \frac{L}{\sqrt{n}}s}
            \lesssim \exp(-cs)
        \end{align*}
        By taking $s = O \rbr{(1+\tau) \log n}$, the proof is then complete.
\end{proof}

\begin{lemma}
    \label{lem: sf_snf}
    Let $\cbr{(X_i,Q_i)}_{i=1}^n$ be i.i.d. samples satisfying Assumption \ref{assumption: X}-\ref{assumption: bdd_Q}. Let 
    $(\cV,\norm{\cdot})$ be a normed vector space and $\psi: \cS_d^{++} \times \Theta \to \cV$ is a  mapping parametrized by $\theta \in \Theta$. Suppose there exists an event $\cE_0$ under which
    \begin{itemize}
        \item $\norm{X_i-\mu} \leq L_n$ for $i=1,...,n$.
        \item $\norm{\psi(Q_i, \theta)} \leq K_n$ for $i=1,...,n$ uniformly for $\theta \in \Theta$.
    \end{itemize}
    Denote
    \begin{align*}
        G(L, \Theta; n)&:=\sup_{x \in B_\mu(L)} \sup_{\theta \in \Theta} \norm{\frac{1}{n}\sum_{i=1}^{n} \sbr{w_{n,\rho}(x,X_i) - w(x,X_i)} \psi(Q_i; \theta) }
    \end{align*}
    Then for any $\rho \in \cbr{0,n^{-1}}$, $L \geq 1$ and any $\tau \geq 0$, there exists constant $C$ indpendent of $n$ such that
    \begin{align*}
        \PP \cbr{ G(L, \Theta; n) \geq C K_n(L_n + 1) \frac{L}{\sqrt{n}} (1+\tau)\log n  }  
        \leq &O \rbr{ n^{-(1+\tau)}} + \PP(\cE_0^c)
    \end{align*}
\end{lemma}

\begin{proof}
    Given $s\geq 1$, define event $\cE_1(s)$
    \begin{align*}
        \cE_1(s) := \sup_{x \in B_\mu(L)} \sup_{\theta \in \Theta} \norm{\frac{1}{n}\sum_{i=1}^{n} \sbr{w_{n,\rho}(x,X_i) - w(x,X_i)} \psi(Q_i; \theta) } \geq 2K_n(L_n + 1)\cdot \frac{L}{\sqrt{n}} s
    \end{align*}
    Then
    \begin{align*}
        \PP \rbr{\cE_1(s)} \leq \PP \rbr{\cE_1(s) \cap \cE_0} + \PP \rbr{\cE_0^c}
    \end{align*}
    By (\ref{eqn: lem_s_sin_1}) in Lemma \ref{lem: s_sin}, one can obtain for any $x$ and $\theta \in \Theta$,
    \begin{align*}
        &\quad \norm{\frac{1}{n}\sum_{i=1}^{n} w_{n,\rho}(x,X_i) \psi(Q_i;\theta) - \frac{1}{n}\sum_{i=1}^{n} w(x,X_i) \psi(Q_i;\theta)} \\
        &= \norm{W_{0n}(x) \frac{1}{n}\sum_{i=1}^{n} \psi(Q_i;\theta) + \frac{1}{n} \sum_{i=1}^{n} W_{1n}(x)^\top X_i \psi(Q_i;\theta)}\\
        &\leq \abs{W_{0n}(x)}\cdot  \frac{1}{n}\sum_{i=1}^n \norm{\psi(Q_i;\theta)} + \norm{W_{1n}(x)}_2 \frac{1}{n} \sum_{i=1}^{n} \norm{X_i}_2 \norm{\psi(Q_i;\theta)}
    \end{align*}
    Hence
    \begin{align*}
        \PP(\cE_1(s) \cap \cE_0) 
        &\leq \PP \cbr{\sup_{x\in B_{\mu}(L)} \sbr{\abs{W_{0n}(x)} \cdot K_n + \twonorm{W_{1n}(x)} \cdot L_n K_n} \geq 2K_n(L_n + 1)\cdot \frac{L}{\sqrt{n}} s }\\
        &= \PP\cbr{\sup_{x\in B_{\mu}(L)} \sbr{\abs{W_{0n}(x)} + \twonorm{W_{1n}(x)}}  \geq 2\cdot \frac{L}{\sqrt{n}} s }\\
        &\leq \PP\cbr{\sup_{x\in B_{\mu}(L)}\abs{W_{0n}(x)}  \geq  \frac{L}{\sqrt{n}} s }+ \PP\cbr{\sup_{x\in B_{\mu}(L)}\twonorm{W_{1n}(x)}  \geq \frac{L}{\sqrt{n}} s }
    \end{align*}
    Applying (\ref{eqn: lem_s_sin_W0}) and (\ref{eqn: lem_s_sin_W0}) in Lemma \ref{lem: s_sin} gives the desired result.
\end{proof}

We will need uniform upper bounds for various quantities of the form
\begin{align*}
    \sup_{\theta \in \Theta} f(Z_1^n;\theta)
\end{align*}
where $Z_1^n$ denotes $(Z_1,...,Z_n)$.
To this end, we decompose the above quantity as follows.
\begin{align*}
    \underbrace{\sup_{\theta \in \Theta} f(Z_1^n;\theta) - \EE \sup_{\theta \in \Theta} f(Z_1^n;\theta)}_{\mathrm{perturbation}} + \underbrace{\EE \sup_{\theta \in \Theta} f(Z_1^n;\theta)}_{\mathrm{expectation}}
\end{align*}
The expectation term is bounded in Lemma \ref{lem: Rademacher} below with a chaining argument, while the perturbation term is shown to concentrate in Lemma \ref{lem: bdd_diff} by exploiting its bounded difference property. The proof of Lemma \ref{lem: Rademacher} is deferred to Appendix \ref{subsubsec: prf_lem_Rademacher}. 
\begin{lemma}
    \label{lem: Rademacher}
    Given  a function class
    \begin{equation*}
        \cF(\Theta)=\left\{ f(\cdot;\theta): \RR^p \to \RR: \theta \in \Theta \right\}
    \end{equation*}
    where $f$ is continuous jointly in $(z,\theta)$ and $(\Theta,d)$ is a separable metric space with finite diameter $D := \sup_{\theta_1, \theta_2 \in \Theta} d(\theta_1,\theta_2) \gtrsim 1$. 
    Suppose a random vector $Z \in \RR^p$ satisfies the following inequality
    \begin{equation}
        \left\|f(Z;\theta_1)-f(Z;\theta_2) \right\|_{\psi_2} \leq  \tau\left( d(\theta_1,\theta_2) \right), \quad \forall \theta_1, \theta_2 \in \Theta
        \label{eqn: lem_rad_lip}
    \end{equation}
    for some increasing function $\tau:\RR^+ \to \RR^+$ that satisfies $\tau(0)=0, \tau(+\infty)= + \infty$. Then
    \begin{itemize}
        \item the following inequality holds
        \begin{equation}
            \begin{aligned}
                \EE \sup_{\theta \in \Theta} \left| f(Z;\theta) \right| \leq  C\int_0^{\tau(D)} \sqrt{\log N(\tau^{-1}(t);\Theta)} \cdot dt + \sup_{\theta \in \Theta} \EE \left| f(Z;\theta) \right|
            \end{aligned}
            \label{eqn: lem_rad_chain}
        \end{equation}
        where $C>0$ is a fixed absolute constant and $N(\epsilon;\Theta)$ is the $\epsilon$-covering number of $\Theta$.
        \item Specifically, if $\tau$ has the form
        \begin{equation*}
            \tau(\epsilon) = \frac{K}{\sqrt{n}} \left( \epsilon \vee \epsilon^{\alpha_0}\right)
        \end{equation*}
        for some constant $\alpha_0 >0$, then for any $D \gtrsim 1$, one has
        \begin{equation}
            \int_0^{\tau(D)} \sqrt{\log N(\tau^{-1}(t);\Theta)} \cdot dt \lesssim  \frac{KD^{1 \vee \alpha_0}}{\sqrt{n}} \sqrt{\log^+ D} 
            \label{eqn: lem_rad_chain_holder}
        \end{equation}
        \item Specifically, if $Z_1,\ldots,Z_n \in \RR^p$ are i.i.d. and $f:(\RR^p)^k \times \Theta \to \RR$ is equal to 
        \begin{equation*}
            \begin{aligned}
                f(Z_1^n;\theta) = \opnorm{\frac{1}{n} \sum_{i=1}^{n} \psi(Z_i;\theta) - \EE \frac{1}{n} \sum_{i=1}^{n} \psi(Z_i;\theta)} , \quad Z_1^n:= (Z_1,\ldots,Z_n)
            \end{aligned}
        \end{equation*}
        where $\psi(z;\theta) \in \cM_{k,s}(\RR^m; \RR^m)$ is a symmetric $k$-linear operator for any $z \in \RR^p,\theta \in \Theta$. Then the following inequality
        \begin{equation}
            \gnorm{f(Z_1^n;\theta)-f(Z_1^n;\ttheta)} \lesssim \frac{1}{\sqrt{n}} \gnorm{\opnorm{\psi(Z_1;\theta) - \psi(Z_1;\ttheta) }}
            \label{eqn: lem_rad_lip_avg}
        \end{equation}
        holds.
    \end{itemize}
\end{lemma}

\begin{remark}
    (\ref{eqn: lem_rad_lip_avg}) will be applied in three ways as follows.
    \begin{itemize}
        \item $m=1, k=1$: $\psi(z;\theta) \in \RR$ and $\opnorm{\psi(z;\theta)}$ reduces to the absolute value of $\psi(z;\theta)$. Examples include $\psi(X,Q;x,S) = w(x,X)W^2(Q,S)$ in Lemma \ref{lem: prf_conv_lem_slow}
        \item $m=d, k=1$: $\psi(z;\theta) \in \cL(\RR^m; \RR^m)$ can be viewed as a matrix. Moreover, $\opnorm{\psi(z;\theta)}$ reduces to the matrix operator norm. Examples include $\psi(X,Q;x) = w(x,X)(T^Q_{Q^*(x)} - I_d)$ in Lemma \ref{lem: conc_tld_exp}.
        \item $m = d\times d, k=1$: this is a special case of the previous one by identifying $d \times d$ matrices as a vector in $\RR^{d^2}$. Examples include $\psi(X,Q;x) = -w(x,X)dT^Q_{Q^*(x)}$ in Lemma \ref{lem: conc_tld_exp}.
        \item $m= d \times d, k =2$: Examples include $\psi(X,Q;x,S) = w(x,X)d^2 T^Q_S$ in Lemma \ref{lem: conc_tld_exp}.
    \end{itemize}
\end{remark}

\begin{lemma}
    \label{lem: bdd_diff}
     Let $\psi: \RR^p \times \Theta \to \RR^q$ be a class of functions indexed by $\theta \in \Theta$. Suppose $\psi$ is uniformly bounded in the sense that there exists a finite constant $K$ such that $\norm{\psi(x;\theta)}_2 \leq K$ for any $x \in \RR^p$ and $\theta \in \Theta$. For any fixed $y_0 \in \cW$, define $f: (\RR^p)^n \times \Theta \to \RR^q$ as 
    \begin{equation*}
        f(x_1^n):=\sup_{\theta \in \Theta}\norm{\frac{1}{n} \sum_{i=1}^{n} \psi(x_i;\theta) - y_0}_2
    \end{equation*}
    Then $f$ satisfies the bounded difference property with parameter $2K/n$, i.e. for any $x_1^n,\tx_1^n \in (\RR^p)^n$ such that $ \sum_{i=1}^n \II(x_i \neq \tx_i) \leq 1$,
    \begin{equation}
        \norm{f(x_1^n) - f(\tx_1^n)}_2 \leq \frac{2K}{n}, 
        \label{eqn: bdd_diff_2}
    \end{equation}
    Moreover, suppose $X_1,\ldots,X_n$ are i.i.d. random element in $\RR^p$, then
    \begin{equation}
        \gnorm{f(X_1^n) - \EE f(X_1^n) } \lesssim \frac{2K}{\sqrt{n}}
        \label{eqn: bdd_diff_3}
    \end{equation}
\end{lemma}

\begin{proof}
    Without loss of generality, assume $x_1 \neq \tx_1$ and $x_j= \tx_j$ for $j=2,\ldots,n$. Then one has
    \begin{align*}
        &\norm{\frac{1}{n} \sum_{i=1}^{n} \psi(x_i;\theta) - y_0}_2 - \sup_{\ttheta \in \Theta}\norm{\frac{1}{n} \sum_{i=1}^{n} \psi(\tx_i;\ttheta) - y_0}_2\\
        \leq & \norm{\frac{1}{n} \sum_{i=1}^{n} \psi(x_i;\theta) - y_0}_2 - \norm{\frac{1}{n} \sum_{i=1}^{n} \psi(\tx_i;\theta) - y_0}_2\\
        \leq & \norm{\frac{1}{n} \sum_{i=1}^{n} \psi(x_i;\theta) - y_0 - \frac{1}{n} \sum_{i=1}^{n} \psi(\tx_i;\theta) + y_0}_2\\
        = & \norm{\frac{1}{n} \psi(x_1;\theta) - \psi(\tx_1;\theta)}_2\\
        \leq & \frac{2K}{n}
    \end{align*}
    Taking supremum over $\theta$ then gives 
    \begin{equation*}
        f(x_1^n) - f(\tx_1^n) \leq \frac{2K}{n}
    \end{equation*}
    The other direction can be obtained with the role $x_1^n$ and $\tx_1^n$ reversed. 
    Therefore we get (\ref{eqn: bdd_diff_2}). By Corollary 2.21 in \citet{wainwright}, we get (\ref{eqn: bdd_diff_3}). The proof is then complete.
    
\end{proof}

Finally, even though Theorem \ref{thm: unfm_conv} shows fast convergence $\hQ_{\rho}(x)\to Q^*(x)$, one needs $\hQ_{\rho}(\bar{X})\to Q^*(\mu)$ when considering power in Theorem \ref{thm: power}, and this is stated in Lemma \ref{lem: barycenter} below. A crucial observation here is that $\hQ_{\rho}(\bar{X})$ and $Q^*(\mu)$ are the empirical and population Fr\'echet mean respectively. The proof is built upon \citet{rig22_fast_alex} and \citet{alt21_averaging}. 
\begin{lemma}
    \label{lem: barycenter}
    Assume Assumption \ref{assumption: X} and \ref{assumption: bdd_Q} hold. Then there exists an event $\tE$ with probability at least $1-n^{-100}$ under which the following holds
    \begin{equation}
        W^2 \rbr{\hQ_{\rho}(\bar{X}), Q^*(\mu)} \lesssim \frac{\polylog{n}}{\sqrt{n}}
        \label{eqn: lem_bary_W}
    \end{equation}
    As a result, under $\tE$, one has
    \begin{equation}
        \Fnorm{\hQ_{\rho}(\bar{X}) - Q^*(\mu)} \lesssim \frac{\polylog{n}}{\sqrt{n}}
        \label{eqn: lem_bary_F}
    \end{equation}
\end{lemma}

\begin{proof}
    Note that $\hQ_{\rho}(\bar{X})$ and $Q^*(\mu)$ are equal to the empirical and population barycenter respectively. For simplicity, we write $Q^*$ for $Q^*(\mu)$.
    
    \paragraph{Proof of (\ref{eqn: lem_bary_W}): }
     From (3.8)-(3.9) in \citet{rig22_fast_alex}, the variance inequality in \citet[Theorem 6]{che20_gd} and results in \citet{alt21_averaging}, one can obtain the following
    \begin{equation}
         W^2 \rbr{\hQ_{\rho}(\bar{X}), Q^*} \leq C_b \norm{\frac{1}{n}\sum_{i=1}^{n} \rbr{T^{Q_i}_{Q^*} - I_d}}_{Q^*}
        \label{eqn: prf_bary_fast}
    \end{equation}
    where $C_b > 0$ is a constant independent of $n$ and $\norm{A}_{Q^*(\mu)}:= \inner{A}{Q^* A}$. Moreover, Lemma \ref{lem: Q*x} implies that $ c_1^{-1} I_d \preceq Q^* \preceq c_1 I_d$; see Assumption \ref{assumption: bdd_Q} for the definition of $c_1$. As a result, one can obtain
    \begin{align}
        \norm{\frac{1}{n}\sum_{i=1}^{n} \rbr{T^{Q_i}_{Q^*} - I_d}}_{Q^*} 
        &\leq \eigmax{Q^* } \cdot \Fnorm{\frac{1}{n}\sum_{i=1}^{n} \rbr{T^{Q_i}_{Q^*} - I_d} }\nonumber\\
        &\leq c_{\Lambda} \Fnorm{\frac{1}{n}\sum_{i=1}^{n} \rbr{T^{Q_i}_{Q^*} - I_d} }
        \label{eqn: prf_bary_Q*2Frob}
    \end{align}
    Note that by the optimality condition for barycenter, one has $\EE \rbr{T^{Q}_{Q^*} - I_d} = 0$. Then one can apply Lemma \ref{lem: subG_norm} to get that
        \begin{align}
            \gnorm{\Fnorm{\frac{1}{n}\sum_{i=1}^{n} \rbr{T^{Q_i}_{Q^*} - I_d}}} 
            &\lesssim \sup_{\Fnorm{U}\leq 1} \gnorm{ \inner{U}{\frac{1}{n}\sum_{i=1}^{n} \rbr{T^{Q_i}_{Q^*} - I_d}}} \nonumber \\
            &\lesssim \sup_{\Fnorm{U}\leq 1} \frac{1}{\sqrt{n}} \gnorm{\inner{U}{T^Q_{Q^*}-I_d}} \nonumber \\
            &\leq \frac{1}{\sqrt{n}} \gnorm{\Fnorm{T^Q_{Q^*}-I_d }}
            \label{eqn: prf_bary_conc}
        \end{align}
    Recall that $T^Q_S = S^{-1/2} \rbr{S^{1/2}QS^{1/2}}^{1/2} S^{-1/2}$, then one has
    \begin{align}
        \gnorm{\Fnorm{T^Q_{Q^* }-I_d }} 
        &\overset{(i)}{\lesssim}  \gnorm{\Onorm{T^Q_{Q^*}-I_d }}\nonumber\\
        &\leq \gnorm{1 + c_{\Lambda}^{3/2} \eigmax{Q}^{1/2}}\nonumber\\
        &\overset{(ii)}{\leq} \gnorm{1 + c_{\Lambda}^{3/2} c_{\Lambda}^{1/2} \norm{X-\mu}^{1/2}}\nonumber\\
        &\leq \gnorm{1 + c_{\Lambda}^{3/2} c_{\Lambda}^{1/2} \rbr{1 \vee \norm{X-\mu}}} \nonumber\\
        &\leq \gnorm{1 + c_{\Lambda}^2 \rbr{1 + \norm{X-\mu}}} \nonumber\\
        &\lesssim 1
        \label{eqn: prf_bary_compute}
    \end{align}
    Here (i) follows since dimension $d$ is fixed and absorbed into the constant factor independent of $n$, (ii) is a result of Assumption \ref{assumption: bdd_Q}.

    Finally, combining (\ref{eqn: prf_bary_fast}) (\ref{eqn: prf_bary_Q*2Frob}) (\ref{eqn: prf_bary_conc}) and (\ref{eqn: prf_bary_compute}) gives
    \begin{equation}
        \gnorm{W^2 \rbr{\hQ_{\rho}(\bar{X}), Q^*(\mu)}} \lesssim \frac{1}{\sqrt{n}}
    \end{equation}
    which implies (\ref{eqn: lem_bary_W}). 

    \paragraph{Proof of (\ref{eqn: lem_bary_F}):} Apply (\ref{eqn: lem_diff_W_Frob}) in Lemma \ref{lem: diff} to get that under the event $\tE$,
    \begin{equation}
        \begin{aligned}
            \Fnorm{\hQ_{\rho}(\bar{X}) - Q^*} 
            &\leq W^2 \rbr{\hQ_{\rho}(\bar{X}), Q^*} \cdot \frac{1 + \eigminn{-1}{Q^*} \eigmax{\hQ_{\rho}(\bar{X})}}{\eigmax{Q^*} \cdot \eigminn{-2}{Q^*}}\\
            &\lesssim W^2 \rbr{\hQ_{\rho}(\bar{X}), Q^*} \cdot \rbr{1 + \eigmax{\hQ_{\rho}(\bar{X})}}\\
            &\leq C \frac{\sqrt{\log n}}{\sqrt{n}}\cdot \rbr{1 + \eigmax{\hQ_{\rho}(\bar{X})}}
        \end{aligned}
        \label{eqn: prf_bary_F_ineq}
    \end{equation}
    for constant $C>0$ large enough. Note that this implies
    \begin{equation*}
        \begin{aligned}
            \eigmax{\hQ_{\rho}(\bar{X})} &\leq \eigmax{Q^*} + \Fnorm{\hQ_{\rho}(\bar{X}) - Q^*}\\
            &\leq c_{\Lambda} + C \frac{\sqrt{\log n}}{\sqrt{n}}\cdot \rbr{1 + \eigmax{\hQ_{\rho}(\bar{X})}}
        \end{aligned}
    \end{equation*} 
    Solving for $\eigmax{\hQ_{\rho}(\bar{X})}$ gives
    \begin{equation}
        \begin{aligned}
            \eigmax{\hQ_{\rho}(\bar{X})} &\leq \frac{c_{\Lambda} + C \frac{\sqrt{\log n}}{\sqrt{n}}}{1- C \frac{\sqrt{\log n}}{\sqrt{n}}}\\
            &\lesssim 1 \qquad \text{for } n \text{ large enough}
        \end{aligned}
        \label{eqn: prf_bary_const}
    \end{equation}
    Finally, plugging (\ref{eqn: prf_bary_const}) back into (\ref{eqn: prf_bary_F_ineq}) gives (\ref{eqn: lem_bary_F}).

\end{proof}

\subsubsection{Proof of Lemma \ref{lem: Rademacher}}
\label{subsubsec: prf_lem_Rademacher}

\paragraph{Proof of (\ref{eqn: lem_rad_chain}):}
Since $f(z,\theta)$ is continuous in $\theta$ and $\Theta$ is separable, we have for any countable, dense subset $\UU \subset \Theta$, the following equality
\begin{equation*}
    \EE \sup_{\theta \in \Theta} \abs{f(Z,\theta)} = \EE \sup_{S \in \UU} \abs{f(Z,\theta) }
\end{equation*}
holds. By the monotone convergence theorem,  it suffices to assume that $\UU$ is finite and get an upper bound that is independent of the cardinality of $\UU$.

For each $k \in \ZZ$, let $\UU_k \subset \UU$ be a minimal $\epsilon_k$-covering set of $\UU$ where $\epsilon_k$ is defined by 
\begin{equation}
    \epsilon_k = \tau^{-1}(2^{-k})
    \label{eqn: prf_lem_unfm_eps_def}
\end{equation}
Let $N(\epsilon_k,\UU)$ be the cardinality of the minimal $\epsilon_k$-covering set $\UU_k$ of $\UU$. Since $\UU$ is a subset of $\Theta$, $N(\epsilon_k,\UU)$ can be upper bounded by
\begin{equation*}
    \log N(\epsilon_k; \UU) = \log \abs{\UU_k}\leq \log N(\epsilon_k;\Theta)
\end{equation*}
Since $\UU$ is finite, there is a largest $\eta \in \ZZ$ and a smallest integer $H \in \ZZ$ such that
\begin{equation*}
    \UU_{\eta} = \left\{ \theta_0 \right\} \text{ for some } \theta_0 \in \UU, \qquad \UU_H = \UU
\end{equation*}
For each $k \in \ZZ$, define the mapping $\pi_k: \Theta \to \UU_k $ via
\begin{equation*}
    \pi_k(\theta) = \argmin_{\ttheta \in \UU_k} d(\ttheta , \theta)
\end{equation*}
so that $\pi_k(S)$ is the best approximation of $\theta \in \Theta$ from the set $\UU_k$.

For any $\theta \in \UU$, apply the triangle inequality to see that
\begin{equation}
    \begin{aligned}
        \EE \abs{\max_{\theta \in \UU} f(Z,\theta)} 
        &\leq  \EE \abs{\max_{\theta \in \UU} \rbr{f(Z,\theta)-  f(Z,\theta_0)}} + \EE \abs{f(Z;\theta_0)}\\
        &\leq \EE \max_{\theta \in \UU} \abs{f(Z,\theta)-  f(Z,\theta_0)} + \EE \abs{f(Z;\theta_0)}\\
        &\overset{(i)}{=} \EE \max_{\theta \in \UU} \abs{\sum_{k=\eta+1}^{H} \rbr{f(Z, \pi_{k}(\theta)) - f(Z, \pi_{k-1}(\theta))}} + \EE \abs{f(Z;\theta_0)}\\
        &\leq \EE \max_{\theta \in \UU} \sum_{k=\eta+1}^{H} \abs{f(Z, \pi_{k}(\theta)) - f(Z, \pi_{k-1}(\theta))}+ \EE \abs{f(Z;\theta_0)}\\
        &\leq \sum_{k=\eta+1}^{H} \EE \max_{\theta \in \UU} \left| f(Z,\pi_k(\theta)) - f(Z,\pi_{k-1}(\theta)) \right|+ \EE \left| f(Z;\theta_0) \right|
    \end{aligned}
    \label{eqn: prf_lem_unfm_chn}
\end{equation}
Here (i) is a consequence of decomposing $f(Z,\theta)-f(Z,\theta_0)$ as a telescoping sum
\begin{equation*}
    f(Z,\theta)-f(Z,\theta_0)=\sum_{k=\eta+1}^{H} \rbr{f(Z, \pi_{k}(\theta)) - f(Z, \pi_{k-1}(\theta))}
\end{equation*}

For any fixed $\theta \in \Theta$, one has
\begin{align*}
    \left\|f(\cdot,\pi_k(\theta)) - f(\cdot,\pi_{k-1}(\theta))  \right\|_{\psi_2} 
    &\leq \left\|f(\cdot,\pi_k(\theta)) - f(\cdot,\theta)  \right\|_{\psi_2} + \left\|f(\cdot,\theta) - f(\cdot,\pi_{k-1}(\theta))  \right\|_{\psi_2}\\
    &\overset{(I)}{\leq} \tau(\epsilon_{k}) + \tau(\epsilon_{k-1})\\
    &\leq 2 \tau(\epsilon_{k-1})\\
    &\overset{(II)}{=}2 \cdot 2^{-(k-1)}
\end{align*}
Here (I) is a result of (\ref{eqn: lem_rad_lip}), and (II) follows from (\ref{eqn: prf_lem_unfm_eps_def}). Then, one can obtain
\begin{equation}
    \begin{aligned}
        \EE \max_{\theta \in \UU} \abs{f(Z,\pi_k(\theta)) - f(Z,\pi_{k-1}(\theta))}
        &\overset{(i)}{\lesssim} \sqrt{\abs{\UU_{k-1}}\cdot \abs{\UU_k}} \cdot 2 \cdot 2^{-(k-1)}\\
        &\lesssim \sqrt{\log N(\epsilon_k;\Theta)} \cdot 2^{-(k-1)}
    \end{aligned}
    \label{eqn: prf_lem_unfm_chn_k}
\end{equation}
Here (i) follows from the properties of the maximum of finitely many sub-Gaussian random variables and the fact that the maximum is taken over at most $\abs{\UU_{k-1}}\cdot \abs{\UU_k} \leq N(\epsilon_k;\Theta)^2$ random variables. Therefore, 
\begin{equation}
    \begin{aligned}
        \sum_{k=\eta+1}^{H} \EE \max_{\theta \in \UU} \left| f(Z,\pi_k(\theta)) - f(Z,\pi_{k-1}(\theta)) \right| 
        &\lesssim \sum_{k=\eta+1}^{H} \sqrt{\log N(\epsilon_k;\Theta)} \cdot 2^{-(k-1)}\\
        &= \sum_{k=\eta+1}^{H} \sqrt{\log N(\tau^{-1}(2^{-k});\Theta)} \cdot 2^{-(k-1)}\\
        &\lesssim \int_0^{\infty} \sqrt{\log N(\tau^{-1}(t);\Theta)} \cdot dt\\
        &\overset{(I)}{=} \int_0^{\tau(D)} \sqrt{\log N(\tau^{-1}(t);\Theta)} \cdot dt
    \end{aligned}
    \label{eqn: prf_lem_unfm_chn_int}
\end{equation}
Here (I) follows by noticing that $\log N(\tau^{-1}(t);\Theta)=0$ for any $t \geq \tau(D)$.

Combine (\ref{eqn: prf_lem_unfm_chn}) and (\ref{eqn: prf_lem_unfm_chn_int}) to see that
\begin{align*}
    \EE \sup_{\theta \in \Theta} \abs{f(Z;\theta)} \leq C\int_0^{\tau(D)} \sqrt{\log N(\tau^{-1}(t);\Theta)} \cdot dt + \sup_{\theta \in \Theta} \EE \abs{f(Z;\theta)}
\end{align*}
for some constant $C>0$.

\paragraph{Proof of (\ref{eqn: lem_rad_chain_holder}):} Consider two cases $\alpha_0 \in (0,1]$ and $\alpha_0 >1$.

\paragraph{Case I $\alpha_0 \in (0,1]$:}

Note that since $\alpha_0 \in (0,1]$, one has

\begin{align}
    \tau^{-1}(t) &= \left( \frac{\sqrt{n} t}{K} \right) \wedge \left( \frac{\sqrt{n} t}{K} \right)^{1/\alpha_0} \nonumber\\
    &=
    \begin{cases}
        \frac{\sqrt{n} t}{K}     & \sqrt{n}t \geq K\\
        \left( \frac{\sqrt{n} t}{K} \right)^{1/\alpha_0}    & \sqrt{n}t \leq K
    \end{cases}
    \label{eqn: prf_lem_rad_tau_inv}
\end{align}
As a result, one has
\begin{equation}
    \begin{aligned}
        &\quad \int_0^{\tau(D)} \sqrt{\log N(\tau^{-1}(t);\Theta)} \cdot dt\\
        &\overset{(i)}{=} \frac{1}{\sqrt{n}}\int_0^{\sqrt{n}\tau(D)} \sqrt{\log N\left( \tau^{-1}\left( \frac{t}{\sqrt{n}} \right);\Theta \right)} \cdot dt\\
        &\overset{(ii)}{=} \frac{1}{\sqrt{n}}\int_0^{K} \sqrt{\log N\left( (t/K)^{1/\alpha_0};\Theta \right)} \cdot dt + \frac{1}{\sqrt{n}} \int_K^{\sqrt{n}\tau(D)} \sqrt{\log N\left( t/K;\Theta \right)} \cdot dt\\
        &\lesssim  \frac{1}{\sqrt{n}}\int_0^{K} \sqrt{\log^+ \left( \frac{D}{(t/K)^{1/\alpha_0}} \right)} \cdot dt + \frac{1}{\sqrt{n}}\int_K^{\sqrt{n}\tau(D)} \sqrt{\log^+ \left( \frac{D}{(t/K)} \right)} \cdot dt\\
        &\overset{(iii)}{=}\frac{K}{\sqrt{n}} \underbrace{\int_0^{1} \sqrt{\log^+ \left( \frac{D}{t^{1/\alpha_0}} \right)} \cdot dt}_{a_1} + \frac{K}{\sqrt{n}} \underbrace{\int_1^{\sqrt{n}\tau(D)/K} \sqrt{\log^+ \left( \frac{D}{t} \right)} \cdot dt}_{a_2}
    \end{aligned}
    \label{eqn: prf_rad_holder_decomp}
\end{equation}
Here (i), (iii) follows from a change of variables, and (ii) follows from (\ref{eqn: prf_lem_rad_tau_inv}).

\begin{itemize}
    \item $a_1$: one has
    \begin{equation}
        \begin{aligned}
            a_1 &\overset{(i)}{\leq} \int_0^1 \sqrt{\log^+ \left( \frac{1}{t^{1/\alpha_0}} \right)}dt + \int_0^1 \sqrt{\log^+ D} dt\\
            &\overset{(ii)}{\lesssim }\sqrt{\log^+ D}
        \end{aligned}
        \label{eqn: prf_rad_holder_a1}
    \end{equation}
    Here (i) follows from the inequality $\sqrt{s+t} \leq \sqrt{s} + \sqrt{t}$ for $s,t \geq 0$, and (ii) follows from the assumption that $D \gtrsim 1$.
    \item $a_2$: for $D \gtrsim 1$,
    \begin{equation}
        \begin{aligned}
            a_2 &\overset{(i)}{\leq }\int_1^{\sqrt{n} \tau(D)/K} \sqrt{\log^+ D}dt\\
            &\overset{(ii)}{\leq} \frac{\sqrt{n} \tau(D)}{K} \cdot \sqrt{\log^+ D}\\
            &= \rbr{D \vee D^{\alpha_0}} \sqrt{\log^+ D}\\
            &\overset{(iii)}{\lesssim} D \sqrt{\log^+ D}
        \end{aligned}
        \label{eqn: prf_rad_holder_a2}
    \end{equation}
    Here (i) a consequence of the fact that $D/t \leq D$ for $t\geq 1$, (ii) results from substituting the definition of $\tau(\cdot)$ and (iii) follows from the assumption $\alpha_0 \in (0,1]$.
\end{itemize}
Combine (\ref{eqn: prf_rad_holder_decomp}), (\ref{eqn: prf_rad_holder_a1}) and (\ref{eqn: prf_rad_holder_a2}) to see that
\begin{equation*}
    \int_0^{\tau(D)} \sqrt{\log N(\tau^{-1}(t);\Theta)} \cdot dt \lesssim \frac{K}{\sqrt{n}}  D \sqrt{\log^+ D}
\end{equation*}
which proves (\ref{eqn: lem_rad_chain_holder}) for $\alpha_0 \in (0,1] $.

\paragraph{Case II $\alpha_0 > 1$:} for $\alpha_0 >1$, one has
\begin{equation}
    \begin{aligned}
        \tau^{-1}(t) &= \left( \frac{\sqrt{n} t}{K} \right) \wedge \left( \frac{\sqrt{n} t}{K} \right)^{1/\alpha_0}\\
        &=
        \begin{cases}
            \frac{\sqrt{n} t}{K}     & \sqrt{n}t \leq K\\
            \left( \frac{\sqrt{n} t}{K} \right)^{1/\alpha_0}    & \sqrt{n}t > K
        \end{cases}
    \end{aligned}
    \label{eqn: prf_lem_rad_tau_inv2}
\end{equation}
As a result, one can obtain
\begin{equation}
    \begin{aligned}
        &\quad \int_0^{\tau(D)} \sqrt{\log N(\tau^{-1}(t);\Theta)} \cdot dt\\
        &\overset{(i)}{=} \frac{1}{\sqrt{n}}\int_0^{\sqrt{n}\tau(D)} \sqrt{\log N\left( \tau^{-1}\left( \frac{t}{\sqrt{n}} \right);\Theta \right)} \cdot dt\\
        &\overset{(ii)}{=} \frac{1}{\sqrt{n}}\int_0^{K} \sqrt{\log N\left( t/K;\Theta \right)} \cdot dt + \frac{1}{\sqrt{n}} \int_K^{\sqrt{n}\tau(D)} \sqrt{\log N\left( (t/K)^{1/\alpha_0};\Theta \right)} \cdot dt\\
        &\lesssim  \frac{1}{\sqrt{n}}\int_0^{K} \sqrt{\log^+ \left( \frac{D}{t/K} \right)} \cdot dt + \frac{1}{\sqrt{n}}\int_K^{\sqrt{n}\tau(D)} \sqrt{\log^+ \left( \frac{D}{(t/K)^{1/\alpha_0}} \right)} \cdot dt\\
        &\overset{(iii)}{=}\frac{K}{\sqrt{n}} \underbrace{\int_0^{1} \sqrt{\log^+ \left( \frac{D}{t} \right)} \cdot dt}_{a_3} + \frac{K}{\sqrt{n}} \underbrace{\int_1^{\sqrt{n}\tau(D)/K} \sqrt{\log^+ \left( \frac{D}{t^{1/\alpha_0}} \right)} \cdot dt}_{a_4}
    \end{aligned}
    \label{eqn: prf_rad_holder_decomp2}
\end{equation}
Here (i), (iii) follows from a change of variables, and (ii) follows from (\ref{eqn: prf_lem_rad_tau_inv2}).

\begin{itemize}
    \item $a_3$: one has
    \begin{equation}
        \begin{aligned}
            a_3 &\overset{(i)}{\leq} \int_0^1 \sqrt{\log^+ \left( \frac{1}{t } \right)}dt + \int_0^1 \sqrt{\log^+ D} dt\\
            &\overset{(ii)}{\lesssim }\sqrt{\log^+ D}
        \end{aligned}
        \label{eqn: prf_rad_holder_a3}
    \end{equation}
    Here (i) follows from the inequality $\sqrt{s+t} \leq \sqrt{s} + \sqrt{t}$ for $s,t \geq 0$, and (ii) follows from the assumption that $D \gtrsim 1$.
    \item $a_4$: for $D \gtrsim 1$,
    \begin{equation}
        \begin{aligned}
            a_4 &\overset{(i)}{\leq }\int_1^{\sqrt{n} \tau(D)/K} \sqrt{\log^+ D}dt\\
            &\overset{(ii)}{\leq} \frac{\sqrt{n} \tau(D)}{K} \cdot \sqrt{\log^+ D}\\
            &= \rbr{D \vee D^{\alpha_0}} \sqrt{\log^+ D}\\
            &\overset{(iii)}{\lesssim} D^{\alpha_0} \sqrt{\log^+ D}
        \end{aligned}
        \label{eqn: prf_rad_holder_a4}
    \end{equation}
    Here (i) a consequence of the fact that $D/\rbr{t^{1/\alpha_0}} \leq D$ for $t\geq 1$, (ii) results from substituting the definition of $\tau(\cdot)$ and (iii) follows from the assumption $\alpha_0 >1 $.
\end{itemize}
Combine (\ref{eqn: prf_rad_holder_decomp2}), (\ref{eqn: prf_rad_holder_a3}) and (\ref{eqn: prf_rad_holder_a4}) to see that
\begin{equation*}
    \int_0^{\tau(D)} \sqrt{\log N(\tau^{-1}(t);\Theta)} \cdot dt \lesssim \frac{K}{\sqrt{n}}  D^{\alpha_0} \sqrt{\log^+ D}
\end{equation*}
which proves (\ref{eqn: lem_rad_chain_holder}) for $\alpha_0 >1$.

\paragraph{Proof of (\ref{eqn: lem_rad_lip_avg}):} Let $\Psi_n(\theta) := \frac{1}{n}\sum_{i=1}^{n} \psi(Z_i;\theta)$, then one has
\begin{equation*}
    \begin{aligned}
        &\quad \gnorm{f(Z_1^n;\theta) - f(Z_1^n;\ttheta)} \\
        &= \gnorm{ \opnorm{\Psi_n(\theta) - \EE \Psi_n(\theta)} - \opnorm{\Psi_n(\ttheta) - \EE \Psi_n(\ttheta)}}\\
        &\leq \gnorm{\opnorm{\Psi_n(\theta) - \EE \Psi_n(\theta)-\Psi_n(\ttheta) + \EE \Psi_n(\ttheta)}}\\
        &\overset{(i)}{\lesssim} \sup_{u,v \in \SS^{m-1}} \gnorm{\inner{u}{ \sbr{\tPsi_n(\theta) - \EE \tPsi_n(\theta)} - \sbr{\tPsi_n(\ttheta) + \tPsi_n(\ttheta)} \cdot v^{\otimes k}}} \\
        &\overset{(ii)}{\lesssim} \frac{1}{\sqrt{n}} \sup_{u,v \in \SS^{m-1}}
        \gnorm{\inner{u}{ \left[ \psi(Z;\theta)  - \psi(Z;\ttheta)\right] - \EE \sbr{\psi(Z;\theta) - \psi(Z;\ttheta)} \cdot v^{\otimes k}} }\\
        &\overset{(iii)}{\lesssim} \frac{1}{\sqrt{n}} \sup_{u,v \in \SS^{m-1}} \gnorm{\inner{u}{\left[ \psi(Z;\theta) - \psi(Z;\ttheta) \right] \cdot v^{\otimes k}}}\\
        &\leq \frac{1}{\sqrt{n}} \sup_{v \in \SS^{m-1}} \gnorm{\Fnorm{\sbr{\psi(Z;\theta) - \psi(Z;\ttheta)} \cdot v^{\otimes k}}}\\
        &\overset{(iv)}{\leq} \frac{1}{\sqrt{n}} \sup_{v \in \SS^{m-1}} \gnorm{\opnorm{\psi(Z;\theta) - \psi(Z;\ttheta)} \cdot \Fnorm{v}^k}\\
        &\leq \frac{1}{\sqrt{n}} \gnorm{\opnorm{\psi(Z;\theta) - \psi(Z;\ttheta) }}
    \end{aligned}
\end{equation*}
Here (i) is a result of Lemma \ref{lem: subG_norm}, (ii) arises due to independence, (iii) follows from Lemma 2.6.8 in \citet{vershynin}, and (iv) is derived from the definition of $\opnorm{\cdot}$. The proof is then complete.

%% file: tech_lem_other.tex
This section collects properties of $F(\cdot,\cdot)$ and $Q^*(\cdot)$ that are needed later in the proof.

\begin{lemma}
    \label{lem: Q*x}
    Assume Assumption \ref{assumption: X} and \ref{assumption: bdd_Q} holds. Then the following inequalities
    \begin{align}
        \gamma_{\Lambda}(\norm{x-\mu})^{-1} I_d \preceq &Q^*(x) \preceq  \gamma_{\Lambda}(\norm{x-\mu}) I_d, \qquad \forall x \in \supp X 
        \label{eqn: lem_Q*_x1}
    \end{align}
    Moreover, there exist constant $C_0^*, C_1^* >0$ which only depend on $c_{\Sigma}, C_{\psi_2}$ from Assumption \ref{assumption: X} and $c_{\Lambda},C_1$ (from Assumption \ref{assumption: bdd_Q}) such that
    \begin{align}
        Q^*(x) \preceq \sbr{C_0^* + C_1^* \norm{x-\mu}}^2 I_d
        \label{eqn: lem_Q^*2}
    \end{align}
\end{lemma}

\begin{proof}
    Without loss of generality, assume $X$ have expectation $\mu=0$ and covariance $\Sigma = I_p$.  Denote $\lambda_{1}(x)= \eigmax{Q^*(x)}$ and $\lambda_{d}(x)= \eigmin{Q^*(x)}$

    \noindent\textit{Proof of (\ref{eqn: lem_Q*_x1}):}
    By the optimality condition of $Q^*(x) = \argmin_S \EE \sbr{W^2(S,Q)\big| X=x}$. By the differential properties of $W^2$ in Lemma \ref{lem: diff}, one can obtain
    \begin{align*}
        I_d = \EE \sbr{T_{Q^*(x)}^{Q} \big| X=x}
    \end{align*}
    Recall that $T_{Q^*(x)}^Q = {Q^*(x)}^{-1/2} \left( Q^*(x)^{1/2}Q Q^*(x)^{1/2}\right)^{1/2} Q^*(x)^{-1/2}$, multiplying both sides of the equation with $Q^*(x)^{1/2}$ then gives
    \begin{equation*}
        Q^*(x)=\EE \sbr{\rbr{Q^*(x)^{1/2} Q Q^*(x)^{1/2}}^{1/2} \big| X=x}
    \end{equation*}
    Then one can obtain
    \begin{equation*}
        \begin{aligned}
            \lambda_{1}(x) &= \eigmax{\EE \sbr{\rbr{Q^*(x)^{1/2} Q Q^*(x)^{1/2}}^{1/2} \mid X=x} } \\
            &\leq \EE \left[ \lambda_{\max} \left( \left( Q^*(x)^{1/2} Q Q^*(x)^{1/2}\right)^{1/2}  \right)\mid X=x \right]\\
            &\leq \gamma_{\Lambda}(\norm{x-\mu})^{1/2} \lambda_{1}(x)^{1/2}
        \end{aligned}
    \end{equation*}
    Here the second line follows from the convexity of the largest eigenvalue function $\lambda_{\max}\left( \cdot \right)$ over PSD matrices. The last inequality follows from Assumption \ref{assumption: bdd_Q}. Therefore, one can readily obtain $\lambda_{1}(x) \leq \gamma_{\Lambda}(\norm{x-\mu})$. The lower bound $\lambda_{d}(x) \geq \gamma_{\Lambda}(\norm{x-\mu})^{-1}$ can be derived similarly by the concavity of the function $\lambda_{\min} \left( \cdot \right)$.
    
    \noindent\textit{Proof of (\ref{eqn: lem_Q^*2}):}
    By the optimality condition of 
    $Q^*(x)=\argmin_S \EE w(x,X)W^2(S,Q)$ , one has
    \begin{align*}
        0 = \EE w(x,X) \rbr{T_{Q^*(x)}^Q - I_d}
    \end{align*}
    Note that $\EE w(x,X) = 1$ for any $x$, hence one has
    \begin{align*}
        I_d = \EE w(x,X) T_{Q^*(x)}^{Q}
    \end{align*}
    Recall that $T_{Q^*(x)}^Q = {Q^*(x)}^{-1/2} \left( Q^*(x)^{1/2}Q Q^*(x)^{1/2}\right)^{1/2} Q^*(x)^{-1/2}$, multiplying both sides of the equation with $Q^*(x)^{1/2}$ then gives
    \begin{align*}
        Q^*(x) &= \EE w(x,X) \left( Q^*(x)^{1/2}Q_i Q^*(x)^{1/2}\right)^{1/2}
    \end{align*}
    Taking the largest eigenvalue on both sides, one has
    \begin{align*}
        \lambda_{1}(x)  &\leq \EE \abs{w(x,X)} \cdot \Onorm{\left( Q^*(x)^{1/2}Q Q^*(x)^{1/2}\right)^{1/2}}\\
        &\leq \EE \abs{w(x,X)} \cdot \lambda_{\max}(Q)^{1/2} \lambda_{1}^{1/2}(x)
    \end{align*}
    which implies
    \begin{align*}
        \lambda_{1}(x)
        &\leq \sbr{\EE \rbr{1+\norm{x} \cdot \norm{X}} \lambda_{\max}(Q)^{1/2} }^2 
    \end{align*}
    Under Assumption \ref{assumption: X} and Assumption \ref{assumption: bdd_Q}, we have $\lambda_{\max}(Q_i) \leq \gamma_{\Lambda}(\norm{X_i})$, which implies that
    \begin{align*}
        \lambda_1(x) &\leq \sbr{\EE \gamma_{\Lambda}(\norm{X})^{1/2} + \norm{x} \EE \norm{X} \gamma_{\Lambda}(\norm{X})^{1/2} }^2\\
        &\leq \sbr{C_0^* + C_1^* \norm{x}}^2
    \end{align*}
    where $C_0^*, C_1^* >0$ are finite constants only depend on $c_{\Sigma}, C_{\psi_2}$ from Assumption \ref{assumption: X} and $c_{\Lambda},C_\Lambda$ (from Assumption \ref{assumption: bdd_Q}).
\end{proof}

\begin{lemma}
    \label{lem: F_ind}
    Assume Assumption \ref{assumption: X} and \ref{assumption: bdd_Q} hold. If $X$ and $Q$ are independent, then the following holds.
    \begin{itemize}
        \item $Q^*(x)\equiv Q^*(\mu)$
        \item Assumption \ref{assumption: minimizer_local} holds for $C_\lambda=0$ and $c_\lambda$ large enough.
        \item Assumption \ref{assumption: minimizer_global} holds for $\delta_F=1$, $\alpha_F=2$, $C_F=1$ and some constant $c_F$ large enough.
    \end{itemize}
\end{lemma}

\begin{proof}
    Independence between $X$ and $Q$ implies that
    \begin{align*}
        F(x,S) &= \EE w(x,X)W^2(S,Q)\\
        &=\EE w(x,X) \cdot \EE W^2(S,Q)\\
        &=\EE W^2(S,Q)
    \end{align*}
    The existence and uniqueness of the minimizer of $F(x,\cdot)$ follow from \citet[Proposition 3.2.3, Proposition 3.2.7]{invitation}. Moreover, since $\EE W^2(S,Q)$ does not depend on $x$, we have $Q^*(x)\equiv Q^*(\mu)$. We denote $Q^*:=Q^*(\mu)$ from now on. Note that Lemma \ref{lem: Q*x} implies that $Q^* \in \cS_d(c_{\Lambda}^{-1},c_{\Lambda})$
    \paragraph{Verification of Assumption \ref{assumption: minimizer_local}:} Again by independence, one has
    \begin{align*}
         \EE \sbr{-w(x,X)dT_{Q^*}^{Q}}&= \EE \sbr{-dT_{Q^*}^{Q}}
    \end{align*}
    Therefore, one can obtain
    \begin{align*}
        \eigmin{\EE \sbr{-w(x,X)dT_{Q^*}^{Q}}} &= \eigmin{\EE \sbr{-dT_{Q^*}^{Q}}}\\
        &\overset{(i)}{\geq} \EE \frac{\lambda_{\min}^{1/2}( Q^{1/2}Q^*Q^{1/2} )}{2} \cdot \lambda_{\min}^2(Q^{*-1})\\
        &\gtrsim \EE \lambda_{\min}^{1/2}(Q)\\
        &\overset{(ii)}{\gtrsim} \EE \sbr{\gamma_{\Lambda}(\norm{X-\mu})}^{-1/2}\\
        &\gtrsim \EE \frac{1}{1\vee \norm{X-\mu}^{C_\Lambda/2}}\\
        &\overset{(iii)}{\geq} c_{\mathrm{eig}}
    \end{align*}
    for some constant $c_{\mathrm{eig}}$ small enough. Here (i) follows form Lemma \ref{lem: diff}, (ii) is a result of Assumption \ref{assumption: bdd_Q} and (iii) is due to Assumption \ref{assumption: X}. The proof is then complete.
    
    \paragraph{Verification of Assumption \ref{assumption: minimizer_global}:}
    
    In order to verify (\ref{eqn: assmptn_min_sep}). Let us denote $Q'=Q^{*-1/2} S Q^{*-1/2}$. Then (\ref{eqn: lem_diff_quad_approx}) and (\ref{eqn: lem_diff_dT_eigen}) in Lemma \ref{lem: diff} imply the following lower bound.
    \begin{equation}
        \begin{aligned}
            &F(x,S) - F(x,Q^*)\\
            =&\EE \sbr{W^2(S,Q) -W^2(Q^*,Q)} \\
            \geq &\EE \inner{I_d-T^Q_{Q^*} }{ S-Q^*} + \EE \frac{2}{\left( 1+\lambda_{\max}^{1/2}(Q'(x)) \right)^2} \inner{-dT_{Q^*}^{Q}(S- Q^*)}{ S-Q^*}\\
            \overset{(i)}{\geq} &\EE \frac{2}{\left( 1+\lambda_{\max}^{1/2}(Q'(x)) \right)^2} \cdot\frac{\lambda_{\min}^{1/2} \rbr{Q^{1/2}Q^*Q^{1/2}}}{2}\cdot \Fnorm{Q^{*-1/2}(S- Q^*)Q^{*-1/2}}^2
        \end{aligned}
        \label{eqn: F_ind_quad1}
    \end{equation}
    Here (i) follows from the fact that $\EE(T^Q_{Q^*}-I_d)=0$ which is a consequence of the optimality condition and independence.

    Triangle inequality implies that
    \begin{align*}
        \cbr{S \in \cS_d^{++}: \delta \leq \Fnorm{S-Q^*}\leq \Delta} \subset  \cbr{S\in \cS_d^{++}: \Onorm{S} \leq c_{\Lambda} +\Delta}
    \end{align*}
    which combined with (\ref{eqn: F_ind_quad1}) implies for any $\Delta \geq \delta$,
    \begin{align*}
        &\inf_{\delta \leq \Fnorm{S-Q^*}\leq \Delta} F(x,S) - F(x,Q^*)\\
        \geq & \inf_{\delta \leq \Fnorm{S-Q^*}\leq \Delta} \EE \frac{2}{\left( 1+\lambda_{\max}^{1/2}(Q') \right)^2} \cdot\frac{\lambda_{\min}^{1/2} \rbr{Q^{1/2}Q^* Q^{1/2}}}{2}\cdot \Fnorm{Q^{*-1/2}(S- Q^*)Q^{*-1/2}}^2\\
        \overset{(i)}{\gtrsim} &\inf_{\delta \leq \Fnorm{S-Q^*}\leq \Delta} \frac{1}{\rbr{1+\Delta^{1/2}}^2}  \cdot \Fnorm{S-Q^*}^2  \\
        \overset{(ii)}{\gtrsim} &\inf_{\delta \leq \Fnorm{S-Q^*}\leq \Delta} \frac{1}{1+\Delta} \cdot \Onorm{S-Q^*}^2\\
        \geq &c_{\mathrm{sep}} \frac{1}{\Delta} \delta^2
    \end{align*}
    where $c_{\mathrm{sep}}>0$ is a constant independent of $n$ that is small enough. Here (i) follows from the inequalities that $\eigmax{Q'} \lesssim \eigmax{S}\lesssim \Delta$, $\eigmin{Q^{1/2}Q^*Q^{1/2}} \gtrsim \eigmin{Q}$ and $\EE \eigmin{Q}^{1/2} \gtrsim 1$; (ii) holds by choosing $\delta_F=1$ so that $\Delta \geq \delta_F = 1$.

    Combining results above, Assumption \ref{assumption: minimizer_global} holds with with $\delta_F=1$, $\alpha_F = 2$, $C_F = 1$ and $c_F = 1/c_{\mathrm{sep}}$. The proof is then  complete.

\end{proof}

Recall the definition of $\gamma_{\Lambda}(\cdot)$ from Assumption \ref{assumption: bdd_Q} and the definition of $\alpha_F,\delta_F, C_F$ from Assumption \ref{assumption: minimizer_global}. Denote $M_L:= \gamma_{\Lambda}(L)$ for any $L>0$. With these notations in place, we now present the continuity theorem for $Q^*(\cdot)$ as follows.
\begin{lemma}[H\"older continuity of $Q^*(\cdot)$]
    \label{lem: holder}
    Suppose Assumption \ref{assumption: X}-\ref{assumption: minimizer_local} hold. Then for any $L\geq e$ that satisfies $M_L>\delta_F$, any $x,\tx \in B_\mu(L)$,
    \begin{equation}
        \norm{Q^*(x) -Q^*(\tx)} \leq C_H \max \cbr{\rbr{L^{C_F}M_L^{C_F+1} }^{1/\alpha_F}\norm{x-\tx}^{1/\alpha_F} , L^{C_F}M_L^{C_F+1} \norm{x-\tx}}
        \label{eqn: lem_holder}
    \end{equation}
    holds as long as the constant $C_H$ is large enough. Here $C_H$ is independent of $n$ but depends on $d$.
\end{lemma}

\begin{proof}
    Note that $M_L \geq 1$ by the properties of $\gamma_{\Lambda}$ in Assumption \ref{assumption: bdd_Q}. The proof is divided into 3 steps.
    \begin{enumerate}
        \item First, we should that $F(x,S)$ is Lipschitz in $x$ for any fixed $S \in \cS_d(M_L^{-1},M_L)$.
        \item Next, the Lipschitzness in step 1 and Assumption \ref{assumption: minimizer_global} imply the local H\"older continuity of $Q^*(x)$ with respect to $x$.
        \item Then, an argument based on linear interpolation between $x$ and $\tx$ implies Lipschitz continuity for $x,\tx$ separated far apart.
    \end{enumerate}
    Combining the 3 step above finally gives (\ref{eqn: lem_holder}).

    To start with, note that Lemma \ref{lem: Q*x} implies 
    \begin{equation*}
        Q^*(x) \in \cS_d(M_L^{-1},M_L), \quad \text{for any } x \in B(\mu,L)
    \end{equation*}
    Recall that $F(x,S)=\EE \sbr{w(x,X)W^2(S,Q)}$.

\paragraph*{1. $F$ is Lipschitz in $x$:} for any $x,\tx \in B_\mu(L)$, one can obtain

\begin{align}
    \abs{F(x,S) - F(\tx, S)} &= \abs{\EE (x-\tx)\Sigma^{-1}(X-\mu)W^2(S,Q)} \nonumber\\
    &\lesssim \norm{x-\tx} \cdot \EE \norm{(X-\mu)W^2(S,Q)} \nonumber\\
    &\overset{(i)}{\lesssim} \norm{x-\tx} \cdot \EE \norm{X-\mu} \left( \eigmax{S}+ \eigmax{Q} \right) \nonumber\\
    &=\norm{x-\tx} \cdot \sbr{\EE \norm{X-\mu}\eigmax{S} + \EE \norm{X-\mu} \eigmax{Q}}\nonumber\\
    &\overset{(ii)}{\lesssim} \norm{x-\tx} \left(\eigmax{S} + 1 \right) \label{eqn: prf_lem_holder_lip_x}
\end{align}
    
Here (i) follows from (\ref{eqn: diff_W_sup}) in Lemma \ref{lem: diff}, and (ii) follows from the concentration of $X$ (Assumption \ref{assumption: X}) and boundedness of $Q$ given $X$ (Assumption \ref{assumption: bdd_Q}).

\paragraph*{2. Local H\"older continuity:} for any $x,\tx \in B_\mu(L)$, one has
\begin{equation}
    \begin{aligned}
        0\leq &F(x,Q^*(\tx)) - F(x,Q^*(x))\\
        = &\rbr{F(x,Q^*(\tx)) - F(\tx, Q^*(\tx))} +  \underbrace{\rbr{F(\tx,Q^*(\tx)) - F( \tx, Q^*(x))}}_{\leq 0} +  \rbr{F\left( \tx, Q^*(x) \right) - F(x,Q^*(x))}\\
        \leq &\abs{F(x,Q^*(\tx)) - F(\tx, Q^*(\tx))} + \abs{F\left( \tx, Q^*(x) \right) - F(x,Q^*(x))}\\
        \overset{(i)}{\lesssim}& \norm{x-\tx}\cdot \rbr{\eigmax{Q^*(x)} + \eigmax{Q^*(\tx)} + 2 }\\
        \overset{(ii)}{\leq} &C M_L \norm{x-\tx}
    \end{aligned}
    \label{eqn: prf_lem_holder_lip_x_*}
\end{equation}
provided that constant $C>0$ is large enough. Here (i) follows from (\ref{eqn: prf_lem_holder_lip_x}) and (ii) follows since $Q^*(x), Q^*(\tx) \in \cS_d(M_L^{-1},M_L)$ for $x, \tx \in B_\mu(L)$ by Lemma \ref{lem: Q*x}.

Note also that for any $x,\tx \in B_\mu(L)$, one has $\norm{Q^*(x) - Q^*(\tx)}_F \leq \sqrt{d} M_L$.
Recall the definition of $\delta_F, \alpha_F, \gamma_2(\cdot,\cdot)$ from Assumption \ref{assumption: minimizer_global}.
Then for any $x,\tx \in B_\mu(L)$ that satisfy
\begin{align}
    \norm{x- \tx}\leq \frac{\delta_F^{\alpha_F}}{C M_L\gamma_2(L,\sqrt{d} M_L)}=:t_0
    \label{eqn: prf_lem_holder_local_x}
\end{align}
one can obtain
\begin{equation}
    \norm{Q^*(x)-Q^*(\tx)}_F \leq  \rbr{C M_L \gamma_2(L,\sqrt{d}M_L)}^{1/\alpha_F} \cdot  \norm{x-\tx}^{1/\alpha_F}
    \label{eqn: prf_lem_holder_local_holder}
\end{equation}
To see this, note that first we have $\norm{Q^*(x) - Q^*(\tx)}_F < \delta_F$ since otherwise Assumption \ref{assumption: minimizer_global} then implies that
\begin{equation*}
    F(x,Q^*(\tx)) - F(x,Q^*(x)) > \frac{\delta_F^{\alpha_F}}{\gamma_2(L,\sqrt{d}M_L)} \overset{(i)}{\geq} C M_L \norm{x-\tx}
\end{equation*}
Here $(i)$ results from (\ref{eqn: prf_lem_holder_local_x}).
This is a contradiction with (\ref{eqn: prf_lem_holder_lip_x_*}). Next applying Assumption \ref{assumption: minimizer_global} with $\delta = \norm{Q^*(x)-Q^*(\tx)}_F$ and $\Delta = \sqrt{d} M_L$ gives
\begin{align*}
    F(x,Q^*(\tx)) - F(x,Q^*(x)) \geq \frac{\norm{Q^*(x)-Q^*(\tx)}^{\alpha_F}}{\gamma_2(L, \sqrt{d}M_L)}
\end{align*}
which combined with (\ref{eqn: prf_lem_holder_lip_x}) implies (\ref{eqn: prf_lem_holder_local_holder}).

\paragraph*{3. $x,\tx$ far apart:} for any $x,\tx \in B_\mu(L)$ such that $\norm{x-\tx}>t_0$, let $K:=  \ceil{\norm{x-\tx}/t_0}$ and 
\begin{equation*}
    x_k := \rbr{1-\frac{k}{K}} x + \frac{k}{K} \tx, \quad k=0,\ldots,K
\end{equation*}
Then $\norm{x_{k+1}-x_k} \leq t_0$, and one can apply (\ref{eqn: prf_lem_holder_local_holder}) to see that
\begin{equation}
    \begin{aligned}
        \norm{Q^*(x)-Q^*(\tx)} &\leq \sum_{k=0}^{K-1} \norm{Q^*(x_{k+1})- Q^*(x_k)}\\
        &\leq \sum_{k=0}^{K-1} \rbr{C M_L\gamma_2(L,\sqrt{d} M_L)}^{1/\alpha_F} \cdot  \norm{x_{k+1}-x_k}^{1/\alpha_F}\\
        &\leq \sum_{k=0}^{K-1} \rbr{C M_L\gamma_2(L, \sqrt{d}M_L)}^{1/\alpha_F} \cdot  t_0^{1/\alpha_F}\\
        &= K \delta_F\\
        &\leq 2 \frac{\norm{x-\tx}}{t_0} \delta_F\\
        &\lesssim M_L \gamma_2(L, \sqrt{d}M_L) \cdot \norm{x-\tx}
    \end{aligned}
    \label{eqn: prf_lem_holder_global_lip}
\end{equation}
Here the last line follows since $\delta_F$ is a constant defined in Assumption \ref{assumption: minimizer_global}.

Finally, combine (\ref{eqn: prf_lem_holder_local_holder}) and (\ref{eqn: prf_lem_holder_global_lip}) to see that for any $x,\tx \in B_{\mu}(L)$, one has
\begin{equation*}
    \norm{Q^*(x)- Q^*(\tx)} \leq C_H \cdot \max \cbr{\rbr{M_L \gamma_2(L,M_L)}^{1/\alpha_F}\norm{x-\tx}^{1/\alpha_F} , M_L \gamma_2(L,M_L) \norm{x-\tx}}
\end{equation*}
for some constant $C_H$ independent of $n$ (but depends on $d$).
\end{proof}

%% file: prf_unfm_conv.tex
\subsection{Proof outline for uniform convergence}

We start by introducing some notations. Define 
\begin{equation}
    \begin{aligned}
        \hF_{\rho}(x,S) &= \frac{1}{n} \sum_{i=1}^{n} w_{n,\rho}(x,X_i)W^2(S,Q_i)\\
        \hA_{\rho} &:= \frac{1}{n} \sum_{i=1}^{n} w_{n,\rho}(x,X_i)\left( T^{Q_i}_{Q^*(x)} - I_d \right)\\
        \hPhi_{\rho}(x) &:= \frac{1}{n} \sum_{i=1}^{n} w_{n,\rho}(x,X_i)  dT^{Q_i}_{Q^*(x)}\\
        \hPsi_{\rho}\rbr{x, S} &:= \frac{1}{n} \sum_{i=1}^{n} w_{n,\rho}(x,X_i)  d^2 T^{Q_i}_{S}
    \end{aligned}
\end{equation}
and 
\begin{equation}
    \begin{aligned}
        \tF_n(x,S)&:= \frac{1}{n} \sum_{i=1}^{n} w(x,X_i)W^2(S,Q_i) \\
        \tA_n(x)& := \frac{1}{n} \sum_{i=1}^{n} w(x,X_i) \left( T_{Q^*(x)}^{Q_i} - I_d \right)\\
        \tPhi_n(x)&:=\frac{1}{n} \sum_{i=1}^{n} w(x,X_i)  dT^{Q_i}_{Q^*(x)}\\
        \tPsi_n(x, S)&:=\frac{1}{n} \sum_{i=1}^{n} w(x,X_i)  d^2 T^{Q_i}_{S} 
    \end{aligned}
    \label{eqn: prf_conv_def_til}
\end{equation}
For any $t \geq 1$, define the event $ E_0(t) $ as
\begin{align*}
    E_0(t) &:= \cbr{\norm{X_i - \mu}\leq t, Q_i \in \cS_d(M_t^{-1}, M_t), i \in [n]}, \quad
    \text{where }   M_t := \gamma_\Lambda(t) \nonumber
\end{align*}
Specifically, let $L_{\tau}= \sqrt{(1+\tau)\log n}$ for $\tau \geq 0$ and define
\begin{align}
    E_0 = E_0 \rbr{C_{\psi_2}L_{\tau}}, \label{eqn: prf_conv_E0}
\end{align}
Under Assumption \ref{assumption: X} and \ref{assumption: bdd_Q}, we have $\PP(E_0^c)\leq 2n^{-(1+\tau)}$. 

We now give an outline of the proof strategy for our uniform convergence theory, namely, Theorems \ref{thm: unfm_conv} and Corollary \ref{cor: W_moment}.
\begin{itemize}
    \item Lemma \ref{lem: bdd} derives a tail probability bound for the largest eigenvalue of $\hQ_{\rho}(x)$.
    \item Lemma \ref{lem: F_conc} obtains uniform concentration of $\hF_{\rho}$.
    \item Lemma \ref{lem: prf_conv_lem_slow} shows the uniform consistency of $\hQ_{\rho}(x)$ with a potentially slow rate;
    \item Lemma \ref{lem: conc_tld_exp}-\ref{lem: conv_smmry} deliver upper bounds for various quantities related to $\hA_{\rho}(x)$, $\hPhi_{\rho}(x)$ and $\hPsi_{\rho}(x;S)$.
    \item Lemma \ref{lem: fast} arrives at the non-asymptotic $\sqrt{n}$-uniform convergence in the Frobenius norm.
    \item Lemma \ref{lem: F_moment} combines Lemma \ref{lem: bdd} and \ref{lem: fast} to get a $\sqrt{n}$-moment bound for the uniform Frobenius error.
    \item Lemma \ref{lem: W_moment} extends the above moment of Frobenius error result to Wasserstein distance.
\end{itemize}
Theorem \ref{thm: unfm_conv} and Corollary \ref{cor: W_moment} then follow immediately by combining Lemmas \ref{lem: bdd}-\ref{lem: W_moment}.

\begin{lemma}
    \label{lem: bdd}
    Suppose Assumption \ref{assumption: X}-\ref{assumption: bdd_Q} hold. If $n \geq 3$, $1 \leq L \leq \sqrt{n}$, then there exist constant $C_b, c_b, c_e >0$ independent of $n$ such that with $t_0 = C_b L$, the following holds:
    \begin{itemize}
        \item when $\rho = 0$, for any $t \geq t_0$, the tail probability satisfies
        \begin{align*}
            \PP \rbr{\sup_{x \in B_{\mu}(L)}\hlambda_1(x) \geq t^2} \lesssim  h_0(t) + \exp(-c_e n)
        \end{align*}
        \item when $\rho = n^{-1}$, the tail probability satisfies
        \begin{align*}
            \PP \rbr{\sup_{x \in B_{\mu}(L)}\hlambda_1(x) \geq t^2} \lesssim
            \begin{cases}
                h_{0}(t) + \exp(-c_e n), & t \in \sbr{t_0, nt_0 }\\
                h_{0}(t) + \exp(-c_e n) \wedge h_0(t/n), & t \in \rbr{nt_0, +\infty}
            \end{cases}
        \end{align*}
    \end{itemize}
    where we denote
    \begin{align*}
        h_0(t):= \exp \rbr{-\sqrt{\frac{t}{c_b L}}}
        +n \cdot\exp  \sbr{- \rbr{\frac{t}{c_b L}}^{2/(C_{\Lambda}+2)}}
    \end{align*}    
    As a result, denote
    \begin{align}
        \tM_L&:= \rbr{c_b \vee C_b}^2 L^2  L_{\tau}^{4+\rbr{4 \vee 2C_1}}\\
        \tE_0&:= \cbr{ \sup_{x \in B_\mu(L)} \lambda_{\max}\left( \hQ_{\rho}(x) \right)  <  \tM_L} 
        \label{eqn: prf_conv_lem_slow1}
    \end{align}
    then
    \begin{align}
        \PP(\tE_0^c) \leq C_{\lambda,\tau} n^{-\tau} \label{eqn: bdd_E0c}
    \end{align}
    for some constant $C_{\lambda, \tau}$ independent of~$n$.
\end{lemma}
\begin{proof}
    See Appendix \ref{subsec: prf_lem_bdd}.
\end{proof}

Denote $F(x,S) = \EE \tF_n(x,S)$.
\begin{lemma}
    \label{lem: F_conc}
    Instate the notations and assumptions in Theorem \ref{thm: unfm_conv} and Lemma \ref{lem: bdd}. Then for any $\tau \geq 0$, 
    \begin{align*}
        \sup_{\genfrac{}{}{0pt}{2}{x\in B_\mu(L)}{S \in \cS_d( 0,\tM_L)} } \abs{\hF_{\rho}(x,S) - F(x,S)} = C_{F,\tau} \frac{L \tM_L^3}{\sqrt{n}} 
    \end{align*}
    with probability greater than $1- c_{F,\tau}n^{-\tau}$ for constants $C_{F,\tau}, c_{F,\tau}$ independent of $n$.
\end{lemma}
\begin{proof}
    See Appendix \ref{subsec: F_conc}.
\end{proof}

\begin{lemma}
    \label{lem: prf_conv_lem_slow}
    Instate the notations and assumptions in Theorem \ref{thm: unfm_conv} and Lemma \ref{lem: F_conc}. Let $1 \leq L = O(\sqrt{\log n})$. Let $\rho \in \cbr{0,1/n}$. Then for any $\tau$, there exist an event $\tE_1 \subset \tE_0$  such that
    \begin{itemize}
        \item $\PP(\tE_1^c)\leq c_{\delta,\tau} n^{-\tau}$
        \item under $\tE_1$,
        \begin{align}
            \sup_{x \in B_\mu(L)} \Fnorm{\hQ_{\rho}(x)-Q^*(x)} =  C_{\delta,\tau} \frac{ \polylog{n}}{n^{1/(2\alpha_F)}}
            \label{eqn: prf_conv_lem_slow2}
        \end{align}
    for constants $C_{\delta,\tau}, c_{\delta,\tau}$ independent of $n$.
    \end{itemize}
    As a result, under $\tE_1$,
    \begin{align}
        \hQ_{\rho}(x) \in \cS_d \left(  (C_{\mathrm{slow}} M_{L})^{-1}, C_{\mathrm{slow}} M_{L}\right),  \quad \forall x \in B_\mu(L)
        \label{eqn: prf_slow_3}
    \end{align}
    for some constant $C_{\mathrm{slow}}>0$ independent of $n$.
\end{lemma}
\begin{proof}
    See Appendix \ref{subsec: prf_conv_lem_slow}.
\end{proof}

\begin{remark}
    We remind readers that in Lemma \ref{lem: prf_conv_lem_slow}, $\alpha_F \geq 1$ is defined in Assumption \ref{assumption: minimizer_global}. Therefore, Lemma \ref{lem: prf_conv_lem_slow} might only lead to a uniform convergence rate slower than $\sqrt{n}$. For example, it is shown in Lemma \ref{lem: F_ind} that $\alpha_F=2$ when $X$ and $Q$ are independent, which results by (\ref{eqn: prf_conv_lem_slow2}) in a uniform convergence rate of $n^{1/4}$. Such a show rate is not enough to derive the asymptotic null distribution of our test statistic in Theorem \ref{thm: test_dist}. Therefore, we apply a finer analysis in the following part of proof to further boost the uniform convergence rate to $n^{1/2}$.
\end{remark}

\begin{lemma}
    \label{lem: conc_tld_exp}
    Instate the notations and assumptions in Theorem \ref{thm: unfm_conv} and Lemma \ref{lem: prf_conv_lem_slow}. Then there exists an event $E_{2,1} \subset \tE_1$ with probability greater than $1- c_{\tau,1}n^{-\tau}$ under which
    \begin{align}
        \sup_{x \in B_\mu(L)} &\Fnorm{\tA_n(x)} \leq C_{\tau,1} \frac{\polylog{n}}{\sqrt{n}}  \label{eqn: conc_lem_unfmA}\\
        \sup_{x \in B_\mu(L)} &\opnorm{\tPhi_n(x) - \EE \tPhi_n(x)} \leq C_{\tau,1} \frac{\polylog{n}}{\sqrt{n}} \label{eqn: conc_lem_unfmPhi}\\
        \sup_{\genfrac{}{}{0pt}{2}{x \in B_\mu(L)}{S \in \cS_d \rbr{(C_{\mathrm{slow}}M_{L})^{-1},C_{\mathrm{slow}} M_{L}}}} &\opnorm{\tPsi_n(x;S)- \EE \tPsi_n(x;S)} \leq C_{\tau,1} \frac{\polylog{n}}{\sqrt{n}} \label{eqn: conc_lem_unfmPsi}\\
        \sup_{ \genfrac{}{}{0pt}{2}{x \in B_\mu(L)}{S \in \cS_d \rbr{(C_{\mathrm{slow}}M_{L})^{-1},C_{\mathrm{slow}}M_{L}}} } &\opnorm{\EE \tPsi_n(x;S)} \leq C_{\tau,1} \polylog{n} \label{eqn: conc_lem_E_Psi}
    \end{align}
    for constants $C_{\tau,1},c_{\tau,1}$ independent of $n$.
\end{lemma}
\begin{proof}
    See Appendix \ref{subsec: conc_tld_exp}.
\end{proof}

\begin{lemma}
    \label{lem: conc_s_sin}
    Instate the notations and assumptions in Theorem \ref{thm: unfm_conv} and Lemma \ref{lem: prf_conv_lem_slow}. Then there exists event $E_{2,2} \subset \tE_1$ with probability greater than $1- c_{\tau,2} n^{-\tau}$ under which
    \begin{align}
        \sup_{x \in B_\mu(L)} &\Fnorm{\hA_{\rho}(x) -  \tA_n(x)} \leq C_{\tau,2} \frac{\polylog{n}}{\sqrt{n}} \label{eqn: conc_lem_s_sin_A}\\
        \sup_{x \in B_\mu(L)} &\opnorm{\hPhi_{\rho}(x) - \tPhi_n(x)} \leq C_{\tau,2} \frac{\polylog{n}}{\sqrt{n}} \label{eqn: conc_lem_s_sin_Phi}\\
        \sup_{\genfrac{}{}{0pt}{2}{x \in B_\mu(L)}{S \in \cS_d \rbr{(C_{\mathrm{slow}}M_{L})^{-1},C_{\mathrm{slow}}M_{L}}}} &\opnorm{\hPsi_{\rho}(x;S)- \tPsi_n(x;S)} \leq  C_{\tau,2} \frac{\polylog{n}}{\sqrt{n}} \label{eqn: conc_lem_s_sin_Psi}
    \end{align}
    for constants $C_{\tau,2},c_{\tau,2}$ independent of $n$.
\end{lemma}
\begin{proof}
    See Appendix \ref{subsec: prf_conc_lem_s_sin}.
\end{proof}
With Lemma \ref{lem: conc_tld_exp} and \ref{lem: conc_s_sin} in hand, we define event $\tE_2 := E_{2,1} \cap E_{2,2}$ under which $\hA_{\rho}(x)$, $\hPsi_{\rho}(x;S)$ and $\hPhi_{\rho}^{-1}(x)$ are uniformly bounded. This is summarized in Lemma \ref{lem: conv_smmry}.
\begin{lemma}
    \label{lem: conv_smmry}
    Instate the notations and assumptions in Theorem \ref{thm: unfm_conv}, Lemma \ref{lem: prf_conv_lem_slow}, \ref{lem: conc_tld_exp} and \ref{lem: conc_s_sin}. Then $\PP(\tE_2^c) \leq  \tc_{\tau,2} n^{-\tau}$ and under $\tE_2$,
    \begin{align}
        \sup_{x \in B_\mu(L)} \Fnorm{\hA_{\rho}(x)} &\leq \tC_{\tau,2} \frac{\polylog{n}}{\sqrt{n}} \label{eqn: lem_smmry_An}\\
        \sup_{\genfrac{}{}{0pt}{2}{x \in B_\mu(L)}{S \in \cS_d \rbr{(C_{\mathrm{slow}}M_{\tL})^{-1},C_{\mathrm{slow}}M_{\tL}}}} \opnorm{\hPsi_{\rho}(x,S)} &\leq \tC_{\tau,2} \polylog{n} \label{eqn: lem_smmry_Psi}\\
        \inf_{x \in B_{\mu(L)}} \lambda_{\min} \left( -\hPhi_{\rho}(x) \right) &\geq \frac{1}{\tC_{\tau,2} \polylog{n}} \label{eqn: lem_smmry_Phi}
    \end{align}
    As a result, under $\tE_2$ the operator $-\hPhi_{\rho}$ is invertible and
    \begin{equation}
        \sup_{x \in B_\mu(L)} \opnorm{-\hPhi_{\rho}^{-1}(x)} \leq \tC_{\tau,2}\polylog{n} \label{eqn: lem_smmry_Phi_inv}
    \end{equation}
\end{lemma}
\begin{proof}
    See Appendix \ref{subsec: prf_conv_lem_smmry}.
\end{proof}

\begin{lemma}
    \label{lem: fast}
    Instate the notations and assumptions in Theorem \ref{thm: unfm_conv}, Lemma \ref{lem: bdd} - \ref{lem: conc_s_sin}.  then under $\tE_2$,
    \begin{align*}
        \sup_{x \in B_\mu(L)} \Fnorm{\hQ_{\rho}(x)-Q^*(x)} \leq  C_{\mathrm{fast},\tau} \frac{ \polylog{n}}{\sqrt{n}}
    \end{align*}
    for constant $C_{\mathrm{fast},\tau}$ independent of $n$.
\end{lemma}
\begin{proof}
    See Appendix \ref{subsec: prf_fast}.
\end{proof}

\begin{lemma}
    \label{lem: F_moment}
    Let $\rho=1/n$. Instate the notations and assumptions in Theorem \ref{thm: unfm_conv}, Lemma \ref{lem: bdd} - \ref{lem: fast}. Then
    \begin{align*}
        \EE \sup_{x \in B_\mu(L)} \Fnorm{\hQ_{\rho}(x)-Q^*(x)} \leq C_{E,F}\frac{\polylog{n}}{\sqrt{n}}
    \end{align*}
    for some constant $C_{E,F}$ independent of $n$.
\end{lemma}
\begin{proof}
    See Appendix \ref{subsec: prf_F_moment}.
\end{proof}

\begin{lemma}
    \label{lem: W_moment}
    Let $\rho=1/n$.
    Instate the notations and assumptions in Theorem \ref{thm: unfm_conv}, Lemma \ref{lem: bdd} - \ref{lem: fast}. Then
    \begin{align*}
        \EE \sup_{x \in B_\mu(L)} W\rbr{\hQ_{\rho}(x),Q^*(x)} \leq C_{E,W}\frac{\polylog{n}}{\sqrt{n}}
    \end{align*}
    for some constant $C_{E,W}$ independent of $n$.
\end{lemma}
\begin{proof}
    See Appendix \ref{subsec: prf_W_moment}.
\end{proof}

\subsection{Proof of Lemma \ref{lem: bdd}}
\label{subsec: prf_lem_bdd}
Without loss of generality, assume $\mu = 0$, $\Sigma = I_p$, $C_{\psi_2}=1$ and $c_{\Lambda}=1$ (defined in Assumption~\ref{assumption: bdd_Q}). Denote $\hlambda_{1}(x)= \eigmax{\hQ_{\rho}(x)}$.

    By the optimality condition for $\hQ_{\rho}(x)$, one has
    \begin{align*}
        0 = \frac{1}{n}\sum_{i=1}^n w_{n,\rho}(x,X_i) \rbr{T_{\hQ_{\rho}(x)}^{Q_i} - I_d}
    \end{align*}
    Note that $n^{-1}\sum_{i=1}^n w_{n,\rho}(x,X_i) = 1$ for any $x$, hence one has
    \begin{align*}
        I_d = \frac{1}{n}\sum_{i=1}^{n} w_{n,\rho}(x,X_i) T_{\hQ_{\rho}(x)}^{Q_i}
    \end{align*}
    Recall that $T_{\hQ_{\rho}(x)}^Q = {\hQ_{\rho}(x)}^{-1/2} \left( \hQ_{\rho}(x)^{1/2}Q \hQ_{\rho}(x)^{1/2}\right)^{1/2} \hQ_{\rho}(x)^{-1/2}$, multiplying both sides of the equation with $\hQ_{\rho}(x)^{1/2}$ then gives
    \begin{align*}
        \hQ_{\rho}(x) &= \frac{1}{n} \sum_{i=1}^{n} w_{n,\rho}(x,X_i) \left( \hQ_{\rho}(x)^{1/2}Q_i \hQ_{\rho}(x)^{1/2}\right)^{1/2}
    \end{align*}
    Taking the largest eigenvalue on both sides, one has
    \begin{align*}
        \hlambda_{1}(x)  &\leq \frac{1}{n} \sum_{i=1}^{n} \abs{w_{n,\rho}(x,X_i)} \cdot \Onorm{\left( \hQ_{\rho}(x)^{1/2}Q_i \hQ_{\rho}(x)^{1/2}\right)^{1/2}}\\
        &\leq \frac{1}{n} \sum_{i=1}^{n} \abs{w_{n,\rho}(x,X_i)} \cdot \lambda_{\max}(Q_i)^{1/2} \hlambda_{1}^{1/2}(x)
    \end{align*}
    which implies
    \begin{align*}
        \hlambda_{1}(x) &\leq \sbr{\frac{1}{n} \sum_{i=1}^{n} \abs{1+(x-\bar{X})^\top \hSigma^{-1}_{\rho} (X_i-\bar{X})} \cdot \lambda_{\max}(Q_i)^{1/2}}^2\\
        &\leq \sbr{\frac{1}{n}\sum_{i=1}^n \rbr{1+\norm{x-\bar{X}} \cdot \Onorm{\hSigma^{-1}_{\rho}} \cdot  \norm{X_i-\bar{X}}} \lambda_{\max}(Q_i)^{1/2} }^2 
    \end{align*}
    Hence
    \begin{align*}
        &\PP \rbr{\sup_{x \in B_{\mu}(L)}\hlambda_1(x) \geq t^2}\\
        \leq &\PP \rbr{\cbr{\Onorm{\hSigma^{-1}_{\rho}}\leq 2} \cap \cbr{\sup_{x \in B_{\mu}(L)}\frac{1}{n}\sum_{i=1}^n \rbr{1+\norm{x-\bar{X}} \cdot  \norm{X_i-\bar{X}}} \lambda_{\max}(Q_i)^{1/2} \geq t/2}}\\
        &+ \PP \rbr{\cbr{2< \Onorm{\hSigma^{-1}_{\rho}}\leq n} \cap \cbr{\sup_{x \in B_{\mu}(L)}\frac{1}{n}\sum_{i=1}^n \rbr{1+\norm{x-\bar{X}} \cdot  \norm{X_i-\bar{X}}} \lambda_{\max}(Q_i)^{1/2} \geq t/n}}\\
        &+ \PP\rbr{\Onorm{\hSigma^{-1}_{\rho}}> n}\\
        \leq &\underbrace{ \PP\rbr{\sup_{x \in B_{\mu}(L)}\frac{1}{n}\sum_{i=1}^n \rbr{1+\norm{x-\bar{X}} \cdot  \norm{X_i-\bar{X}}} \lambda_{\max}(Q_i)^{1/2} \geq t/2}}_{\mathbf{(I)}}\\
        &+ \underbrace{\PP \rbr{ 2< \Onorm{\hSigma^{-1}_{\rho}}\leq n } \wedge \PP \rbr{\sup_{x \in B_{\mu}(L)}\frac{1}{n}\sum_{i=1}^n \rbr{1+\norm{x-\bar{X}} \cdot  \norm{X_i-\bar{X}}} \lambda_{\max}(Q_i)^{1/2} \geq t/n}}_{\mathbf{(II)}}\\
        &+ \underbrace{\PP\rbr{\Onorm{\hSigma^{-1}_{\rho}}> n}}_{(\mathbf{III})}
    \end{align*}
    
    \begin{claim}
        \label{claim: bdd_1}
        Suppose Assumption \ref{assumption: X}-\ref{assumption: bdd_Q} hold with $\mu=0$, $\Sigma =I_p$, $C_{\psi_2}=1$ and $c_{\Lambda} =1$. If $n \geq 3$, $1 \leq L \leq \sqrt{n}$.
        then for any $ t \geq 6L$,
        \begin{align}
            &\PP \cbr{\Onorm{\hSigma_{\rho}^{-1}}> 2} \lesssim  \exp (- c_e n)\\
            &\PP \cbr{\sup_{x \in B_{\mu}(L)}\frac{1}{n}\sum_{i=1}^n \rbr{1+\norm{x-\bar{X}} \cdot  \norm{X_i-\bar{X}}} \lambda_{\max}(Q_i)^{1/2} \geq t} \lesssim h_0(t)
        \end{align}
        where 
        \begin{align*}
            h_0(t):= \exp \rbr{-\sqrt{\frac{t}{c_b L}}}
            +n \cdot\exp  \sbr{- \rbr{\frac{t}{c_b L}}^{2/(C_{\Lambda}+2)}}
        \end{align*}
        for some constants $c_b, c_e>0$ independent of $n$.
    \end{claim}
    See Appendix \ref{subsubsec: claim_bdd1} for the proof.

    Combining above results, one can obtain that 
    \begin{itemize}
        \item If $\rho = 0$, then
        \begin{align*}
            \mathbf{(II)} + \mathbf{(III)} \leq \PP\rbr{\Onorm{\hSigma^{-1}_{\rho}}> 2}
        \end{align*}
        As a result, for any $t \geq 6L$, the tail probability satisfies
        \begin{align*}
            \PP \rbr{\sup_{x \in B_{\mu}(L)}\hlambda_1(x) \geq t^2} \lesssim  h_0(t) + \exp(-c_e n)
        \end{align*}
        \item If $\rho = n^{-1}$, then $\mathbf{(III)}=0$. As a results, the tail probability satisfies
        \begin{align*}
            \PP \rbr{\sup_{x \in B_{\mu}(L)}\hlambda_1(x) \geq t^2} \lesssim
            \begin{cases}
                h_{0}(t) + \exp(-c_e n), & t \in \sbr{t_0,n t_0}\\
                h_{0}(t) + \exp(-c_e n) \wedge h_0(t/n), & t \in \rbr{nt_0, +\infty}
            \end{cases}
        \end{align*}
        with $t_0 = 6L$.
    \end{itemize}
    Finally, set
    \begin{align*}
        \tilde{t}_L =  \rbr{C_b \vee c_b} L_{\tau}^{2+\rbr{2 \vee 2C_1}} L
    \end{align*}
    Note that $\tilde{t}_L \geq t_0$, hence 
    \begin{align*}
        \PP \rbr{\sup_{x \in B_{\mu}(L)}\hlambda_1(x) \geq \tilde{t}_L^2} &\lesssim h_0(\tilde{t}_L) + \exp(-c_e n)\\
        &\lesssim n^{-(1+\tau)} + \exp(-c_e n)\\
        &\leq  C_{\tau} n^{-(1+\tau)}
    \end{align*}
    for some constant $C_{\tau}$ independent of $n$ (but depends on $\tau$).
    
    This finishes the proof of Lemma \ref{lem: bdd}.

\subsubsection{Proof of Claim \ref{claim: bdd_1}}
\label{subsubsec: claim_bdd1}
The first inequality follows by noticing that $\hSigma_{\rho} = n^{-1}\sum_{i=1}^n X_i X_i^\top - \bar{X}\bar{X}^\top + \rho I_p$ with $\rho \in \cbr{0,1/n}$. Apply triangle inequality, one can obtain that for $n \geq 3$,
\begin{align*}
    \PP \cbr{\Onorm{\hSigma_{\rho}^{-1}} > 2} &
    \leq  \PP \cbr{\Onorm{\hSigma_{\rho}-I_p} > 1/2}\\
    &\leq \PP\rbr{\Onorm{n^{-1}\sum_{i=1}^n X_i X_i^\top - I_p} > 1/12} + \PP\rbr{\Onorm{\bar{X}\bar{X}^\top} > 1/12}
\end{align*}
Under the sub-Gaussian assumption in Assumption \ref{assumption: X}, one can obtain that 
\begin{align*}
    &\PP \cbr{\Onorm{n^{-1}\sum_{i=1}^n X_i X_i^\top - I_p} > 1/12} \lesssim \exp (- c n)\\
    &\PP \cbr{\Onorm{\bar{X}\bar{X}^\top} > 1/12} \lesssim  \exp (-c n)
\end{align*}

Now we move on to the second inequality. Note that
\begin{align*}
    &\quad \ \frac{1}{n}\sum_{i=1}^n \rbr{1+\norm{x-\bar{X}} \cdot  \norm{X_i-\bar{X}}} \lambda_{\max}(Q_i)^{1/2}\\
    &\leq\frac{1}{n}\sum_{i=1}^n \rbr{1+\norm{x} \cdot  \norm{X_i}} \lambda_{\max}(Q_i)^{1/2} 
    + \frac{1}{n}\sum_{i=1}^n \norm{\bar{X}} \cdot  \norm{X_i} \lambda_{\max}(Q_i)^{1/2}\\
    &+\frac{1}{n}\sum_{i=1}^n \norm{x} \cdot \norm{\bar{X}}  \lambda_{\max}(Q_i)^{1/2} 
    + \frac{1}{n}\sum_{i=1}^n \norm{\bar{X}}^2 \lambda_{\max}(Q_i)^{1/2}\\
    &= \rbr{1+ \norm{x}\norm{\bar{X}} + \norm{\bar{X}}^2}\cdot \frac{1}{n} \sum_{i=1}^n \lambda_{\max}(Q_i)^{1/2} + \rbr{\norm{\bar{X}} + \norm{x}}\cdot \frac{1}{n}\sum_{i=1}^n \norm{X_i} \lambda_{\max}(Q_i)^{1/2} 
\end{align*}
Therefore, for any $L \geq 1$ and $t>0$, one has
\begin{align*}
    &\quad \ \PP\cbr{\sup_{x \in B_{\mu}(L)}
    \frac{1}{n}\sum_{i=1}^n \rbr{1+\norm{x-\bar{X}} \cdot  \norm{X_i-\bar{X}}} \lambda_{\max}(Q_i)^{1/2} \geq t}\\
    &\leq \PP \cbr{\rbr{1+ L\norm{\bar{X}} + \norm{\bar{X}}^2}\cdot \frac{1}{n} \sum_{i=1}^n \lambda_{\max}(Q_i)^{1/2} \geq \frac{1}{1+L}t}\\
    &+ \PP \cbr{\rbr{\norm{\bar{X}} + \norm{x}}\cdot \frac{1}{n}\sum_{i=1}^n \norm{X_i} \lambda_{\max}(Q_i)^{1/2} \geq \frac{L}{1+L}t}\\
    &\leq \underbrace{\PP \cbr{\rbr{1+ L\norm{\bar{X}} + \norm{\bar{X}}^2}\cdot \frac{1}{n} \sum_{i=1}^n \lambda_{\max}(Q_i)^{1/2} \geq \frac{t}{2L}}}_{\mathbf{(A)}}\\
    &+ \underbrace{\PP \cbr{\rbr{\norm{\bar{X}} + L}\cdot \frac{1}{n}\sum_{i=1}^n \norm{X_i} \lambda_{\max}(Q_i)^{1/2} \geq \frac{t}{2}}}_{\mathbf{(B)}}\\
\end{align*}

\noindent\textbf{Analysis for (A):} for any $1 \leq L \leq \sqrt{n}$,
\begin{align*}
    \mathbf{(A)} &\leq \PP \cbr{\frac{1}{n} \sum_{i=1}^n \lambda_{\max}(Q_i)^{1/2} \geq \frac{1}{1+L/\sqrt{n} + 1/n} \cdot \frac{t}{2L}} \\
    &+ \PP \cbr{L \norm{\bar{X}} \frac{1}{n} \sum_{i=1}^n \lambda_{\max}(Q_i)^{1/2} \geq \frac{L/\sqrt{n}}{1+L/\sqrt{n} + 1/n} \cdot \frac{t}{2L}}\\
    &+ \PP \cbr{\norm{\bar{X}}^2 \frac{1}{n} \sum_{i=1}^n \lambda_{\max}(Q_i)^{1/2} \geq \frac{1/n}{1+L/\sqrt{n} + 1/n} \cdot \frac{t}{2L}}\\
    &\leq \PP \cbr{\frac{1}{n} \sum_{i=1}^n \lambda_{\max}(Q_i)^{1/2} \geq \frac{1}{3} \cdot \frac{t}{2L}} +\underbrace{ \PP \cbr{\norm{\bar{X}} \frac{1}{n} \sum_{i=1}^n \lambda_{\max}(Q_i)^{1/2} \geq \frac{1/\sqrt{n}}{3} \cdot \frac{t}{2L}}}_{\mathbf{A_1}}\\
    &+ \underbrace{\PP \cbr{\norm{\bar{X}}^2 \frac{1}{n} \sum_{i=1}^n \lambda_{\max}(Q_i)^{1/2} \geq \frac{1/n}{3} \cdot \frac{t}{2L}}}_{\mathbf{A_2}}
\end{align*}
\begin{itemize}
    \item $\mathbf{A_1}$:
    \begin{align*}
        \mathbf{A_1} \leq \PP \cbr{\norm{\bar{X}} \geq \frac{1}{\sqrt{n}} \cdot \sqrt{\frac{t}{6L}}}+
        \PP \cbr{\frac{1}{n} \sum_{i=1}^n \lambda_{\max}(Q_i)^{1/2} \geq \sqrt{\frac{t}{6L}}}
    \end{align*}
    \item $\mathbf{A_2}$:
    \begin{align*}
        \mathbf{A_2} \leq \PP \cbr{\norm{\bar{X}}^2 \geq \frac{1}{n} \cdot \sqrt{\frac{t}{6L}}} + 
        \PP \cbr{\frac{1}{n} \sum_{i=1}^n \lambda_{\max}(Q_i)^{1/2} \geq \sqrt{\frac{t}{6L}}}
    \end{align*}
\end{itemize}
Hence for $t \geq 6L$,
\begin{align*}
    \mathbf{(A)} &\lesssim \PP \cbr{\frac{1}{n} \sum_{i=1}^n \lambda_{\max}(Q_i)^{1/2} \geq \sqrt{\frac{t}{6L}}} + \PP \cbr{\norm{\bar{X}}^2 \geq \frac{1}{n} \cdot \sqrt{\frac{t}{6L}}}\\
    &\lesssim \PP \cbr{\frac{1}{n} \sum_{i=1}^n \lambda_{\max}(Q_i)^{1/2} \geq \sqrt{\frac{t}{6L}}}
    + \exp \rbr{-c\sqrt{\frac{t}{6L}}}
\end{align*}

\noindent\textbf{Analysis for (B):}
\begin{align*}
    \mathbf{(B)}&\leq \PP \cbr{L \cdot \frac{1}{n}\sum_{i=1}^n \norm{X_i} \lambda_{\max}(Q_i)^{1/2} \geq \frac{\sqrt{n}L}{\sqrt{n}L + 1} \frac{t}{2} }\\
    &+ \PP \cbr{\norm{\bar{X}} \cdot \frac{1}{n}\sum_{i=1}^n \norm{X_i} \lambda_{\max}(Q_i)^{1/2} \geq \frac{1}{\sqrt{n}L + 1} \frac{t}{2} }\\
    &\leq \PP \cbr{\frac{1}{n}\sum_{i=1}^n \norm{X_i} \lambda_{\max}(Q_i)^{1/2} \geq \frac{t}{4L} }
    + \PP \cbr{\norm{\bar{X}} \cdot \frac{1}{n}\sum_{i=1}^n \norm{X_i} \lambda_{\max}(Q_i)^{1/2} \geq \frac{t}{4\sqrt{n}L} }\\
    &\leq \PP \cbr{\frac{1}{n}\sum_{i=1}^n \norm{X_i} \lambda_{\max}(Q_i)^{1/2} \geq \frac{t}{4L} }\\
    &+ \PP \cbr{\norm{\bar{X}} \geq \sqrt{\frac{t}{4nL}}} + \PP \cbr{\frac{1}{n}\sum_{i=1}^n \norm{X_i} \lambda_{\max}(Q_i)^{1/2} \geq \sqrt{\frac{t}{4L}} }
\end{align*}
Hence for $t \geq 4L$,
\begin{align*}
    \mathbf{(B)} &\lesssim \PP \cbr{\norm{\bar{X}} \geq \sqrt{\frac{t}{4nL}}} + \PP \cbr{\frac{1}{n}\sum_{i=1}^n \norm{X_i} \lambda_{\max}(Q_i)^{1/2} \geq \sqrt{\frac{t}{4L}} }\\
    &\lesssim \exp(-ct/L)+ \PP \cbr{\frac{1}{n}\sum_{i=1}^n \norm{X_i} \lambda_{\max}(Q_i)^{1/2} \geq \sqrt{\frac{t}{4L}} }
\end{align*}
Combining results for $\mathbf{(A)}$ and $\mathbf{(B)}$, one can obtain that for $t\geq 6L$,
\begin{align*}
    &\quad \ \mathbf{(A)} + \mathbf{(B)}\\ 
    &\lesssim \exp \rbr{-c\sqrt{\frac{t}{6L}}} + \exp(-ct/L)+ \PP \cbr{\frac{1}{n} \sum_{i=1}^n \lambda_{\max}(Q_i)^{1/2} \geq \sqrt{\frac{t}{6L}}}\\
    & + \PP \cbr{\frac{1}{n}\sum_{i=1}^n \norm{X_i} \lambda_{\max}(Q_i)^{1/2} \geq \sqrt{\frac{t}{4L}} }\\
    &\lesssim \exp \rbr{-c\sqrt{\frac{t}{6L}}} +\PP \cbr{\frac{1}{n} \sum_{i=1}^n \lambda_{\max}^{1/2}(Q_i) \geq \sqrt{\frac{t}{6L}}}+ \PP \cbr{\frac{1}{n}\sum_{i=1}^n \norm{X_i} \lambda_{\max}^{1/2}(Q_i) \geq \sqrt{\frac{t}{4L}} }
\end{align*}
Note that for any $s \geq 1$,
\begin{align*}
    \PP \cbr{\frac{1}{n} \sum_{i=1}^n \lambda_{\max}(Q_i)^{1/2} \geq s} &\leq n \cdot \PP \cbr{\lambda_{\max}(Q_i)^{1/2} \geq s}\\
    &\leq n \cdot \PP \cbr{ (1 \vee \norm{X_i})^{C_{\Lambda}/2} \geq s }\\
    &= n \cdot \PP \cbr{  \norm{X_i} \geq s^{2/C_{\Lambda}} }\\
    &\lesssim n \cdot \exp \rbr{- cs^{4/C_{\Lambda}}}
\end{align*}
and 
\begin{align*}
    \PP \cbr{\frac{1}{n} \sum_{i=1}^n \norm{X_i}\lambda_{\max}(Q_i)^{1/2} \geq s} &\leq n \cdot \PP \cbr{\norm{X_i}\lambda_{\max}(Q_i)^{1/2} \geq s}\\
    &\leq n \cdot \PP \cbr{ \norm{X_i}(1 \vee \norm{X_i})^{C_{\Lambda}/2} \geq s }\\
    &= n \cdot \PP \cbr{  \norm{X_i}^{1+C_{\Lambda}/2} \geq s }\\
    &\lesssim n \cdot \exp \rbr{- cs^{4/(C_{\Lambda}+2)}}
\end{align*}
Therefore, for any $t \geq 6L$,
\begin{align*}
    \mathbf{(A)} + \mathbf{(B)} &\lesssim \exp \rbr{-c\sqrt{\frac{t}{L}}} + n \cdot\exp  \sbr{- c\rbr{t/L}^{2/C_{\Lambda}}}
    +n \cdot\exp  \sbr{- c\rbr{t/L}^{2/(C_{\Lambda}+2)}}\\
    &\lesssim \exp \rbr{-c\sqrt{\frac{t}{L}}}+n \cdot\exp  \sbr{- c\rbr{t/L}^{2/(C_{\Lambda}+2)}}
\end{align*}
The proof of Claim \ref{claim: bdd_1} is complete by adjusting constant $c$ here.

\subsection{Proof of Lemma \ref{lem: F_conc}}
\label{subsec: F_conc}
Without loss of generality, assume $\mu = 0$, $\Sigma = I_p$, $C_{\psi_2}=1$.
Recall notation $\cS_d(a,b):= \cbr{A \in \cS_d: a I_d \prec A \prec b I_d}$. Denote parameter $\theta=(x,S)$ and  parameter space $\Theta = B_\mu(L) \times \cS_d( 0,\tM_L)$. Let $F(\theta)$ denote $F(x,S)$ and similarly for other functions of $(x,S)$.
By triangle inequality, one can obtain that
\begin{equation}
    \sup_{\theta \in \Theta} \left| \hF_{\rho}(\theta) - F(\theta) \right| \leq \underbrace{\sup_{\theta \in \Theta} \left| \hF_{\rho}(\theta) - \tF_n(\theta) \right|}_{\zeta_1} + \underbrace{\sup_{\theta \in \Theta} \left| \tF_n(\theta) - F(\theta) \right| }_{\zeta_2}
    \label{eqn: prf_conv_lem_slow_E1}
\end{equation}

For any $t_1,t_2>0$, define $C_{\sup}(t_1,t_2) := \sup \cbr{W^2(Q,S):Q \in \cS_d(0,t_1), S \in \cS_d(0,t_2)}$. By Lemma \ref{lem: diff}, one has 
\begin{align}
    C_{\sup}(t_1,t_2) &\lesssim t_1+t_2
    \label{eqn: C_sup}
\end{align}
Here dimension $d$ is viewed as a fixed constant and absorbed into $\lesssim$.

\begin{itemize}
    \item $\zeta_1$:
    under event $E_0$ (defined in (\ref{eqn: prf_conv_E0})), one has
    \begin{itemize}
        \item $\norm{X_i -\mu} \leq L_{\tau} = \sqrt{(1+\tau)\log n}$ for $i=1,...,n$.
        \item $W^2(Q_i,S) \lesssim  M_{L_\tau} \vee \tM_{L} $ for any $S \in \cS_d(0, \tM_{L})$ and $i \in [n]$.
    \end{itemize}
    Lemma \ref{lem: sf_snf} then implies that 
    \begin{align}
        \zeta_1 & \lesssim \rbr{M_{L_\tau} \vee \tM_L} \frac{ L L_{\tau}^3  }{\sqrt{n}}
        \label{eqn: prf_conv_lem_slow_gamma1}
    \end{align}
    with probability greater than $1-O\rbr{n^{-(1+\tau)}}$.
        
    \item $\zeta_2$ can be bounded by truncation and the uniform concentration in Lemma \ref{lem: Rademacher}. Let $\tL \geq 1$ be a parameter to be specified later, and let $M_{\tL}:= \gamma_{\Lambda}(\tL)$. Then Assumption \ref{assumption: bdd_Q} implies that a.s.
        \begin{equation}
            X \in B_\mu(\tL) \implies Q \in \cS_d \rbr{M_{\tL}^{-1},M_{\tL}}
            \label{eqn: prf_conv_lem_slow_gamma2_tL_tM}
        \end{equation}
        
        Let $\tX$ be a truncated form of $X$ defined as
        \begin{equation}
            \tX := \tX(\tL) = X \1( \norm{X} \leq \tL)
            \label{eqn: prf_conv_lem_slow_trunc1}
        \end{equation}
        Let $\tX_{>}:= X - \tX$. By definition, one has $X=\tX + \tX_>$ , $\norm{\tX} \leq \tL$ and
        \begin{equation}
            \tX_>=
            \begin{cases}
                0,  & \norm{X}\leq \tL\\
                X, & \norm{X}> \tL
            \end{cases}
            \label{eqn: prf_conv_lem_slow_trunc2}
        \end{equation}  
        As a result, one can obtain  the following decomposition.
        \begin{equation*}
            \begin{aligned}
                &\quad \  w(x,X)W^2(S,Q)\\
                &= \rbr{1 + x^\top X} W^2(S,Q)\\
                &= \rbr{1+ x^\top \tX } W^2(S,Q) + x^\top \tX_> W^2(S,Q)\\
                &= \rbr{1+ x^\top\tX } W^2(S,Q) \1 \rbr{Q \in \cS \rbr{M_{\tL}^{-1},M_{\tL}}} \\
                &\quad + \rbr{1+ x^\top \tX } W^2(S,Q) \1 \rbr{Q \notin \cS \rbr{M_{\tL}^{-1},M_{\tL}}} + x^\top \tX_> W^2(S,Q)\\
                &\overset{(i)}{=}\rbr{1+ x^\top \tX } W^2(S,Q) \1 \rbr{Q \in \cS \rbr{M_{\tL}^{-1},M_{\tL}}} \\
                &\quad +  W^2(S,Q) \1 \rbr{Q \notin \cS \rbr{M_{\tL}^{-1},M_{\tL}}}+ x^\top \tX_> W^2(S,Q)\\
            \end{aligned}
        \end{equation*}
        Here (i) follows from (\ref{eqn: prf_conv_lem_slow_gamma2_tL_tM}). Indeed, if $Q \notin \cS \rbr{M_{\tL}^{-1},M_{\tL}}$, then as a result of (\ref{eqn: prf_conv_lem_slow_gamma2_tL_tM}), one has $\norm{X}> \tL$ which then implies $\tX  = 0$.

        As a consequence, one can obtain then following decomposition of $\tF_n(\theta) - F(\theta)$.
        \begin{equation*}
            \begin{aligned}
                \tF_n(\theta) - F(\theta) &= \alpha_1(\theta) + \alpha_2(\theta) + \alpha_3(\theta), \quad \text{where}\\
                \alpha_1(\theta)&:=\frac{1}{n} \sum_{i=1}^{n}  \rbr{1 + x^\top \tX_i} W^2(S,Q_i) \1 \rbr{Q_i \in \cS \rbr{M_{\tL}^{-1},M_{\tL}}} \\
                &\quad - \EE \left( 1+ x^\top \tX  \right)W^2(S,Q) \1\rbr{Q \in \cS \rbr{M_{\tL}^{-1},M_{\tL}}}\\
                \alpha_2(\theta)&:=\frac{1}{n} \sum_{i=1}^{n} x^\top \tX_{i,>} W^2(S,Q_i) - \EE x^\top \tX_{>} W^2(S,Q)\\
                \alpha_3(\theta)&:= \frac{1}{n}\sum_{i=1}^{n} W^2(S,Q_i) \1 \rbr{Q_i \notin \cS \rbr{M_{\tL}^{-1},M_{\tL}}}\\
                &\quad - \EE W^2(S,Q) \1 \rbr{Q \notin \cS \rbr{M_{\tL}^{-1},M_{\tL}}}
            \end{aligned}
        \end{equation*}
    
        Then triangle inequality gives
        \begin{equation}
            \zeta_2 \leq \sup_{\theta \in \Theta} \abs{\alpha_1(\theta)} + \sup_{\theta \in \Theta} \abs{\alpha_2(\theta)} + \sup_{\theta \in \Theta} \abs{\alpha_3(\theta)}
            \label{eqn: prf_conv_lem_slow_trunc_zeta2}
        \end{equation}
        \begin{itemize}
            \item $\alpha_2(\theta), \alpha_3(\theta)$: $\sup \abs{\alpha_2(\theta)}$ and $\sup \abs{\alpha_3(\theta)}$ can be upper bounded as in Claim \ref{claim: truncation}. The proof is deferred to Appendix \ref{subsubsec: prf_truncation}.
            
            \begin{claim}
                \label{claim: truncation}
                There exists constant $\tC$ independent of $n$ such that 
                by taking $\tL= \tC C_{\psi_2}\sqrt{(1+\tau)\log n}$, the following inequality
                \begin{align*}
                    \sup_{\theta \in \Theta} \abs{\alpha_2(\theta)} \vee \sup_{\theta \in \Theta} \abs{\alpha_3(\theta)} &=
                    O \rbr{ \frac{\tM_L \tL + \tL^{1+C_{\Lambda}}}{n^{1+\tau}} \cdot L}
                \end{align*}
                holds with probability greater than $1-O \rbr{n^{-\tau}}$.
            \end{claim}

            \item $\alpha_1(\theta)$: Apply triangle inequality to see that
            \begin{equation}
                \begin{aligned}
                    \sup_{\theta \in \Theta} \left| \alpha_1(\theta) \right| &= \underbrace{\sup_{\theta \in \Theta} \left| \alpha_1(\theta) \right| - \EE \sup_{\theta \in \Theta} \left| \alpha_1(\theta) \right|}_{b_1} + \underbrace{\EE \sup_{\theta \in \Theta} \left| \alpha_1(\theta) \right|}_{b_2}
                \end{aligned}
                \label{eqn: prf_conc_lem_slw_clm_Rad_a1_decomp}
            \end{equation}
            After truncation, $\alpha_1(\theta)$ is uniformly bounded. Hence $b_1$ can be upper bounded by exploiting the bounded difference property (Lemma \ref{lem: bdd_diff}). For $b_2$, we first show that $W^2(Q,S)$ is only H\"older continuous in $S$ which then implies H\"older (rather than  Lipschitz) continuity of the corresponding sub-Gaussian norm. By generalizing Dudley’s integral inequality via chaining in Lemma \ref{lem: Rademacher}, we arrive at an upper bound for $b_2$. The results are summarized below in 
            in Claim \ref{claim: Rademacher} whose proof is deferred to Appendix \ref{subsubsec: prf_claim_Rademacher}.
            
            \begin{claim}
                Instate the notations and assumptions in Lemma \ref{lem: prf_conv_lem_slow} and Claim \ref{claim: truncation}.
                \begin{equation*}
                    \sup_{\theta \in \Theta} \left| \alpha_1(\theta) \right| = O \rbr{\frac{L \tM_L^2}{\sqrt{n}}   \sqrt{\log \tM_L }}
                \end{equation*}
                with probability at least $1- O\rbr{n^{-(1+\tau)}}$.
                \label{claim: Rademacher}
            \end{claim}
            \item combine Claim \ref{claim: truncation}, Claim \ref{claim: Rademacher} and (\ref{eqn: prf_conv_lem_slow_trunc_zeta2}) to see that
            \begin{equation}
                \zeta_2 = O \rbr{\frac{L \tM_L^3}{\sqrt{n}} }
                \label{eqn: prf_conv_lem_slow_gamma2}
            \end{equation}
            with probability at least $1- O \rbr{n^{-\tau}}$.
        \end{itemize}
        \item Combining results for $\zeta_1$ and $\zeta_2$, one has
        \begin{align*}
            \sup_{\theta \in \Theta} \abs{\hF_{\rho}(\theta) - F(\theta) } &\leq \zeta_1 + \zeta_2 \\
            &=O \rbr{\frac{L \tM_L^3}{\sqrt{n}} }
        \end{align*}
        with probability greater than $1-O\rbr{n^{-\tau}}$.
    \end{itemize}
    The proof of Lemma \ref{lem: F_conc} is now complete.
    
\subsubsection{Proof of Claim \ref{claim: truncation}}
\label{subsubsec: prf_truncation}
Without loss generality, assume $\gnorm{X} \leq \gnorm{\norm{X}}\leq 1$.

\noindent \textbf{Proof for $\alpha_{2}(\theta)$}: a crude upper bounded suffices here. By triangle inequality, one has
\begin{align*}
    &\sup_{\theta \in \Theta} \abs{\alpha_2(\theta)}\\
    \leq & \sup_{\theta \in \Theta} \abs{\EE x^\top  \tX_{>} W^2(S,Q)} + \sup_{\theta \in \Theta} \abs{\frac{1}{n} \sum_{i=1}^{n} x^\top \tX_{i,>} W^2(S,Q_i)} \\
    \leq & \sup_{\theta \in \Theta} \norm{x} \cdot \EE \norm{\tX_>} \cdot W^2(S,Q) + \sup_{\theta \in \Theta} \frac{1}{n} \sum_{i=1}^{n} \norm{x} \cdot \norm{\tX_{i,>}} \cdot W^2(S,Q_i)
\end{align*}
Therefore, for any $t > 0$, one has
\begin{align*}
    \PP \cbr{\sup_{\theta \in \Theta} \abs{\alpha_2(\theta)} \geq t } \leq &\underbrace{\PP \cbr{\sup_{\theta \in \Theta} \norm{x} \cdot \EE \norm{\tX_>} \cdot W^2(S,Q) \geq t}}_{=:\mathbf{I}(t)}\\
    +& \underbrace{\PP \cbr{\sup_{\theta \in \Theta} \frac{1}{n} \sum_{i=1}^{n} \norm{x} \cdot \norm{\tX_{i,>}} \cdot W^2(S,Q_i) >0}}_{=: \mathbf{II}(t)}
\end{align*}

\noindent\textbf{I(t):} Lemma \ref{lem: diff} and Assumption \ref{assumption: bdd_Q} imply that 
\begin{align*}
    \sup_{\theta \in \Theta} W^2(S,Q) \lesssim \tM_L + \gamma_{\Lambda}(\norm{X})
\end{align*}
which further gives
\begin{align*}
    \mathbf{(I)} &\leq \PP \cbr{ \EE \norm{\tX_{>}} \rbr{\tM_L + \norm{X}^{C_{\Lambda}}} \geq t/(Lc)}\\
    &\leq \PP \cbr{ \EE \norm{\tX_{>}} \geq  \frac{t}{2c L\tM_L}} + \PP \cbr{ \EE \norm{\tX_{>}}  \norm{X}^{C_{\Lambda}} \geq t/(2Lc)}
\end{align*}
By the definition of $\tX_>$, one can obtain that $\norm{\tX_{>}}  \norm{X}^{C_{\Lambda}} = \norm{\tX_>}^{1+C_{\Lambda}}$ and for any $s \geq 0$,
\begin{align*}
    \PP \cbr{\norm{\tX_>} \geq \tL + s} \leq 2 \exp \rbr{ - \tL^2} \exp \rbr{-s^2}
\end{align*}
Therefore, one can obtain that for any $\tL \geq 1$, one has
\begin{align*}
    \EE \norm{\tX_>}
    &\leq 2\exp(-\tL^2) \int_{0}^{\infty} (\tL +s ) \exp(-s^2)\\
    &\lesssim \exp(-\tL^2) \tL
\end{align*}
and 
\begin{align*}
    \EE \norm{\tX_{>}}^{1+C_{\Lambda}} &\leq 2 \exp(-\tL^2) \int_0^{\infty} (\tL+s)^{1+C_{\Lambda}} \exp(-s^2)ds\\
    &\lesssim \exp(-\tL^2) \tL^{1+C_{\Lambda}} 
\end{align*}
Therefore, by taking $\tL = O(L_{\tau})$, there exists constant $C_I>0$ such that
\begin{align*}
    \mathbf{I(t)} = 0, \quad \forall t \geq C_I n^{-(1+\tau)} \cdot \rbr{\tM_L + M_{\tL}} L \tL
\end{align*}

\noindent\textbf{II:} Note that $\norm{\tX_{i,>}} > 0$ if and only only if $\norm{X_i} > \tL $. Hence 
\begin{align*}
    \mathbf{(II)} &\leq n \cdot \PP \cbr{\norm{\norm{X_i}>\tL}}\\
    &\lesssim n \cdot \exp(-\tL^2)\\
    &\lesssim n^{-\tau}
\end{align*}





\noindent\textbf{Combining (I), (II),}  there exists constant $\tC_2>0$ such that for any $s \geq 0$,
\begin{align*}
    &\quad \ \PP \cbr{\sup_{\theta \in \Theta} \abs{\alpha_2(\theta)} \geq n^{-(1+\tau)}\tC_2 L \tL \rbr{\tM_L + \tL^{C_{\Lambda}}} } \lesssim n^{-\tau}
\end{align*}


\paragraph{Proof for $\alpha_3(\theta)$:} 
\begin{align*}
    \alpha_3(\theta)&:= \frac{1}{n}\sum_{i=1}^{n} W^2(S,Q_i) \1 \rbr{Q_i \notin \cS \rbr{M_{\tL}^{-1},M_{\tL}}}\\
    &\quad - \EE W^2(S,Q) \1 \rbr{Q \notin \cS \rbr{M_{\tL}^{-1},M_{\tL}}}
\end{align*}
Following a similar argument to the one above, there exists constant $\tC_3>0$ such that
\begin{align*}
    \PP \cbr{\sup_{\theta \in \Theta} \abs{\alpha_3(\theta)} \geq n^{-(1+\tau)} \tC_3 L (\tM_L +\tL^{C_{\Lambda}})} \lesssim n^{-\tau}
\end{align*}
Finally, combining above results on $\alpha_2(\theta)$ and $\alpha_3(\theta)$ proves Claim \ref{claim: truncation}.

\subsubsection{Proof of Claim \ref{claim: Rademacher}}
\label{subsubsec: prf_claim_Rademacher}
Recall that we assume without loss of generality that $C_{\psi_2} =1$, $\mu=0$, $\Sigma = I_p$, $ \tL = L_{\tau}$, $M_{\tL}= \gamma_{\Lambda}(\tL)$ and $\tM_L$ is defined in Lemma \ref{lem: bdd} as
\begin{align*}
    \tM_L \asymp L^2  L_{\tau}^{4+\rbr{4 \vee 2C_1}} \gtrsim c_{\Lambda}(1 \vee \tL)^{C_{\Lambda}}= M_{\tL}
\end{align*}

Denote $Z=(X,Q)$ and define  $\tX$ as in (\ref{eqn: prf_conv_lem_slow_trunc1}). Define 

\begin{equation*}
    \begin{aligned}
        f(Z;\theta)&:=\left( 1+ (x-\mu)^\top \Sigma^{-1}(\tX - \mu) \right)W^2(S,Q) \mathbbm{1} \rbr{Q \in \cS \rbr{M_{\tL}^{-1},M_{\tL}}}
    \end{aligned}
\end{equation*}
Then one has $\alpha_1(\theta)= n^{-1}\sum_{i=1}^{n} f(Z_i;\theta) - \EE f(Z;\theta)$. 

\paragraph{Analysis of $b_1$:} due to truncation, $f$ is uniformly bounded. To see this, note that recall the definition of $\Theta$, 
    \begin{equation*}
        \theta \in \Theta \iff x \in B_\mu(L), S \in \cS_d(0, \tM_L)
    \end{equation*}
    Then for any $\theta \in \Theta$, one has
    \begin{equation}
        \begin{aligned}
            \sup_{\theta \in \Theta} \sup_{Z} \abs{f(Z;\theta)} &= \sup_{\theta \in \Theta} \sup_{Z}  \abs{(1+x^\top \tX )W^{2}(Q,S)} \mathbbm{1} \rbr{Q \in \cS_d \rbr{M_{\tL}^{-1},M_{\tL}}}\\
            &\leq  \sup_{\theta \in \Theta} \sup_{Z} \rbr{1+ \norm{x-\mu}\cdot \norm{\tX - \mu}} \cdot \abs{W^{2}(Q,S)} \mathbbm{1} \rbr{Q \in \cS_d \rbr{M_{\tL}^{-1},M_{\tL}}}\\
            &\lesssim \rbr{1+ L \tL} \cdot (M_{\tL} + \tM_L) \\
            &\lesssim L\tL \tM_L
            \label{eqn: prf_conc_lem_slw_clm_unfm1_bd}
        \end{aligned}
    \end{equation}    
    By Lemma \ref{lem: bdd_diff}, one has $\gnorm{b_1(x)} \lesssim \frac{L \tL \tM_L}{\sqrt{n}}$ which implies
    \begin{equation}
        b_1 = O \rbr{\frac{L \tL \tM_L L_{\tau}}{\sqrt{n}}}
        \label{eqn: prf_conc_lem_slw_clm_unfm_b1}
    \end{equation}
    with probability at least $1- O\rbr{n^{-(1+\tau)}}$ for some absolute constant $C>0$.

\paragraph{Analysis of $b_2$:} to apply Lemma \ref{lem: Rademacher}, we follow the steps below.
    \begin{itemize}
        \item first we define a metric $d$ on $\Theta$ by
        \begin{equation*}
            d(\theta,\ttheta) = \max\left\{ \twonorm{x-\tx}, \Fnorm{S-\tS} \right\}
        \end{equation*}
        for any $\theta = (x,S)$ and $\tS=(\tx,\tS)$. Since $\cS_d(0,\tM_L)$ is a subset of $\left\{ S \in \RR^{d \times d}: \left| S_{ij} \right| \leq \tM_L \right\}$ (by noticing that for any $S \in \cS^+_{\tM_L}$, one has $0 \leq S_{ii} \leq \tM_L$ and $\left| S_{ij} \right| \leq \sqrt{S_{ii}S_{jj}}\leq \tM_L$), the diameter $D$ and metric entropy of $\Theta$ can be upper bounded by 
        \begin{equation}
            D \leq \max\left\{ \mathrm{diam}(B_\mu(\tL)), \mathrm{diam}(\cS_d(\tM_L^{-1},\tM_L)) \right\}\lesssim L \vee \tM_L \overset{(i)}{\lesssim} \tM_L
            \label{eqn: prf_conc_lem_slw_clm_unfm_up_bdds1}
        \end{equation}
        and 
        \begin{align}
            \log N(\epsilon;\Theta) &\leq \log \left( N(\epsilon; B_\mu(\tL)) \cdot N(\epsilon;\cS^+_{\tM_L}) \right) \nonumber\\
            &\lesssim \log^+ \left( \frac{L \vee \tM_L}{\epsilon} \right) \nonumber\\
            & \overset{(ii)}{\lesssim} \log^+ \left( \frac{\tM_L}{\epsilon} \right)
            \label{eqn: prf_conc_lem_slw_clm_unfm_up_bdds2}
        \end{align}
        Here both (i) and (ii) arise due to $\tM_L \gtrsim L$.
        \item Next, let us consider the sub-Gaussian norm of $\alpha_1(\theta)-\alpha_1(\ttheta)$ as (\ref{eqn: lem_rad_lip}) in Lemma \ref{lem: Rademacher}. By (\ref{eqn: lem_rad_lip_avg}) in Lemma \ref{lem: Rademacher}, one has

        \begin{equation}
            \begin{aligned}
                \gnorm{\alpha_1(\theta)-\alpha_1(\ttheta)} &\lesssim \frac{1}{\sqrt{n}} \gnorm{ \abs{f(Z;\theta) - f(Z;\ttheta)} }\\
                &= \frac{1}{\sqrt{n}}\gnorm{ \abs{f(Z;x,S) - f(Z;\tx,\tS)} }\\ 
                &\leq \frac{1}{\sqrt{n}}\underbrace{ \gnorm{\abs{f(Z;x,S) - f(Z;x,\tS)}} }_{\beta_1} + \frac{1}{\sqrt{n}} \underbrace{\gnorm{\abs{f(Z;x,\tS) - f(Z;\tx,\tS)}}}_{\beta_2}
            \end{aligned}
            \label{eqn: prf_conc_lem_slow_clm_Rad_a1_lip}
        \end{equation}
        \begin{itemize}
            \item $\beta_1$: 
            Let $g(Q;S):=(Q^{1/2}SQ^{1/2})^{1/2}$. Then for any $Q \in \cS_d(M_{\tL}^{-1}, M_{\tL})$,
            \begin{equation*}
                \begin{aligned}
                    \Fnorm{g(Q;S) - g(Q;\tS)} &= \Fnorm{(Q^{1/2}SQ^{1/2})^{1/2} - (Q^{1/2}\tS Q^{1/2})^{1/2}}\\
                    &\lesssim \Fnorm{Q^{1/2}SQ^{1/2} - Q^{1/2}\tS Q^{1/2}}^{1/2} \\
                    &\leq  \Fnorm{S -\tS}^{1/2} \Onorm{Q}^{1/2}\\
                    &\leq M_{\tL}^{1/2}\left\|S -\tS \right\|_F^{1/2}
                \end{aligned}
            \end{equation*}
            Here in the second inequality, we exploit the following the H\"older continuity of matrix square root in \citet{09_mtrx_holder,18_mtr_pert}.
            \begin{equation*}
                \left\|A^{1/2}-B^{1/2} \right\|_F \leq C_d \left\|A-B \right\|_F^{1/2}
            \end{equation*}
            where $C_d$ is a constant only depending on dimension $d$. Hence for any $\theta, \ttheta \in \Theta$,
            \begin{align}
                \beta_1 &= \gnorm{\abs{f(X,Q;x,S) - f(X,Q;x,\tS)}} \nonumber\\
                &\leq \gnorm{\abs{1+x^\top \tX }}
                 \cdot \Fnorm{\rbr{g(Q;S)  - g(Q;\tS)} \cdot \mathbbm{1} \rbr{Q \in \cS_d \rbr{M_{\tL}^{-1},M_{\tL}}}} \nonumber\\
                &\leq \left( 1+ \norm{x} \cdot \gnorm{\twonorm{\tX}} \right) \cdot M_{\tL}^{1/2}\Fnorm{S -\tS}^{1/2}\nonumber\\
                &\lesssim L  \cdot  M_{\tL} ^{1/2} \Fnorm{S- \tS}^{1/2} 
                \label{eqn: prf_conc_lem_slow_clm_Rad_b1}
            \end{align}
            \item $\beta_2$: for any $\theta, \ttheta \in \Theta$, one has
            \begin{align}
                \beta_2 &= \gnorm{(x-\tx)^\top \tX W^{2}(Q,S) \mathbbm{1} \rbr{Q \in \cS_d \rbr{M_{\tL}^{-1},M_{\tL}}}} \nonumber\\
                &\overset{(i)}{\lesssim} \norm{x-\tx} \gnorm{\tX } \cdot \rbr{M_{\tL} + \tM_L}\nonumber\\
                &\lesssim \tM_L \left\|x-\tx \right\|_2
                \label{eqn: prf_conc_lem_slow_clm_Rad_b2}
            \end{align}
            Here (i) follows from (\ref{eqn: C_sup}).
            \item As a result of (\ref{eqn: prf_conc_lem_slow_clm_Rad_b1}) and (\ref{eqn: prf_conc_lem_slow_clm_Rad_b2}), one can obtain
                \begin{align}
                    \left\| f(Z;\theta) - f(Z;\ttheta)\right\|_{\psi_2} 
                    &\leq \beta_1 + \beta_2 \nonumber\\
                    &\lesssim L M_{\tL}^{1/2} d(\theta,\ttheta)^{1/2} + \tM_L d(\theta,\ttheta)\nonumber\\
                    &\lesssim L \tM_L \left(  d(\theta,\ttheta)^{1/2} \vee d(\theta,\ttheta)\right)
                    \label{eqn: prf_conc_lem_slow_clm_Rad_f_lip}
                \end{align}
        \end{itemize}

        Combining (\ref{eqn: prf_conc_lem_slow_clm_Rad_a1_lip}) and ({\ref{eqn: prf_conc_lem_slow_clm_Rad_f_lip}}) then gives
        \begin{equation*}
            \begin{aligned}
                \left\| \alpha_1(\theta)  - \alpha_1(\ttheta)  \right\|_{\psi_2} \lesssim \frac{L \tM_L}{\sqrt{n}} \left( d(\theta,\ttheta)^{1/2} \vee d(\theta,\ttheta)\right)
            \end{aligned}
        \end{equation*}
        Hence $\tau (\epsilon)$ in Lemma \ref{lem: Rademacher} can be chosen as
        \begin{equation}
            \begin{aligned}
                \tau(\epsilon) =  \frac{K}{\sqrt{n}}\left( \epsilon^{1/2} \vee \epsilon\right), \quad  K=C_{f} L \tM_L
            \end{aligned}
            \label{eqn: prf_conc_lem_slow_unfm1_tau_n_inv}
        \end{equation}
        
        \item by Lemma \ref{lem: Rademacher}, one has
        \begin{equation*}
            \begin{aligned}
                b_2 := \EE \sup_{\theta \in \Theta} \left| \alpha_1(Z;\theta) \right| 
                &\lesssim \underbrace{\frac{KD}{\sqrt{n}} \sqrt{\log^+ D}}_{\zeta_1} + \underbrace{\sup_{\theta \in \Theta} \EE \left| \alpha_1(Z;\theta) \right|}_{\zeta_2}
            \end{aligned}
        \end{equation*}
        \begin{itemize}
            \item $\zeta_1$: by (\ref{eqn: lem_rad_chain_holder}) in Lemma \ref{lem: Rademacher} and (\ref{eqn: prf_conc_lem_slw_clm_unfm_up_bdds1}), one can obtain
            \begin{equation}
                \zeta_1 \lesssim \frac{L \tM_L^2}{\sqrt{n}}  \sqrt{\log \tM_L}
                \label{eqn: prf_conc_lem_slow_clm_Rad_g1}
            \end{equation}
            \item $\zeta_2$: For any $\theta \in \Theta$, we have
            \begin{equation*}
                \begin{aligned}
                    \gnorm{f(Z;\theta)} &\leq \gnorm{\abs{f(Z;\theta)}}\\
                    &\leq \gnorm{\abs{1+(x-\mu)^\top \Sigma^{-1}(\tX -\mu)}}\\
                    &\quad \cdot \sup_{S \in \cS_d(0,\tM_L)} \sbr{W^2(Q,S) \mathbbm{1} \rbr{Q \in \cS_d \rbr{M_{\tL}^{-1},M_{\tL}}}}\\
                    &\lesssim \left( 1 + \left\|x-\mu \right\| \left\|\left\|\tX - \mu \right\| \right\|_{\psi_2} \right) \cdot \rbr{M_{\tL}+\tM_L}\\
                    &\lesssim L \tM_L
                \end{aligned}
            \end{equation*}
            Therefore for any $\theta \in \Theta$, one has $\left\|\alpha_1(\theta) \right\|_{\psi_2} \lesssim L \tM_L/\sqrt{n}$ and 
            \begin{equation*}
                \begin{aligned}
                    \EE \abs{\alpha_1(Z;\theta)} \lesssim \gnorm{\alpha_1(\theta)} \lesssim \frac{L\tM_L}{\sqrt{n}}
                \end{aligned}
            \end{equation*}
            Take supremum over $\theta \in \Theta$ to see that
            \begin{equation}
                \zeta_2 \lesssim \frac{L\tM_L}{\sqrt{n}}
                \label{eqn: prf_conc_lem_slow_clm_Rad_g2}
            \end{equation}
            \item As a result of (\ref{eqn: prf_conc_lem_slow_clm_Rad_g1}) and (\ref{eqn: prf_conc_lem_slow_clm_Rad_g2}), one can obtain
            \begin{equation}
                \begin{aligned}
                    b_2  
                    &\lesssim \frac{L \tM_L^2}{\sqrt{n}}  \sqrt{\log \tM_L }
                \label{eqn: prf_conc_lem_slw_clm_unfm_b2}
                \end{aligned}
            \end{equation}
            
        \end{itemize}        
    \end{itemize}

Finally, combine (\ref{eqn: prf_conc_lem_slw_clm_unfm_b1}) and (\ref{eqn: prf_conc_lem_slw_clm_unfm_b2}) to see that
\begin{equation}
    \begin{aligned}
        \sup_{\theta \in \Theta}\left| \alpha_1(\theta) \right| 
        &\leq  b_1 + b_2 \lesssim   
        C\frac{L \tM_L^2}{\sqrt{n}}   \sqrt{\log \tM_L }
    \end{aligned}
\end{equation}
with probability at least $1-O\rbr{ n^{-(1+\tau)} }$. The proof of Claim \ref{claim: Rademacher} is then complete.

\subsection{Proof of Lemma \ref{lem: prf_conv_lem_slow}}
\label{subsec: prf_conv_lem_slow}
Without loss of generality, assume $\mu = 0$, $\Sigma = I_p$ and $C_{\psi_2}=1$ in Assumption \ref{assumption: X}-Assumption \ref{assumption: bdd_Q}. 
To prove the convergence rate of $\hQ_{\rho}$ (\ref{eqn: prf_conv_lem_slow2}), observe that for any $\delta_n,\epsilon_n >0$, one has
\begin{align}
    &\quad \cbr{\sup_{x \in B_\mu(L)} \Fnorm{\hQ_{\rho}(x) - Q^*(x)} \leq \delta_n }  \nonumber\\
    &\overset{(i)}{\supset} \cbr{\inf_{S: \Fnorm{S-Q^*(x)} \leq \delta_n} \hF_{\rho}(x,S) < \inf_{S: \Fnorm{S-Q^*(x) } \geq \delta_n} \hF_{\rho}(x,S) , \forall x \in B_{\mu}(L)} \nonumber\\
    &\overset{(ii)}{\supset} \cbr{ \hF_{\rho}(x,Q^*(x)) <  \inf_{S: \Fnorm{S-Q^*(x)} \geq \delta_n} \hF_{\rho}(x,S) , \forall x \in B_{\mu}(L)} \cap \cbr{\hQ_{\rho}(x) \preceq \tM_L I_d  , \forall x \in B_{\mu}(L)} \nonumber\\
    &\overset{(iii)}{\supset} \left\{ \hF_{\rho}(x,Q^*(x)) <  \inf_{   S:  \underset{S \preceq \tM_L I_d}{\Fnorm{S-Q^*(x)} \geq \delta_n} } \hF_{\rho}(x,S),\forall x \in B_{\mu}(L) \right\} \cap \left\{ \hQ_{\rho}(x) \preceq \tM_L I_d , \forall x \in B_{\mu}(L) \right\} \nonumber\\
    &\supset \underbrace{ \cbr{\sup_{  \genfrac{}{}{0pt}{2}{x\in B_\mu(L)}{S \preceq \tM_L I_d} } \abs{ \hF_{\rho}(x,S)-F(x,S)} \leq \epsilon_n}  }_{=:\cE_1} \nonumber\\
    &\quad \cap \underbrace{ \cbr{F(x,Q^*(x)) \leq  \inf_{ \genfrac{}{}{0pt}{2}{S:\Fnorm{S-Q^*(x)} \geq \delta_n}{S \preceq \tM_L I_d}    } F(x,S) - 3\epsilon_n ,\forall x \in B_{\mu}(L)} }_{=:\cE_2} \nonumber\\
    &\quad \cap \underbrace{ \cbr{\hQ_{\rho}(x) \preceq \tM_L I_d , \forall x \in B_{\mu}(L)} }_{=\tE_0 \text{ in Lemma } \ref{lem: bdd}} \label{eqn: prf_conv_lem_slow_set}
\end{align}
Here (i) and (iii) follows since $\hQ_{\rho}(x)$ is the minimizer of $F_{n,\rho}(x,\cdot)$, (ii) is obtained by setting~$S=Q^*(x)$.

With (\ref{eqn: prf_conv_lem_slow_set}) in place, it suffices to choose appropriate $\delta_n,\epsilon_n$ such that $\cE_1 \cap \cE_2$ occurs with high probability and define $E_2 = \cE_1 \cap \cE_2 \cap E_0$. To this end, we first find the uniform convergence rate $\epsilon_n$ so that $\PP(\cE_1^c) \lesssim n^{-\tau}$. Next, $\delta_n$ is set based on $\epsilon_n$ and properties of $F$ so that $\PP(\cE_2^c) \lesssim n^{-\tau}$. For $\cE_3$, we already showed in (\ref{eqn: prf_conv_lem_slow1}) that $\PP(\cE_3^c)\lesssim n^{-\tau}$. Combining results on $\cE_1,\cE_2$ and $\cE_3$, the proof is then finished.


\paragraph{Analysis of $\cE_1$:} Denote $\epsilon_0=\delta_F^{\alpha_F}/\gamma_F(L,\sqrt{d} \tM_L)$ ($\delta_F$ and $\gamma_F(\cdot,\cdot)$ are defined in Assumption \ref{assumption: minimizer_global}) and set
    \begin{align*}
        \epsilon_n &:= C_{F,\tau} \frac{L \tM_L^3}{\sqrt{n}} \wedge \frac{\epsilon_0}{3}
    \end{align*}
    Note that for $L=O\rbr{\sqrt{\log n}}$, $C_{F,\tau} \frac{L \tM_L^3}{\sqrt{n}} \leq \epsilon_0/3$ for large $n$.
    Therefore,  Lemma \ref{lem: F_conc} implies that for any $\tau\geq 0$, 
    \begin{align*}
        \PP \cbr{\cE_1} \geq 1 - c_{\delta,\tau} n^{-\tau}
    \end{align*}
    for some constant $c_{\delta,\tau}>0$ independent of $n$.
  
\paragraph{Analysis of $\cE_2$:}    
    Note that for any $x \in B_{\mu}(L)$ and $S \preceq \tM_L I_d$, one has
    \begin{align*}
        \Fnorm{S-Q^*(x)}\leq \sqrt{d}\max\cbr{\Onorm{S}, \Onorm{Q^*(x)}} \leq \sqrt{d} \tM_L
    \end{align*}
    Hence one has
    \begin{align*}
        \inf_{S:  \underset{S \preceq \tM_L I_d}{\Fnorm{S-Q^*(x)} \geq \delta_n} } F(x,S) - F(x,Q^*(x)) 
        \geq \inf_{S:  \delta_n \leq \Fnorm{S-Q^*(x)} \leq \sqrt{d} \tM_L } F(x,S) - F(x,Q^*(x)) 
    \end{align*}
    Set
    \begin{align*}
        \delta_n := \sbr{3 \epsilon_n \gamma_F(L, \sqrt{d}\tM_L)}^{1/\alpha_F}
    \end{align*}
    By definition, $\delta_n \leq \delta_F$.
    Therefore, Assumption \ref{assumption: minimizer_global} implies that $\PP(\cE_2)=1$ for $n$ large enough.
    
\paragraph{Analysis of $\cE_3$:} Note that $\cE_3 = \tE_0$ which is defined in Lemma \ref{lem: bdd}. Hence
\begin{align*}
    \PP(\cE_3^c) \leq C_{\lambda,\tau} n^{-\tau}
\end{align*}

\paragraph{Combining $\cE_1,\cE_2,\cE_3$:} Finally, combining results on $\cE_1,\cE_2$ and $\cE_3$ above, one can obtain
\begin{align*}
    \PP\rbr{ \sup_{x\in B_{\mu}(L)}\Fnorm{\hQ_{\rho}(x) - Q^*(x)} \leq C_{\delta,\tau} \sbr{\frac{L \tM_L^3 \gamma_F(L,\sqrt{d} \tM_L)}{\sqrt{n}}}^{1/\alpha_F} } \geq 1- c_{\delta,\tau} n^{-\tau}
\end{align*}
for constants $C_{\delta,\tau}, c_{\delta,\tau}$ independent of $n$.
Note that by definition, $L,\tM_L, \gamma_F(L,\sqrt{d}\tM_L)\lesssim \polylog{n}$, one can then arrive at (\ref{eqn: prf_conv_lem_slow2}).

Finally, Eq. (\ref{eqn: prf_slow_3}) follows  by noticing that $\delta_F \wedge \delta_n = o \rbr{M_L}$.

This finishes the proof of Lemma \ref{lem: prf_conv_lem_slow}.

\subsection{Proof of Lemma \ref{lem: conc_tld_exp}}
\label{subsec: conc_tld_exp}
Without loss of generality, assume $\mu = 0$, $\Sigma = I_p$ and $C_{\psi_2}=1$ in Assumption \ref{assumption: X}-Assumption \ref{assumption: bdd_Q}. 
If $X$ is unbounded, then one can apply the truncation as (\ref{eqn: prf_conv_lem_slow_trunc1}) and follow similar arguments as in Lemma \ref{lem: prf_conv_lem_slow}. Hence we can assume without loss of generality that $\left\|X-\mu \right\|\leq \tL$ almost surely with  $\tL = C_{\psi_2} L_{\tau}=C_{\psi_2} \sqrt{(1+\tau)\log n}$. Then we have a.s.
\begin{equation}
    \begin{aligned}
        \norm{X-\mu} &\leq \tL\\
        Q, Q^*(X) &\in \cS_d(M_{\tL}^{-1},M_{\tL})\\
        \mathbbm{1}(E_0) & =1
    \end{aligned}
    \label{eqn: conc_tld_exp_wlog}
\end{equation}
where we remind reader that $E_0 := \cbr{\norm{X_i - \mu}\leq \tL, Q_i \in \cS_d(M_{\tL}^{-1}, M_{\tL}), i \in [n]}$ is defined in (\ref{eqn: prf_conv_E0}).

\paragraph{Proof of (\ref{eqn: conc_lem_unfmA}):}  one can obtain the following decomposition.

\begin{equation}
    \begin{aligned}
        \sup_{x \in B_\mu(L)} \left\|\tA_n(x) \right\|_F \leq \underbrace{\sup_{x \in B_\mu(L)} \left\|\tA_n(x) \right\|_F - \EE \sup_{x \in B_\mu(L)} \left\|\tA_n(x) \right\|_F}_{a_1} + \underbrace{\EE \sup_{x \in B_\mu(L)} \left\|\tA_n(x) \right\|_F}_{a_2}
    \end{aligned}
    \label{eqn: prf_slow_A_decomp}
\end{equation}
Define $\varphi(Z;x)=w(x,X)\left( T^{Q}_{Q^*(x)}-I_d \right)$.

\noindent \textit{Analysis of $a_1$:} By (\ref{eqn: conc_tld_exp_wlog}), one has almost surely that for any $x \in B_{\mu}(L)$,
    \begin{equation*}
        \begin{aligned}
            \Fnorm{\varphi(Z;x)} &\lesssim \left( 1+ \norm{x-\mu}\cdot \norm{X-\mu} \right) \cdot  \sup_{ \genfrac{}{}{0pt}{2}{Q \in \cS_d(M_{\tL}^{-1},M_{\tL})}{S \in \cS_d((M_{L})^{-1}, M_{L})} } \Fnorm{T^Q_{S}-I_d}\\
            &\overset{(i)}{\lesssim} L \tL M_{\tL}^{1/2} M_{L}^{3/2}
        \end{aligned}
    \end{equation*}
     Here (i) follows by noticing that 
    \begin{equation*}
        \begin{aligned}
            \sup_{ \genfrac{}{}{0pt}{2}{Q \in \cS_d(M_{\tL}^{-1},M_{\tL})}{S \in \cS_d((M_{L})^{-1}, M_{L})} } \Onorm{T^Q_{S}} = \sup_{\genfrac{}{}{0pt}{2}{Q \in \cS_d(M_{\tL}^{-1},M_{\tL})}{S \in \cS_d((M_{L})^{-1}, M_{L})}}\Onorm{S^{-1/2} (S^{1/2} Q S^{1/2})^{1/2} S^{-1/2}}\leq M_{\tL}^{1/2} M_{L}^{3/2}
        \end{aligned}
    \end{equation*}
    Lemma \ref{lem: bdd_diff} then implies that $\gnorm{a_1} \lesssim \frac{L \tL M_{\tL}^{1/2} M_{L}^{3/2} }{\sqrt{n}}$. Therefore, one can obtain
    \begin{equation}
        a_1 \lesssim \frac{L \tL M_{\tL}^{1/2} M_{L}^{3/2} L_{\tau}}{\sqrt{n}}
        \label{eqn: conc_lem_unfmA_a1}
    \end{equation}
    with probability at least $1-O(n^{-(1+\tau)})$.

\noindent\textit{Analysis of $a_2$:} to apply Lemma \ref{lem: Rademacher}, we follow the steps below.
\begin{itemize}
    \item First, we consider the sub-Gaussian norm of $\Fnorm{\tA_n(x)} - \Fnorm{\tA_n(\tx)}$ as (\ref{eqn: lem_rad_lip}) in Lemma \ref{lem: Rademacher}. For any $x,\tx \in B_\mu(L)$, one has
    \begin{equation}
        \begin{aligned}
            \gnorm{\Fnorm{\tA_n(x)} - \Fnorm{\tA_n(\tx)}} &\leq \gnorm{\Fnorm{\tA_n(x) - \tA_n(\tx)}}
        \end{aligned}
        \label{eqn: conc_lem_unfmA_a2_triangle}
    \end{equation}
    Note that
    \begin{equation*}
        \begin{aligned}
            \tA_n(x) - \tA_n(\tx) = \frac{1}{n} \sum_{i=1}^{n} \left[ w(x,X_i)\left( T^{Q_i}_{Q^*(x)}-I_d \right) - w(\tx,X_i)\left( T^{Q_i}_{Q^*(\tx)}-I_d \right) \right]
        \end{aligned}
    \end{equation*}
    Since $\EE \tA_n(x)=0$ by (\ref{eqn: prf_conv_opt_Q*}), one has
    \begin{equation}
        \begin{aligned}
            \gnorm{\Fnorm{\tA_n(x) - \tA_n(\tx)}}
            &\overset{(i)}{\lesssim} \frac{1}{\sqrt{n}} \gnorm{\Fnorm{w(x,X)\left( T^{Q}_{Q^*(x)}-I_d \right) - w(\tx,X)\left( T^{Q}_{Q^*(\tx)}-I_d \right)}}\\
            &\leq \frac{1}{\sqrt{n}} \underbrace{ \gnorm{\Fnorm{w(x,X)\left( T^{Q}_{Q^*(x)}-I_d \right) - w(x,X)\left( T^{Q}_{Q^*(\tx)}-I_d \right)}} }_{b_1} \\
            &\quad + \frac{1}{\sqrt{n}} \underbrace{ \gnorm{\Fnorm{w(x,X)\left( T^{Q}_{Q^*(\tx)}-I_d \right) - w(\tx,X)\left( T^{Q}_{Q^*(\tx)}-I_d \right)}} }_{b_2}\\
        \end{aligned}
        \label{eqn: conc_lem_unfmA_a2_lip}
    \end{equation}
    Here (i) follows from (\ref{eqn: lem_rad_lip_avg}) in Lemma \ref{lem: Rademacher}.
    
    \begin{itemize}
        \item $b_1$: For any $x,\tx \in B_\mu(L)$, one has $Q^*(x),Q^*(\tx) \in \cS_d(M_L^{-1},M_L)$. Then one can obtain
        \begin{align}
            b_1 &= \gnorm{\Fnorm{w(x,X)\left( T^{Q}_{Q^*(x)}- T^{Q}_{Q^*(\tx)} \right)}}\nonumber\\
            &\lesssim \gnorm{\abs{1+ \norm{x-\mu}\cdot \norm{X-\mu}}\cdot \Fnorm{T^{Q}_{Q^*(x)}- T^{Q}_{Q^*(\tx)}}} \nonumber\\
            &\lesssim \gnorm{\abs{1+ \norm{x-\mu}\cdot \norm{X-\mu}}} \cdot \sup_{Q \in\cS(M_{\tL}^{-1},M_{\tL})} \Fnorm{T^{Q}_{Q^*(x)} - T^{Q}_{Q^*(\tx)}} \nonumber\\
            &\lesssim L \cdot\sup_{\genfrac{}{}{0pt}{2}{Q \in \cS_d(M_{\tL}^{-1},M_{\tL})}{x,\tx \in B_{\mu}(L)} } \Fnorm{T^{Q}_{Q^*(x)} - T^{Q}_{Q^*(\tx)}} \nonumber\\
            &\overset{(i)}{=} L \cdot \sup_{ \genfrac{}{}{0pt}{2}{Q \in \cS_d(M_{\tL}^{-1},M_{\tL})}{Q' \in \cS_d(M_{L}^{-1} , M_{L})} } \Fnorm{dT^{Q}_{Q'} \cdot (Q^*(x)- Q^*(\tx))} \nonumber\\
            &\overset{(ii)}{\leq} L \cdot \Fnorm{Q^*(x) -Q^*(\tx)}\cdot \sup_{ \genfrac{}{}{0pt}{2}{Q \in \cS_d(M_{\tL}^{-1},M_{\tL})}{S \in \cS_d(M_{L}^{-1} , M_{L})} } \opnorm{dT^{Q}_{S}}\nonumber\\
            &\overset{(iii)}{\lesssim} L \poly{M_L,M_{\tL}} \cdot \Fnorm{Q^*(x) -Q^*(\tx)}
            \label{eqn: conc_lem_unfmA_b1}
        \end{align}       
        Here (i) is a consequence of the mean value theorem \citep[Theorem 5.3]{dudley} for some $Q'$ that lies between $Q^*(x)$ and $Q^*(\tx)$, (ii) arises from Lemma \ref{lem: f_cal}, and (iii) follows from Lemma \ref{lem: diff}.
        \item $b_2$:
        \begin{align}
            b_2 &\leq \gnorm{\abs{(x-\tx)\Sigma^{-1}(X-\mu)}} \sup_{\genfrac{}{}{0pt}{2}{Q \in \cS_d(M_{\tL}^{-1},M_{\tL})}{S \in \cS_d(M_{L}^{-1} , M_{L})} } \Fnorm{T^Q_{S} -I_d} \nonumber\\
            &\overset{(i)}{\lesssim} \norm{x-\tx} \gnorm{\norm{X-\mu}} \poly{M_L, M_{\tL}} \nonumber\\
            &\lesssim L \poly{M_L, M_{\tL}} \norm{x-\tx}
            \label{eqn: conc_lem_unfmA_b2}
        \end{align}
        almost surely. Here (i) arises from Lemma \ref{lem: subG_norm} and the bounds on $T^Q_S$ in Lemma \ref{lem: diff}.
    \end{itemize}
    
    Combining (\ref{eqn: conc_lem_unfmA_a2_triangle}), (\ref{eqn: conc_lem_unfmA_a2_lip}), (\ref{eqn: conc_lem_unfmA_b1}) and (\ref{eqn: conc_lem_unfmA_b2}) gives that for any $x,\tx \in B_\mu(L)$,
    \begin{equation}
        \begin{aligned}
            &\quad \ \gnorm{\Fnorm{\tA_n(x)} - \Fnorm{\tA_n(\tx)}}\\ 
            &\lesssim \frac{1}{\sqrt{n}}\cdot L \poly{M_L, M_{\tL}} \left( \norm{x-\tx} + \left\|Q^*(x)-Q^*(\tx) \right\|_F \right)\\
            &\overset{(i)}{\lesssim }\frac{1}{\sqrt{n}}\cdot L \poly{M_L, M_{\tL}} \left( \norm{x-\tx} + M_L \gamma_F(L,M_L)\rbr{\norm{x-\tx}_2 \vee \norm{x-\tx}^{1/\alpha_F}} \right)\\
            &\overset{(ii)}{\leq } \frac{C}{\sqrt{n}}\cdot L^{1+C_F} M_L^{4 + C_F} \rbr{\norm{x-\tx}_2 \vee \norm{x-\tx}^{1/\alpha_F}}
        \end{aligned}
        \label{eqn: prf_conc_lem_slow_A_lip}
    \end{equation}
    Here (i) follows from Lemma \ref{lem: holder} and the fact that $M \gamma_F(L,M) \geq 1$.

    \item With the H\"older continuity (\ref{eqn: prf_conc_lem_slow_A_lip}) in place, we can apply Lemma \ref{lem: Rademacher} with $\tau (\epsilon)$ chosen as
    \begin{equation*}
        \begin{aligned}
            \tau(\epsilon) =  \frac{K}{\sqrt{n}}\left( \epsilon^{1/\alpha_F} \vee \epsilon\right), \quad  K=CL^{1+C_F} M_L^{4 + C_F}
        \end{aligned}
    \end{equation*} 
    which gives
    \begin{align*}
        a_2 &\lesssim \frac{KL^{1 \vee \alpha_F^{-1}}}{\sqrt{n}} \sqrt{\log^+ L}+ \sup_{x \in B_\mu(L)} \EE \Fnorm{\tA_n(x)}\\
        &\overset{(i)}{=}\underbrace{\frac{KL}{\sqrt{n}} \sqrt{\log^+ L}}_{\zeta_1} + \underbrace{\sup_{x \in B_\mu(L)} \EE \Fnorm{\tA_n(x)} }_{\zeta_2}
    \end{align*}
    Here (i) follows since $\alpha_F\geq 1$ as defined in Assumption \ref{assumption: minimizer_global}.
    \begin{itemize}
        \item $\zeta_1$: direct calculation gives that
        \begin{equation}
            \zeta_1 \lesssim \frac{L^{2+C_F} M_L^{4+C_F}}{\sqrt{n}} \sqrt{\log^+ L}
            \label{eqn: conc_lem_unfmA_g1}
        \end{equation}
        \item $\zeta_2$: for any $U \in \RR^{d \times d}$ with unit Frobenius norm, one has            
        \begin{equation*}
            \begin{aligned}
                \sup_{x \in B_{\mu}(L)} \gnorm{\Fnorm{\tA_n(x)}} 
                &\overset{(i)}{\lesssim} \sup_{x \in B_{\mu}(L)} \sup_{\Fnorm{U}=1}\gnorm{\inner{U}{\tA_n(x)}}\\
                &\overset{(ii)}{\lesssim} \sup_{x \in B_{\mu}(L)} \sup_{\Fnorm{U}=1} \gnorm{\inner{U}{\frac{1}{n}\sum_{i=1}^{n} \varphi(Z_i;x)- \EE \varphi(Z;x)}} \\
                &\lesssim \sup_{x \in B_{\mu}(L)} \sup_{\Fnorm{U}=1} \frac{1}{\sqrt{n}} \gnorm{ \inner{U}{\varphi(Z;x)}}\\
                &\leq  \sup_{x \in B_{\mu}(L)} \frac{1}{\sqrt{n}} \gnorm{ \Fnorm{\varphi(Z;x)} }\\
                &\leq \frac{1}{\sqrt{n}} \sup_{x \in B_\mu(L)}\gnorm{\abs{w(x,X)}} \cdot \sup_{\genfrac{}{}{0pt}{2}{Q \in \cS_d(M_{\tL}^{-1},M_{\tL})}{S \in \cS_d(M_{L}^{-1} , M_{L})} } \Fnorm{T^Q_S - I_d}\\
                &\lesssim \frac{1}{\sqrt{n}} L \cdot \poly{M_L, M_{\tL}}
            \end{aligned}
        \end{equation*}
        Here (i) results from Lemma \ref{lem: subG_norm}, (ii) follows from independence. 
        Therefore,
        \begin{equation}
            \begin{aligned}
                \zeta_2 \lesssim \frac{L \cdot \poly{M_L, M_{\tL}}}{\sqrt{n}} 
                \label{eqn: conc_lem_unfmA_g2}
            \end{aligned}
        \end{equation}
    \end{itemize}
    
    \item As a result of (\ref{eqn: conc_lem_unfmA_g1}) and (\ref{eqn: conc_lem_unfmA_g2}), one can obtain
        \begin{equation}
            a_2 \lesssim \frac{\poly{L, M_L, M_{\tL}}}{\sqrt{n}} \sqrt{\log^+ L}
            \label{eqn: conc_lem_unfmA_a2}
        \end{equation}
\end{itemize}
Finally, combining (\ref{eqn: prf_slow_A_decomp}), (\ref{eqn: conc_lem_unfmA_a1}) and (\ref{eqn: conc_lem_unfmA_a2}) gives (\ref{eqn: conc_lem_unfmA}).

\paragraph{Proof of (\ref{eqn: conc_lem_unfmPhi}):} one can obtain the following decomposition.
\begin{equation}
    \begin{aligned}
        \sup_{x \in B_\mu(L)} \opnorm{\tPhi_n(x) - \EE \tPhi_n(x)} &\leq \underbrace{\sup_{x \in B_\mu(L)} \opnorm{\tPhi_n(x) - \EE \tPhi_n(x)} - \EE \sup_{x \in B_\mu(L)} \opnorm{\tPhi_n(x) - \EE \tPhi_n(x)}}_{a_3}\\
        &\quad + \underbrace{\EE \sup_{x \in B_\mu(L)} \opnorm{\tPhi_n(x) - \EE \tPhi_n(x)}}_{a_4}
    \end{aligned}
    \label{eqn: prf_slow_Phi_decomp}
\end{equation}
Let $Z=(X,Q)$. Define $\phi(Z;x) = w(x,X)dT^Q_{Q^*(x)} $.

\noindent\textit{Analysis of $a_3$:} By (\ref{eqn: conc_tld_exp_wlog}), one has almost surely that for any $x \in B_{\mu}(L)$,
    \begin{equation*}
        \begin{aligned}
            \opnorm{\phi(Z;x) }
            &\lesssim L \tL \cdot \sup_{\genfrac{}{}{0pt}{2}{Q \in \cS_d(M_{\tL}^{-1},M_{\tL})}{S \in \cS_d(M_{L}^{-1} , M_{L})} } \opnorm{dT^Q_S}\\
            &\lesssim L \tL  \poly{M_L, M_{\tL}}
        \end{aligned}
    \end{equation*}
    By Lemma \ref{lem: bdd_diff}, one then can obtain
    \begin{equation}
        \gnorm{a_3 } \lesssim \frac{L \tL  \poly{M_L, M_{\tL}}}{\sqrt{n}}
        \label{eqn: conc_lem_unfmPhi_a3}
    \end{equation}

\noindent\textit{Analysis of $a_4$:} to apply Lemma \ref{lem: Rademacher}, we follow the steps below.
\begin{itemize}
    \item First, let us consider the sub-Gaussian norm of $\opnorm{\tPhi_n(x) - \EE \tPhi_n(x)} -  \opnorm{\tPhi_n(\tx) - \EE \tPhi_n(\tx)}$. For any $x,\tx \in B_\mu(L)$, 
    by (\ref{eqn: lem_rad_lip_avg}) in Lemma \ref{lem: Rademacher}, one can obtain
    \begin{equation}
        \begin{aligned}
            \gnorm{ \opnorm{\tPhi_n(x) - \EE \tPhi_n(x)} -  \opnorm{\tPhi_n(\tx) - \EE \tPhi_n(\tx)}}
            \lesssim \frac{1}{\sqrt{n}} \left\| \opnorm{\phi(Z;x) - \phi(Z;\tx) }\right\|_{\psi_2}
        \end{aligned}
        \label{eqn: conc_lem_unfmA_a4_triangle}
    \end{equation}

    Moreover,
    \begin{equation}
        \begin{aligned}
            &\quad \ \gnorm{\opnorm{\phi(Z;x) -\phi(Z;\tx)}}\\
            &= \gnorm{\opnorm{w(x,X)dT^Q_{Q^*(x)}-w(\tx,X)dT^Q_{Q^*(\tx)} }}\\
            &\leq \underbrace{\gnorm{\opnorm{w(x,X)dT^Q_{Q^*(x)}-w(x,X)dT^Q_{Q^*(\tx)} }}}_{b_3} \\
            &+\underbrace{\gnorm{ \opnorm{w(x,X)dT^Q_{Q^*(\tx)}-w(\tx,X)dT^Q_{Q^*(\tx)} }}}_{b_4}
        \end{aligned}
        \label{eqn: conc_lem_unfmA_a4_lip}
    \end{equation}    
    
    \begin{itemize}
        \item $b_3$:  we have for any $x,\tx \in B_\mu(L)$,
        \begin{align}
            b_3 & \overset{(i)}{=} \gnorm{\opnorm{w(x,X) d^2 T^Q_{Q'} \cdot (Q^*(x) - Q^*(\tx))}} \nonumber\\
            &\lesssim L \gnorm{X-\mu} \cdot   \Fnorm{Q^*(x)- Q^*(\tx)} \cdot \sup_{ \genfrac{}{}{0pt}{2}{Q \in \cS_d(M_{\tL}^{-1},M_{\tL})}{S \in \cS_d(M_{L}^{-1} , M_{L})} }\opnorm{ d^2T^Q_{S} } \nonumber\\
            &\overset{(ii)}{\lesssim} L \cdot \poly{M_L, M_{\tL}} \Fnorm{Q^*(x) -Q^*(\tx)}
            \label{eqn: conc_lem_unfmA_b3}
        \end{align}
        Here (i) is a consequence of the mean value theorem \citep[Theorem 5.3]{dudley} for some $Q'$ that lies between $Q^*(x)$ and $Q^*(\tx)$ and (ii) follows from Lemma \ref{lem: diff}.
        
        \item $b_4$: for any $x,\tx \in B_\mu(L)$,
        \begin{align}
            b_4 &\leq \left\|\left|w(x,X) - w(\tx,X) \right| \right\|_{\psi_2} \sup_{\genfrac{}{}{0pt}{2}{Q \in \cS_d(M_{\tL}^{-1},M_{\tL})}{S \in \cS_d(M_{L}^{-1} , M_{L})} } \opnorm{dT^Q_S} \nonumber\\
            &\lesssim \poly{M_L,M_{\tL}} \left\|x-\tx \right\|_2
            \label{eqn: conc_lem_unfmA_b4}
        \end{align}
    \end{itemize}
    Combining (\ref{eqn: conc_lem_unfmA_a4_triangle}), (\ref{eqn: conc_lem_unfmA_a4_lip}), (\ref{eqn: conc_lem_unfmA_b3}) and (\ref{eqn: conc_lem_unfmA_b4}), one can obtain that
    \begin{equation*}
        \begin{aligned}
            &\quad \gnorm{\opnorm{\tPhi_n(x) - \EE \tPhi_n(x)} -  \opnorm{\tPhi_n(\tx) - \EE \tPhi_n(\tx)}} \\
            &\lesssim \frac{1}{\sqrt{n}}\cdot L \poly{M_L, M_{\tL}} \rbr{\norm{x-\tx} + \Fnorm{Q^*(x) - Q^*(\tx)}}\\
            &\lesssim \frac{1}{\sqrt{n}}\cdot L\poly{M_L,M_{\tL}} \left( \norm{x-\tx}\vee \norm{x-\tx}^{1/\alpha_F} \right)
        \end{aligned}
    \end{equation*}

    \item With the H\"older continuity above, we can apply Lemma \ref{lem: Rademacher} with $\tau (\epsilon)$ chosen as
    \begin{equation*}
        \begin{aligned}
            \tau(\epsilon) =  \frac{K}{\sqrt{n}}\left( \epsilon^{1/\alpha_F} \vee \epsilon\right), \quad  K= \poly{L, M_L}
        \end{aligned}
    \end{equation*}
    which gives
    \begin{equation*}
        \begin{aligned}
            a_4 \lesssim \underbrace{\frac{KL}{\sqrt{n}} \sqrt{\log^+ L}}_{\zeta_3} + \underbrace{\sup_{x \in B_\mu(L)} \EE \opnorm{\tPhi_n(x)}}_{\zeta_4}
        \end{aligned}
    \end{equation*}
    \begin{itemize}
        \item $\zeta_3$: direct calculation gives
        \begin{equation}
            \zeta_3 \lesssim  \frac{\poly{L, M_L}}{\sqrt{n}}
            \label{eqn: conc_lem_unfmPhi_g3}
        \end{equation}
        \item $\zeta_4$: for any $U,V \in \RR^{d \times d}$ with unit Frobenius norm, one has
        \begin{equation*}
            \begin{aligned}
                \sup_{x \in B_\mu(L)} \gnorm{\opnorm{\tPhi_n(x)}} 
                &\overset{(i)}{\lesssim} \sup_{x \in B_\mu(L)} \sup_{\Fnorm{U}=\Fnorm{V}=1}\gnorm{\inner{V}{\tPhi_n(x) \cdot U}}\\
                &\overset{(ii)}{\lesssim} \sup_{x \in B_\mu(L)} \sup_{\Fnorm{U}=\Fnorm{V}=1} \frac{1}{\sqrt{n}} \gnorm{\inner{V}{\sbr{\phi(Z;x)-E \phi(Z;x)} \cdot U}} \\
                &\lesssim \sup_{x \in B_\mu(L)} \sup_{\Fnorm{U}=\Fnorm{V}=1} \frac{1}{\sqrt{n}} \gnorm{ \inner{V}{\phi(Z;x)\cdot U}}\\
                &\leq  \sup_{x \in B_\mu(L)} \frac{1}{\sqrt{n}} \gnorm{ \opnorm{\phi(Z;x)} }\\
                &\leq \frac{1}{\sqrt{n}} \sup_{x \in B_\mu(L)}\gnorm{\abs{w(x,X)}} \cdot \sup_{\genfrac{}{}{0pt}{2}{Q \in \cS_d(M_{\tL}^{-1},M_{\tL})}{S \in \cS_d(M_{L}^{-1} , M_{L})} } \opnorm{dT^Q_S}\\
                &\overset{(iii)}{\lesssim} \frac{1}{\sqrt{n}} L \poly{M_L,M_{\tL}}
            \end{aligned}
        \end{equation*}
        Here (i) results from Lemma \ref{lem: subG_norm}, (ii) follows from independence and (iii) is due to Lemma \ref{lem: diff}. Hence
        \begin{equation}
            \zeta_4 \lesssim \frac{L \poly{M_L, M_{\tL}}}{\sqrt{n}}
            \label{eqn: conc_lem_unfmPhi_g4}
        \end{equation}
    \end{itemize}
    \item As a result of (\ref{eqn: conc_lem_unfmPhi_g3}) and (\ref{eqn: conc_lem_unfmPhi_g4}), one can obtain
    \begin{equation}
        a_4 \lesssim \frac{\poly{L, M_L, M_{\tL}}}{\sqrt{n}}
        \label{eqn: conc_lem_unfmPhi_a4}
    \end{equation}
\end{itemize}

Finally, combining (\ref{eqn: prf_slow_Phi_decomp}), (\ref{eqn: conc_lem_unfmPhi_a3}) and (\ref{eqn: conc_lem_unfmPhi_a4}) gives (\ref{eqn: conc_lem_unfmPhi}).

\paragraph{Proof of (\ref{eqn: conc_lem_unfmPsi}):}
Denote $\Theta =  B_\mu(L) \times \cS_d((C_{\mathrm{slow}} M_{\tL})^{-1},C_{\mathrm{slow}} M_{\tL})$ and $\theta = (x,S) \in \Theta$.

\begin{equation}
    \begin{aligned}
        \sup_{\theta \in \Theta} \opnorm{\tPsi_n(\theta) - \EE \tPsi_n(\theta)} &\leq \underbrace{\sup_{\theta \in \Theta} \opnorm{\tPsi_n(\theta) - \EE \tPsi_n(\theta)} - \EE \sup_{\theta \in \Theta} \opnorm{\tPsi_n(\theta) - \EE \tPsi_n(\theta)}}_{a_5}\\
        &\quad + \underbrace{\EE \sup_{\theta \in \Theta} \opnorm{\tPsi_n(\theta) - \EE \tPsi_n(\theta)}}_{a_6}
    \end{aligned}
    \label{eqn: prf_slow_Psi_decomp}
\end{equation}
Let $Z=(X,Q)$. Define $\psi(Z;\theta) = w(x,X)d^2T^Q_{S} $ and $\bar{\psi}(Z;\theta) = \psi(Z;\theta) - \EE \psi(Z;\theta)$

\noindent\textit{Analysis of $a_5$:} By (\ref{eqn: conc_tld_exp_wlog}), one has almost surely that 
    \begin{equation*}
        \begin{aligned}
            \sup_{Z,\theta}\opnorm{\psi(Z;\theta) }
            &\leq \sup_{ \genfrac{}{}{0pt}{2}{x \in B_{\mu}(L)}{X \in B_{\mu}(\tL)} } \left| w(x,X) \right| \cdot \sup_{ \genfrac{}{}{0pt}{2}{Q \in \cS_d(M_{\tL}^{-1},M_{\tL})}{S \in \cS_d((C_{\mathrm{slow}}M_{L})^{-1},C_{\mathrm{slow}}M_{L})}} \opnorm{d^2 T^Q_S}\\
            &\lesssim L \tL\cdot \poly{M_{L}, M_{\tL}}
        \end{aligned}
    \end{equation*}
    Lemma \ref{lem: bdd_diff} then implies that , $\gnorm{a_5} \lesssim \frac{L \tL\cdot \poly{M_{\tL} M_L}}{\sqrt{n}}$. Therefore, one can obtain
    \begin{equation}
        a_5 \lesssim L \tL \cdot \poly{M_L M_{\tL}}\frac{L_\tau}{\sqrt{n}}
        \label{eqn: conc_lem_unfmPsi_a5}
    \end{equation}
    with probability at least $1-O(n^{-(1+\tau)})$.

\noindent\textit{Analysis of $a_6$:} to apply Lemma \ref{lem: Rademacher}, we follow the steps below.
\begin{itemize}
    \item First, let us consider the sub-Gaussian norm of $\opnorm{\tPsi_n(\theta) - \EE \tPsi_n(\theta)} -  \opnorm{\tPsi_n(\ttheta) - \EE \tPsi_n(\ttheta)}$. For any $\theta, \ttheta \in \Theta$, one has
    \begin{equation}
        \begin{aligned}
            \gnorm{ \opnorm{\tPsi_n(\theta) - \EE \tPsi_n(\theta)} -  \opnorm{\tPsi_n(\ttheta) - \EE \tPsi_n(\ttheta)}}
            \overset{(i)}{\lesssim} \frac{1}{\sqrt{n}} \left\| \opnorm{\psi(Z;\theta) - \psi(Z;\ttheta) }\right\|_{\psi_2}
        \end{aligned}
        \label{eqn: conc_lem_unfmA_a6_triangle}
    \end{equation}
    Here (i) follows from Lemma \ref{lem: Rademacher}. Moreover, one has
    \begin{equation}
        \begin{aligned}
            &\quad \gnorm{\opnorm{\psi(Z;\theta) -\psi(Z;\ttheta)}}\\
            &= \gnorm{\opnorm{w(x,X)d^2 T^Q_{S}-w(\tx,X)d^2 T^Q_{\tS} }}\\
            &\leq \underbrace{\gnorm{\opnorm{w(x,X)d^2T^Q_{S}-w(x,X)d^2T^Q_{\tS} }}}_{b_5} +\underbrace{\gnorm{ \opnorm{w(x,X)d^2 T^Q_{\tS}-w(\tx,X)d^2T^Q_{\tS} }}}_{b_6}
        \end{aligned}
        \label{eqn: conc_lem_unfmA_a6_lip}
    \end{equation}
    \begin{itemize}
        \item $b_5$: for any $\theta, \ttheta \in \Theta$, one can obtain
            \begin{align}
                b_5 &\overset{(i)}{=} \gnorm{\opnorm{w(x,X) d^3 T^Q_{Q'} \cdot (S - \tS)}} \nonumber\\
                &\leq \gnorm{ \abs{w(x,X)}} \cdot \Fnorm{S- \tS} \cdot \sup_{ \genfrac{}{}{0pt}{2}{Q \in \cS_d(M_{\tL}^{-1},M_{\tL})}{S \in \cS_d((C_{\mathrm{slow}}M_{L})^{-1},C_{\mathrm{slow}}M_{L})} }\opnorm{ d^3T^Q_{S} } \nonumber\\
                &\overset{(ii)}{\lesssim} L \cdot \poly{M_L,M_{\tL}} \Fnorm{S- \tS}
                \label{eqn: conc_lem_unfmA_b5}
            \end{align}
        Here (i) is a consequence of the mean value theorem \citep[Theorem 5.3]{dudley} for some $Q'$ that lies between $Q^*(x)$ and $Q^*(\tx)$ and (ii) follows from Lemma \ref{lem: diff}.
        
        \item $b_6$: for any $\theta, \ttheta \in \Theta$, one has
        \begin{align}
            b_6 &\leq \gnorm{\abs{w(x,X) - w(\tx,X)}} \sup_{ \genfrac{}{}{0pt}{2}{Q\in \cS_d (M_{\tL}^{-1},M_{\tL})}{S\in \cS_d ((C_{\mathrm{slow}}M_{L})^{-1},C_{\mathrm{slow}}M_{L})}   } \opnorm{d^2 T^Q_S} \nonumber\\
            &\lesssim \poly{M_L, M_{\tL}} \norm{x-\tx}
            \label{eqn: conc_lem_unfmA_b6}
        \end{align}
    \end{itemize}
    Combining (\ref{eqn: conc_lem_unfmA_a4_triangle}), (\ref{eqn: conc_lem_unfmA_a6_lip}), (\ref{eqn: conc_lem_unfmA_b5}) and  (\ref{eqn: conc_lem_unfmA_b6}), one can obtain
    \begin{equation*}
        \begin{aligned}
            \gnorm{ \opnorm{\tPsi_n(\theta) - \EE \tPsi_n(\theta)} -  \opnorm{\tPsi_n(\ttheta) - \EE \tPsi_n(\ttheta)}}
            &\lesssim \frac{L\poly{M_L, M_{\tL}}}{\sqrt{n}}\cdot  d(\theta,\ttheta)
        \end{aligned}
    \end{equation*}

    \item With the Lipschitz continuity above, we can apply Lemma \ref{lem: Rademacher} with $\tau (\epsilon)$ chosen as
    \begin{equation*}
        \begin{aligned}
            \tau(\epsilon) =  \frac{K}{\sqrt{n}}\epsilon, \quad  K=C \cdot L \poly{M_L, M_{\tL}}
        \end{aligned}
    \end{equation*}
    which gives
    \begin{equation*}
        \begin{aligned}
            a_6 \lesssim \underbrace{\frac{K M_{\tL}}{\sqrt{n}} \sqrt{\log^+ M_{\tL}}}_{\zeta_5} + \underbrace{\sup_{\theta \in \Theta} \EE \opnorm{\tPsi_n(\theta)}}_{\zeta_6}
        \end{aligned}
    \end{equation*}
    \begin{itemize}
        \item $\zeta_5$: direct calculation gives
        \begin{equation}
            \zeta_5 \lesssim \frac{L\poly{ M_L M_{\tL}}}{\sqrt{n}} 
            \label{eqn: conc_lem_unfmPsi_g5}
        \end{equation}
        \item $\zeta_6$: for any $U,V \in \RR^{d \times d}$ with unit Frobenius norm, one has
        \begin{equation*}
            \begin{aligned}
                \gnorm{\inner{V}{\tPsi_n(\theta) \cdot U}} &\lesssim \frac{1}{\sqrt{n}} \gnorm{\inner{V}{\sbr{\psi(Z;\theta) - \EE \psi(Z;\theta)}\cdot U}} \\
                &\lesssim \frac{1}{\sqrt{n}} \gnorm{ \inner{V}{\psi(Z;\theta)\cdot U}}\\
                &\leq  \frac{1}{\sqrt{n}} \gnorm{ \opnorm{\psi(Z;\theta)} }\\
                &\leq \frac{1}{\sqrt{n}} \gnorm{\abs{w(x,X)}} \cdot \sup_{ \genfrac{}{}{0pt}{2}{Q\in \cS_d (M_{\tL}^{-1},M_{\tL})}{S\in \cS_d ((C_{\mathrm{slow}}M_{L})^{-1}, C_{\mathrm{slow}}M_{L})}  } \opnorm{d^2 T^Q_S}\\
                &\lesssim \frac{1}{\sqrt{n}} L \cdot \poly{M_L M_{\tL}}
            \end{aligned}
        \end{equation*}
        which combined with  Lemma \ref{lem: subG_norm} implies
        \begin{equation}
            \zeta_6 \lesssim \frac{L \poly{M_L M_{\tL}}}{\sqrt{n}}
            \label{eqn: conc_lem_unfmPsi_g6}
        \end{equation}
    \end{itemize}
    \item As a result of (\ref{eqn: conc_lem_unfmPsi_g5}) and (\ref{eqn: conc_lem_unfmPsi_g6}), one can obtain
    \begin{equation}
        a_6 \lesssim \frac{ L\poly{M_L M_{\tL}}}{\sqrt{n}}
        \label{eqn: conc_lem_unfmPsi_a6}
    \end{equation}
\end{itemize}
Finally, combining (\ref{eqn: prf_slow_Psi_decomp}), (\ref{eqn: conc_lem_unfmPsi_a5}) and (\ref{eqn: conc_lem_unfmPsi_a6}) gives (\ref{eqn: conc_lem_unfmPsi}).

\paragraph{Proof of (\ref{eqn: conc_lem_E_Psi}):} by definition, one has
\begin{align*}
    \EE \tPsi_n(x;S)= \EE_{(X,Q)} \sbr{w(x,X)d^2 T^Q_S}
\end{align*}
Therefore, by the truncation assumption (\ref{eqn: conc_tld_exp_wlog}), one can obtain

\begin{align*}
    &\sup_{ \genfrac{}{}{0pt}{2}{x \in B_\mu(L)}{S \in \cS_d \left( (C_{\mathrm{slow}}M_{\tL})^{-1},C_{\mathrm{slow}}M_{\tL} \right)} } \opnorm{\EE \tPsi_n(x;S)}\\
    \leq &\sbr{\sup_{x \in B_{\mu}(L)} \EE \abs{w(x,X)}} \cdot \sbr{\sup_{ \genfrac{}{}{0pt}{2}{Q \in \cS_d(M_{\tL}^{-1},M_{\tL})}{S \in \cS_d \left( (C_{\mathrm{slow}}M_{L})^{-1},C_{\mathrm{slow}}M_{L} \right)}  } \opnorm{d^2 T^Q_S}}\\
    \overset{(i)}{\lesssim} & L \cdot \poly{M_L M_{\tL}}
\end{align*}
Here (i) follows from Lemma \ref{lem: diff}. 

Finally, note that by definition, $\tL \asymp \sqrt{\log n}$ and $M_{\tL} = \poly{\tL}$, one can obtain (\ref{eqn: conc_lem_E_Psi}).

\subsection{Proof of Lemma \ref{lem: conc_s_sin}}
\label{subsec: prf_conc_lem_s_sin}

As argued at the beginning of Appendix \ref{subsec: conc_tld_exp}, we can assume without loss of generality that $\norm{X-\mu}\leq \tL$ almost surely with $\tL$ defined in (\ref{eqn: prf_conv_E0}). As a result, we have almost surely that
\begin{equation}
    \begin{aligned}
        \norm{X-\mu} &\leq \tL\\
        Q, Q^*(X) &\in \cS_d(M_{\tL}^{-1},M_{\tL})\\
        \mathbbm{1}(E_0)&=1
    \end{aligned}
    \label{eqn: prf_s_sn_trunc}
\end{equation}
where $E_0 := \cbr{\norm{X_i - \mu}\leq L, Q_i \in \cS_d(M_{\tL}^{-1}, M_{\tL}), i \in [n]}$.

With boundedness condition (\ref{eqn: prf_s_sn_trunc}) in place, Lemma \ref{lem: diff} implies  the following upper bounds
\begin{equation*}
    \begin{aligned}
        \sup_{ \genfrac{}{}{0pt}{2}{Q \in \cS_d(M_{\tL}^{-1},M_{\tL})}{S \in \cS_d \left( (C_{\mathrm{slow}}M_{L})^{-1},C_{\mathrm{slow}}M_{L} \right)} } \Fnorm{T^Q_S - I_d} &\leq \poly{M_{L}, M_{\tL}} \\
        \sup_{ \genfrac{}{}{0pt}{2}{Q \in \cS_d(M_{\tL}^{-1},M_{\tL})}{S \in \cS_d \left( (C_{\mathrm{slow}}M_{L})^{-1},C_{\mathrm{slow}}M_{L} \right)}} \opnorm{dT^Q_S} &\leq \poly{M_L, M_{\tL}}\\
        \sup_{ \genfrac{}{}{0pt}{2}{Q \in \cS_d(M_{\tL}^{-1},M_{\tL})}{S \in \cS_d \left( (C_{\mathrm{slow}}M_{L})^{-1},C_{\mathrm{slow}}M_{L} \right)}} \opnorm{d^2 T^Q_S} &\leq \poly{M_L, M_{\tL}}
    \end{aligned}
\end{equation*}
almost surely. Applying Lemma \ref{lem: sf_snf} then gives the desired results.

\subsection{Proof of Lemma \ref{lem: conv_smmry}}
\label{subsec: prf_conv_lem_smmry}

\paragraph{Proof of (\ref{eqn: lem_smmry_An}):} apply the triangle inequality to see that
\begin{equation*}
    \begin{aligned}
        \sup_{x \in B_\mu(L)} \Fnorm{\hA_{\rho}(x)} &\leq \sup_{x \in B_\mu(L)} \Fnorm{\hA_{\rho}(x) - \tA_n(x)} + \sup_{x \in B_\mu(L)} \Fnorm{\tA_n(x)}\\
        &\overset{(i)}{\leq} \frac{\polylog{n}}{\sqrt{n}}
    \end{aligned}
\end{equation*}
Here (i) follows from Lemma \ref{lem: conc_tld_exp} and Lemma \ref{lem: conc_s_sin}

\paragraph{Proof of (\ref{eqn: lem_smmry_Psi}):} apply the triangle inequality to see that
\begin{equation*}
    \begin{aligned}
        &\sup_{ \genfrac{}{}{0pt}{2}{x \in B_\mu(L)}{S \in \cS_d \left( (2M_L)^{-1},2M_L \right)} } \opnorm{\hPsi_{\rho}(x,S)} \\
        &\leq \sup_{\genfrac{}{}{0pt}{2}{x \in B_\mu(L)}{S \in \cS_d \left( (2M_L)^{-1},2M_L \right)}} \opnorm{\hPsi_{\rho}(x,S) - \tPsi_n(x,S)} 
        + \sup_{\genfrac{}{}{0pt}{2}{x \in B_\mu(L)}{S \in \cS_d \left( (2M_L)^{-1},2M_L \right)} } \opnorm{\tPsi_n(x,S)- \EE \tPsi_n(x,S)}\\
        &\qquad + \sup_{\genfrac{}{}{0pt}{2}{x \in B_\mu(L)}{S \in \cS_d \left( (2M_L)^{-1},2M_L \right)} } \opnorm{\EE \tPsi_n(x,S)}\\
        &\overset{(i)}{\leq} \frac{\polylog{n}}{\sqrt{n}}+\frac{\polylog{n}}{\sqrt{n}} + \polylog{n}\\
        &\leq \polylog{n}
    \end{aligned}
\end{equation*}
Here (i) follows from Lemma \ref{lem: conc_tld_exp} and Lemma \ref{lem: conc_s_sin}.

\paragraph{Proof of (\ref{eqn: lem_smmry_Phi}), (\ref{eqn: lem_smmry_Phi_inv}):} By the remark after Assumption \ref{assumption: minimizer_local} and the eigenvalue stability inequality, one can obtain
\begin{equation*}
    \begin{aligned}
        &\quad \ \inf_{x \in B_{\mu(L)}}\lambda_{\min} \left( -\hPhi_{\rho}(x) \right)\\
        &\geq \inf_{x \in B_{\mu(L)}} \lambda_{\min} \left( -\EE \tPhi_n(x) \right) - \sup_{x \in B_{\mu(L)}}\opnorm{\tPhi_n(x) - \EE \tPhi_n(x)} - \sup_{x \in B_{\mu(L)}}\opnorm{\tPhi_n(x) - \hPhi_{\rho}(x)}\\
        &\overset{(i)}{\geq} \frac{1}{\polylog{n}} - \frac{\polylog{n}}{\sqrt{n}} - \frac{\polylog{n}}{\sqrt{n}}\\
        &\geq \frac{1}{\polylog{n}}
    \end{aligned}
\end{equation*}
Here (i) follows by noticing that $\EE (-\tPhi_n(x))= \EE \left( -w(x,X)dT^Q_{Q^*(x)} \right)$, Assumption \ref{assumption: minimizer_local} and applying Lemma \ref{lem: conc_tld_exp}, Lemma \ref{lem: conc_s_sin}. 

By the remark after Assumption \ref{assumption: minimizer_local} and (\ref{eqn: lem_smmry_Phi}), one then has
\begin{equation*}
    \opnorm{-\hPhi_{\rho}^{-1}(x)}  = \lambda_{\max} \left( -\hPhi_{\rho}^{-1}(x) \right) \leq \frac{1}{\lambda_{\min} \left( -\hPhi_{\rho}(x) \right)}
\end{equation*}
Then (\ref{eqn: lem_smmry_Phi_inv}) follows from (\ref{eqn: lem_smmry_Phi}).

\subsection{Proof of Lemma \ref{lem: fast}}
\label{subsec: prf_fast}

Recall definitions of $Q^*(x)$ and $\hQ_{\rho}(x)$ that
\begin{align*}
    Q^*(x) &= \argmin_{S \in \cS_d^{++}} \EE \sbr{w(x,X)W^2(S,Q)}\\
    \hQ_{\rho}(x) &= \argmin_{S \in \cS_d^{++}} \frac{1}{n} \sum_{i=1}^n w_{n,\rho}(x,X_i)W^2(S,Q_i)
\end{align*}
First, the differential properties of $W^2$ (Lemma \ref{lem: diff}) imply the following optimality conditions for $Q^*(x)$ and $\hQ_{\rho}(x)$.
\begin{align}
    \EE w(x,X) \rbr{T_{Q^*(x)}^Q - I_d} &= 0 \label{eqn: prf_conv_opt_Q*} \\
    \frac{1}{n} \sum_{i=1}^n w_{n,\rho}(x,X_i) \rbr{T_{\hQ_{\rho}(x)}^{Q_i} - I_d} &= 0
    \label{eqn: prf_conv_opt_Qn}
\end{align}
Next, one can apply Lemma \ref{lem: diff} again to get the 2nd order Taylor expansion of (\ref{eqn: prf_conv_opt_Qn}) around $Q^*(x)$ as follows.
\begin{align*}
    0 &= \frac{1}{n} \sum_{i=1}^n w_{n,\rho}(x,X_i) \rbr{T_{Q^*(x)}^{Q_i} - I_d} + \frac{1}{n} \sum_{i=1}^n w_{n,\rho}(x,X_i) dT_{Q^*(x)}^{Q_i} \cdot \rbr{\hQ_{\rho}(x) - Q^*(x)}\\
    \qquad &+ \frac{1}{2n} \sum_{i=1}^n w_{n,\rho}(x,X_i) d^2 T_{\tQ_n(x)}^{Q_i} \cdot \rbr{\hQ_{\rho}(x) - Q^*(x)}^{\otimes 2}
\end{align*}
where $\tQ_n(x)$ lies between $Q^*(x)$ and $\hQ_{\rho}(x)$. Rearranging then gives
\begin{equation}
    \begin{aligned}
        &\hQ_{\rho}(x)-Q^*(x) = \left( -\underbrace{\frac{1}{n} \sum_{i=1}^{n} w_{n,\rho}(x,X_i)  dT^{Q_i}_{Q^*(x)} }_{=\hPhi_{\rho}(x)}\right)^{-1} 
        \cdot \underbrace{\frac{1}{n} \sum_{i=1}^{n} w_{n,\rho}(x,X_i)\left( T^{Q_i}_{Q^*(x)} - I_d \right)}_{=\hA_{\rho}(x)}  \\
        &+ \left( -\frac{1}{n} \sum_{i=1}^{n} w_{n,\rho}(x,X_i)  dT^{Q_i}_{Q^*(x)}\right)^{-1} 
        \cdot \frac{1}{2}\cdot \underbrace{\frac{1}{n} \sum_{i=1}^{n} w_{n,\rho}(x,X_i)  d^2 T^{Q_i}_{\tQ_n(x)} }_{=\hPsi_{\rho}\rbr{x; \tQ_{n,\rho}(x)}} \cdot (\hQ_{\rho}(x)-Q^*(x))^{\otimes 2} 
    \end{aligned}
    \label{eqn: prf_conv_Qn_Q}
\end{equation}
Taking Frobenius norm on both sides of (\ref{eqn: prf_conv_Qn_Q}) gives the following quadratic inequality for $\Fnorm{\hQ_{\rho}(x)-Q^*(x)}$.
\begin{align}
    &\Fnorm{\hQ_{\rho}(x)-Q^*(x)} \nonumber\\
    \leq &\opnorm{-\hPhi_{\rho}(x)^{-1}} \cdot \Fnorm{\hA_{\rho}(x)} + \frac{1}{2}\opnorm{-\hPhi_{\rho}(x)^{-1}} \cdot \opnorm{\hPsi_{\rho}(x,\tQ(x))}\cdot \Fnorm{\hQ_{\rho}(x)-Q^*(x)}^2
    \label{eqn: prf_conv_Qn_Q_x}
\end{align}

With Lemma \ref{lem: conv_smmry} in place, one can then derive from (\ref{eqn: prf_conv_Qn_Q_x}) that the following quadratic inequality holds uniformly for any $x \in B_{\mu}(L)$ under $\tE_2$.


\begin{equation*}
     \Fnorm{\hQ_{\rho}(x)-Q^*(x)} \leq a_0 + a_2 \Fnorm{\hQ_{\rho}(x)-Q^*(x)}^2
\end{equation*}
Here $a_0,a_2>0$ are uniform for any $x \in B_\mu(L)$ and satisfies
\begin{equation*}
    \begin{aligned}
        a_0 &= \tC_{\tau,1}' \frac{\polylog{ n}}{\sqrt{n}}\\
        a_2 &= \tC_{\tau,2}' \polylog{n} \polylog{n}
    \end{aligned}
\end{equation*}
for constants $\tC_{\tau,1}', \tC_{\tau,2}'>0$ independent of $n$.

Therefore, taking the supremum over $x$ gives
\begin{equation*}
    \sup_{x \in B_\mu(L)}\Fnorm{\hQ_{\rho}(x)-Q^*(x)} \leq a_0 + a_2 \sup_{x \in B_\mu(L)} \Fnorm{\hQ_{\rho}(x)-Q^*(x)}^2
\end{equation*}
Solving the above quadratic inequality for $\sup_{x \in B_\mu(L)}\Fnorm{\hQ_{\rho}(x)-Q^*(x)}$ then gives
\begin{equation*}
    \sup_{x \in B_\mu(L)}\Fnorm{\hQ_{\rho}(x)-Q^*(x)} \in \sbr{0,\frac{1 - \sqrt{1-4a_0a_2}}{2a_2}} \cup \left[\frac{1 + \sqrt{1-4a_0a_2}}{2a_2}, +\infty \right)
\end{equation*}
Since $\tE_2 \subset \tE_1$ and the slow rate of convergence (Lemma \ref{lem: prf_conv_lem_slow}) implies that only the smaller branch should be retained. Therefore, one has
\begin{equation*}
    \begin{aligned}
        \sup_{x \in B_\mu(L)}\Fnorm{\hQ_{\rho}(x)-Q^*(x)} &\leq \frac{1 - \sqrt{1-4a_0a_2}}{2a_2}\\
        &=\frac{4a_0a_2}{2a_2 (1 + \sqrt{1-4a_0a_2})}\\
        &\leq C_{\mathrm{fast}, \tau}\frac{\cdot \polylog{n}}{\sqrt{n}}
    \end{aligned}
\end{equation*}
for some constant $C_{\mathrm{fast}, \tau}$ independent of $n$.

The proof of Lemma \ref{lem: fast} is then complete.

\subsection{Proof of Lemma \ref{lem: F_moment}}
\label{subsec: prf_F_moment}

Fix $\tau =10$. Denote
\begin{align*}
    \Delta_{n, \rho} &= \sup_{x \in B_\mu(L)} \Fnorm{\hQ_{n, \rho}(x)-Q^*(x)}\\
    \tE_{\rho} &= \cbr{\sup_{x \in B_{\mu}(L)} \Onorm{\hQ_{n, \rho}(x)} \leq n^2 \tM_L}, 
\end{align*}
where $\tM_L = \rbr{C_B \vee c_b}^2 L^2  L_{\tau}^{4+\rbr{4 \vee 2C_1}} \geq t_0^2$ is defined in Lemma \ref{lem: bdd}.
By definition, one has $\tE_2 \subset \tE_0 \subset \tE_{\rho}$. Hence one can obtain the following decomposition
\begin{align*}
    \EE \Delta_{n, \rho} &= \EE \Delta_{n, \rho} \1 \cbr{\tE_2} + \EE \Delta_{n, \rho} \1 \cbr{\tE_2^c \cap \tE_\rho} + \EE \Delta_{n, \rho} \1 \cbr{\tE_\rho^c}\\
    &\leq C_{\mathrm{fast}, \tau} \frac{\polylog{n}}{\sqrt{n}} + + \EE \Delta_{n, \rho} \1 \cbr{\tE_2^c \cap \tE_\rho} + \EE \Delta_{n, \rho} \1 \cbr{\tE_\rho^c}
\end{align*}
\begin{itemize}
    \item Under event $\tE_2^c \cap \tE_\rho$, one has
    \begin{align*}
        \Delta_{n, \rho} &\lesssim  \sup_{x \in B_\mu(L)} \Onorm{\hQ_{\rho}(x)-Q^*(x)}\\
        &\leq \sup_{x \in B_\mu(L)} \sbr{\Onorm{\hQ_{\rho}(x)} \vee \Onorm{Q^*(x)}}\\
        &\overset{(i)}{\lesssim} n^2 \tM_L
    \end{align*}
    Here (i) follows from Lemma \ref{lem: Q*x}, \ref{lem: bdd} and the definition of $\tE_{\rho}$. Hence
    \begin{align*}
        \EE \Delta_{n, \rho} \1 \cbr{\tE_2^c \cap \tE_\rho} &\lesssim n^2 \tM_L \PP(\tE_2^c)\\
        &\lesssim  \frac{\polylog{n}}{n^{\tau-2}}
    \end{align*}
    \item Under event $\tE_\rho^c$, one has
    \begin{align*}
        \Delta_{n, \rho} &\lesssim \sup_{x \in B_{\mu}(L)} \sbr{\Onorm{\hQ_{\rho}(x)} + \Onorm{Q^*(x)}}\\
        &\lesssim \sup_{x \in B_{\mu}(L)} \Onorm{\hQ_{\rho}(x)} 
    \end{align*}
    Lemma  \ref{lem: bdd} implies that under $\tE_\rho^c$, for any $t \geq n^2 \tM_L$,
    \begin{align*}
        &\quad \ \PP \cbr{\sup_{x \in B_{\mu}(L)} \Onorm{\hQ_{\rho}(x)} \geq t}\\
        &\lesssim h_0(\sqrt{t}/n)\\
        &\lesssim n \cdot\exp  \sbr{- \rbr{\frac{\sqrt{t}}{c_b Ln}}^{2/(C_{\Lambda}\vee 2 + 2)}}\\
        &= n \cdot\exp  \sbr{- \rbr{\frac{t}{c_b^2 L^2 n^2}}^{1/(C_{\Lambda}\vee 2 + 2)}}
    \end{align*}
    Hence
    \begin{align*}
        \EE \Delta_{n, \rho}\1 \cbr{\tE_\rho^c} &\lesssim  \EE \sup_{x \in B_{\mu}(L)} \Onorm{\hQ_{\rho}(x)} \1\cbr{\sup_{x\in B_{\mu}(L)}\Onorm{\hQ_{\rho}(x)} > n^2 \tM_L}\\
        &\lesssim \int_{n^2 \tM_L}^{\infty} n t \exp  \sbr{- \rbr{\frac{t}{c_b^2 L^2 n^2}}^{1/(C_{\Lambda}\vee 2 + 2)}} dt\\
        &= \int_{\tM_L}^{\infty} n^5 t \exp  \sbr{- \rbr{\frac{t}{c_b^2 L^2}}^{1/(C_{\Lambda}\vee 2 + 2)}} dt\\
        &\overset{(i)}{\lesssim} \int_{\tM_L}^{\infty} n^5 t \cbr{\exp  \sbr{- \rbr{\frac{\tM_L}{c_b^2 L^2}}^{1/(C_{\Lambda}\vee 2 + 2)} - \rbr{\frac{t}{c_b^2 L^2}}^{1/(C_{\Lambda}\vee 2 + 2)}}} dt\\
        &\lesssim \int_{\tM_L}^{\infty} n^5 t \cbr{ n^{-(1+\tau)}\exp  \sbr{ - \rbr{\frac{t}{c_b^2 L^2}}^{1/(C_{\Lambda}\vee 2 + 2)}}} dt\\
        &\lesssim n^{4-\tau} L^4 \int_{\tM_L/(c_b L^2)}^{\infty} t \exp  \sbr{ - t^{1/(C_{\Lambda}\vee 2 + 2)}} dt\\
        &\lesssim n^{4-\tau} L^4
    \end{align*}
    Here (i) follows from the inequality
    \begin{align*}
        (a+b)^{\alpha} \geq (a^{\alpha} + b^{\alpha})/2^{1-\alpha}, \quad a,b \geq 0, 0 < \alpha \leq 1
    \end{align*}

\end{itemize}
Finally, it follows by combining results above that for $\tau \geq 9/2$, 
\begin{align*}
\EE \Delta_{n, \rho} \leq \frac{\polylog{n}}{\sqrt{n}}
\end{align*}
This finishes the proof of Lemma \ref{lem: F_moment}.

\subsection{Proof of Lemma \ref{lem: W_moment}}
\label{subsec: prf_W_moment}
Let $\tau > 9/2$.
Lemma \ref{lem: diff} \ref{eqn: lem_diff_W_Frob} implies that under $\tE_2$, one also has
\begin{align*}
    \sup_{x \in B_{\mu}(L)} W\rbr{\hQ_{\rho}(x), Q^*(x)} \leq \frac{\polylog{n}}{\sqrt{n}}
\end{align*}
The proof then follows similar steps to those in Lemma \ref{lem: F_moment}. Specifically, denote 
\begin{align*}
    \tDelta_{n,\rho} = \sup_{x \in B_{\mu}(L)} W\rbr{\hQ_{\rho}(x), Q^*(x)}
\end{align*}
Then one can obtain
\begin{align*}
    \EE \tDelta_{n, \rho} &= \EE \tDelta_{n, \rho} \1 \cbr{\tE_2} + \EE \tDelta_{n, \rho} \1 \cbr{\tE_2^c \cap \tE_\rho} + \EE \tDelta_{n, \rho} \1 \cbr{\tE_\rho^c}\\
    &\leq \frac{\polylog{n}}{\sqrt{n}} + + \EE \tDelta_{n, \rho} \1 \cbr{\tE_2^c \cap \tE_\rho} + \EE \tDelta_{n, \rho} \1 \cbr{\tE_\rho^c}
\end{align*}
\begin{itemize}
    \item Under event $\tE_2^c \cap \tE_\rho$, Lemma \ref{lem: diff} \ref{eqn: diff_W_sup} implies that 
    \begin{align*}
        \tDelta_{n, \rho} &\lesssim \sup_{x \in B_\mu(L)} \sbr{\Onorm{\hQ_{\rho}(x)} \vee \Onorm{Q^*(x)}}\\
        &\overset{(i)}{\lesssim} n^2 \tM_L
    \end{align*}
    Here (i) follows from Lemma \ref{lem: Q*x}, \ref{lem: bdd} and the definition of $\tE_{\rho}$. Hence
    \begin{align*}
        \EE \tDelta_{n, \rho} \1 \cbr{\tE_2^c \cap \tE_\rho} &\leq n^2 \tM_L \PP(\tE_2^c)\\
        &\leq \frac{\polylog{n}}{n^{\tau-2}}
    \end{align*}
    \item Under event $\tE_\rho^c$, one can similarly obtain that
    \begin{align*}
        \tDelta_{n, \rho} 
        &\lesssim \sup_{x \in B_{\mu}(L)} \Onorm{\hQ_{\rho}(x)} 
    \end{align*}
    and
    \begin{align*}
        \EE \tDelta_{n, \rho}\1 \cbr{\tE_\rho^c} &\lesssim  \EE \sup_{x \in B_{\mu}(L)} \Onorm{\hQ_{\rho}(x)} \1\cbr{\sup_{x\in B_{\mu}(L)}\Onorm{\hQ_{\rho}(x)} > n^2 \tM_L}\\
        &\lesssim n^{-(\tau-4)}
    \end{align*}

\end{itemize}
Finally, one can arrive at the conclusion that for $\tau \geq 9/2$, 
\begin{align*}
\EE \tDelta_{n, \rho} \leq \frac{\polylog{n}}{\sqrt{n}}
\end{align*}
This finishes the proof of Lemma \ref{lem: W_moment}.

%% file: proof_of_est.tex
To simplify notation, we fix $x$ and write $Q^*$ for $Q^*(x)$, $\hQ_{n,\rho}$ for $\hQ_{n,\rho}(x)$ when there is no ambiguity. 

Following the same argument as (\ref{eqn: prf_conv_Qn_Q}) in Appendix \ref{sec: prf_unfm_conv}, one can obtain
\begin{align}
    &\EE_{(X,Q) \sim \PP} \  w(x,X)\left( T_{Q^*}^Q - I_d \right)=0 \label{eqn: prf_clt_opt_Q*}\\
    &\sqrt{n}\rbr{\hQ_{n,\rho}(x)-Q^*(x)} = \left( -\underbrace{\frac{1}{n} \sum_{i=1}^{n} w_{n,\rho}(x,X_i)  dT^{Q_i}_{Q^*(x)} }_{=\hPhi_{\rho}(x)}\right)^{-1} 
    \cdot \underbrace{\frac{1}{\sqrt{n}} \sum_{i=1}^{n} w_{n,\rho}(x,X_i)\left( T^{Q_i}_{Q^*(x)} - I_d \right)}_{=:a_n(x)}  \nonumber\\
    +& \left( -\frac{1}{n} \sum_{i=1}^{n} w_{n,\rho}(x,X_i)  dT^{Q_i}_{Q^*(x)}\right)^{-1} 
    \cdot \underbrace{\frac{\sqrt{n}}{2}\cdot \frac{1}{n} \sum_{i=1}^{n} w_{n,\rho}(x,X_i)  d^2 T^{Q_i}_{\tQ_n(x)}  \cdot (\hQ_{n,\rho}(x)-Q^*(x))^{\otimes 2} }_{=:b_n(x)} \label{eqn: prf_clt_taylor}
\end{align}

\paragraph{Analysis of $\hPhi_{\rho}(x)$:} Lemma \ref{lem: conc_tld_exp} and \ref{lem: conc_s_sin} together imply that for any fixed $x$,
\begin{align*}
    \hPhi_{\rho}(x)  \xlongrightarrow{p} \EE \sbr{- w(x,X) dT^Q_{Q^*(x)}}
\end{align*}
Then by Assumption \ref{assumption: minimizer_local}, one has
\begin{align}
    \sbr{\hPhi_{\rho}(x)}^{-1}  \xlongrightarrow{p} \sbr{\EE \rbr{- w(x,X) dT^Q_{Q^*(x)}}}^{-1}
    \label{eqn: prf_clt_Phi_inv}
\end{align}

\paragraph{Analysis of $b_n(x)$:} Theorem \ref{thm: unfm_conv}, Lemma \ref{lem: conc_tld_exp} and \ref{lem: conc_s_sin} together imply that for any fixed $x$, with probability at least $1-O(n^{-100})$,
\begin{align}
    \abs{b_n(x)} &\leq \frac{\sqrt{n}}{2} \cdot \opnorm{\frac{1}{n} \sum_{i=1}^{n} w_{n,\rho}(x,X_i)  d^2 T^{Q_i}_{\tQ_n(x)}} \cdot \Fnorm{\hQ_{n,\rho}(x)-Q^*(x)}^2 \nonumber\\
    &\leq \frac{\sqrt{n}}{2} \cdot \polylog{n} \cdot \frac{\polylog{n}}{n} \nonumber\\
    &= \frac{\polylog{n}}{\sqrt{n}} \label{eqn: prf_clt_bn}
\end{align}

\paragraph{Analysis of $a_n(x)$:} by Lemma \ref{lem: weights}, one has $w(x,X) = \vec{x}^\top \vec{\Sigma}^{-1}\vec{X} $ and $w_{n,\rho}(x,X) = \vec{x}^\top \hat{\vec{\Sigma}}_{\rho}^{-1}\vec{X} $. Hence one can obtain
\begin{equation}
    \begin{aligned}
        a_n(x) &= \frac{1}{\sqrt{n}} \sum_{i=1}^n \vec{x}^\top \hat{\vec{\Sigma}}_{\rho}^{-1} \vec{X}_i\left( T^{Q_i}_{Q^*} - I_d \right)\\
        &= \left( \vec{x}^\top \hat{\vec{\Sigma}}_{\rho}^{-1} \otimes I_d \right) \cdot \frac{1}{\sqrt{n}} \sum_{i=1}^n  \vec{X}_i \otimes \left( T^{Q_i}_{Q^*} - I_d \right)\\
        &= \underbrace{\rbr{\vec{x}^\top \vec{\Sigma}^{-1} \otimes I_d} \cdot \frac{1}{\sqrt{n}} \sum_{i=1}^n  \vec{X}_i \otimes \left( T^{Q_i}_{Q^*} - I_d \right)}_{=:a_{n,1}(x)}\\
        &\quad  + \underbrace{\sqrt{n}\rbr{\vec{x}^\top (\hat{\vec{\Sigma}}_{\rho}^{-1} - \vec{\Sigma}^{-1}) \otimes I_d} \cdot \EE \vec{X} \otimes \rbr{ T^{Q}_{Q^*} - I_d }}_{=:a_{n,2}(x)}\\
        &\quad + \underbrace{\sqrt{n}\rbr{\vec{x}^\top (\hat{\vec{\Sigma}}_{\rho}^{-1} - \vec{\Sigma}^{-1}) \otimes I_d} \cdot \sbr{\frac{1}{n} \sum_{i=1}^n  \vec{X}_i \otimes \left( T^{Q_i}_{Q^*} - I_d \right) - \EE \vec{X} \otimes \rbr{ T^{Q}_{Q^*} - I_d }}}_{=:a_{n,3}(x)}
    \end{aligned}
    \label{eqn: proof_est_a2_decom}
\end{equation}

\begin{itemize}
    \item $a_{n,1}(x)$: by (\ref{eqn: prf_clt_opt_Q*}) and Lemma \ref{lem: weights}, one can obtain 
    \begin{align}
        \EE a_{n,1}(x) = 0 \label{eqn: prf_est_a21}
    \end{align}
    To see this, it suffices to show $\EE \rbr{\vec{x}^\top \vec{\Sigma}^{-1} \otimes I_d} \cdot \rbr{\vec{X} \otimes (T_{Q^*}^Q - I_d) } = 0$. Direct computation shows that the LHS of the equality is equal to
    \begin{align*}
        \EE \vec{x}^\top \vec{\Sigma}^{-1} \vec{X} \rbr{T_{Q^*}^Q - I_d} &= \EE w(x,X) (T_{Q^*}^{Q} - I_d) = 0
    \end{align*}
    Here the equality follows from the optimality condition (\ref{eqn: prf_clt_opt_Q*}).
    \item $a_{n,2}(x)$: using the formula $A^{-1}-B^{-1} = A^{-1}(B-A)B^{-1} $, one can obtain
    \begin{align}
        a_{n,2}(x) &= \sqrt{n}\rbr{\vec{x}^\top \hat{\vec{\Sigma}}_{\rho}^{-1}(\vec{\Sigma} - \hat{\vec{\Sigma}}_{\rho}^{-1})\vec{\Sigma}^{-1} \otimes I_d} \cdot \EE \vec{X} \otimes \rbr{ T^{Q}_{Q^*} - I_d } \nonumber\\
        &= \sqrt{n}\rbr{\vec{x}^\top \vec{\Sigma}^{-1}(\vec{\Sigma} - \hat{\vec{\Sigma}}_{\rho})\vec{\Sigma}^{-1} \otimes I_d} \cdot \EE \vec{X} \otimes \rbr{ T^{Q}_{Q^*} - I_d } \nonumber \\
        &\quad - \sqrt{n}\rbr{\vec{x}^\top (\vec{\Sigma}^{-1}-\hat{\vec{\Sigma}}_{\rho}^{-1})(\vec{\Sigma} - \hat{\vec{\Sigma}}_{\rho}^{-1})\vec{\Sigma}^{-1} \otimes I_d} \cdot \EE \vec{X} \otimes \rbr{ T^{Q}_{Q^*} - I_d } \nonumber\\
        &= - \frac{1}{\sqrt{n}} \rbr{\sum_{i=1}^n \vec{x}^\top \vec{\Sigma}^{-1} (\vec{X}_i \vec{X}_i^\top - \vec{\Sigma})\vec{\Sigma}^{-1} \otimes I_d} \cdot \rbr{\EE \vec{X} \otimes (T^{Q}_{Q^*} - I_d)}  + O_p \rbr{ \frac{1}{\sqrt{n}}} 
        \label{eqn: prf_clt_an2}
    \end{align}
    Here the last line follows from the fact that $\hat{\vec{\Sigma}}_{\rho}-\vec{\Sigma} = O_p(n^{-1/2})$ which is itself due to the sub-Gaussianity of $X$.
    \item $a_{n,3}(x)$: note that $\sbr{\frac{1}{n} \sum_{i=1}^n  \vec{X}_i \otimes \left( T^{Q_i}_{Q^*} - I_d \right) - \EE \vec{X} \otimes \rbr{ T^{Q}_{Q^*} - I_d }} = O_p(n^{-1/2})$ again by the sub-Gaussianity of $X$ and Assumption \ref{assumption: bdd_Q}, one can arrive at
    \begin{align}
        a_{n,3}(x) = O_p(n^{-1/2}) \label{eqn: prf_clt_an3}
    \end{align}
\end{itemize}

Combining (\ref{eqn: proof_est_a2_decom}), (\ref{eqn: prf_est_a21}), (\ref{eqn: prf_clt_an2}), (\ref{eqn: prf_clt_an3}) and the functional central limit theorem  \citep[Theorem 7.7.6]{hsing}, one can obtain
\begin{align}
    a_n(x) &=  \frac{1}{\sqrt{n}} \sum_{i=1}^n \Big[ \rbr{\vec{x}^\top \vec{\Sigma}^{-1} \otimes I_d} \cdot \rbr{\vec{X}_i \otimes ( T^{Q_i}_{Q^*} - I_d )} - 
    \nonumber\\
    &\qquad \qquad \quad \  \rbr{ \vec{x}^\top \vec{\Sigma}^{-1} (\vec{X}_i \vec{X}_i^\top - \vec{\Sigma}) \otimes I_d} \cdot \rbr{\EE \vec{X} \otimes (T^{Q}_{Q^*} - I_d)} \Big]  + o_p(1) \nonumber\\
    &\overset{w}{\to} Z_x
    \label{eqn: prf_clt_an}
\end{align}
where $Z_x \sim \cN \rbr{0, \Xi_x}$ is a Gaussian random matrix with mean $0$ and covariance $\Xi_x$. Here $\Xi_x := \EE V_x \otimes V_x$ with
\begin{align*}
    &V_x = w(x,X) \rbr{ T^{Q}_{Q^*} - I_d } - 
     \rbr{ \vec{x}^\top \vec{\Sigma}^{-1} (\vec{X} \vec{X}^\top - \vec{\Sigma}) \otimes I_d} \cdot \rbr{\EE \vec{X} \otimes (T^{Q}_{Q^*} - I_d)}
\end{align*}

Finally, (\ref{eqn: prf_clt_taylor}), (\ref{eqn: prf_clt_Phi_inv}), (\ref{eqn: prf_clt_bn}) and (\ref{eqn: prf_clt_an}) together imply
\begin{equation*}
    \sqrt{n}\left( \hQ_{n,\rho} - Q^* \right) \overset{w}{\to} \left( -\EE w(x,X)dT^{Q}_{Q^*} \right)^{-1} \cdot Z_x
\end{equation*}
which completes the proof.

%% file: clt_cor.tex
The optimality condition for $Q^*(x)$ implies that
\begin{align*}
    \EE w(x,X)\rbr{T_{Q^*(x)}^Q  - I_d} = 0
\end{align*}
Note that independence between $X$ and $Q$ implies that
\begin{align*}
    \EE w(x,X)\rbr{T_{Q^*(x)}^Q  - I_d} &= \EE w(x,X) \EE \rbr{T_{Q^*(x)}^Q  - I_d}\\
    &\overset{(i)}{=} \EE \rbr{T_{Q^*(x)}^Q  - I_d}
\end{align*}
Here (i) follows from the fact that $\EE w(x,X)=1$. Therefore, one has
\begin{align*}
    \EE \rbr{T_{Q^*(x)}^Q  - I_d} = 0
\end{align*}
By independence, the above equality then implies 
\begin{align*}
    \EE \vec{X} \otimes (T^{Q}_{Q^*} - I_d) = 0
\end{align*}
Hence one has $V_{x,2}=0$, the proof is then complete.

%% file: proof_test_dist.tex
First, we demonstrate in Lemma \ref{lem: prf_null_lem_uni} below uniform fast convergence under Assumption \ref{assumption: X}, \ref{assumption: bdd_Q}, \ref{assumption: minimizer_local}, \ref{assumption: cond_ind} and the null hypothesis. In order to apply Theorem \ref{thm: unfm_conv}, it suffices to verify Assumption \ref{assumption: minimizer_local} and Assumption \ref{assumption: minimizer_global}. Note that Assumption \ref{assumption: cond_ind} and the null implies the independence between $X$ and $Q$, then both Assumption \ref{assumption: minimizer_local} and \ref{assumption: minimizer_global} are consequences of Lemma \ref{lem: F_ind}.
\begin{lemma}
    \label{lem: prf_null_lem_uni}
    Instate the assumptions in Theorem \ref{thm: test_dist}. Then with probability at least $1-O(n^{-100})$, one has
    \begin{align*}
        \sup_{x\in B_{\mu}(L_n)} \Fnorm{\hQ_{\rho}(x) - Q^*(x)} &\leq \frac{\polylog{n}}{\sqrt{n}}\\
        \sup_{1\leq i\leq n} \Fnorm{\hQ_{\rho}(X_i) - Q^*(X_i)} &\leq \frac{\polylog{n}}{\sqrt{n}}
    \end{align*}
\end{lemma}

With Lemma \ref{lem: prf_null_lem_uni} in place, we consider the Taylor expansion of the optimality condition for $\hQ_{\rho}(x)$. Recall that the optimality condition for $\hQ_{\rho}(x)$ gives
\begin{align*}
    \frac{1}{n} \sum_{i=1}^n w_{n,\rho}(x,X_i) \rbr{T^{Q_i}_{\hQ_{\rho}(x)} - I_d} = 0
\end{align*}
Then one can apply Lemma \ref{lem: diff} to get the following 2nd order Taylor expansion around $Q^*(x)$.
\begin{align}
    0 &= \frac{1}{n} \sum_{i=1}^n w_{n,\rho}(x,X_i) \rbr{T_{Q^*(x)}^{Q_i} - I_d} + \frac{1}{n} \sum_{i=1}^n w_{n,\rho}(x,X_i) dT_{Q^*(x)}^{Q_i} \cdot \rbr{\hQ_{\rho}(x) - Q^*(x)} + R_2(x) 
    \label{eqn: test_taylor}
\end{align}
where $R_2(x)$ is the 2nd order remainder term with some $\tQ_n(x)$ lying between $Q^*(x)$ and $\hQ_{\rho}(x)$ defined as follows.
\begin{align*}
    R_2(x) := \frac{1}{2n} \sum_{i=1}^n w_{n,\rho}(x,X_i) d^2 T_{\tQ_n(x)}^{Q_i} \cdot \rbr{\hQ_{\rho}(x) - Q^*(x)}^{\otimes 2}
\end{align*}
Under the null hypothesis (\ref{eqn: testing_null}), one has $Q^*(x) \equiv Q^*$ and (\ref{eqn: test_taylor}) reduces to
\begin{align}
    0 &= \frac{1}{n} \sum_{i=1}^n w_{n,\rho}(x,X_i) \rbr{T_{Q^*}^{Q_i} - I_d} + \frac{1}{n} \sum_{i=1}^n w_{n,\rho}(x,X_i) dT_{Q^*}^{Q_i} \cdot \rbr{\hQ_{\rho}(x) - Q^*(x)} + R_2(x) 
    \label{eqn: test_taylor_null}
\end{align}
Setting $x=\bar{X}$ in (\ref{eqn: test_taylor_null}) then gives
\begin{align}
    0 &= \frac{1}{n} \sum_{i=1}^n  \rbr{T_{Q^*}^{Q_i} - I_d} + \frac{1}{n} \sum_{i=1}^n dT_{Q^*}^{Q_i} \cdot \rbr{\hQ_{\rho}(x) - Q^*(x)} + R_2(x) 
    \label{eqn: test_taylor_null_xbar}
\end{align}
Take difference between (\ref{eqn: test_taylor_null}) and (\ref{eqn: test_taylor_null_xbar}) and rearrange, one can obtain
\begin{equation}
    \begin{aligned}
        &\underbrace{\left( -\frac{1}{n}\sum_{i=1}^n dT_{Q^*}^{Q_i} \right) \cdot (\hQ_{\rho}(x) -\hQ_{\rho}(\bar{X}))}_{=:\tau(x)}\\
        = &  
        \underbrace{\frac{1}{n}\sum_{i=1}^{n} (w_{n,\rho}(x,X_i)-1) (T_{Q^*}^{Q_i} - I_d)}_{=:\alpha_0(x)} + \underbrace{\frac{1}{n}\sum_{i=1}^{n} (w_{n,\rho}(x,X_i)-1) dT^{Q_i}_{Q^*}\cdot (\hQ_{\rho}(x) -\hQ_{\rho}(\bar{X})) }_{=:\alpha_1(x)}\\
        &+ \underbrace{R_2(x) - R_2(\bar{X})}_{=:\alpha_2(x)}
    \end{aligned}
    \label{eqn: prf_null_decomp}
\end{equation}
Before delving further into the proof, we pause to introduce more notations. Define 
\begin{equation*}
    \begin{aligned}
        \htau(x) &= \left( -\frac{1}{n}\sum_{i=1}^n dT_{\hQ_{\rho}(\bar{X})}^{Q_i} \right) \cdot (\hQ_{\rho}(x) -\hQ_{\rho}(\bar{X}))
    \end{aligned}
\end{equation*}
Then by definition, the test statistic $\hat{\cT}_{\rho} = \sum_{i=1}^{n} \Fnorm{\htau(X_i)}^2$, and we have the following decomposition of $\tau(x)$.
\begin{equation*}
    \begin{aligned}
        \htau(x) &= \tau(x) + \htau(x) - \tau(x) \\
        &=\tau(x) +  \underbrace{\left( \frac{1}{n}\sum_{i=1}^n dT_{Q^*}^{Q_i}-\frac{1}{n}\sum_{i=1}^n dT_{\hQ_{\rho}(\bar{X})}^{Q_i} \right) \cdot (\hQ_{\rho}(x) -\hQ_{\rho}(\bar{X}))}_{=:\alpha_3(x)}\\
        &= \alpha_0(x) + \underbrace{\alpha_1(x) +\alpha_2(x) + \alpha_3(x)}_{=:R(x)}
    \end{aligned}
\end{equation*}
We also define counterparts of $\alpha_0,...,\alpha_2$ and $R_2$ by replacing $w_{n,\rho}(\cdot, \cdot)$ with $w(\cdot, \cdot)$ as follows.
\begin{equation*}
    \begin{aligned}
        \talpha_0(x)&:=\frac{1}{n}\sum_{i=1}^{n} (w(x,X_i)-1) (T_{Q^*}^{Q_i} - I_d)\\
        \talpha_1 (x)&:= \frac{1}{n}\sum_{i=1}^{n} (w(x,X_i)-1) dT^{Q_i}_{Q^*} \cdot (\hQ_{\rho}(x) -\hQ_{\rho}(\bar{X}))\\
        \talpha_2(x) &:= \tR_2(x) - \tR_2(\bar{X})
    \end{aligned}
\end{equation*}
where 
\begin{equation*}
    \tR_2(x)= \frac{1}{2n}\sum_{i=1}^{n} w(x,X_i)  d^2 T^{Q_i}_{\tQ_n(x)} \cdot (\hQ_{\rho}(x)-Q^*)^{\otimes 2}
\end{equation*}
With the above notation in place, one can obtain
\begin{equation}
    \begin{aligned}
        \hat{\cT}_{\rho} &= \sum_{k=1}^{n} \Fnorm{\alpha_0(X_k)}^2 + \underbrace{2 \sum_{k=1}^{n} \inner{\alpha_0(X_k)}{R(X_k)} + \sum_{k=1}^{n} \Fnorm{R(X_k)}^2 }_{\rem_n}
    \end{aligned}
    \label{eqn: prf_null_T_op1}
\end{equation}
The proof is then divided into three steps.
\begin{itemize}
    \item First, we give upper bounds for $\talpha_i$ as well as $\talpha_i - \alpha_i$ and their consequences.
    \item Next, we show that the remainder term $\rem_n$ is negligible.
    \item Then, we demonstrate that $\sum_i \Fnorm{\alpha_0(X_i)}^2$ converges weakly to the desired asymptotic null distribution (\ref{eqn: asymp_null_dist}).
\end{itemize}

\paragraph{Analysis of $\talpha$ and $\talpha-\alpha$:} We give uniform upper bounds for $\talpha_i(x)-\alpha_i(x)$ for $i=0,1,2$ in Lemma \ref{lem: prf_null_lem_s_sin} as uniform upper bounds for $\talpha_i(x)$ and $\alpha_3(x)$ in Lemma \ref{lem: prf_null_lem_til}; see Appendix \ref{subsec: proof_lem_diff_ab}, \ref{subsec: prf_null_lem_til} for the proof.
\begin{lemma}
    \label{lem: prf_null_lem_s_sin}
    Instate the notations and assumptions in Theorem \ref{thm: test_dist}.
    \begin{align}
        \sup_{x \in B_\mu(L)} \left\| \alpha_0(x) - \talpha_0(x)\right\|_F &\leq \frac{\polylog{n}}{\sqrt{n}} \label{eqn: prf_null_lem_s_sin_0}\\
        \sup_{x \in B_\mu(L)} \left\| \alpha_1(x) - \talpha_1(x)\right\|_F &\leq \frac{\polylog{n}}{n} \label{eqn: prf_null_lem_s_sin_1}\\
        \sup_{x \in B_\mu(L)} \left\| \alpha_2(x) - \talpha_2(x)\right\|_F &\leq \frac{\polylog{n}}{n^{3/2}} \label{eqn: prf_null_lem_s_sin_2}
    \end{align} 
    with probability at least $1-O\left( n^{-99} \right)$.
\end{lemma}

\begin{lemma}
    \label{lem: prf_null_lem_til}
    Instate the notations and assumptions in Theorem \ref{thm: test_dist}.
    \begin{align}
        \sup_{x \in B_\mu(L)} \left\| \talpha_0(x)\right\|_F &\leq \frac{\polylog{n}}{\sqrt{n}} \label{eqn: prf_null_lem_til_0}\\
        \sup_{x \in B_\mu(L)} \left\| \talpha_1(x)\right\|_F &\leq \frac{\polylog{n}}{n} \label{eqn: prf_null_lem_til_1}\\
        \sup_{x \in B_\mu(L)} \left\| \talpha_2(x)\right\|_F &\leq \frac{\polylog{n}}{n} \label{eqn: prf_null_lem_til_2}\\
        \sup_{x \in B_\mu(L)} \left\| \alpha_3(x)\right\|_F &\leq \frac{\polylog{n}}{n} \label{eqn: prf_null_lem_til_3}
    \end{align}
    with probability at least $1-O\left( n^{-99} \right)$.
\end{lemma}
With the above lemmas in place, one can readily obtain that with probability at least $1-O(n^{-99})$,
\begin{align}
    \sup_{i \in [n]} \Fnorm{\alpha_0(X_i)}
    &\leq \sup_{i \in [n]} \Fnorm{\talpha_0(X_i) - \alpha_0(X_i)} + \sup_{i \in [n]} \Fnorm{\talpha_0(X_i)}
    \leq \frac{\polylog{n}}{\sqrt{n}} \label{eqn: prf_null_a0}\\
    \sup_{i \in [n]} \Fnorm{\alpha_1(X_i)}
    &\leq \sup_{i \in [n]} \Fnorm{\talpha_1(X_i) - \alpha_1(X_i)} + \sup_{i \in [n]} \Fnorm{\talpha_1(X_i)}
    \leq \frac{\polylog{n}}{n} \label{eqn: prf_null_a1}\\
    \sup_{i \in [n]} \Fnorm{\alpha_2(X_i)}
    &\leq \sup_{i \in [n]} \Fnorm{\talpha_2(X_i) - \alpha_2(X_i)} + \sup_{i \in [n]} \Fnorm{\talpha_2(X_i)}
    \leq \frac{\polylog{n}}{n} \label{eqn: prf_null_a2}
\end{align}
which then implies
\begin{equation}
    \sup_{i \in [n]} \Fnorm{R(X_i)} \leq \sum_{k=1}^{3} \sup_{i \in [n]} \Fnorm{\alpha_k(X_i)} \leq \frac{\polylog{n}}{n}
    \label{eqn: prf_null_R}
\end{equation}

\paragraph{Negligibility of $\rem_n$:} We consider two terms $\sum_{k=1}^{n} \inner{\alpha_0(X_k)}{R(X_k)}$ and $\sum_{k=1}^{n} \Fnorm{R(X_k)}^2$ separately.

\noindent\textit{Analysis of $\sum_{k=1}^{n} \inner{\alpha_0(X_k)}{R(X_k)}$:} By (\ref{eqn: prf_null_a0}) and (\ref{eqn: prf_null_R}), one has with probability at least $1-O(n^{-99})$,
\begin{align*}
    \sum_{k=1}^{n} \inner{\alpha_0(X_k)}{R(X_k)} &\leq \sum_{k=1}^n \Fnorm{\alpha_0(X_k)} \cdot \Fnorm{R(X_k)}\\
    &\leq \sum_{k=1}^n \frac{\polylog{n}}{n^{3/2}}\\
    &= \frac{\polylog{n}}{\sqrt{n}}
\end{align*}

\noindent\textit{Analysis of $\sum_{k=1}^{n} \Fnorm{R(X_k)}^2$:} Similarly, by (\ref{eqn: prf_null_R}), one has with probability at least $1-O(n^{-99})$,
\begin{align*}
    \sum_{k=1}^{n} \Fnorm{R(X_k)}^2 \leq \frac{\polylog{n}}{n}
\end{align*}

Therefore, the above results imply that with probability at least $1-O(n^{-99})$,
\begin{align}
    \rem_n \leq \frac{\polylog{n}}{\sqrt{n}}
    \label{eqn: prf_null_rem}
\end{align}

\paragraph{Analysis of $\alpha_0$:} To consider the main term $\sum_{i=1}^{n} \Fnorm{\alpha_0(X_i)}^2$, one has
\begin{equation*}
    \begin{aligned}
        &\sum_{k=1}^{n} \Fnorm{\alpha_0(X_k)}^2 \\
        = &\frac{1}{n^2} \sum_{k=1}^{n} \inner{\sum_{i=1}^{n} (X_k-\bar{X})^\top \hSigma_{\rho}^{-1} (X_i - \bar{X})(T_{Q^*}^{Q_i} - I_d)}{\sum_{j=1}^{n} (X_k-\bar{X})^\top \hSigma_{\rho}^{-1} (X_j - \bar{X})(T_{Q^*}^{Q_j} - I_d)}\\
        = &\frac{1}{n^2} \sum_{i,j=1}^{n} \sum_{k=1}^{n} (X_i - \bar{X})^\top \hSigma_{\rho}^{-1} (X_k - \bar{X}) (X_k - \bar{X})^\top \hSigma_{\rho}^{-1}(X_j - \bar{X}) \inner{T_{Q^*}^{Q_i} - I_d}{T_{Q^*}^{Q_j} - I_d}\\
        = &\frac{1}{n} \sum_{i,j =1}^{n} (X_i - \bar{X})^\top \hSigma_{\rho}^{-1} (X_j - \bar{X})\inner{T_{Q^*}^{Q_i} - I_d}{T_{Q^*}^{Q_j} - I_d}\\
        = & \Fnorm{\frac{1}{\sqrt{n}} \sum_{i=1}^{n} \hSigma_{\rho}^{-1/2}(X_i - \bar{X}) \otimes  (T_{Q^*}^{Q_i}- I_d)}^2
    \end{aligned}
\end{equation*}
Note that
\begin{align}
    &\quad \ \frac{1}{\sqrt{n}} \sum_{i=1}^{n} \hSigma_{\rho}^{-1/2}(X_i - \bar{X}) \otimes  (T_{Q^*}^{Q_i}- I_d) \nonumber\\ 
    &= \sbr{(\hSigma_{\rho}^{-1/2} \Sigma^{1/2})\otimes I_d} \cdot \sbr{\frac{1}{\sqrt{n}} \sum_{i=1}^{n} \Sigma^{-1/2}(X_i - \bar{X}) \otimes  (T_{Q^*}^{Q_i}- I_d)} \nonumber\\
    &= (1+o_p(1))\frac{1}{\sqrt{n}} \sum_{i=1}^{n} \Sigma^{-1/2}(X_i - \bar{X}) \otimes  (T_{Q^*}^{Q_i}- I_d)\label{eqn: proof_null_alpha1_1}
\end{align}
Also, one can obtain
\begin{equation}
    \begin{aligned}
        &\left( \frac{1}{\sqrt{n}} \sum_{i=1}^{n} \Sigma^{-1/2}(X_i - \bar{X}) \otimes  (T_{Q^*}^{Q_i}- I_d) \right) - \left( \frac{1}{\sqrt{n}} \sum_{i=1}^{n} \Sigma^{-1/2}(X_i - \mu) \otimes  (T_{Q^*}^{Q_i}- I_d) \right) \\
        = &\frac{1}{\sqrt{n}} \sum_{i=1}^{n} \Sigma^{-1/2} (\mu- \bar{X}) \otimes (T_{Q^*}^{Q_i}- I_d)\\
        =& \Sigma^{-1/2} (\mu- \bar{X}) \otimes \sbr{  \frac{1}{\sqrt{n}} \sum_{i=1}^n \rbr{T_{Q^*}^{Q_i}- I_d}}\\
        = &o_p(1)
    \end{aligned}
    \label{eqn: proof_null_alpha1_2}
\end{equation}
Here the last line follows since $\bar{X}-\mu = o_p(1)$ and $\frac{1}{\sqrt{n}} \sum_{i=1}^{n} (T_{Q^*}^{Q_i}- I_d)$ is asymptotically normal with zero mean. The zero mean is justified in the following claim whose proof is deferred to Appendix \ref{subsec: prf_null_claim_ind}.
\begin{claim}
    \label{claim: null_ind}
    Under the null hypothesis (\ref{eqn: testing_null}) and Assumption \ref{assumption: minimizer_global}, \ref{assumption: cond_ind}, $X$ and $Q$ are independent and one has
    \begin{align}
        \EE (X-\mu) \otimes \rbr{T_{Q^*}^{Q}- I_d}  &= 0 \label{eqn: claim_ind_otimes}\\
        \EE \rbr{T_{Q^*}^{Q}- I_d}  &=0 \label{eqn: claim_ind_T}
    \end{align}
\end{claim}

Claim \ref{claim: null_ind} also implies that $\EE \Sigma^{-1/2}(X-\mu) \otimes (T_{Q^*}^{Q_i} - I_d)=0 $. Therefore, by the functional central limit theorem  \citep[Theorem 7.7.6]{hsing}, one can obtain
\begin{equation}
    \frac{1}{\sqrt{n}} \sum_{i=1}^{n} \Sigma^{-1/2}(X_i - \mu) \otimes  (T_{Q^*}^{Q_i}- I_d) \overset{w}{\to} \cN\left( 0, I_p \otimes \EE \left[ (T_{Q^*}^{Q} - I_d) \otimes (T_{Q^*}^{Q} - I_d) \right]\right)
    \label{eqn: proof_null_alpha1_3}
\end{equation}
Combining (\ref{eqn: proof_null_alpha1_1}), (\ref{eqn: proof_null_alpha1_2}) and (\ref{eqn: proof_null_alpha1_3}), we arrive at
\begin{equation*}
    \begin{aligned}
        \frac{1}{\sqrt{n}} \sum_{i=1}^{n} \hSigma_{\rho}^{-1/2}(X_i - \bar{X}) \otimes  (T_{Q^*}^{Q_i}-I_d) \overset{w}{\to} \cN\left( 0, I_p \otimes \EE \left[ (T_{Q^*}^{Q} - I_d) \otimes (T_{Q^*}^{Q} - I_d) \right]\right)
    \end{aligned}
\end{equation*}
which then implies that
\begin{equation}
    \sum_{k=1}^{n} \left\|\alpha_0(X_k) \right\|^2 \overset{w}{\to} \sum_{i}^{} \lambda_i w_i
    \label{eqn: proof_null_sum_alpha0}
\end{equation}
where $w_i$ are i.i.d. $\chi_p^2$ random variables and $\lambda_i$ are the eigenvalues of $\EE \left[ (T_{Q_*}^Q -I_d) \otimes (T_{Q_*}^Q -I_d) \right]$.

Finally, taking (\ref{eqn: prf_null_T_op1}) (\ref{eqn: prf_null_rem}) and (\ref{eqn: proof_null_sum_alpha0}) collectively yields
\begin{equation*}
    \hat{\cT}_{\rho} \overset{w}{\to} \sum_{i}^{}  \lambda_i w_i
\end{equation*}



\subsection{Proof of Lemma \ref{lem: prf_null_lem_s_sin}}
\label{subsec: proof_lem_diff_ab}

Apply Lemma \ref{lem: prf_null_lem_uni} and triangle inequality to see that under the null (\ref{eqn: testing_null}), with probability at least $1-O(n^{-100})$, one has
\begin{equation}
    \begin{aligned}
        \sup_{x \in B_\mu(L)} \Fnorm{\hQ_{\rho}(x) - Q^*} &\leq \frac{\polylog{n}}{\sqrt{n}}\\
        \sup_{x \in B_\mu(L)} \Fnorm{\hQ_{\rho}(x) -\hQ_{\rho}(\bar{X})} &\leq \sup_{x \in B_\mu(L)} \Fnorm{\hQ_{\rho}(x) -Q^*} + \sup_{x \in B_\mu(L)} \Fnorm{\hQ_{\rho}(\bar{X}) -Q^*} \leq \frac{\polylog{n}}{\sqrt{n}}
    \end{aligned}
    \label{eqn: prf_null_fast}
\end{equation}

\paragraph{Proof of (\ref{eqn: prf_null_lem_s_sin_0}):} note that
\begin{equation*}
    \begin{aligned}
        \alpha_0(x)-\talpha_0(x) &= \frac{1}{n}\sum_{i=1}^{n} \left( w_{n,\rho}(x,X_i)- w(x,X_i)\right) (T_{Q^*}^{Q_i} - I_d)\\
        &=A_n(x)- \tA_n(x)
    \end{aligned}
\end{equation*}
Hence (\ref{eqn: prf_null_lem_s_sin_0}) follows from Lemma \ref{lem: conc_s_sin}.

\paragraph{Proof of (\ref{eqn: prf_null_lem_s_sin_1}):} one can obtain
\begin{equation*}
    \begin{aligned}
        \Fnorm{\alpha_1(x) - \talpha_1(x)} &= \Fnorm{\frac{1}{n}\sum_{i=1}^{n} (w_{n,\rho}(x,X_i)-w(x,X_i)) dT^{Q_i}_{Q^*}\cdot (\hQ_{\rho}(x) -\hQ_{\rho}(\bar{X})) }\\
        &\leq \opnorm{\frac{1}{n}\sum_{i=1}^{n} (w_{n,\rho}(x,X_i)-w(x,X_i)) dT^{Q_i}_{Q^*}} \cdot \Fnorm{\hQ_{\rho}(x) -\hQ_{\rho}(\bar{X})}\\
        &= \opnorm{\Phi_n(x)- \tPhi_n(x)} \cdot \Fnorm{\hQ_{\rho}(x) -\hQ_{\rho}(\bar{X})}
    \end{aligned}
\end{equation*}
Hence (\ref{eqn: prf_null_lem_s_sin_1}) follows from Lemma \ref{lem: conc_s_sin} and (\ref{eqn: prf_null_fast}).

\paragraph{Proof of (\ref{eqn: prf_null_lem_s_sin_2}):} one can obtain
\begin{equation*}
    \begin{aligned}
        \Fnorm{\alpha_2(x) - \talpha_2(x)} &= \Fnorm{\frac{1}{n}\sum_{i=1}^{n} (w_{n,\rho}(x,X_i)-w(x,X_i)) d^2 T^{Q_i}_{\tQ_n(x)} \cdot (\hQ_{\rho}(x)-Q^*)^{\otimes 2} }\\
        &\leq \opnorm{\frac{1}{n}\sum_{i=1}^{n} (w_{n,\rho}(x,X_i)-w(x,X_i)) d^2 T^{Q_i}_{\tQ_n(x)}} \cdot \Fnorm{\hQ_{\rho}(x) -\hQ_{\rho}(\bar{X})}^2 \\
        &= \opnorm{\Psi_n(x,\tQ_n(x))- \tPsi_n(x,\tQ_n(x))} \cdot \Fnorm{\hQ_{\rho}(x) -\hQ_{\rho}(\bar{X})}^2
    \end{aligned}
\end{equation*}
Hence (\ref{eqn: prf_null_lem_s_sin_2}) follows from Lemma \ref{lem: conc_s_sin} and (\ref{eqn: prf_null_fast}).

\subsection{Proof of Lemma \ref{lem: prf_null_lem_til}}
\label{subsec: prf_null_lem_til}

Apply Lemma \ref{lem: prf_null_lem_uni} and triangle inequality to see that under the null (\ref{eqn: testing_null}), with probability at least $1-O(n^{-100})$, one has
\begin{equation}
    \begin{aligned}
        \sup_{x \in B_\mu(L)} \Fnorm{\hQ_{\rho}(x) - Q^*} &\leq \frac{\polylog{n}}{\sqrt{n}}\\
        \sup_{x \in B_\mu(L)} \Fnorm{\hQ_{\rho}(x) -\hQ_{\rho}(\bar{X})} &\leq \sup_{x \in B_\mu(L)} \Fnorm{\hQ_{\rho}(x) -Q^*} + \sup_{x \in B_\mu(L)} \Fnorm{\hQ_{\rho}(\bar{X}) -Q^*} \leq \frac{\polylog{n}}{\sqrt{n}}
    \end{aligned}
    \label{eqn: prf_null_lem2_fast}
\end{equation}

\paragraph{Proof of (\ref{eqn: prf_null_lem_til_0}):} 
Recall the definition of $\tA_n(x)$ in (\ref{eqn: prf_conv_def_til}), one can see that (\ref{eqn: prf_null_lem_til_0}) follows from (\ref{eqn: conc_lem_unfmA}) in Lemma \ref{lem: conc_tld_exp} with a slight and straightforward modification to accommodate the $w(x,X)-1$ term here. For brevity, we omit the proof.

\paragraph{Proof of (\ref{eqn: prf_null_lem_til_1}):}
\begin{equation}
    \begin{aligned}
        \Fnorm{\talpha_1} \leq \opnorm{\frac{1}{n}\sum_{i=1}^{n} (w(x,X_i)-1) dT^{Q_i}_{Q^*}} \cdot \Fnorm{\hQ_{\rho}(x) -\hQ_{\rho}(\bar{X})} 
    \end{aligned}
    \label{eqn: prf_null_lem_til1_1}
\end{equation}
Recall the definition of $\tPhi_n(x)$ in (\ref{eqn: prf_conv_def_til}), with a slight modification of the proof of (\ref{eqn: conc_lem_unfmPhi}) in Lemma \ref{lem: conc_tld_exp}, one can obtain
\begin{equation}
    \sup_{x \in B_\mu(L)}\opnorm{\frac{1}{n}\sum_{i=1}^{n} (w(x,X_i)-1) dT^{Q_i}_{Q^*}} \leq \frac{\polylog{n}}{\sqrt{n}}
    \label{eqn: prf_null_lem_til1_2}
\end{equation}
For brevity, the proof is omitted.

Combining (\ref{eqn: prf_null_lem2_fast}), (\ref{eqn: prf_null_lem_til1_1}) and (\ref{eqn: prf_null_lem_til1_2}) gives (\ref{eqn: prf_null_lem_til_1}).

\paragraph{Proof of (\ref{eqn: prf_null_lem_til_2}):} Apply triangle inequality to see that
\begin{equation*}
    \begin{aligned}
        \sup_{x \in B_\mu(L)} \Fnorm{\talpha_2(x)} &\leq \Fnorm{\tR_2(\bar{X})} + \sup_{x \in B_\mu(L)}\Fnorm{\tR_2(x)} \\
        &\leq 2\sup_{x \in B_\mu(L)}\Fnorm{\tR_2(x)}
    \end{aligned}
\end{equation*}
Moreover, one can obtain
\begin{equation*}
    \begin{aligned}
        \sup_{x \in B_\mu(L)} \Fnorm{\tR_2(x)} &\lesssim  \sup_{x \in B_\mu(L)}\opnorm{\frac{1}{n}\sum_{i=1}^{n} w(x,X_i) d^2 T^{Q_i}_{\tQ_n(x)}} \cdot \Fnorm{\hQ_{\rho}(x)-Q^*}^2\\
        &\leq \sup_{ \genfrac{}{}{0pt}{2}{x \in B_\mu(L)}{S \in \cS_d((2M_L)^{-1}, 2M_L)} } \opnorm{\tPsi_n(x,S)} \cdot \sup_{x \in B_\mu(L)}\Fnorm{\hQ_{\rho}(x)-Q^*}^2\\
        &\overset{(i)}{\leq} \frac{\polylog{n}}{n}
    \end{aligned}
\end{equation*}
Here (i) follows from (\ref{eqn: conc_lem_unfmPsi}), (\ref{eqn: conc_lem_E_Psi}) in Lemma \ref{lem: conc_tld_exp} as well as (\ref{eqn: prf_null_lem2_fast}).

\paragraph{Proof of (\ref{eqn: prf_null_lem_til_3}):} one can obtain
\begin{align}
    \sup_{x \in B_{\mu}(L)} \Fnorm{\alpha_3(x)} &= \sup_{x \in B_{\mu}(L)} \Fnorm{\left( \frac{1}{n}\sum_{i=1}^n dT_{Q^*}^{Q_i}-\frac{1}{n}\sum_{i=1}^n dT_{\hQ_{\rho}(\bar{X})}^{Q_i} \right) \cdot (\hQ_{\rho}(x) -\hQ_{\rho}(\bar{X}))} \nonumber\\
    &\leq \opnorm{ \frac{1}{n}\sum_{i=1}^n dT_{Q^*}^{Q_i}-\frac{1}{n}\sum_{i=1}^n dT_{\hQ_{\rho}(\bar{X})}^{Q_i} } \cdot \sup_{x \in B_{\mu}(L)} \Fnorm{\hQ_{\rho}(x) -\hQ_{\rho}(\bar{X})} \label{eqn: prf_null_lem_til1_a3}
\end{align}
Moreover, one has
\begin{align}
    \opnorm{ \frac{1}{n}\sum_{i=1}^n dT_{Q^*}^{Q_i}-\frac{1}{n}\sum_{i=1}^n dT_{\hQ_{\rho}(\bar{X})}^{Q_i} }
    &\overset{(i)}{\leq} \frac{1}{n} \sum_{i=1}^{n} \opnorm{dT_{Q^*}^{Q_i} - dT_{\hQ_{\rho}(\bar{X})}^{Q_i}} \nonumber\\
    &\overset{(ii)}{\leq} \frac{1}{n} \sum_{i=1}^{n} \opnorm{d^2 T^{Q_i}_{Q_{i,n}'} \cdot (\hQ_{\rho}(\bar{X}) - Q^*)} \nonumber\\
    &\overset{(iii)}{\leq}  \frac{1}{n} \sum_{i=1}^{n} \opnorm{d^2 T^{Q_i}_{Q_{i,n}'}} \cdot \Fnorm{\hQ_{\rho}(\bar{X}) - Q^*} \label{eqn: prf_null_lem_til1_dT}
\end{align}
Here (i) follows from triangle inequality, (ii) from the mean value theorem \citep[Theorem 5.3]{dudley} for some $Q'_{i,n}$ that lies on the segment between $Q^*$ and $\hQ_{\rho}(\bar{X})$, and (iii) arises from Lemma \ref{lem: f_cal}. Therefore, with a slight modification of the proof of (\ref{eqn: conc_lem_E_Psi}) in Lemma \ref{lem: conc_tld_exp}, one can obtain with probability at least $1-O(n^{-100})$,
\begin{equation}
    \frac{1}{n} \sum_{i=1}^{n} \opnorm{d^2 T^{Q_i}_{Q_{i,n}'}} \leq \polylog{n}
    \label{eqn: prf_null_lem_til1_d2T_sum}
\end{equation}
For brevity, the proof is omitted.

Finally, combining (\ref{eqn: prf_null_lem_til1_a3}), (\ref{eqn: prf_null_lem_til1_dT}) and (\ref{eqn: prf_null_lem_til1_d2T_sum}) gives
\begin{align*}
    \sup_{x \in B_{\mu}(L)} \Fnorm{\alpha_3(x)} 
    &\lesssim \sbr{\frac{1}{n} \sum_{i=1}^{n} \opnorm{d^2 T^{Q_i}_{Q_{i,n}'}}} \cdot \sbr{\sup_{x \in B_{\mu}(L)} \Fnorm{ \hQ_{\rho}(x) - Q^*(x) } }^2\\
    &\leq \frac{\polylog{n}}{n}
\end{align*}
which finishes the proof for (\ref{eqn: prf_null_lem_til_3}).

\subsection{Proof of Claim \ref{claim: null_ind}}
\label{subsec: prf_null_claim_ind}
The independence follows directly from Assumption \ref{assumption: cond_ind} and the null hypothesis (\ref{eqn: testing_null}).

\paragraph{Proof of (\ref{eqn: claim_ind_otimes}):} by independence, we have
\begin{align*}
    \EE (X-\mu) \otimes \rbr{T_{Q^*}^Q - I_d} &= \sbr{\EE(X-\mu)} \otimes \sbr{\EE\rbr{T_{Q^*}^Q - I_d}}\\
    &= 0
\end{align*}

\paragraph{Proof of (\ref{eqn: claim_ind_T}):} The optimality condition of $Q^*(x)$ gives that
\begin{align*}
    \EE w(x,X)\rbr{T_{Q^*(x)}^Q - I_d} = 0
\end{align*}
By independence, one then has
\begin{align*}
    \EE \rbr{T_{Q^*(x)}^Q - I_d} &\overset{(i)}{=} \sbr{\EE w(x,X)} \cdot \sbr{\EE w(x,X)\rbr{T_{Q^*(x)}^Q - I_d}}\\
    &= \EE w(x,X)\rbr{T_{Q^*(x)}^Q - I_d} \\
    &= 0
\end{align*}
Here (i) follows from the fact that $\EE w(x,X)\equiv 1$.

%% file: prf_size.tex
Note that under the null (\ref{eqn: testing_null}), it holds that $Q^*(\bar{X})=Q^*$. Then Theorem \ref{thm: unfm_conv} implies that with probability at least $1-O(n^{-100})$,
\begin{align*}
    \Fnorm{\hQ_{\rho}(\bar{X}) - Q^*} \leq \frac{\polylog{n}}{\sqrt{n}}
\end{align*}
Therefore, one has with probability at least $1-O(n^{-100})$,
\begin{align}
    \hQ_{\rho}(\bar{X}) \in \cS_d((2c_1)^{-1},2c_1)
    \label{eqn: prf_size_set}
\end{align}
for $n$ large enough. Here we recall that $c_1 \geq 1$ is a constant defined in Assumption \ref{assumption: bdd_Q}.

One can apply the triangle inequality to see that
\begin{align}
    &\Fnorm{\frac{1}{n}\sum_{i=1}^{n} \rbr{T^{Q_i}_{\hQ_{\rho}(\bar{X})} - I_d} \otimes \rbr{T^{Q_i}_{\hQ_{n}(\bar{X})} - I_d} - \EE \rbr{T^{Q}_{Q^*} - I_d} \otimes \rbr{T^{Q}_{Q^*} - I_d}} \nonumber\\
    \leq & \Fnorm{\frac{1}{n}\sum_{i=1}^{n} \rbr{T^{Q_i}_{\hQ_{\rho}(\bar{X})} - I_d} \otimes \rbr{T^{Q_i}_{\hQ_{n}(\bar{X})} - I_d} - \EE \rbr{T^{Q}_{\hQ_{\rho}(\bar{X})} - I_d} \otimes \rbr{T^{Q}_{\hQ_{\rho}(\bar{X})} - I_d}} \nonumber\\
    &+ \Fnorm{\EE \rbr{T^{Q}_{\hQ_{\rho}(\bar{X})} - I_d} \otimes \rbr{T^{Q}_{\hQ_{\rho}(\bar{X})} - I_d} - \EE \rbr{T^{Q}_{Q^*} - I_d} \otimes \rbr{T^{Q}_{Q^*} - I_d}} \nonumber\\
    \overset{(i)}{\leq} & \underbrace{\sup_{S \in \cS_d((2c_1)^{-1},2c_1)} \Fnorm{\frac{1}{n}\sum_{i=1}^{n} \rbr{T^{Q_i}_{S} - I_d} \otimes \rbr{T^{Q_i}_{S} - I_d} - \EE \rbr{T^{Q}_{S} - I_d} \otimes \rbr{T^{Q}_{S} - I_d}}}_{\zeta_1} \nonumber\\
    &+ \underbrace{\sup_{S: \Fnorm{S-Q^*}\leq \polylog{n}/\sqrt{n}}  \Fnorm{\EE \rbr{T^{Q}_{S} - I_d} \otimes \rbr{T^{Q}_{S} - I_d} - \EE \rbr{T^{Q}_{Q^*} - I_d} \otimes \rbr{T^{Q}_{Q^*} - I_d}} }_{\zeta_2}
    \label{eqn: prf_size_approx_F_decomp}
\end{align} 
Here (i) follows from (\ref{eqn: prf_size_set}).

Upper bounds for $\zeta_1,\zeta_2$ in (\ref{eqn: prf_size_approx_F_decomp}) are summarized in the lemma below whose proof is deferred to Appendix \ref{subsec: prf_size_lem_lem_approx}. Note that Lemma \ref{lem: prf_size_lem_approx} do not assume the null hypothesis so that it can be reused later for the proof of the power (Theorem \ref{thm: power}).
\begin{lemma}
    \label{lem: prf_size_lem_approx}
    Suppose Assumption \ref{assumption: X}-\ref{assumption: cond_ind} hold. Then with probability at least $1-O(n^{-100})$,
    \begin{align}
        \zeta_1 &\leq \frac{\polylog{n}}{\sqrt{n}} \label{eqn: prf_size_lem_z1}\\
        \zeta_2 & \leq \frac{\polylog{n}}{\sqrt{n}} \label{eqn: prf_size_lem_z2}
    \end{align}
\end{lemma}
Lemma \ref{lem: prf_size_lem_approx} combined with (\ref{eqn: prf_size_approx_F_decomp}) then implies
\begin{align}
    \Fnorm{\frac{1}{n}\sum_{i=1}^{n} \rbr{T^{Q_i}_{\hQ_{\rho}(\bar{X})} - I_d} \otimes \rbr{T^{Q_i}_{\hQ_{n}(\bar{X})} - I_d} - \EE \rbr{T^{Q}_{Q^*} - I_d} \otimes \rbr{T^{Q}_{Q^*} - I_d}} \leq \frac{\polylog{n}}{\sqrt{n}}
    \label{eqn: prf_size_approx_F}
\end{align}
As a result, if $\rbr{\lambda_i}_{i \in [d^2]}$ are sorted in order, then one has $\hlambda_i \to \lambda_i$ uniformly for $i \in [d^2]$ in probability which further implies that 
\begin{equation*}
    \sum_{i=1}^{d^2} \hlambda_{i}w_i \overset{p}{\rightarrow} \sum_{i=1}^{d^2} \lambda_{i}w_i
\end{equation*}
Then the continuous mapping theorem implies that $\hq_{1-\alpha} \to q_{1-\alpha}$ in probability. Then one can obtain
\begin{equation*}
    \begin{aligned}
        \PP\rbr{ \hat{\cT}_{\rho} > \hq_{1-\alpha}} \leq \PP \rbr{\hat{\cT}_{\rho} > q_{1-\alpha}-\epsilon} + \PP\rbr{\abs{\hq_{1-\alpha} - q_{1-\alpha}}>\epsilon}
    \end{aligned}
\end{equation*}
Taking the limit as $n \to \infty$, followed by letting $\epsilon \to 0$ to get that
\begin{equation}
    \limsup_{n \to \infty}\PP\rbr{ \hat{\cT}_{\rho} > \hq_{1-\alpha}} \leq \alpha
    \label{eqn: prf_size_limsup}
\end{equation}
A similar lower bound shows
\begin{equation}
    \liminf_{n \to \infty}\PP\rbr{ \hat{\cT}_{\rho} > \hq_{\alpha}} \geq \alpha
    \label{eqn: prf_size_liminf}
\end{equation}
Finally, combining (\ref{eqn: prf_size_limsup}) and (\ref{eqn: prf_size_liminf}) completes the proof.

\subsection{Proof of Lemma \ref{lem: prf_size_lem_approx}}
\label{subsec: prf_size_lem_lem_approx}

\paragraph{Proof of (\ref{eqn: prf_size_lem_z1}):} the proof is similar to Lemma \ref{lem: conc_tld_exp}, and is hence omitted for brevity.

\paragraph{Proof of (\ref{eqn: prf_size_lem_z2}):} First, recall differential properties (Lemma \ref{lem: diff}) that
\begin{equation}
    \begin{aligned}
        \Onorm{T^Q_S} &\leq \eigmax{Q}\cdot \eigmin{Q}^{-1/2} \cdot \eigmin{S}^{-1/2}\\
        \opnorm{dT^Q_S} &\leq \frac{1}{2} \eigmax{Q}^{1/2} \cdot \eigmin{S}^{-2}\cdot \eigmax{S}^{1/2}
    \end{aligned}
    \label{eqn: prf_size_lem_diff}
\end{equation}
Also, denote $\phi(Q,S):= \rbr{T^{Q}_{S} - I_d} \otimes \rbr{T^{Q}_{S} - I_d}$. 

With these results in place, one can obtain
\begin{equation*}
    \begin{aligned}
        \Fnorm{ \phi(Q,S) -  \phi(Q,Q^*)} 
        &\leq \Fnorm{ \rbr{T^{Q}_{S} - I_d} \otimes \rbr{T^{Q}_{S} - T^Q_{Q^*}}} + \Fnorm{ \rbr{T^Q_{S} - T^Q_{Q^*}} \otimes \rbr{T^Q_{Q^* }- I_d}}\\
        &\leq \Fnorm{T^{Q}_{S} - I_d} \cdot \Fnorm{T^{Q}_{S} - T^Q_{Q^*}} + \Fnorm{T^Q_{S} - T^Q_{Q^*}} \cdot \Fnorm{T^Q_{Q^*}- I_d}\\
        &\leq \Fnorm{T^{Q}_{S} - I_d} \cdot \opnorm{dT^{Q}_{S'}}\cdot \Fnorm{S-Q^*} + \Fnorm{T^Q_{Q^* }- I_d}\cdot \opnorm{dT^{Q}_{S'}}\cdot \Fnorm{S-Q^*}
    \end{aligned}
\end{equation*}
where $S'$ lies between $Q^*$ and $S$. Note that Assumption \ref{assumption: bdd_Q} and the condition $\Fnorm{S-Q^*}\leq \polylog{n}/\sqrt{n}$ implies that
\begin{align*}
    S,Q^* \in \cS_d((2c_1)^{-1}, 2c_1)
\end{align*}
Then for any $S \in \cS_d((2c_1)^{-1}, 2c_1)$, one has
\begin{equation*}
    \begin{aligned}
        \EE \Fnorm{T^{Q}_{S} - I_d} \cdot \opnorm{dT^{Q}_{S'}} &\overset{(i)}{\lesssim} \EE \rbr{ \norm{X-\mu}^{3C_1/2} + 1} \cdot \norm{X-\mu}^{C_1/2}\\
        &\overset{(ii)}{\lesssim} 1
    \end{aligned}
\end{equation*}
Here $C_1$ is defined in Assumption \ref{assumption: bdd_Q}, (i) follows from (\ref{eqn: prf_size_lem_diff}) and (ii) is a result of the sub-Gaussianity of $X$. Similarly, one has
\begin{equation*}
    \EE \Fnorm{T^Q_{Q^* }- I_d}\cdot \opnorm{dT^{Q}_{S'}} \lesssim 1
\end{equation*}
Combining results above, one can obtain
\begin{align*}
    \zeta_2 &\leq  \sup_{S: \Fnorm{S-Q^*}\leq \polylog{n}/\sqrt{n}} \EE_Q \Fnorm{ \phi(Q,S) -  \phi(Q,Q^*)}\\
    &\lesssim \sup_{S: \Fnorm{S-Q^*}\leq \polylog{n}/\sqrt{n}} 1 \cdot \Fnorm{S-Q^*}\\
    &\leq \frac{\polylog{n}}{\sqrt{n}}
\end{align*}
The proof is then complete.

%% file: proof_power.tex
As argued at the beginning of Appendix \ref{subsec: conc_tld_exp}, one can assume without loss of generality that $\norm{X-\mu}\leq L$ almost surely with $L=C_L \sqrt{\log n}$ for some constant $C_L$ large enough as in (\ref{eqn: prf_conv_E0}). Recall the notation that $M_L:=\gamma_{\Lambda}(L)$.

\paragraph{Power under Frobenius norm:} First we demonstrate concentration of various quantities of interest and derive their consequences in Lemma \ref{lem: prf_power} below. The proof is in Appendix \ref{subsubsec: prf_power_lem}.
\begin{lemma}
    \label{lem: prf_power}
     Instate the notations and assumptions in Theorem \ref{thm: power}. Assume in addition that 
    $\norm{X-\mu}\leq L$ almost surely.
    Then there exists an event $E_n$ that satisfies $\PP(E_n) \geq 1-O(n^{-100})$ for any $\PP \in \mathfrak{P}$, under which the following inequalities
    \begin{align}
        \norm{X_i-\mu} \leq L, &\quad  Q_i \in \cS_d (M_L^{-1},M_L) \qquad \text{for} \quad i \in [n] \label{eqn: prf_pow_XQ}\\
        \sup_{x \in B_\mu(L)}\Fnorm{\hQ_{\rho}(x) - Q^*(x)} &\leq \frac{\polylog{n}}{\sqrt{n}} \label{eqn: prf_pow_hQ_Q}\\
        \cbr{\hQ_{\rho}(x):x \in B_{\mu}(L)} &\subset \cS_d \rbr{ \rbr{2^{1/6}M_L}^{-1}, 2^{1/6}M_L} \label{eqn: prf_pow_hQ_in}\\
        \sum_{i=1}^{n} \Fnorm{Q^*(X_i) - Q^*(\mu)}^2 &\geq \frac{n a_n^2}{2} \label{eqn: prf_pow_sum1}\\
        \Fnorm{\hQ_{\rho}(\bar{X})-Q^*(\mu)} &\leq \frac{\polylog{n}}{\sqrt{n}} \label{eqn: prf_pow_bary}\\
        \abs{\hlambda_i} &\leq 2 \lambda_1, \qquad i \in \sbr{d^2} \label{eqn: prf_pow_eig}
    \end{align}
    hold for $n$ large enough.
\end{lemma}

Next, note that (\ref{eqn: lem_diff_dT_eigen}) in Lemma \ref{lem: diff}  implies that for any $Q,S \in \cS_d \rbr{ \rbr{2^{1/6}M_L}^{-1}, 2^{1/6}M_L}$ (here $2^{1/6}$ is chosen only for technical computation), one has
\begin{equation*}
    \frac{1}{2\sqrt{2} M_L^3}\leq \eigmin{-dT^Q_S} \leq \eigmax{-dT^Q_S} \leq \frac{\sqrt{2}}{2}M_L^3
\end{equation*}
which then implies that under $E_n$, the following holds.
\begin{equation*}
    \begin{aligned}
        \eigmin{\hat{H}} \geq \frac{1}{n}\sum_{i=1}^{n} \eigmin{-dT^{Q_i}_{\hQ_{\rho}(\bar{X})}} \geq \frac{1}{2\sqrt{2} M_L^3}
    \end{aligned}
\end{equation*}
Therefore, under $E_n$, one has 
\begin{equation}
    \begin{aligned}
        \hat{\cT}_{\rho} \geq \frac{1}{8 M_L^6} \sum_{i=1}^{n} \Fnorm{\hQ_{\rho}(X_i) - \hQ_{\rho}(\bar{X})}^2
    \end{aligned}
    \label{eqn: prf_pow_test2Frob}
\end{equation}

Then from the following decomposition 
\begin{equation*}
    \hQ_{\rho}(X_i) - \hQ_{\rho}(\bar{X}) = Q^*(X_i) - Q^*(\mu) + \underbrace{\hQ_{\rho}(X_i) - Q^*(X_i) + Q^*(\mu) - \hQ_{\rho}(\bar{X})}_{\Delta_i}
\end{equation*}
one can obtain
\begin{align}
    &\sum_{i=1}^{n} \Fnorm{\hQ_{\rho}(X_i) - \hQ_{\rho}(\bar{X})}^2 \nonumber\\
    \overset{(i)}{=}& \sum_{i=1}^{n} \Fnorm{Q^*(X_i) - Q^*(\mu)}^2 + 2\sum_{i=1}^{n} \inner{Q^*(X_i) - Q^*(\mu)}{\Delta_i}+\sum_{i=1}^{n} \Fnorm{\Delta_i}^2 \nonumber\\
    \overset{(ii)}{\geq}& \sum_{i=1}^{n} \Fnorm{Q^*(X_i) - Q^*(\mu)}^2 - 2\rbr{\sum_{i=1}^{n}  \Fnorm{Q^*(X_i) - Q^*(\mu)}^2}^{1/2} \cdot \rbr{\sum_{i=1}^{n}  \Fnorm{\Delta_i}^2}^{1/2} +\sum_{i=1}^{n} \Fnorm{\Delta_i}^2 \nonumber\\
    \geq & \rbr{\sum_{i=1}^{n} \Fnorm{Q^*(X_i) - Q^*(\mu)}^2} \cdot \rbr{1- 2\rbr{\frac{\sum_{i=1}^{n}  \Fnorm{\Delta_i}^2}{\sum_{i=1}^{n} \Fnorm{Q^*(X_i) - Q^*(\mu)}^2}}^{1/2}}
    \label{eqn: prf_pow_Fro_lower1}
\end{align}
Here (i) follows by developing the square, (ii) is a consequence of the Cauchy-Schwarz inequality.

Lemma \ref{lem: prf_power} implies that under $E_n$, one has
\begin{align}
    \Fnorm{\Delta_i} &\leq \Fnorm{\hQ_{\rho}(X_i) - Q^*(X_i)} + \Fnorm{Q^*(\mu) - \hQ_{\rho}(\bar{X})} \nonumber\\
    &\lesssim \frac{\polylog{n}}{\sqrt{n}}
    \label{eqn: prf_pow_Delta}
\end{align}
Therefore, (\ref{eqn: prf_pow_Fro_lower1}) and (\ref{eqn: prf_pow_Delta}) together imply the following inequality for $n$ large enough.
\begin{align}
    \sum_{i=1}^{n} \Fnorm{\hQ_{\rho}(X_i) - \hQ_{\rho}(\bar{X})}^2 \overset{(I)}{\geq} \frac{na_n^2}{2} \cdot \frac{3}{4}
    \label{eqn: prf_pow_Fro_lower}
\end{align}
Here (I) arises due to Lemma \ref{lem: prf_power}, (\ref{eqn: prf_pow_Delta}) as well as the fact that $\polylog{n} = o \rbr{na_n^2}$.

Combining (\ref{eqn: prf_pow_test2Frob}) and (\ref{eqn: prf_pow_Fro_lower}) then gives that under $E_n$, one has
\begin{align}
    \hat{\cT}_{\rho} &\geq \frac{3n a_n^2}{64 M_L^6} \nonumber\\
    &\geq \sqrt{na_n^2} \label{eqn: prf_pow_hTau_lower}
\end{align}
for $n$ large enough.

Denote $\tq_{1-\alpha}$ the $1-\alpha$ quantile of $\sum_{i=1}^{d^2} 2\lambda_1w_i$, which is a fixed constant. Then as a consequence of Lemma \ref{lem: prf_power},  $\sum_{i=1}^{d^2} \hlambda_i w_i$ is stochastically dominated by $\sum_{i=1}^{d^2} 2\lambda_1w_i$  under $E_n$, which then implies that
\begin{equation*}
    \hq_{1-\alpha} < \tq_{1-\alpha}, \qquad \text{under } E_n
\end{equation*}
Therefore, for any $\PP \in \mathfrak{P}$, one can obtain
\begin{equation*}
    \begin{aligned}
        \PP \rbr{ \hat{\cT}_{\rho} > \hq_{1-\alpha}} &\geq \PP \rbr{\cbr{\hat{\cT}_{\rho} > \tq_{1-\alpha}} \cap \cbr{ \tq_{1-\alpha}>\hq_{1-\alpha}}} \\
        &\geq \PP \rbr{\cbr{\hat{\cT}_{\rho} > \tq_{1-\alpha}} \cap E_n} \\
        &\overset{(i)}{\geq} \PP \rbr{\cbr{\sqrt{na_n^2} > \tq_{1-\alpha}}\cap E_n }    \end{aligned}
\end{equation*}
for $n$ large enough. Here (i) follows from (\ref{eqn: prf_pow_hTau_lower}).

Finally, one can take $n \to \infty$ to see that
\begin{equation*}
    \begin{aligned}
        \liminf_{n \to \infty} \inf_{\PP \in \mathfrak{P}}\PP \rbr{ \hat{\cT}_{\rho} > \hq_{1-\alpha}}
        & \geq \liminf_{n \to \infty} \inf_{\PP \in \mathfrak{P}} \PP \rbr{\cbr{\sqrt{na_n^2} > \tq_{1-\alpha}}\cap E_n }\\
        &\geq 1
    \end{aligned}
\end{equation*}
which implies 
\begin{equation*}
    \lim_{n \to \infty}\inf_{\PP \in \mathfrak{P}} \PP \rbr{ \hat{\cT}_{\rho} > \hq_{1-\alpha}} = 1
\end{equation*}
The proof is then complete.

\paragraph{Power under Wasserstein distance:} by Lemma \ref{lem: diff} and the boundedness assumption, one has
\begin{align*}
    \EE \Fnorm{Q^*(X)-Q^*(\mu)}^2 \geq \frac{1}{\polylog{n}} \EE W^2\rbr{Q^*(X),Q^*(\mu)}
\end{align*}
Therefore, the 1st part of the proof (Frobenius norm) can be applied.

\subsection{Proof of Lemma \ref{lem: prf_power}}
\label{subsubsec: prf_power_lem}

(\ref{eqn: prf_pow_XQ})-(\ref{eqn: prf_pow_hQ_in}) are shown in the proof of Theorem \ref{thm: unfm_conv} and (\ref{eqn: prf_pow_bary}) is due to Lemma \ref{lem: barycenter}.

\paragraph{Proof of (\ref{eqn: prf_pow_sum1}):}

With these in place, we have almost surely that
\begin{align}
    \Fnorm{Q^*(X) - Q^*(\mu)}^2 &\leq d \Onorm{Q^*(X) - Q^*(\mu)}^2 \nonumber\\
    &\leq d M_L^2 =:M_F \asymp \polylog{n}\label{eqn: prf_power_lem_trunc}
\end{align}
Define random variable $Y_i$ as follows.
\begin{align*}
    Y_i&:=  \Fnorm{Q^*(X_i) - Q^*(\mu)}^2 - \EE \Fnorm{Q^*(X) - Q^*(\mu)}^2
\end{align*}
With (\ref{eqn: prf_power_lem_trunc}) in place, one then has $Y_i \in [-M_F,M_F]$, $\EE Y_i = 0$ and
\begin{align*}
    \Var{Y_i}&= \EE Y_i^2 \\
    &\leq \EE \abs{Y_i} M_F \\
    &\leq 2\EE \Fnorm{Q^*(X_i) - Q^*(\mu)}^2 \cdot M_F \\
    &= 2a_n^2 \cdot M_F
\end{align*}
To prove (\ref{eqn: prf_pow_sum1}), it suffices to show that ,
\begin{align}
    \abs{\sum_{i}^{n} Y_i} \leq \frac{n a_n^2}{2} \quad \text{with probability at least } 1-O(n^{-100})
    \label{eqn: prf_pow_lem_suff}
\end{align}
To this end, note that by Bernstein's inequality \citep[Theorem 2.8.4]{vershynin}, one has with probability at least $1-O(n^{-100})$,
\begin{align}
    \abs{\sum_{i}^{n} Y_i} &\lesssim M_F \log n + \sqrt{n \Var{Y_i}} \cdot \sqrt{\log n} \nonumber\\
    &\lesssim M_F \log n + \sqrt{n a_n^2 M_F} \cdot \sqrt{\log n} \nonumber\\
    &\overset{(i)}{\leq} \polylog{n} \cdot \rbr{1+\sqrt{n a_n^2}} \label{eqn: prf_pow_lem_berns}
\end{align}
Then (\ref{eqn: prf_pow_lem_berns}) implies (\ref{eqn: prf_pow_lem_suff}) for $n$ large enough due to the assumption that $na_n^2 \gtrsim n^{2\alpha_2}$ for some constant $\alpha_2>0$. The proof is then complete.

\paragraph{Proof of (\ref{eqn: prf_pow_eig}):} Note that the proof of Lemma \ref{lem: prf_size_lem_approx} does not assume the null hypothesis. Therefore, with (\ref{eqn: prf_pow_bary}) in place, one can exactly follow (\ref{eqn: prf_size_set})-(\ref{eqn: prf_size_approx_F}) as in the proof of Lemma \ref{lem: prf_size_lem_approx} to get 
\begin{align*}
    \Fnorm{\frac{1}{n}\sum_{i=1}^{n} \rbr{T^{Q_i}_{\hQ_{\rho}(\bar{X})} - I_d} \otimes \rbr{T^{Q_i}_{\hQ_{\rho}(\bar{X})} - I_d} - \EE \rbr{T^{Q}_{Q^*} - I_d} \otimes \rbr{T^{Q}_{Q^*} - I_d}} \leq \frac{\polylog{n}}{\sqrt{n}}
\end{align*}
which implies (\ref{eqn: prf_pow_eig}).

%% file: more_sim.tex
We conduct simulations to compare different initialization methods and step sizes, as described in Section \ref{subsec: algm}.

\subsection{Initialization}
To assess the performance of different initialization methods, we conduct simulations based on Example \ref{example_nc} with parameters $n = 200$, $p = 5$ and $d = 6$, setting $\eta = 1$, $T=30$ and $\mathrm{eps} = 10^{-6}$ in Algorithm \ref{alg: gd}. We consider three different initialization methods:
\begin{enumerate}
    \item Identity Initialization : $S_0 = I_d$.
    \item Random Initialization: $S_0 = U \exp(M) U^\top$ where $U$ and $M$ are independent; $U \in \cO_d$ is a random orthogonal matrix following the Haar measure, and $M$ is a random diagonal matrix with i.i.d. diagonal entries $M_{kk} \sim \cN(0,1/\sqrt{d})$, $k \in [d]$.
    \item Mean Initialization: $S_0 = \sum_{i=1}^n Q_i$.
\end{enumerate}
Table \ref{tab: init} reports the number of steps before termination and the relative error with respect to the identity initialization, averaged over 100 simulations. In each simulation, after generating $n = 200$ pairs of $(X_i, Q_i)$, we randomly generate $\tX \sim \mathrm{Uniform}[-1,1]^p$ and compute $\hQ_{\rho}(\tX)$ using Algorithm \ref{alg: gd} with the three initialization methods described above.
\begin{table}[!htb] 
    \centering 
    \caption{Comparison of different initialization methods.} 
    \label{tab: init} 
    \begin{tabular}{c c|c c}
        \hline
        $\delta$ & Initialization method & Steps & Relative error ($\times \mathrm{eps}$) \\
        \hline
        \multirow{3}{*}{$0$} & $I_d$ & 3 & 0 \\ 
        & Random & 3.4 & 5.3e-4 \\ 
        & Mean & 3 & 4.8e-6 \\
        \hline
        \multirow{3}{*}{$0.2$} & $I_d$ & 3 & 0 \\ 
        & Random & 3.4 & 5.9e-4 \\ 
        & Mean & 3 & 4.8e-6 \\
        \hline
        \multirow{3}{*}{$0.4$} & $I_d$ & 3 & 0 \\ 
        & Random & 3.3 & 6.0e-4 \\ 
        & Mean & 3 & 5.9e-6 \\
        \hline
    \end{tabular}
\end{table}
It can be observed that all three methods converge to the same point within nearly the same number of steps, as the relative error is significantly smaller than the threshold $\eps$.

\subsection{Step size}
To examine the performance of different step sizes, we conduct simulations based on Example~\ref{example_nc} with parameters $n = 200$, $p = 5$ and $d = 6$, setting $S_0 = I_d$, $T=30$ and $\mathrm{eps} = 10^{-6}$ in Algorithm \ref{alg: gd}.
Table \ref{tab: step_size} reports the number of steps before termination and the relative error with respect to step size $1$, averaged over 100 simulations. In each simulation, after generating $n = 200$ pairs of $(X_i, Q_i)$, we randomly generate $\tX \sim \mathrm{Uniform}[-1,1]^p$ and compute $\hQ_{\rho}(\tX)$ using Algorithm \ref{alg: gd} with step size $\eta \in \cbr{1,0.9,0.8,0.7,0.6}$

\begin{table}[!htb] 
    \centering 
    \caption{Comparison of different step sizes.} 
    \label{tab: step_size} 
    \begin{tabular}{c c|c c}
        \hline
        $\delta$ & $\eta$ & Steps & Relative error ($\times \mathrm{eps}$) \\
        \hline
        \multirow{5}{*}{$0$}
        & 1 & 3 & 0 \\ 
        & 0.9 & 8 & 0.013 \\ 
        & 0.8 & 10 & 0.13 \\
        & 0.7 & 13 & 0.20\\
        & 0.6 & 17 & 0.22\\
        \hline
        \multirow{5}{*}{$0.2$}
        & 1 & 3 & 0 \\ 
        & 0.9 & 8 & 0.013 \\ 
        & 0.8 & 10 & 0.13 \\
        & 0.7 & 13 & 0.20\\
        & 0.6 & 17 & 0.22\\
        \hline
        \multirow{5}{*}{$0.4$}
        & 1 & 3 & 0 \\ 
        & 0.9 & 7.94 & 0.017 \\ 
        & 0.8 & 10 & 0.13 \\
        & 0.7 & 13 & 0.20\\
        & 0.6 & 15.97 & 0.22\\
        \hline
    \end{tabular}
\end{table}
It can be observed that all three methods converge to the same point, as the relative error is smaller than the threshold $\eps$. Additionally, setting $\eta = 1$ results in the fastest convergence.